\documentclass[a4paper,10pt]{article} 
\usepackage[centertags]{amsmath}
\usepackage{amsmath}
\usepackage[dutch,british]{babel}
\usepackage{amsfonts}
\usepackage{amssymb}
\usepackage{amsthm}
\usepackage{newlfont}
\usepackage{color}
\usepackage{subfig}
 \usepackage{bbm}

\usepackage{stmaryrd}
\usepackage{mathrsfs}
\usepackage{pstricks,pst-node}
\usepackage{palatino}  
\usepackage{fancyhdr}
\usepackage{graphicx}
\usepackage{fancybox}
\usepackage{multibox}
\usepackage{geometry}
 \usepackage{psfrag}
 
 \usepackage{multirow}

\theoremstyle{plain}
  \newtheorem{theorem}{Theorem}[section]
  
  \newtheorem{proposition}[theorem]{Proposition}
  \newtheorem{lemma}[theorem]{Lemma}
  \newtheorem{remark}[theorem]{Remark}
\theoremstyle{definition}
  \newtheorem{definition}{Definition}[section]
  \newtheorem{assumption}[theorem]{Assumption}

\theoremstyle{remark}

\numberwithin{equation}{section}

\DeclareMathOperator{\Tr}{Tr}

 \DeclareMathOperator{\supp}{Supp}
\renewcommand{\Re}{\mathrm{Re}\, }
\renewcommand{\Im}{\mathrm{Im}\,}

\newcommand\otimesal{\mathop{\hbox{\raise 1.6 ex
  \hbox{$\scriptscriptstyle\mathrm{al}$}
\kern -0.92 em \hbox{$\otimes$}}}}
\newcommand\oplusal{\mathop{\hbox{\raise 1.6 ex
  \hbox{$\scriptscriptstyle\mathrm{al}$}
\kern -0.92 em \hbox{$\oplus$}}}}
\newcommand\Gammal{\hbox{\raise 1.7 ex
\hbox{$\scriptscriptstyle\mathrm{al}$}\kern -0.50 em $\Gamma$}}
\renewcommand\i{\mathrm{i}}

\let\al=\alpha \let\be=\beta \let\de=\delta \let\ep=\epsilon

\let\ve=\varepsilon  \let\ga=\gamma 
\let\ka=\kappa \let\la=\lambda \let\om=\omega 
\let\si=\sigma

 \let\Ga=\Gamma \let\La=\Lambda \let\Om=\Omega
  \let\Si=\Sigma

\newcommand{\caA}{{\mathcal A}}
\newcommand{\caB}{{\mathcal B}}
\newcommand{\caC}{{\mathcal C}}
\newcommand{\caD}{{\mathcal D}}
\newcommand{\caE}{{\mathcal E}}
\newcommand{\caF}{{\mathcal F}}
\newcommand{\caG}{{\mathcal G}}
\newcommand{\caH}{{\mathcal H}}
\newcommand{\caI}{{\mathcal I}}
\newcommand{\caJ}{{\mathcal J}}

\newcommand{\caL}{{\mathcal L}}
\newcommand{\caM}{{\mathcal M}}
\newcommand{\caN}{{\mathcal N}}

\newcommand{\caP}{{\mathcal P}}

\newcommand{\caR}{{\mathcal R}}

\newcommand{\caT}{{\mathcal T}}
\newcommand{\caU}{{\mathcal U}}
\newcommand{\caV}{{\mathcal V}}
\newcommand{\caW}{{\mathcal W}}

\newcommand{\caZ}{{\mathcal Z}}

\newcommand{\scrB}{{\mathscr B}}

\newcommand{\scrE}{{\mathscr E}}

\newcommand{\scrG}{{\mathscr G}}
\newcommand{\scrH}{{\mathscr H}}

\newcommand{\scrS}{{\mathscr S}}

\newcommand{\bbC}{{\mathbb C}}

\newcommand{\bbN}{{\mathbb N}}

\newcommand{\bbR}{{\mathbb R}}
\newcommand{\bbS}{{\mathbb S}}
\newcommand{\bbT}{{\mathbb T}}

\newcommand{\bbZ}{{\mathbb Z}}

\newcommand{\opunit}{\text{1}\kern-0.22em\text{l}}

\newcommand{\frh}{{\mathfrak h}}

\newcommand{\frl}{{\mathfrak l}}

\newcommand{\frt}{{\mathfrak t}}

\newcommand{\frC}{{\mathfrak C}}
\newcommand{\frD}{{\mathfrak D}}

\newcommand{\frL}{{\mathfrak L}}

\newcommand{\e}{{\mathrm e}}

\renewcommand{\d}{{\mathrm d}}
\newcommand{\sys}{{\mathrm S}}
\newcommand{\res}{{\mathrm R}}

\renewcommand{\sp}{\mathrm{sp}}
\newcommand{\Ran}{\mathrm{Ran}}

\newcommand{\Dom}{\mathrm{Dom}}

\newcommand{\beq}{ \begin{equation} }
\newcommand{\eeq}{ \end{equation} }
\newcommand{\bet}{ \begin{theorem} }
\newcommand{\eet}{ \end{theorem} }

\newcommand{\Symm}{\mathrm{Sym}}

\newcommand{\baq}{\begin{eqnarray}}
\newcommand{\eaq}{\end{eqnarray}}

\newcommand{\dsi}{\si}

\newcommand{\dsil}{   \dsi_{l }  }

\renewcommand{\supp}{\mathrm{Supp}}

 \newcounter{smallarabics}
\newenvironment{arabicenumerate}
{\begin{list}{{\normalfont\textrm{\arabic{smallarabics})}}}
  {\usecounter{smallarabics}\setlength{\itemindent}{0cm}
  \setlength{\leftmargin}{5ex}\setlength{\labelwidth}{4ex}
  \setlength{\topsep}{0.75\parsep}\setlength{\partopsep}{0ex}
   \setlength{\itemsep}{0ex}}}
{\end{list}}

\newcounter{smallroman}

\newcommand{\ben}{\begin{arabicenumerate}}
\newcommand{\een}{\end{arabicenumerate}}

\newcommand{\sfock}{\Ga_{\mathrm{s}}}

\newcommand{\norm}{ \|}
\newcommand{\str}{ |}

\newcommand{\interspace}{ \scrS }

\newcommand{\lat}{ \bbZ^d }
\newcommand{\tor}{ {\bbT^d}  }
\newcommand{\initialres}{\rho_\res^\be}

\newcommand{\links}{L}
\newcommand{\rechts}{R}
\newcommand{\adjoint}{\mathrm{ad}}
\newcommand{\ad}{\adjoint}

\newcommand{\bosondispersion}{\omega}
\newcommand{\meleq}{\mathop{\leq}\limits_{m.e.}}
\newcommand{\cone}{c(\ga)}
\newcommand{\ctwo}{c(\ga,\la)}
\newcommand{\cthree}{c'(\ga,\la)}

\begin{document}
\begin{center}
\large{ \bf{Diffusion of a massive quantum particle  coupled to a quasi-free thermal medium. }} \\
\vspace{15pt} \normalsize

%

{\bf   W.  De Roeck\footnote{
email: {\tt
 w.deroeck@thphys.uni-heidelberg.de}}  }\\
\vspace{10pt} 
{\it   Institut f\"ur Theoretische Physik  \\ Universit\"at Heidelberg \\
D69120 Heidelberg,  Germany 
} \\

\vspace{20pt}

{\bf   J. Fr\"ohlich   }\\
\vspace{10pt} 
{\it   Institute for Theoretical Physics \\
ETH Z\"urich \\
CH-8093 Z\"urich, Switzerland}
\vspace{20pt} 

\end{center}

\vspace{20pt} \footnotesize \noindent {\bf Abstract: }
We consider a heavy quantum particle with an internal degree of freedom moving on the $d$-dimensional lattice $\bbZ^d$ (e.g., a heavy atom  with finitely many internal states).   The particle is coupled to a thermal medium (bath) consisting of free relativistic bosons (photons or Goldstone modes) through an  interaction of  strength $\la$ linear in creation and annihilation operators. The mass of the quantum particle is assumed to be of order $\la^{-2}$,  and we assume that  the internal degree of freedom is  coupled ``effectively" to the thermal medium. 
We prove that the motion of the  quantum particle is diffusive in $d\geq 4$ and for $\la$ small enough.

\vspace{5pt} \footnotesize \noindent {\bf KEY WORDS:}  diffusion, weak
coupling limit, quantum Boltzmann equation, quantum field theory  \vspace{20pt}
\normalsize

\section{Introduction}\label{sec: intro}

\subsection{Diffusion }

Diffusion and Brownian motion are  central phenomena in the theory of transport processes and nonequilibrium statistical physics in general. 
One can think of the  diffusion of a tracer particle in interacting particle systems,  the diffusion of energy in coupled oscillator chains, and many other examples. 

From a heuristic point of view, diffusion is  rather well-understood in most of these examples. It can often be successfully described by  some Markovian approximation, e.g.\ the Boltzmann equation or Fokker-Planck equation, depending on the example  under study. In fact, this has been the strategy of Einstein in his ground breaking work of 1905, in which he modeled diffusion as a random walk. 

However, up to this date, there is no rigorous derivation of diffusion from classical Hamiltonian mechanics or unitary quantum mechanics, except for some special chaotic systems; see Section \ref{sec: results classical}. Such a derivation ought to allow us, for example, to prove that the motion of a tracer particle that interacts with its environment is diffusive at large times. In other words, one would like to prove a central limit theorem for the position of such a particle. 

In recent years, some promising steps towards this goal have been taken. We provide a brief review of previous results in Section \ref{sec: related results}.
In the present paper, we rigorously exhibit diffusion for a quantum particle weakly coupled to a thermal reservoir. However, our method is restricted to spatial dimension $d \geq 4$.

\subsection{Informal description of the model and  main results}

We consider a quantum particle hopping on the lattice $\bbZ^d$, and interacting with a reservoir of bosons (photons or phonons) at temperature $\be^{-1}>0$. In the present section, we describe  the system in a way that is appropriate at zero temperature, but is formal when $\be <\infty$.   The total Hilbert space, $\scrH$, of the coupled system is a tensor product of the particle space, $\scrH_\sys$, with a reservoir space, $\scrH_\res$. Thus 
\beq
\scrH:= \scrH_\sys \otimes  \scrH_{\res}.  
\eeq
The particle space $\scrH_\sys$ is given by $l^2(\bbZ^d) \otimes \scrS$, where the Hilbert space $\scrS$ describes the internal degrees of freedom of the particle, e.g., a (pseudo-)spin or dipole moment, and the particle Hamiltonian is given by the sum of the kinetic energy and the energy of the internal degrees of freedom
\beq
H_\sys := H_{\sys,\mathrm{kin}} \otimes 1+  1 \otimes H_{\sys,\mathrm{spin}} 
\eeq
The kinetic energy is chosen to be small in comparison with the interaction energy, and this is made manifest in its definition by a factor  $\la^2$, where $\la$ is the coupling strength between the particle and the reservoir (to be introduced below). Hence we set
\beq
H_{\sys,\mathrm{kin}} = \la^2 \varepsilon(P),
\eeq
where the function $\varepsilon$ is the dispersion law of the particle and $P$ is the lattice-momentum operator. The most natural choice is to take $\varepsilon(P)$ to be (minus) the discrete lattice Laplacian, $-\Delta$. The energy of states of the internal degree of freedom is to a large extent arbitrary
\beq
H_{\sys,\mathrm{spin}}:= Y, \qquad \textrm{for some Hermitian matrix} \,  Y,
\eeq
the main requirement being that  $Y$ not be equal to a multiple of the identity. 

 The reservoir is described by a free boson field;  creation and annihilation operators creating/annihilating bosons with momentum $q \in \bbR^d$  are written as $ a^{*}_{q}, a_{q}$, respectively. They satisfy the canonical commutation relations
\beq
 [a^{\#}_{q} , a^{\#}_{q'}] =0, \qquad  [a_{q} , a^{*}_{q'}] =\delta(q-q'),
\eeq
where $a^{\#}$ stands for either $a$ or $a^{*}$. The energy of a reservoir mode $q$ is given by the dispersion law $\om(q) \geq 0$.
To describe the coupling of the particle to the reservoir, we introduce a Hermitian matrix $W $ on $\scrS$ and we write $X$ for the position operator on $l^2(\lat)$. 

The total Hamiltonian of the system is taken  to be
\beq  \label{def: hamiltonian1}
H_\la: = H_\sys    +    \int_{\bbR^d} \d q  \,  \bosondispersion(q)     a^{*}_{q} a_{q}   + \la   \int_{\bbR^d} \d q    \left( \e^{\i q \cdot X} \otimes W  \otimes  \phi(q) a_q+   \e^{-\i q \cdot X} \otimes W  \otimes  \overline{\phi(q)} a^*_q  \right)  
\eeq 
acting on $\scrH_\sys \otimes  \scrH_\res$. The function  $\phi(q)$ is a form factor and $\la \in \bbR$ is the coupling strength. We write $H_\sys$ instead of $H_\sys \otimes 1 $, etc.

We introduce three important assumptions:  
\ben
\item{The kinetic energy is small w.r.t.\  the coupling term in the Hamiltonian, as has already been indicated by the inclusion of $\la^2$ in the definition of $H_{\sys,\mathrm{kin}}$. Physically, this means that the particle is heavy.}
\item{We require a linear dispersion law for the reservoir modes, $\om(q) \equiv \str q\str
$, in order to have good  decay estimates at low speed.  This means that the reservoir consists of photons, phonons or Goldstone modes of a Bose-Einstein condensate. }
\item{ We assume that the amplitude of the wave front of a reservoir excitation (located on the light cone) has integrable  (in time) decay.  This is satisfied if the  dimension of space\footnote{Since the integrability in time is only needed for reservoir excitations, we can  in principle also treat models in which the particle is $3$-dimensional, but the reservoir is effectively $4$-dimensional.}  is at least $4$.
}
\een
Additional assumptions will concern the smoothness of the form factor $\phi$ and the ``effectiveness" of the coupling to the heat bath (e.g.,  the interaction between the internal degrees of freedom and the reservoir, described by the matrix $W$, should not vanish.)

The initial state, $\initialres $, of the reservoir is chosen to be an equilibrium state at  temperature $\be^{-1}>0$.
 For mathematical details on the construction of infinite reservoirs, see e.g.\ \cite{derezinski1,bratellirobinson,arakiwoods}. The initial state of the whole system, consisting of the particle and the reservoir, is a product state $\rho_{\sys}\otimes \initialres$, with $\rho_{\sys}$ a density  matrix for the particle that will be specified later. 
The time-evolved density matrix of the particle ('subsystem') is called $\rho_{\sys,t}$ and  is obtained by ``tracing out the reservoir degrees of freedom'' after the time-evolution has acted on the initial state during a time $t$, i.e., formally,
\beq \label{def: formal reduced evolution}
\rho_{\sys,t}   :=\Tr_{\scrH_\res} \left[ \e^{-\i t H_\la}  \left(\rho_{\sys}\otimes \initialres\right)   \e^{\i t H_\la}      \right],
\eeq
where $\Tr_{\scrH_\res}$ is the partial trace over $\scrH_\res$.  We warn the reader that the above formula does not make sense mathematically for an infinitely extended reservoir, since the reservoir state $\initialres $ is not a density matrix on $\scrH_\res$. This is a consequence of the fact that the reservoir is described from the start in the thermodynamic limit  and, hence,  the reservoir modes form a continuum.  Nevertheless, the LHS of formula \eqref{def: formal reduced evolution} can be given a meaning in the thermodynamic limit.

The density matrix $\rho_{\sys,t}$ obviously depends on the coupling strength $\la$, but we do not indicate this explicitly. We also drop the subscript $\sys$ and we  simply write $\rho_t$, instead of $\rho_{\sys,t} $, in what follows. 

We will often represent $\rho_t$ as a $ \scrB(\scrS)$-valued kernel on $\lat \times \lat$:
\beq
\rho_t(x^{}_{\links},x^{}_{\rechts}) \in \scrB(\scrS), \qquad  x^{}_{\links}, x^{}_{\rechts} \in \lat.
\eeq
Although this is not necessary for many of our results, we require the initial state of the particle to be exponentially localized near the origin of the lattice, i.e., 
\beq
\norm \rho_t(x^{}_{\links},x^{}_{\rechts}) \norm_{\scrB(\scrS)} \leq C \e^{- \delta' \str x^{}_{\links} \str } \e^{- \delta' \str x^{}_{\rechts}\str }, \qquad \textrm{for some constants}\, C, \delta'>0 
\eeq

Our first result concerns the diffusion of the position of the particle.
\subsubsection{Diffusion} \label{sec: brief diffusion}
  We define the probability density 
\beq
\mu_t(x) := \Tr_{\scrS}\rho_t(x,x) 
\eeq
where $\Tr_{\scrS}$ denotes the partial trace over the internal degrees of freedom. The number $\mu_t(x) $ is the probability to find the particle at site $x$ after time $t$.

 By diffusion, we mean that, for large $t$, 
\beq\label{def: diffusion1} 
\mu_t(x)    \sim  \left( \frac{1}{2 \pi t}\right)^{d/2} (\mathrm{det}D)^{-1/2} \exp \{ -  \frac{1}{2}\left(  \frac{x  }{\sqrt{t}}  \cdot D^{-1} \frac{x  }{\sqrt{t}}  \right) \},
\eeq
where the diffusion tensor $D \equiv D_{\la} $ is a  strictly positive matrix with real entries; actually, if the particle dispersion law $\varepsilon$ is invariant under lattice rotations, then the tensor $D$ is isotropic and hence a scalar. 
The magnitude of  $D$ is inferred from the following reasoning: 
The particle undergoes collisions with the reservoir modes.  Let $t_{m}$ be the mean time  between two collisions, and let $ v_m$ be the mean speed of the particle (the direction of the particle velocity is assumed to be random). Then the mean free path is $v_m \times t_m$ and  the central limit theorem suggests that the particle diffuses with diffusion constant
\beq
D \sim \frac{(v_m \times t_m)^2}{t_m}
\eeq
The mean time $t_{m}$ is of order $t_{m} \sim\la^{-2}$ since the interaction with the reservoir contributes only in second order. The mean velocity $v_m$ is of order $v_m \sim \la^2$ because of the factor $\la^2$ in the definition of the kinetic energy. Hence $D \sim \la^2$.

We now move towards quantifying \eqref{def: diffusion1}.  Let us fix a time $t$. Since $\mu_t(x)$ is a probability measure, one can think of $x_t$ as a random variable with
\beq\label{def: random variable}
\mathrm{Prob}(x_t=x): = \mu_t(x).
\eeq
The claim that the random variable $\frac{x_t}{\sqrt{t}}$ converges in distribution, as $t \nearrow \infty$, to a Gaussian random variable with mean $0$ and variance $D$ is called a Central Limit Theorem (CLT). 
It is equivalent to pointwise convergence of the characteristic function, i.e., 
\beq \label{eq: conv characteristic function}
 \sum_{x \in \lat} \e^{ -\frac{\i }{\sqrt{t}} x \cdot q}  \mu_t(x)   \quad \mathop{\longrightarrow}\limits_{t \uparrow \infty} \quad  \e^{- \frac{1}{2}q \cdot D q}, \quad  \textrm{for all}  \, q \in \bbR^d,
\eeq
and it is this statement which is our main result, Theorem \ref{thm: diffusion}. 
A stronger version of the convergence in \eqref{eq: conv characteristic function} (also included in Theorem \ref{thm: diffusion}) implies that the rescaled  moments of $\mu_t$ converge. For example, for $i,j =1,\ldots,d$,
 \baq 
 \frac{1}{t} \Tr[ \rho_t X_i]   =   \frac{1}{t} \sum_{x}   x_i \mu_t(x) \quad & \mathop{\longrightarrow}\limits_{t \uparrow \infty}& \quad  0
 \label{eq: convergence of first moment} \\[2mm]
    \frac{1}{t} \Tr[\rho_t X_i X_j ]  =   \frac{1}{t} \sum_{x}   x_i x_j \mu_t(x) \quad  &\mathop{\longrightarrow}\limits_{t \uparrow \infty} &\quad   D_{i,j},  \label{eq: convergence of second moment}
 \eaq
 In the form as stated here,  these results are expected only if one assumes that the model has space inversion symmetry, which is  assumed throughout.

\subsubsection{Equipartition}\label{sec: brief equipartition}

 Our second result concerns the asymptotic expectation value of the kinetic energy of the particle and the internal degrees of freedom. 
The equipartition theorem suggests that the energy of all degrees of freedom of the particle, the  translational and internal degrees of freedom, thermalizes at the temperature $\be^{-1}$ of the heat bath. We will establish this property up to a correction that is small in the coupling strength $\la$. This is acceptable, since  the interaction effectively modifies the Gibbs state  of the particle. 
We prove that, for all bounded functions $F$,
\baq
\Tr [\rho_t F(H_{\sys,kin} )]  \qquad  &\mathop{\longrightarrow}\limits_{t \nearrow \infty} &  \frac{1}{Z}  \int_{\tor} \d k \ F(\la^2 \ve(k)) \e^{-\be \la^2 \ve(k) } +o(\str \la\str^0)  \\[2mm]
\Tr [\rho_t  F( H_{\sys,spin})]  \qquad     &\mathop{\longrightarrow}\limits_{t \nearrow \infty} &   \qquad  \frac{1}{Z'}   \sum_{e \in \sp Y}  F(e) \e^{-\be e } +o(\str \la\str^0) , \qquad \textrm{as} \, \,   \la \searrow 0
\eaq
where $Z,Z'$ are normalization constants and the sum $\sum_{e \in \sp Y} $ ranges over all eigenvalues of the Hamiltonian $Y$. 
We note that the factor $\e^{-\be \la^2 \ve(k) }  $ can be replaced by $1$ (as in Theorem \ref{thm: equipartition}) since we anyhow allow a correction term that is small in $\la$ and the function $ \ve(k)$ is bounded.    For this reason, one could say that, for very small values of $\la$, the translational degrees of freedom thermalize at infinite temperature ($\be=0$).

\subsubsection{Decoherence}\label{sec: brief decoherence}

By decoherence we mean that  off-diagonal elements $\rho_t(x,y)$ of the density matrix $\rho_t$ in the position representation  fall off rapidly as a function of $\str x-y\str$. Of course, this property can only hold at large enough times when the effect of the reservoir on the particle has destroyed all initial long-distance coherence, i.e., after a time of order $\la^{-2}$.
Thus, there is a decoherence length $1/{\ga_{dch}}$ and a decay rate $g$ such that 
\beq
\norm \rho_t(x^{}_{\links},x^{}_{\rechts}) \norm_{\scrB(\scrS)}  \leq C \e^{-\ga_{dch} \str x^{}_{\links}-x^{}_{\rechts} \str } + C' \e^{-\la^2 g t}, \qquad  \textrm{as} \,\,  t \nearrow \infty
\eeq
for some constants $C,C'$.
The magnitude of the inverse decoherence rate ${\ga_{dch}}$ is determined as follows: The time the reservoir needs to destroy coherence is of the order of the mean free time $t_m$, while the time that is needed for coherence to be built up over a distance $1/{\ga_{dch}}$ is given by $({\ga_{dch}} \times v_m)^{-1}$, where $v_m$ is the mean velocity of the particle. Equating these two times yields
\beq
   \ga_{dch} \sim   (t_m\times v_m)^{-1}
\eeq
and hence, recalling that  $ t_m \sim \la^{-2}$ and $v_m \sim \la^2$, as argued in Section \ref{sec: brief diffusion}, we find that $\ga_{dch}$ does not scale with $\la$.

\subsection{Related results and discussion} \label{sec: related results}   

\subsubsection{Classical mechanics} \label{sec: results classical}

Diffusion has been established for the two-dimensional finite horizon billiard in \cite{bunimovichsinai}.   In that setup, a point particle travels in a periodic, planar array of fixed hard-core scatterers. The \emph{finite-horizon condition} refers to the fact that the particle cannot move further than a fixed distance without hitting an obstacle.  

Knauf \cite{knaufergodic} replaced the hard-core scatterers by a planar lattice of attractive Coulombic potentials, i.e., the potential is $V(x) =- \sum_{j \in \bbZ^2 }   \frac{1}{\str x-j \str}  $.
In that case, the motion of the particle can be mapped to the free motion on a manifold with strictly negative curvature, and one can again prove  diffusion. 

Recently, a  different approach was taken in \cite{bricmontkupiainencoupledmaps}:  Interpreted freely, the model in \cite{bricmontkupiainencoupledmaps} consists of a  $d=3$ lattice of confined particles that interact locally with chaotic maps such that the energy of the particles is preserved but their momenta are randomized. Neighboring particles can exchange energy via collisions and one proves diffusive behaviour of the energy profile.

\subsubsection{Quantum mechanics for extended systems} \label{sec: qmextended}

 The earliest result for extended quantum systems that we are aware of, \cite{ovchinnikoverikhman},  treats a quantum  particle interacting with  a time-dependent random potential that has no memory (the time-correlation function is $\delta(t)$). Recently, this was generalized in \cite{kangschenker} to the case of time-dependent random potentials where the time-dependence is given by a Markov process with a gap (hence, the free time-correlation function of the environment is exponentially decaying).  In \cite{deroeckfrohlichpizzo}, we treated a quantum particle interacting with independent heat reservoirs at each lattice site. This model also has an exponentially decaying  free reservoir time-correlation function and as such, it  is very similar to   \cite{kangschenker}.  Notice also that, in spirit, the model with independent heat baths is comparable to the model of \cite{bricmontkupiainencoupledmaps}, but, in practice, it is easier since quantum mechanics is linear!

The most serious shortcoming of these results is the fact that the assumption of exponential decay of the correlation function in time is unrealistic.  In the model of the present paper, the space-time correlation function,  called $\psi(x,t)$ in what follows, is the correlation function of freely-evolving excitations in the reservoir, created by interaction with the particle.
 Since momentum is conserved locally, these excitations cannot decay exponentially in time $t$, uniformly in $x$.  For example, if the dispersion law of the reservoir modes is linear, then $\psi(x,t)$ is a solution of the linear wave equation. In $d=3$, it behaves qualitatively as
\beq
\psi(x,t) \sim \frac{1}{\str x \str}  \delta (c \str t \str - \str x \str), \qquad \textrm{with} \, \, c\, \, \textrm{the propagation speed of the reservoir modes}
\eeq
In higher dimensions, one has better dispersive estimates, namely $\sup_{x} \str\psi(x,t)\str \leq t^{\frac{d-1}{2}}$ (under certain conditions), and this is the reason why, for the time being, our approach is restricted to $d\geq 4$. In the Anderson model, the analogue of the correlation function does not decay at all, since the potentials are fixed in time. Indeed, the Anderson model is different from our particle-reservoir model: diffusion is only expected to occur for small values of the coupling strength, whereas the particle gets trapped (Anderson localization) at large coupling.

Finally, we mention a recent and exciting development: in \cite{disertorispencerzirnbauer}, the existence of a delocalized phase in three dimensions is proven for a  supersymmetric model which is interpreted as a toy version of the Anderson model. 

\subsubsection{Quantum mechanics for confined systems}

The theory of confined quantum systems, i.e., multi-level atoms,  in contact with quasi-free thermal reservoirs has been intensively studied in the last decade, e.g.\ by \cite{bachfrohlichreturn, jaksicpillet2, derezinskijaksicreturn}. 
In this setup, one proves approach to equilibrium for the multi-level atom.  Although at first sight, this problem is different from ours (there is no analogue of diffusion), the techniques are quite similar and we were mainly inspired by these results.   However, an important difference is that, due to its confinement,  the multi-level atom experiences a free reservoir correlation function with better decay properties than that of our model.  For example, in \cite{jaksicpillet2}, the free reservoir correlation function is actually exponentially decaying. 

\subsubsection{Scaling limits}

Up to now, most of the rigorous results on diffusion starting from deterministic dynamics are formulated in a \emph{scaling limit}. This means that one does not fix one dynamical system and study its behaviour in the long-time limit, but, rather,  one compares a family of dynamical systems at different times. 
 The precise definition of the scaling limit differs from model to model, but, in general,  one scales time, space and the coupling strength (and possibly also the initial state) such that the Markovian approximation to the dynamics becomes exact. 
In our model the natural scaling limit is the so-called weak coupling limit: one introduces the macroscopic time $\tau:=\la^{2}t$ and one takes the limit $\la \searrow 0, t \nearrow \infty$ while keeping $\tau$ fixed. In that limit, the dynamics of the particle becomes Markovian in $\tau$ (as if the heat bath had no memory) and it is described by a Lindblad evolution.  The long-time behavior of this Lindblad evolution is diffusive. This is explained in detail in Section \ref{sec: lindblad}. One may say that, in this scaling limit, the heuristic reasoning  employed  in the previous sections to deduce the $\la$-dependence of the diffusion constant and the decoherence length becomes exact.  The same scaling is known very well in the theory of confined open quantum systems as it gives rise to the Pauli master equation. This was first made precise in \cite{davies1}.  

If we had set up the model with a kinetic energy of  $O(1)$ (instead of $O(\la^2)$), then one should also rescale space by introducing the macroscopic space-coordinate $\chi:= \la^{2}x$. The reason for this additional rescaling is that, between two collisions, a particle with mass of order $1$ moves during a time of order $\la^{-2}$, and hence it  travels a distance of order $\la^{-2}$.  
The resulting scaling limit  \beq x \to  \la^{-2}x, \quad t \to  \la^{-2}t, \quad \la \searrow 0 \eeq is often called the \emph{kinetic limit}. 
In the kinetic limit the dynamics of the particle is described by a linear Boltzmann equation (LBE) in the variables $(\chi,\tau)$.  The convergence of the particle dynamics to the LBE has been proven in \cite{erdos} for a quantum particle coupled to a heat bath, and in \cite{erdosyauboltzmann} for a quantum particle coupled to a random potential (Anderson model).   The long-time, large-distance limit of the Boltzmann equation is the heat equation, which suggests that one should be able to derive the heat equation directly in the limiting regime corresponding to
\beq x \to  \la^{-(2+ \ka) }x,\quad  t \to  \la^{-(2+2 \ka)}t, \quad \la \searrow 0, \qquad  \textrm{for some}\,\,  \ka >0.  \eeq
This was accomplished in \cite{erdossalmhoferyaunonrecollision, erdossalmhoferyaurecollision} for the Anderson model.    An analogous result was obtained in \cite{komorowskiryzhikdiffusion} for a classical particle moving in a random force field. 

\subsubsection{Limitations to our result}\label{sec: restrictions}

Two striking features of our model are the large mass, of order $\la^{-2}$, and the internal degrees of freedom described by the Hamiltonian $H_{\sys,\mathrm{spin}}=Y$. Physically speaking, these choices are of course not necessary for diffusion, they just make our task of proving it easier.  Let us explain why this is so. 
 First of all, once the mass is chosen to be of order $\la^{-2}$, the internal degrees of freedom are necessary to make the model diffusive in second order perturbation theory. Without the internal degrees of freedom, it would be ballistic. This is explained in Section \ref{sec: momentum representation of caM}; in particular, it can be deduced immediately from conservation of momentum and energy for the processes in Figure \ref{fig: markov}. Note also that the dependence on $\la$ is chosen such that the kinetic term $H_{\sys,\mathrm{kin}} = \la^2 \varepsilon(P)$ is comparable to the particle-reservoir interaction in second order of perturbation theory (both are of order $\la^2$). The large mass ensures that the position of the particle remains well-defined  for a  time of order $\la^{-2}$, which permits us to sum up  Feynman diagrams in real space. 

Further, we  note that our result requires an  analyticity assumption on the form factor $\phi$, see Assumption \ref{ass: analytic form factor}.  This assumption ensures that the free reservoir correlation function  $\psi(x,t)$  is exponentially decaying for small $x$, even though it has slow decay on the lightcone, as explained in Section \ref{sec: qmextended}.

\subsection{Outline of the paper}
The model is introduced in Section \ref{sec: model} and the results are stated in Section \ref{sec: result}.  In Section \ref{sec: lindblad}, we describe the Markovian approximation to our model. This approximation provides most of the intuition and it is a key ingredient of the proofs.  Section \ref{sec: strategy} describes the main ideas of the proof, which is contained in the remaining Sections \ref{sec: polymer models}-\ref{sec: model only long} and the four appendices \ref{app: reservoirs}-\ref{app: combinatorics}.

\section*{Acknowledgements}
W. De Roeck thanks J. Bricmont and H. Spohn for helpful discussions and suggestions, and for pointing out several references.   He has also greatly benefited from collaboration with J. Clark and  C. Maes. In particular, the results described in Section \ref{sec: lindblad} and Appendix \ref{app: caM} were essentially obtained in \cite{clarkderoeckmaes}.  After the first version was submitted, some inaccuracies were pointed out by L. Erd\"os, A. Knowles, H.-T. Yau and J. Yin. Most importantly, the remarks of an anonymous referee allowed for serious improvements in the presentation of the proof. 
At the time when this work was completed, W.D.R.\ was a postdoctoral fellow  of the Flemish research fund `FWO-Vlaanderen'. The support of this institution is gratefully acknowledged. 

\section{ The model} \label{sec: model}
After fixing conventions in Section \ref{sec: conventions},  we introduce the model.  
Section \ref{sec: particle} describes the particle, while Section \ref{sec: reservoir} deals with the reservoir. In Section \ref{sec: dynamics of coupled system}, we  couple the particle to the reservoir, and we define the reduced particle dynamics $\caZ_t$. Section \ref{sec: translation invariance and fiber decomposition} introduces the fiber decomposition.

\subsection{Conventions and notation} \label{sec: conventions}

Given a Hilbert space $\scrE$, we use the standard notation 
\beq \scrB_p(\scrE):= \left\{  S \in \scrB(\scrE), \Tr\left[(S^*S)^{p/2}\right] < \infty  \right\}  ,\qquad   1 \leq p \leq \infty,  \eeq
with $\scrB_\infty(\scrE)\equiv \scrB(\scrE)$ the bounded operators on $\scrE$,
and
\beq
\norm S \norm_p := \left(\Tr\left[(S^*S)^{p/2}\right]\right)^{1/p}, \qquad  \norm S \norm:=\norm S \norm_{\infty}.
\eeq

For bounded  operators acting on $\scrB_p(\scrE)$, i.e. elements of $\scrB(\scrB_p(\scrE))$, we use in general  the calligraphic font: $\caV,\caW,\caT,\ldots$.
An operator $X \in \scrB(\scrE)$ determines an operator  $\adjoint(X) \in \scrB(\scrB_p(\scrE))$ by 
\beq
\adjoint(X) S: =[X,S]= XS-SX, \qquad   S \in \scrB_p(\scrE).
\eeq
 The norm of  operators in $\scrB(\scrB_p(\scrE))$ is defined by \beq\label{def: norm on operators}
\norm \caW \norm := \sup_{S \in \scrB_p(\scrE)}   \frac{\norm \caW(S) \norm_p}{\norm S\norm_p}.
\eeq
We will mainly work with Hilbert-Schmidt operators  ($p=2$) and, unless mentioned otherwise, the notation $\norm \caW \norm$ will refer to this case.

For vectors $\upsilon \in \bbC^d$, we let $\Re \upsilon , \Im  \upsilon$ denote the vectors $(\Re  \upsilon_1, \ldots, \Re  \upsilon_d)$ and  $(\Im  \upsilon_1, \ldots, \Im \ \upsilon_d)$, respectively. The scalar product on $\bbC^d$ is written as 
$ \upsilon \cdot  \upsilon'$ and the norm  as $\str  \upsilon \str := \sqrt{ \upsilon \cdot  \upsilon}$. 

The scalar product on a general Hilbert space $\scrE$  is written as $\langle \cdot, \cdot \rangle$, or, occasionally, as $\langle \cdot, \cdot \rangle_{\scrE}$. All scalar products are defined to be  linear in the second argument and anti-linear in the first one.
We use the physicist's notation 
\beq \label{def: ket bra notation}
\str \varphi \rangle \langle \varphi' \str \qquad \textrm{for the rank-1 operator in $ \scrB(\scrE)$ acting as} \quad  \varphi''  \mapsto     \langle \varphi', \varphi'' \rangle   \varphi
\eeq

We write $\sfock(\scrE)$ for the symmetric (bosonic) Fock space over the Hilbert space $\scrE$ and we refer to  \cite{derezinski1} for definitions and discussion. 
If $\om$ is a self-adjoint operator on $\scrE$, then its (self-adjoint) second quantization, $\d \sfock (\om)$, is defined by 
\beq
\d \sfock (\om)  \Symm (\varphi_1 \otimes \ldots \otimes \varphi_n)   := \sum_{i=1}^n    \Symm (\varphi_1 \otimes\ldots \otimes  \om \varphi_i \otimes \ldots \otimes \varphi_n),  
\eeq
where $ \Symm$ projects on the symmetric subspace of $ \otimes^{n}\scrE$ and $\varphi_1, \ldots, \varphi_n \in \scrE$.  

We use $C,C'$ to denote constants whose precise value can change from equation to  equation.

\subsection{The particle} \label{sec: particle}

We choose a finite-dimensional Hilbert space $\interspace$, which can be thought of as the state space of some internal degrees of freedom of the particle, such as spin or a dipole moment.  
The total Hilbert space of the particle is given by  $\scrH_\sys := l^2(\lat, \interspace)=   l^2(\lat) \otimes\interspace $  (the subscript $\sys$ refers to 'system', as is customary in system-reservoir models). 

We define the position operators, $X_j$, on $\scrH_\sys$ by 
\beq
(X_j \varphi)(x)= x_j \varphi(x) ,\qquad         x \in \lat, \quad \varphi \in l^2(\lat, \interspace), \quad j=1,\ldots,d
\eeq
In what follows, we will almost always drop the component index $j$ and write $X \equiv (X_j)$ to denote the vector-valued position operator.
We will often consider the space $\scrH_\sys$ in its dual representation, i.e.\ as $L^2(\bbT^d, \interspace) $, where $\bbT^d$ is the $d$-dimensional torus (momentum space), which is identified with $L^2([-\pi,\pi]^d, \interspace)$.  We formally define the  `momentum' operator $P$ as multiplication by $k \in \bbT^d$, i.e., 
\beq
P\varphi(k)=k\varphi(k), \quad  k \in [-\pi,\pi]^d,  \,\, \varphi \in L^2(\bbT^d, \interspace)
\eeq
Although $P$ is well-defined as a bounded operator, it does not correspond to a continuous function on $\tor$, and  it is not true that $[X_j,P_j]=-\i$.
Throughout the paper, we will only use operators  $F(P)$ where $F$ is a function on $\tor$ that is extended periodically to $\bbR^d$. We choose such a periodic function, $\varepsilon$,  of $P$ to determine the dispersion law of the particle.
 The kinetic energy of our particle is given by 
$\la^2 \varepsilon(P)$, where $\la$ is a small parameter, i.e., the 'mass' of the particle is of order $\la^{-2}$

The energy  of the internal degrees of freedom is given by a self-adjoint operator $Y \in \caB(\interspace)$, acting on $\scrH_\sys$ as
$
(Y\varphi)(k)= Y(\varphi(k))
$.
The Hamiltonian of the particle is 
\beq
H_\sys:=     \la^2  \varepsilon(P) \otimes 1 +1\otimes  Y 
\eeq
As in Section \ref{sec: intro}, we will mostly write $ \varepsilon(P) $ instead of $ \varepsilon(P) \otimes 1$ and $Y$ instead of $1\otimes  Y $.

 Our first assumption ensures that the Hamiltonian $H_\sys= Y +\la^2 \varepsilon(P)$ has good regularity properties
\begin{assumption}[Analyticity of the particle dynamics] \label{ass: analytic dispersion} \emph{
The function $\ve$, defined originally on $\bbT^d$, extends to an analytic function  in a neighborhood of the complex multistrip of width $ \delta_{\ve} >0$.  That is, when viewed as a periodic function on $\bbR^d$,  $\ve$ is analytic (and bounded) in a neighborhood of $(\bbR+ \i [-\delta_{\ve}, \delta_{\ve}])^d $.   
Moreover,  $\varepsilon$ is symmetric with respect to space inversion, i.e.,
\beq \label{eq: space inversion for varepsilon}
 \varepsilon(k)=\varepsilon(-k).
 \eeq
Furthermore, we assume there is no $\upsilon \in \bbR^d$ such that the function $k \mapsto \upsilon \cdot \nabla \varepsilon(k)$ vanishes identically and that $\ve$ does not have a smaller periodicity than that of $\tor$, i.e., we assume that
\beq \label{eq: no smaller periodicity}
\ve (k)=\ve(z+k)\,\, \textrm{for all} \, \, k \in \tor   \qquad  \Leftrightarrow     \qquad z \in (2\pi \bbZ)^d. 
\eeq
}
\end{assumption}

The most natural choice for $\varepsilon$ is $\varepsilon(k)= \sum_{i=1}^d 2 (1-\cos (k_i))$, which corresponds to $-\varepsilon(P)$ being the lattice Laplacian.
As already indicated in Section \ref{sec: brief diffusion}, the symmetry assumption \eqref{eq: space inversion for varepsilon} is necessary to exclude an asymptotic drift of the particle.

By a simple Paley-Wiener argument, Assumption \ref{ass: analytic dispersion} implies that one has exponential propagation estimates for the evolution generated by the operator $\varepsilon(P)$. Indeed, from the relation
\beq  \label{eq: combes thomas}
\left\norm ( \e^{\i \nu \cdot X }    \e^{- \i t \ve(P) }   \e^{- \i \nu \cdot X }) \right\norm  =   \left\norm  \e^{- \i t \ve(P+ \nu) }  \right\norm    \leq   \e^{ q_{\ve}(\str \Im\nu \str) \str t\str }, \qquad  \textrm{for}\, \,  \str \Im \nu \str \leq \delta_{\ve}
\eeq
with $q_{\ve}(\ga):= \sup_{\str\Im p\str \leq \ga  }  \str \Im \varepsilon(p)\str $, one obtains
\beq \label{eq: bound on matrix elements laplacian}
\left\norm ( \e^{- \i t \ve(P) } )(x_{\links},x_{\rechts}) \right\norm_{\scrS}\leq  \e^{-\ga \str  x_{\links}-x_{\rechts} \str} \e^{q_{\ve}(\ga)  \str t\str }, \qquad \textrm{for} \, \,  \ga \leq \delta_{\varepsilon}
\eeq
where we write $S({x_{\links},x_{\rechts}}) $ for a $\scrB(\scrS)$-valued 'matrix element' of $S \in \scrB(\scrH_\sys)$.

\subsection{The reservoirs} \label{sec: reservoir}

\subsubsection{The reservoir space}
We introduce a one-particle reservoir  space
$\frh = L^2(\bbR^d)  $  and a positive one-particle Hamiltonian $\omega \geq 0$. 
The coordinate $q \in \bbR^d$ should be thought of as a momentum coordinate, and $\omega$ acts by multiplication with a function $\omega(q)$, 
\beq
(\omega\varphi)(q)= \omega(q) \varphi(q)
\eeq
In other words, $\omega$ is the dispersion law of the reservoir particles.
The full reservoir Hilbert space, $\scrH_\res$,  is the symmetric Fock space (see Section \ref{sec: conventions} or \cite{derezinski1}) over the one-particle space $\frh$, 
\beq \scrH_\res:= \sfock (\frh) \eeq 
The reservoir Hamiltonian, $H_\res$, acting on  $\scrH_\res$, is then the second quanitzation of $\om$
 \beq H_\res:= \d \Gamma_{\mathrm{s}} (\omega  ) =  \mathop{\int}\limits_{\bbR^d} \d q\,  \om(q) a^*_q a_q.\eeq  
with the creation/annihilation operators $a_q^*,a_q$ to be introduced below.

\subsubsection{The system-reservoir coupling}
The coupling between system and reservoir is assumed to be translation invariant. We choose  a `form factor' $\phi \in L^2(\bbR^d)$  and a self-adjoint operator $W=W^* \in \scrB(\interspace)$ with $\norm W \norm \leq 1 $, and we define
the interaction Hamiltonian $H_{\sys\res}$ by 
\beq
H_{\sys\res}:= \int \d q    \left( \e^{\i q \cdot X} \otimes W  \otimes  \phi(q) a_q+   \e^{-\i q \cdot X} \otimes W  \otimes   \overline{\phi(q)} a^*_q  \right)  \quad   \textrm{on} \quad 
\scrH_\sys \otimes  \scrH_\res,  
\eeq
where $a_q,a^*_q$ are the creation/annihilation operators (actually, operator-valued distributions) on $\frh$ satisfying the canonical commutation relations (CCR)
\beq
[a_q, a^*_{q'}] = \delta{(q-q')}, \qquad   [a^{\#}_q, a^{\#}_{q'}] = 0
\eeq
with $a^{\#}$ standing for either $a$ or $a^*$. We also introduce the smeared creation/annihilation operators
\beq
a^{*}(\varphi) := \int_{\bbR^d} \d q \, \varphi(q)   a^{*}_q, \qquad    a(\varphi) := \int_{\bbR^d} \d q \,  \overline{\varphi(q)}   a_q, \qquad  \varphi \in L^2(\bbR^d).
\eeq 

In what follows we will specify our assumptions on $H_{\sys\res}$, but we already mention that  we need $[W, Y]\neq 0$ for the internal degrees of freedom to be coupled effectively to the field.

\subsubsection{Thermal states} \label{sec: thermal states}

Next, we put some tools in place to describe the positive temperature state  of the reservoir. 
We introduce the density operator 
\beq T_{\beta}= (\e^{\be \omega}-1)^{-1}\quad   \textrm{on} \quad \frh=L^2(\bbR^d). \eeq
Let $\frC$ be the $^{*}$algebra consisting of polynomials in the creation and annihilation operators $a(\varphi),a^*(\varphi')$ with $\varphi, \varphi'  \in \frh$.
We define $\initialres $ as  a quasi-free state defined on $\frC$.  It is a linear functional on $\frC$, fully specified by the following properties:
\ben
\item{ Gauge-invariance 
\beq \initialres \left[  a^*(\varphi)
 \right]=
 \initialres \left[  a(\varphi) \right]= 0    \label{def: gauge invariance} \eeq}
 \item{The choice of the two-particle correlation function
  \beq \label{def: thermal state canonical}  \left( \begin{array}{cc} \initialres \left[ a^*(\varphi)
a(\varphi') \right]
& \initialres \left[a^*(\varphi) a^*(\varphi') \right]  \\
\initialres \left[ a(\varphi)a(\varphi')\right]& \initialres \left[ a(\varphi) a^*(\varphi')
\right]
\end{array} \right)  =  \left(\begin{array}{cc}
 \langle \varphi' | T_{\beta} \varphi \rangle & 0    \\
0 & \langle \varphi |(1+ T_{\beta}) \varphi'\rangle
\end{array}\right) \eeq
 }
\item{The state   $\initialres $ is quasifree. This means that the  higher correlation functions are related to the two-particle correlation function via Wick's theorem
\beq \label{eq: gaussian property}
\initialres \left[a^{\#}(\varphi_1)  \ldots a^{\#}(\varphi_{2n}) \right]  =   \sum_{  \pi \in \caP_n} \prod_{(i,j) \in \pi}   \initialres \left[a^{\#}(\varphi_i) a^{\#}(\varphi_j) \right]
\eeq
where  $a^{\#}$ stands for either $a^*$ or $a$, and
$\caP_n$  is the set of pairings $\pi$,  partitions of $\{1,\ldots,2n \}$ into $n$ pairs  $(r,s)$. By convention, we  fix the order within the pairs such that  $r<s$. }
\een

The reason why it suffices to specify the state on $\frC$ has been explained in many places, see e.g.\ \cite{bratellirobinson, froehlichmerkli, derezinski1}

\subsubsection{Assumptions on the reservoir}

 Next, we state our main assumption restricting the type of reservoir and the dimensionality of space.

\begin{assumption}[Relativistic reservoir and $d\geq 4$]\label{ass: space} \emph{
We assume that the
\beq
 \textrm{dimension of space} \, \, d  \geq 4
 \eeq
Further, we assume the dispersion law of the reservoir particles to be linear; 
\beq
\om(q) := \str q \str
\eeq
}
\end{assumption}

For simplicity, we will assume that the form factor $\phi$ is rotationally symmetric and we write
\beq \label{eq: rotation symmetry form factor}
\phi(q) \equiv  \phi(\str q \str), \qquad  q \in \bbR^d 
\eeq
We define the "effective squared form factor"  as
\beq \label{def: psi}
\hat\psi (\om) :=    \str \om \str^{(d-1)}  \left\{ \begin{array}{ll}     \frac{1  }{1-\e^{-\be \om}}    \, \str\phi( \str\om\str )\str^2  \     &    \om \geq 0 \\[2mm] 
\frac{1  }{\e^{-\be \om}-1}    \, \str\phi( \str\om\str )\str^2     &    \om<0  \end{array} \right.
\eeq
where we  are abusing the notation by letting $\om$  denote  a variable in $\bbR$. Previously, $\om$ was the  energy operator on the one-particle Hilbert space and as such, it could assume only positive values.   Indeed, at positive temperature, the function $\hat\psi (\om) $ plays a similar role as $\str\phi( \str\om\str )\str^2$ at zero-temperature: It describes the intensity of the coupling to the reservoir modes of frequency $\om$. Modes with $\om<0$ appear only at positive temperature and they correspond physically to "holes".  One checks that  $\frac{\hat\psi(\om)}{\hat\psi(-\om)}= \e^{\be \om}$, which is Einstein's emission-absorption law (i.e.\ detailed balance).
This particle-hole point of view can be incorporated into the formalism by  the Araki-Woods representation, see e.g.\ \cite{bratellirobinson, froehlichmerkli, derezinski1}.

The next assumption restricts the  ``effective squared form factor'' $\hat \psi$. 
\begin{assumption}[Analytic form factor] \label{ass: analytic form factor} \emph{
Let the form factor be rotation-symmetric $\phi(q)\equiv \phi(\str q \str)$, as in \eqref{eq: rotation symmetry form factor}, and let $\hat\psi$ be defined as in \eqref{def: psi}. 
We assume that $\hat\psi(0)=0$ and that  the  function $\om  \rightarrow     \hat\psi (\om)  $ has  an analytic extension to  a neighborhood of the strip  $ \bbR + \i[\delta_\res, \delta_\res] $, for some $\delta_\res>0$,  such that
 \beq\label{eq: hardy class}
 \mathop{\sup}\limits_{-\delta_\res  \leq  \chi \leq \delta_\res} \mathop{\int}\limits_{\bbR+\i \chi} \d \om   \str    \hat \psi (\om)  \str   < \infty.
\eeq
}
\end{assumption}
We note that  Assumption \ref{ass: analytic form factor}  is satisfied (in $d \geq 4$) if one chooses:
\beq
\phi(\str q \str) := \frac{1}{\sqrt{\str q \str}} \vartheta(\str q \str) \eeq
with $\vartheta$ a function on $\bbR$ with $\vartheta(-\om)=\vartheta(\om)$ and analytic in the strip of width $\de_\res$, and such that \eqref{eq: hardy class} holds with $\str\vartheta(\om)\str^2$ substituted for $\str\hat\psi(\om)\str$.

The motivation for Assumptions \ref{ass: space} and \ref{ass: analytic form factor} will become clear in Section \ref{sec: reservoir correlation function},  where we discuss the reservoir space-time correlation function $\psi(x,t)$. 

The last assumption is a Fermi Golden Rule condition that ensures that the spin degrees of freedom are effectively coupled to the reservoir. To state it, we need the following operators 
\beq \label{def: Wa}
         W_{a}  : =  \mathop{\sum}\limits_{\footnotesize{\left. \begin{array}{c}  e,e' \in \sp Y  \\ e-e'=a   \end{array}\right. }}  1_{e'} (Y)  W 1_e(Y), \qquad  a \in \sp (\adjoint(Y)) 
\eeq
Note that the variable $a$ labels the Bohr-frequencies of the internal degrees of freedom of the particle.
\begin{assumption}[Fermi Golden Rule] \label{ass: fermi golden rule} \emph{
Recall the function $\hat\psi$ as defined in \eqref{def: psi}. The set of matrices
\beq
\scrB_{W} := \left\{     \hat\psi(a) W_{a}, a  \in \sp (\adjoint(Y))   \right\}  \subset \scrB(\scrS)
\eeq
generates the complete algebra $\scrB(\scrS)$. This means that any $S \in \scrB(\scrS)$ which commutes with all operators in $\scrB_{W} $ is necessarily a multiple of the identity. 
We also require the following non-degeneracy condition
\begin{itemize}
\item 
Every eigenvalue of $Y$ is nondegenerate (multiplicity $1$)
\item
For all eigenvalues $e,e',e'',e'''$ of $Y$ such that $e\neq e'$, we have
\beq
e'-e= e'''-e'' \qquad  \Rightarrow    \qquad   \left( e'=e'''  \,  \textrm{and} \,   e''=e   \right)
\eeq
\end{itemize}
This condition implies in particular that all eigenvalues of $\adjoint(Y)$ are nondegenerate, except for the eigenvalue $0$, whose multiplicity is given by $\dim \scrS$. }
\end{assumption} 
 
 The strict nondegeneracy condition on $Y$, in contrast to the condition on $\scrB_{W} $, is not crucial to our technique of proof, but it allows us to be more concrete in some stages of the calculation. In particular, the matrices  $W_{a \neq 0}$, introduced above in   \eqref{def: Wa}, can be rewritten as
 \beq 
W_{a} =       \langle e', W e \rangle  \times  \str e'\rangle \langle e \str,
 \eeq
 where $e,e'$ are the unique eigenvalues s.t.\,  $e-e'=a \neq 0$, and we have denoted the corresponding eigenvectors by the same symbols $e,e'$ (cfr.\ \eqref{def: ket bra notation}). The condition that $\scrB_{W}$ generates the complete algebra, can then be rephrased as follows: 
Consider an undirected graph with vertex set $\sp Y$ and let the vertices $e$ and $e'$ be connected by an edge if and only if \beq  \hat \psi(e'-e) \str\langle e, W e' \rangle\str^2 \neq 0   \label{def: connectedness} \eeq (note that this condition is indeed symmetric in $e,e'$, as long as $\be < \infty$).
Then Assumption \ref{ass: fermi golden rule} is satisfied if and only if this graph is connected. 

Assumptions of the type above have their origin in a criterion for ergodicity of quantum master equations due to \cite{spohnapproach, frigerio}, that is the noncommutative analogue of the Perron-Frobenius theorem. In our analysis, too, Assumption \ref{ass: fermi golden rule} is used to ensure that the Markovian semigroup $\La_t$ (to be introduced in Section \ref{sec: lindblad}) has good ergodic properties. This can be seen in Section \ref{sec: app proof of proposition} in Appendix \ref{app: caM}.

\subsection{The dynamics of the coupled system} \label{sec: dynamics of coupled system}

Consider the Hilbert space $ \scrH:= \scrH_\sys \otimes \scrH_\res$. The  Hamiltonian $H_\la$ (with coupling constant $\la$) on $\scrH$ is (formally) given by
\beq \label{def= Hlambda}
H_\la:= H_\sys + H_\res+ \la H_{\sys\res}
\eeq
If the following condition is satisfied
\beq \label{eq: condition for relative boundedness}
\langle \phi, \om^{-1} \phi \rangle_{\frh}   < \infty,
\eeq
then $H_{\sys\res}$ is a relatively bounded perturbation of $H_\sys + H_\res$ and hence $H_\la$ is a self-adjoint operator. One easily checks  that \eqref{eq: condition for relative boundedness} is implied by Assumptions  \ref{ass: space} and \ref{ass: analytic form factor}.

 For the purposes of our analysis, it is important to understand the dynamics of the coupled system at positive temperature. To this end, we introduce the reduced dynamics of the quantum particle. 

By  a slight abuse of notation, we use $\initialres$ to denote the conditional expectation $\scrB(\scrH_\sys) \otimes \frC \to \scrB(\scrH_\sys) $, given by 
\beq
 \initialres [S \otimes R]= S \initialres[R]
\eeq
where $\initialres(R)$ is defined by (\ref{def: gauge invariance}-\ref{def: thermal state canonical}-\ref{eq: gaussian property}) for $R \in \frC$, i.e.\, a  polynomial in creation and annihilation operators. 

Formally, the reduced dynamics in the Heisenberg  picture is given by 
\beq \label{def: z}
\caZ_{t}^{\star} (S) :=     \initialres \left[  \e^{\i t H_\la} \,  ( S \otimes 1)  \,  \e^{-\i t H_\la} \right].
\eeq
However, this definition does not make sense a priori, since   $ \e^{\i t H_\la} \,  ( S \otimes 1)  \,  \e^{-\i t H_\la}  \notin \scrB(\scrH_\sys) \otimes \frC$ in general. A mathematically precise definition of $\caZ_{t}^{\star}$ is the subject of the upcoming Lemma \ref{lem: definition dynamics}. 

Since both the initial reservoir state $  \initialres $ and the Hamiltonian $H_\la$ are translation-invariant, we expect that the reduced evolution $\caZ_{t}^{\star}$ is also translation invariant in the sense 
that 
\beq  \label{eq: expected translation invariance}
\caT_z \caZ_{t}^{\star}  \caT_{-z} =  \caZ_{t}^{\star} , \qquad \textrm{where}   \,\, (\caT_z S)(x_{\links}, x_{\rechts}):=  S(x_{\links}+z, x_{\rechts}+z)
\eeq
By the requirement $\varepsilon(k)=\varepsilon(-k)$ in Assumption \ref{ass: analytic dispersion} and the requirement that $\phi(q)=\phi(-q)$ in Assumption \ref{ass: analytic form factor}, the Hamiltonian $H_\la$ is also invariant with respect to space-inversion $x \mapsto -x$, or, equivalently, $k \mapsto -k$.  Since the initial reservoir state is also invariant with respect to space inversion (this follows from the fact that $\om(q)=\om(-q)$), we expect that
\beq    \label{eq: expected inversion invariance}
\caT_E \caZ_{t}^{\star}  \caT_{E} =  \caZ_{t}^{\star} , \qquad \textrm{where}   \,\, (\caT_E S)(x_{\links}, x_{\rechts}):=  S(-x_{\links}, -x_{\rechts})
\eeq
Finally, the unitarity of the microscopic time-evolution implies that 
\beq    \label{eq: expected self-adjointness}
\caT_J \caZ_{t}^{\star}  \caT_{J} =  \caZ_{t}^{\star} , \qquad \textrm{where}   \,\, (\caT_J S)(x_{\links}, x_{\rechts}):=  S^*(x_{\rechts}, x_{\links})
\eeq
where the $^*$ in $S^*(\cdot, \cdot)$ is the Hermitian conjugation on $\scrB(\scrS)$.

\begin{lemma}\label{lem: definition dynamics}
Assume Assumptions \ref{ass: analytic dispersion}, \ref{ass: space} and \ref{ass: analytic form factor}, and let
\beq
H_0:=  H_\sys+ H_\res, \qquad   H_{\sys\res} (t) :=   \e^{\i t H_0}   H_{\sys\res}    \e^{-\i t H_0}   
\eeq
The Lie-Schwinger series
    \baq \label{eq: lie schwinger}
 \caZ_t^{\star}(S)&  :  =&\mathop{\sum}_{n\in \bbN  }    (\i\la)^{n}   
 \mathop{\int}\limits_{\scriptsize{\left.\begin{array}{c}0 \leq t_1 \leq \ldots \leq t_n \leq t   \end{array}\right.}    } 
   \d t_1 \ldots \d t_n \,    \\[1mm]   &&  \initialres \left(   \adjoint(H_{\sys\res} (t_1)) \adjoint(H_{\sys\res} (t_2)) \ldots   \adjoint(H_{\sys\res} (t_n))  \,  \e^{\i  t \adjoint( H_0)}  \,   (S \otimes 1)  \right)
  \nonumber
    \eaq
 is well-defined for all $\la,t \in \bbR$, that is, the RHS is a norm convergent family of operators and $\caZ_t^{*}$ has the following  properties
    \ben
    \item{$\caZ_t^{\star}(1)=1$.}
    \item{ $\caT_z \caZ_{t}^{\star}  \caT_{-z} =  \caZ_{t}^{\star} $\,  with $\caT_z$ as defined in \eqref{eq: expected translation invariance}.}
            \item{$ \caT_E \caZ_{t}^{\star}  \caT_{E} =  \caZ_{t}^{\star}  $ \, with $\caT_E$ as defined in \eqref{eq: expected inversion invariance}.  }
                 \item{$ \caT_J \caZ_{t}^{\star}  \caT_{J} =  \caZ_{t}^{\star}  $ \, with $\caT_J$ as defined in \eqref{eq: expected self-adjointness}.  }
      \item{$ \norm \caZ_t^{\star}(S) \norm_{\infty} \leq \norm S \norm_\infty$. }
      \item{$  \caZ_t^{\star}(S) \geq 0$ for $S \geq 0$   }

        \item{ For $S \in \scrB_2(\scrH_\sys)$, the map $S \mapsto  \caZ_t^{\star}(S)$ is  continuous in $t$ in the Hilbert-Schmidt norm $\norm \cdot \norm_{2}$.}
    \een
\end{lemma}
These  properties  of $\caZ_{t}^{\star}$ should not come as a surprise, they hold true trivially if one pretends that the initial reservoir state $\initialres$ is a density matrix and  $\caZ_{t}^{\star}$ is obtained by taking the partial trace over the reservoir space, as in \eqref{def: formal reduced evolution}. 
One can prove this lemma, under  much less restrictive conditions than the stated assumptions,  by estimates on the RHS. For this purpose, the estimates given in the present paper amply suffice. However, one can also define the system-reservoir dynamics as a dynamical system on a von Neumann algebra through the Araki-Woods representation and this is the usual approach in the mathematical physics literature, see e.g.\ \cite{bachfrohlichreturn, derezinskijaksicreturn, derezinski1,froehlichmerkli, jaksicpillet2}.

We also define $\caZ_{t}: \scrB_1(\scrH_\sys) \to \scrB_1(\scrH_\sys)$, the reduced dynamics in the Schr{\"o}dinger representation,  by  duality 
\beq \label{eq: duality}
\Tr  [ S  \caZ_{t}^{\star} (S') ]  =     \Tr  [ \caZ_{t}(  S ) S'   ]
\eeq
Physically, $ \caZ_{t}^{\star} $ is the reduced dynamics on observables of the system and $\caZ_t$ is the reduced dynamics on states.

\subsection{Translation invariance and the fiber decomposition} \label{sec: translation invariance and fiber decomposition}

In this section, we introduce concepts and notation that will prove useful in the analysis of the reduced evolution $\caZ_t$. These concepts will be used in Section \ref{sec: asymptotic density matrix}. However, Section \ref{sec: diffusion, decoherence and equipartition}, which contains the main results, can be understood without the concepts introduced in the present section. 

Consider the space of Hilbert-Schmidt operators  \beq \scrB_2(\scrH_\sys)\sim \scrB_2(l^2(\lat) \otimes \scrS)  \sim L^2(\bbT^d \times \bbT^d, \scrB_2(\scrS), \d k_{\links}\d k_{\rechts}) \eeq and define 
\beq
    \hat S     (k_{\links},k_{\rechts}) :=\sum_{x_{\links},x_{\rechts} \in \lat}    S(x_{\links},x_{\rechts})  \e^{- \i (x_{\links} k_{\links}-x_{\rechts} k_{\rechts}) }      , \qquad  S \in    \scrB_2(l^2(\lat) \otimes \scrS) .
\eeq
Note the asymmetric normalization of the Fourier transform, which serves to eliminate factors of $2 \pi$ in the bulk of the paper.
In what follows, we will write $S$ for $\hat S$ to keep the notation simple, since the arguments $x \leftrightarrow k$ will indicate whether we are dealing with $S$ or $\hat S$.
To deal conveniently with the translation invariance of our model, we  change  variables, see also Figure \ref{fig: Brioullin}. 
\beq
k= \frac{k_{\links}+k_{\rechts}}{2}, \qquad  p=k_{\links}-k_{\rechts}, \qquad  k, p \in \tor
\eeq
and,   for a.e.\ $p \in \tor$, we obtain a  function $ S_p \in L^2(\bbT^d, \scrB_2(\scrS))$ by putting
\beq \label{def: Sfiber}
(S_p)(k) := S (k+\frac{p}{2},k-\frac{p}{2}).
\eeq
This follows from the fact that the Hilbert space  $\scrB_2(\scrH_\sys) \sim L^2(\bbT^d \times \bbT^d,   \scrB_2(\scrS),   \d k_{\links}\d k_{\rechts})$ can be represented as a direct integral
\beq \label{def: fiber decomposition}
\scrB_2(\scrH_\sys) = \mathop{\int}\limits_{\bbT^d}^{\oplus} \d p \,    \scrG^p , \qquad     S =  \mathop{\int}\limits_{\bbT^d}^{\oplus} \d p \, S_p,
\eeq
where each `fiber space' $\scrG^p$ is naturally identified with $ \scrG \equiv L^2(\bbT^d, \scrB_2(\scrS))$.
Elements of $\scrG$ will often be denoted by $\xi,\xi'$  and the scalar product is 
\beq
\langle \xi, \xi' \rangle_{\scrG} :=   \mathop{\int}\limits_{\tor} \d k \,   \Tr_{\scrS} [\xi^*(k)  \xi'(k)]
\eeq
with $\Tr_{\scrS}$ the trace over the space of internal degrees of freedom $\scrS$.
  
\begin{figure}[h!] 
\vspace{0.5cm}
  \centering
\psfrag{kl}{ $k_{\links}$}
\psfrag{kr}{ $k_{\rechts}$}
\psfrag{k}{ $k$}
\psfrag{p}{$p$}
\psfrag{kp}{$k=\pi$}
\psfrag{km}{$k=-\pi$}
\psfrag{u}{$p=\pi$}
\psfrag{pm}{$p=-\pi$}
 \includegraphics[width = 12cm, height=12cm]{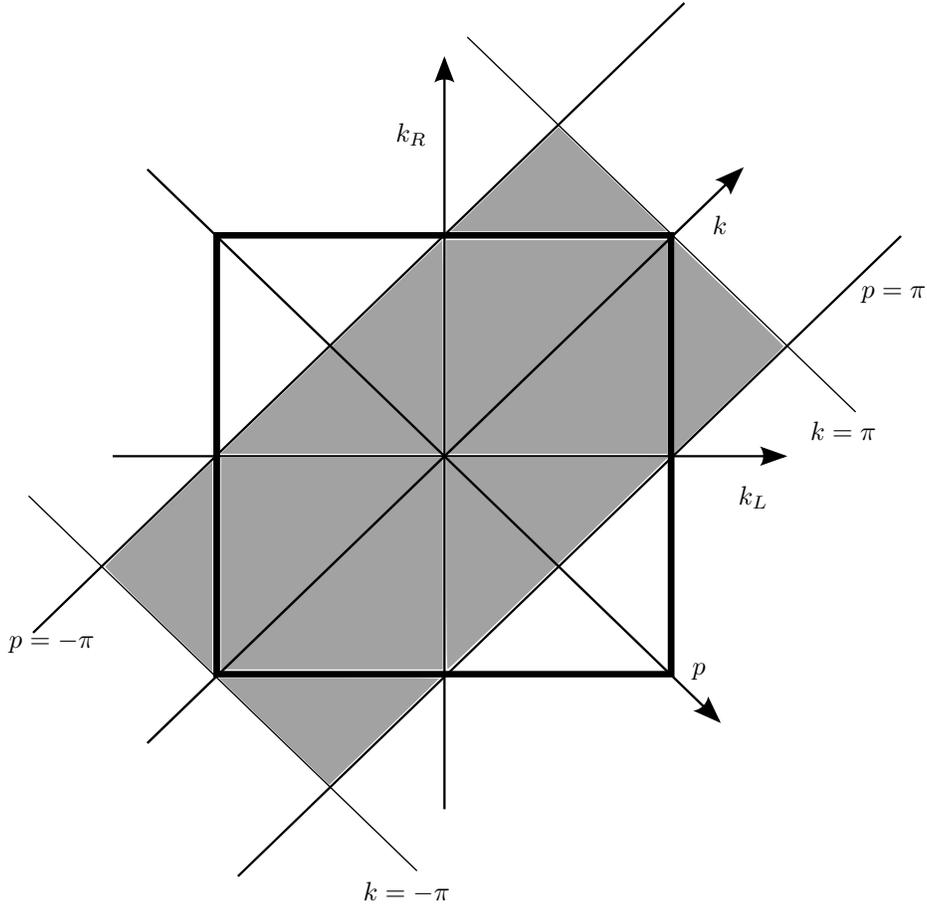}    
\caption{ \footnotesize{The thick black square $[-\pi, \pi] \times [-\pi, \pi] $ is the momentum space $\tor \times \tor$ (drawn here for $d=1$), with $k_{\links}, k_{\rechts} \in \tor$. After changing variables to $(k,p) \in \tor \times \tor$, the momentum space is transformed into the gray rectangle.  On sees that the four triangles which lie inside the square but ouside the rectangle, are identified with the four triangles inside the rectangle but outside the square.}  }
\label{fig: Brioullin}
\end{figure}

 Let $\caT_z, z \in \lat$ be the lattice translation defined in \eqref{eq: expected translation invariance}. In momentum space,
\beq 
(\caT_z S)_p  =  \e^{-\i pz}S_p, \qquad  S \in \scrB_2(\scrH_\sys).
\eeq
Since $H_\la$ and $\initialres$ are translation invariant, it follows that
\beq \label{translation inv of Z}
   \caT_{-z} \caZ_{t} \caT_z =  \caZ_{t}.
\eeq

Let $\caW \in \scrB(\scrB_2(\scrH_\sys))$ be translation invariant in the sense  that $  \caT_{-z} \caW \caT_z=\caW$ (cf.\   \eqref{translation inv of Z}).
Then it follows that, in the representation  \eqref{def: fiber decomposition}, $\caW$ acts diagonally in $p$, i.e.\ 
  $(\caW S)_p$ depends only on $ S_p $ and we define 
${\caW }_p$ by 
\beq
({\caW S})_p  = {\caW} _p   S_p, \qquad  S_p \in \scrG, \caW_p \in \scrB(\scrG)
\eeq
For the sake of clarity, we give  an explicit expression for ${\caW }_p$. 
Define the kernel $\caW_{x^{}_{\links},x^{}_{\rechts};x'_{\links},x'_{\rechts}}$ by 
\beq
  (\caW S)(x'_{\links},x'_{\rechts})=\sum_{x^{}_{\links},x^{}_{\rechts} \in \lat}  \caW_{x^{}_{\links},x^{}_{\rechts};x'_{\links},x'_{\rechts}}   S(x^{}_{\links},x^{}_{\rechts}), \qquad  x'_{\links},x'_{\rechts}  \in \lat.
\eeq
Translation invariance is expressed by 
\beq
\caW_{x^{}_{\links},x^{}_{\rechts};x'_{\links},x'_{\rechts}} = \caW_{x^{}_{\links}+z,x^{}_{\rechts}+z;x'_{\links}+z,x'_{\rechts}+z},   \qquad  z \in \lat,
\eeq
and, as an integral kernel,  $ {\caW}_p \in \scrB(L^2(\tor, \scrB_2(\scrS)))$  is given  by
\beq  \label{eq: position representation of caWp}
 {\caW}_p (k',k)=       \mathop{  \sum}\limits_{\left.\begin{array}{c}x^{}_{\rechts}, x'_{\links},x'_{\rechts} \in \lat  \\  x^{}_{\links}=0  \end{array}\right.} \e^{\i k(x^{}_{\links}-x^{}_{\rechts} ) -\i k' (x'_{\links}-x'_{\rechts}) }    \e^{-\i  \frac{p}{2}((x'_{\links}+x'_{\rechts})- (x^{}_{\links}+x^{}_{\rechts}))}   \caW_{x^{}_{\links},x^{}_{\rechts};x'_{\links},x'_{\rechts}} .
\eeq
To avoid confusion with other subscripts we will often write  
\beq
\left\{ S \right\}_p \,\,  \textrm{instead of}  \,\,   S_p \qquad  \textrm{and} \qquad  \left\{ \caW \right\}_p    \,\,   \textrm{instead of}  \,\,  \caW_p
\eeq

We also introduce the following transformations.  For $\nu \in \bbT^d$, let $U_{\nu}$ be the unitary operator acting on the fiber spaces $\scrG$ as 
\beq
(U_{\nu} \xi)(k)=\xi(k+\nu), \qquad   \xi \in \scrG   \label{def: Unu}
\eeq

Next, let  $\ka=(\ka^{}_{\links}, \ka^{}_{\rechts}) \in \bbC^d \times \bbC^d$ and define the operators $\caJ_\ka$ by
\beq  \label{def: caJ}
(\caJ_\ka S)(x^{}_{\links},x^{}_{\rechts}) : = \e^{\i \frac{1}{2} \ka^{}_{\links} \cdot x^{}_{\links}}  S(x^{}_{\links},x^{}_{\rechts})     \e^{-\i  \frac{1}{2}  \ka^{}_{\rechts} \cdot x^{}_{\rechts} } 
\eeq
Note that $\caJ_{\ka}$ is unbounded if $\ka \notin \bbR^d\times \bbR^d$.

The relation between  the operators $\caJ_{\ka}$ and the fiber decomposition is given by the relation
\baq \label{eq: relation kappa and fibers}
\left\{\caJ_{\ka}  \caW \caJ_{-\ka} \right\}_p &=&     U_{ -\frac{\ka^{}_{\links} +\ka^{}_{\rechts}}{4}} \{\caW\}_{p - \frac{\ka^{}_{\links} -\ka^{}_{\rechts}}{2} }     U_{ \frac{\ka^{}_{\links} +\ka^{}_{\rechts}}{4}},
\eaq
as follows from  \eqref{eq: position representation of caWp} and the definition   \eqref{def: caJ}.
 From \eqref{eq: position representation of caWp} and \eqref{eq: relation kappa and fibers}, we check that
\baq
p \qquad &\textrm{is conjugate to} & \qquad  \frac{1}{2}\left((x'_{\links}+x'_{\rechts})- (x_{\links}+x_{\rechts}) \right)  \label{eq: conjugate1}  \\
\nu \qquad &\textrm{is conjugate to} & \qquad (x_{\links}-x_{\rechts})- (x'_{\links}-x'_{\rechts})    \label{eq: conjugate2}. 
\eaq
We state an important lemma on the fiber decomposition.
\begin{lemma} \label{lem: fiber decomposition}
Let $S \in \caB_1(L^2(\tor, \scrS))$. Then,  $S_p$ is well-defined, for every $p$, as a function in $L^1(\tor,\scrB_2(\scrS))$ and 
\beq \label{eq: first statement of fibered lemma}
\Tr \caJ_{\ka} S= \sum_{x \in \lat} \e^{-\i p x} S(x,x) = \langle 1, S_p \rangle_{\scrG}, \qquad \textrm{with}\,p= - \frac{\ka^{}_{\links} -\ka^{}_{\rechts}}{2} \, \, \textrm{and}\, \ka = (\ka^{}_{\links}, \ka^{}_{\rechts})
\eeq
where $1$ stands for the constant function on $\tor$ with value $1 \in \scrB(\scrS)$.
If, moreover, $\caJ_\ka S $ is a Hilbert-Schmidt operator for $\str \Im\ka^{}_{\links,\rechts} \str \leq \delta' $, then the function 
\beq
  \tor \mapsto \scrG: \qquad  p \mapsto S_p,
\eeq
as defined in \eqref{def: Sfiber}, is well-defined for all $p \in \tor$ and has a bounded-analytic extension to the strip $\str \Im p  \str < \delta' $. 
\end{lemma}
The first statement of the lemma follows from the singular-value decomposition for trace-class operators. In fact, the correct statement asserts that one can \emph{choose} $S_p$ such that \eqref{eq: first statement of fibered lemma} holds. Indeed, one can change the value of the kernel  $S(k^{}_{\links},k^{}_{\rechts})$ on the line $k^{}_{\links}-k^{}_{\rechts}=p$ without changing the operator $S$, and hence $S_p$ in \eqref{eq: first statement of fibered lemma} can not be defined via \eqref{def: Sfiber}  in general, if the only condition on $S$ is $S \in \scrB_1$.

 The second statement of Lemma \ref{lem: fiber decomposition} is the well-known relation between exponential decay of functions and analyticity of their Fourier transforms.   Since we will always demand the initial density matrix $\rho_0$ to be such that $\norm \caJ_\ka \rho_0\norm_2 $ is finite for $\ka$ in a complex domain, we will mainly need the second statement of Lemma \ref{lem: fiber decomposition}.

By employing  Lemma \ref{lem: fiber decomposition}  and the properties of $\caZ^*_t$ listed in Lemma \ref{lem: definition dynamics},  it is easy to show that the function
\beq
k \mapsto \left\{ \caZ_t \rho_0 \right\}_0 (k)  \in \scrB(\scrS) \eeq
takes  values in the positive matrices on $\scrS$ and is normalized, i.e., 
\beq \label{eq: normalization zero fiber}
\int \d k  \Tr_{\scrS} [ \left\{ \caZ_t \rho_0 \right\}_0 (k)]  = \langle 1, \left\{ \caZ_t \rho_0 \right\}_0  \rangle_{\scrG}=1
\eeq
Further, the space-inversion symmetry  and self-adjointness of the density matrix (the third and fourth property in Lemma \ref{lem: definition dynamics} ) imply that 
\beq\label{eq: space inversion zero fiber}
 E \left\{ \caZ_t \right\}_{p} E  =   \left\{ \caZ_t \right\}_{-p}, \qquad \textrm{where}\, \, (E \xi)(k) :=  \xi(-k), \quad \textrm{for} \,\, \xi \in \scrG.
\eeq
\beq\label{eq: hermiticity  zero fiber}
 J \left\{ \caZ_t \right\}_{p} J  =   \left\{ \caZ_t \right\}_{-p}, \qquad \textrm{where}\, \, (J \xi)(k) :=  (\xi(k))^*, \quad \textrm{for} \,\, \xi \in \scrG.
\eeq
where the $^*$ on $\xi(k)$ is the Hermitian conjugation on $\scrB(\scrS)$. 

\section{Results} \label{sec: result}

In this section, we describe our main results. In Section \ref{sec: diffusion, decoherence and equipartition}, we state the results in a direct way, emphasizing  the physical phenomena. 
In Section \ref{sec: asymptotic density matrix}, we describe  more general statements that imply all the results stated in Section \ref{sec: diffusion, decoherence and equipartition}.

\subsection{Diffusion, decoherence and equipartition} \label{sec: diffusion, decoherence and equipartition}

We choose the initial state of the particle to be a density matrix $\rho \in \scrB_1(\scrH_\sys)$  satisfying 
\beq \label{eq: conditions initial state}
 \rho > 0, \qquad   \Tr [\rho]=1 \qquad    \norm \caJ_{\ka} \rho \norm_2 < \infty,
\eeq
for $\ka  $ in some neighborhood of $0 \in \bbC^d \times \bbC^d$. The condition $ \norm \caJ_{\ka} \rho \norm_2 < \infty$ reflects the fact that, at time $t=0$, the particle is exponentially localized near the origin.

Our results describe  the time-evolved density matrix $\rho_t:= \caZ_t \rho$.  Note that $\rho_t$  depends on $\la$, too.
First, we state that the particle exhibits diffusive motion. 

 Define  the probability density  $\mu_t\equiv \mu_t^\la$, depending on the initial state $\rho \in \scrB_1(\scrH_\sys)$, by
 \beq \label{def: density}
 \mu_t(x) :=     \Tr_{\scrS} \left[    \rho_t (x,x) \right].
 \eeq
 It  is easy to see that
 \beq
 \mu_t(x) \geq 0, \qquad    \sum_{x \in \lat} \mu_t(x) =  \Tr [\rho_t]= 1 .\eeq
 
 The following theorem states that the family of probability densities $\mu_t(\cdot)$ converges in distribution and in the sense of moments to a Gaussian, after rescaling space as $x \rightarrow \frac{x}{\sqrt{t}}$.

\bet [Diffusion]\label{thm: diffusion}
Assume Assumptions \ref{ass: analytic dispersion}, \ref{ass: space}, \ref{ass: analytic form factor} and \ref{ass: fermi golden rule}.
 Let the initial state $\rho$ satisfy  condition \eqref{eq: conditions initial state} and let $\mu_t$ be as defined in \eqref{def: density}. 
 
 There is a  positive constant $\la_0$ such that, for $0 <\str\la\str \leq \la_0$,
 \beq \label{eq: diffusion}
      \sum_{x \in \lat}   \mu_t(x)  \e^{- 
     \frac{ \i }{\sqrt{t}} q \cdot x}  \qquad   \mathop{\longrightarrow}\limits_{t \nearrow \infty} \qquad     \e^{- \frac{1}{2}q \cdot D_\la q}
 \eeq
 with the \emph{diffusion matrix} $D_\la$  given by
\beq
D_\la = \la^2 (D_{rw}+o(\la))
\eeq
where $D_{rw} $ is  the diffusion matrix of the Markovian approximation to our model, to be defined in Section \ref{sec: lindblad}. Both $D_\la$ and $D_{rw}$ are strictly positive matrices (i.e., all eigenvalues are strictly positive) with real entries.

The convergence of $\mu_t(\cdot)$ to a Gaussian also holds in the sense of moments:
For any natural number $\ell \in \bbN$, we have 
 \beq \label{eq: convergence of moments}
   (\nabla_q)^{\ell} \left(  \sum_{x \in \lat}   \mu_t(x)  \e^{- 
     \frac{ \i }{\sqrt{t}} q \cdot x} \right) \Bigg\str_{q=0}  \qquad     \mathop{\longrightarrow}\limits_{t \nearrow \infty} \qquad      (\nabla_q)^{\ell} \left( \e^{- \frac{1}{2}q \cdot D_\la q}\right) \Big\str_{q=0},
 \eeq
In particular, for $\ell=2$, this means that 
 \beq \label{eq: convergence of moments: second moment}
  \frac{1}{t} \sum_{x \in \lat}  \,   x_i x_j  \,  \mu_t(x)    \qquad     \mathop{\longrightarrow}\limits_{t \nearrow \infty} \qquad (D_\la)_{i,j}
 \eeq

\eet

Our next result describes the asymptotic 'state' of the particle. Not all observables reach a stationary value as $t \nearrow \infty$. For example, as stated in Theorem \ref{thm: diffusion}, the position diffuses.  The asymptotic state applies to the internal degrees of freedom of the particle and to functions of its momentum. Hence, we look at observables of the form 
\beq \label{eq: class of observables}
F(P) \otimes A, \qquad   F=\overline{F} \in L^{\infty}(\tor), \,\, A=A^* \in \scrB(\scrS).
\eeq
with $P=P\otimes1$ the lattice momentum operator defined in Section \ref{sec: particle}.
Such observables can be represented as elements of the Hilbert space $L^2(\tor) \otimes \scrB_2(\scrS) \sim L^2(\tor, \scrB_2(\scrS))=\scrG $ (recall that $\scrS$ is finite-dimensional) by the obvious mapping
\beq
F(P) \otimes A  \mapsto F \otimes A 
\eeq
since $L^{\infty}(\tor) \subset L^{2}(\tor) $.
 Consequently, the asymptotic state is not described by a density matrix on $\scrH_\sys$, but by a functional on the Hilbert space $\scrG$.  This functional is called  $\xi^{eq}\equiv \xi^{eq}_{\la}$ ('eq' for equilibrium) and we identify it with an element of $\scrG$. The  asymptotic expectation value of $F \otimes A$ is given by
 \beq \label{eq: mapping observables inside fibers}
 \langle F \otimes A, \xi^{eq} \rangle_{\scrG} =  \int_{\tor} \d k  \,  F(k)  \Tr_{\scrS}\left[ \xi^{eq}(k) A \right]  
 \eeq

We also state a result on decoherence: Equation \eqref{eq: decoherence} expresses that the  off-diagonal elements of  $\rho_t$ in position representation are exponentially damped in the distance from the diagonal. Note that this is not in contradiction with Theorem \ref{thm: diffusion} as the latter speaks about diagonal elements of $\rho_t$.

\bet[Equipartition and decoherence] \label{thm: equipartition}
Assume Assumptions \ref{ass: analytic dispersion}, \ref{ass: space}, \ref{ass: analytic form factor} and \ref{ass: fermi golden rule}.
Let the same conditions on the coupling constant $\la$ and the initial  state $\rho$ be satisfied as in Theorem \ref{thm: diffusion}. 
Let $A,F$ be as defined above. Then
 \beq \label{eq: equiparition}
       \Tr[\rho_t (   F(P ) \otimes A  )] =  \langle F \otimes A, \xi^{eq} \rangle_{\scrG}   + O( \e^{-g \la^2 t }) , \qquad  t  \nearrow \infty
 \eeq
  for some decay rate $g>0$.  The function $\xi^{eq} \equiv  \xi^{eq} _{\la}  \in \scrG$ is given by
 \beq
 \xi^{eq}(k) =   \frac{1}{Z(\be)}  \e^{-\be Y} + o(\str \la \str^0), \qquad \textrm{for all} \, k \in \tor, \qquad     \la \searrow 0 
 \eeq
 with the normalization constant 
$
 Z(\be) :=  (2\pi)^d \Tr (\e^{-\be Y})
$. 

  Further, there is a decoherence length $(\ga_{dch})^{-1} >0$ such that 
 \beq \label{eq: decoherence}
      \norm  \rho_t (x,y) \norm_{\scrB(\scrS)}  \leq C \e^{-\ga_{dch}  \str x-y\str }      + O(\e^{-g \la^2 t }) , \qquad    t  \nearrow \infty
 \eeq

\eet 
 In particular, Theorem \ref{thm: equipartition} implies that the inverse decoherence length $\ga_{dch} $ remains strictly positive  as $\la \searrow 0$.
Theorems \ref{thm: diffusion} and \ref{thm: equipartition} are derived from more general statements in the next section.

\subsection{Asymptotic form of the reduced evolution} \label{sec: asymptotic density matrix}

In the following theorem, we present a more general statement about the asymptotic form of the reduced evolution $\caZ_t$. The two previous results, Theorems \ref{thm: diffusion} and \ref{thm: equipartition}, are in fact immediate consequences of this more general statement. 

As argued in Section \ref{sec: translation invariance and fiber decomposition}, the operator $\caZ_t$ is translation invariant and hence it can be decomposed along the fibers,
\beq 
\caZ_t= \mathop{\int}_{\tor}^{\oplus} \d p   \left\{   \caZ_t \right\}_p, \qquad     \left\{   \caZ_t \right\}_p \in \scrB(\scrG)
 \eeq
 
 The next result, Theorem \ref{thm: main}, lists some long-time properties of the operators $\left\{\caZ_t \right\}$ and $
 U_{\nu}  \left\{   \caZ_t \right\}_p   U_{-\nu}
$
with $U_{\nu}$ as defined in  \eqref{def: Unu}. To fix the domains of the parameters $p$ and $\nu$, we define
\baq
 \frD^{\footnotesize{low}} & :=&     \left\{ p \in \tor+\i \tor, \nu   \in \tor+\i \tor \Big\str \, \str  \Re p \str < p^*,\,  \str \Im p \str  <\de, \,    \str \Im \nu \str < \de     \right\}   \label{def: domain low} \\ [2mm]
 \frD^{\footnotesize{high}} &:=&   \left\{ p \in \tor+\i \tor, \nu   \in \tor+\i \tor \Big\str \, \str  \Re p \str > p^*/2,\,  \str \Im p \str <
  \de, \,    \str \Im \nu \str < \de     \right\}    \label{def: domain high} 
\eaq
depending on some positive constants $p^*, \de >0$.
  
  \bet[Asymptotic form of reduced evolution] \label{thm: main}
Assume Assumptions \ref{ass: analytic dispersion}, \ref{ass: space}, \ref{ass: analytic form factor} and \ref{ass: fermi golden rule}, and
let the same conditions on the coupling constant $\la$ and the initial state  $\rho$ be satisfied as in Theorem \ref{thm: diffusion}. Then there are  positive constants $p^*>0$ and $\de>0$, determining the sets $\frD^{low},\frD^{high}$ above, such that the following properties hold:

\ben
\item{ For small fibers $p$, i.e., such that  $(p,0) \in  \frD^{\footnotesize{low}}$, there are 
 rank-1 operators  $P(p,\la)$,  bounded operators $R^{low}(t,p,\la)$  and numbers $f(p,\la)$, analytic in $p$ on $ \frD^{\footnotesize{low}}$  and satisfying 
\baq
\mathop{\sup}\limits_{(p,\nu) \in  \frD^{\footnotesize{low}}}  \norm  U_{\nu}P(p,\la) U_{-\nu} \norm &<&C   \label{eq: bound P}\\[3mm]
\mathop{\sup}\limits_{  (p,\nu) \in \frD^{\footnotesize{low}} } \sup_{t \geq 0}  \norm U_{\nu} R^{low}(t,p,\la) U_{-\nu} \norm & <&C    \label{eq: bound Rlow}
\eaq
such that
\baq
 & \left\{   \caZ_t \right\}_p    = \e^{ f(p,\la)t} P(p,\la)      + R^{\footnotesize{low}}(t,p,\la) \e^{-(\la^2 g^{\footnotesize{low}}) t }&   \label{eq: asymptotic expression caZ} 
\eaq
for a positive rate $g^{\footnotesize{low}} >0$.

}
\item{ For large fibers $p$, i.e., such that  $(p,0) \in  \frD^{\footnotesize{high}}$, there are   bounded operators $ R^{\footnotesize{high}}(t,p,\la)$, analytic in $p$ on $ \frD^{\footnotesize{high}}$  and satisfying
\beq  \label{eq: bound Rhigh}
\mathop{\sup}\limits_{(p, \nu) \in \frD^{\footnotesize{high}}  } \sup_{t \geq 0}  \norm U_{\nu} R^{\footnotesize{high}}(t,p,\la) U_{-\nu} \norm =O(1), \qquad  \la \searrow 0 
\eeq
and
\beq
 \left\{   \caZ_t \right\}_p  =  R^{\footnotesize{high}}(t,p,\la) \e^{-(\la^2 g^{\footnotesize{high}}) t }, \qquad  t \nearrow \infty   
\eeq
for some positive rate $g^{high}>0$.

}
\item{The function $f(p,\la)$ and rank-1 operator $P(p,\la)$ satisfy
\baq
\mathop{\sup}\limits_{(p, 0) \in \frD^{\footnotesize{low}}} \left\str f(p,\la)- \la^2 f_{rw}(p) \right\str &=& o(\str\la\str^2) \label{eq: f and projector close to random walk1} \\[2mm] 
\mathop{\sup}\limits_{(p, \nu) \in \frD^{\footnotesize{low}}} 
 \left\norm  U_{\nu}P(p,\la) U_{-\nu}- U_{\nu} P_{rw}(p)U_{-\nu} \right\norm &=&o(\str\la\str^0), \qquad \la \searrow 0 \label{eq: f and projector close to random walk2} \eaq
where the function $ f_{rw}(p)$ and the projection operator $P_{rw}(p)$ are defined in Section \ref{sec: lindblad}. 
}
\een
\eet
The main conclusion of this theorem is presented in Figure \ref{fig: spectrum of the real model}. Let $\caR(z)$ be the Laplace transform of the reduced evolution $\caZ_t$ and $  \left\{\caR(z) \right\}_p$ its fiber decomposition, i.e.,
\beq
\caR(z) := \mathop{\int}\limits_{\bbR^+} \d t \,  \e^{-t z}  \caZ_t \qquad  \textrm{and} \qquad     \caR(z)= \mathop{\int}_{\tor}^{\oplus} \d p\,    \left\{\caR(z) \right\}_p.
\eeq
The figure shows  the singular points, $z=f(p,\la)$, of $ \left\{\caR(z) \right\}_p$. Those singular points  determine the large time asymptotics. If we had not integrated out the reservoirs, i.e., if $\caZ_t$ were the unitary dynamics, then one could identify $f(p,\la)$ with resonances of the generator of $\caZ_t$. 
\begin{figure}[h!] 
\vspace{0.5cm}
  \centering
\psfrag{gaplow}{ $g^{low}$}
\psfrag{gaphigh}{ $g^{high}$}
\psfrag{functionf}{ $f(p,\la)$}
\psfrag{spectrumM}{$\Re z$}
\psfrag{momentump}{fiber $p$}
\psfrag{pcritical}{ $p^*$}
 \includegraphics[width = 12cm, height=8cm]{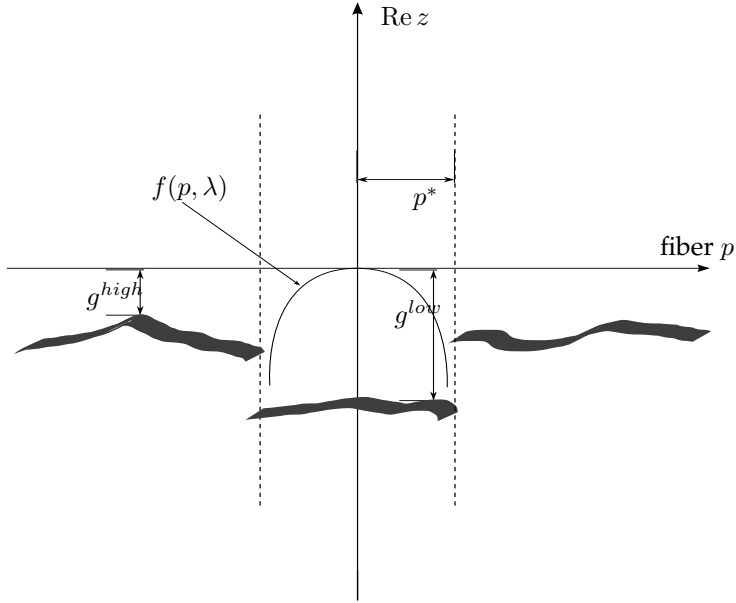}    
\caption{ \footnotesize{The singular points of $\left\{\caR(z) \right\}_p$ as a function of the fiber momentum $p$. Above the irregular black line, the only singular points are given by  $f(p,\la)$, in every small fiber $p$.  Below the irregular black lines, we have no control.}  }
\label{fig: spectrum of the real model}
\end{figure}

The proof of  Theorem \ref{thm: main} forms the bulk of the present paper. 

\subsection{Connection between Theorem \ref{thm: main} and the results in Section \ref{sec: diffusion, decoherence and equipartition} }

In this section, we show how to derive Theorems \ref{thm: diffusion} and \ref{thm: equipartition} from Theorem \ref{thm: main}. 

Since  $P(p,\la) $ is a rank-1 operator, we can write
\beq P(p,\la)    = \big\str \xi(p,\la)\big\rangle \big\langle   \tilde\xi(p,\la) \big \str, \qquad \textrm{for some} \,\, \xi(p,\la), \tilde\xi(p,\la) \in \scrG
\eeq
using the notation introduced in \eqref{def: ket bra notation}. 
We derive a  bound on the   eigenvectors $\xi(p,\la)$ and $\tilde\xi(p,\la)$, analytically continued in the coordinate $k$. This bound follows from  the  analyticity and uniform boundedness of $ U_{\nu} P(p,\la) U_{-\nu} $  on $\frD^{low} $ and  straightforward symmetry arguments;
\begin{lemma}\label{lem: boundedness eigenvectors}
The vectors $U_{\nu}\xi(p,\la)$ and $U_{\nu}\tilde\xi(p,\la)$  can be chosen bounded-analytic on $\frD^{low}$. In other words, the operator 
$P(p,\la)$ has a kernel 
\beq   \label{eq: rank one as kernel}
P(p,\la)(k,k')  =  \left\str \phantom{\tilde I}\!\!  \xi(p,\la)(k)  \right\rangle  \left\langle \tilde\xi(p,\la)(k') \right\str  
\eeq
which is bounded-analytic in both $k$ and $k'$ in the domain $\str \Im k  \str, \str \Im k'  \str  < \delta $
\end{lemma}
Note that for fixed $k,k'$, the RHS of \eqref{eq: rank one as kernel} belongs to $  \scrB(\scrB_2(\scrS))$.
\begin{proof}
Since $P(p,\la)$ is the dominant contribution to $\{\caZ_t\}$ for large $t$, the properties (\ref{eq: space inversion zero fiber}- \ref{eq: hermiticity  zero fiber}) imply  that 
\beq
(J E) P(p,\la) (JE) =P(p,\la)  
\eeq
 (note that $JE$ is an anti-unitary involution).  Consequently, the eigenvectors   $ \xi(p,\la)$ and $ \tilde\xi(p,\la)$ can  be chosen such that 
$J E   \xi(p,\la)=  \xi(p,\la)$ and $J E  \tilde \xi(p,\la)= \tilde \xi(p,\la)$. Then 
\beq
\norm U_{\nu}  \xi(p,\la) \norm =   \norm  U_{\nu}  JE   \xi(p,\la) \norm =  \norm JE  U_{-\nu}  \xi(p,\la) \norm  =     \norm U_{-\nu}  \xi(p,\la) \norm 
\eeq
Since  $U_\nu=\e^{-\i \nu \nabla_k}$, we have also
\beq
\norm \xi(p,\la)  \norm \leq  \norm 2\cosh ( \Im\nu  \nabla_k) \xi(p,\la)  \norm =  \norm (U_{\i \Im\nu }+ U_{-\i \Im\nu})\xi(p,\la)  \norm \leq 2   \norm U_{\nu}\xi \norm, \qquad \textrm{for any} \, \nu \in \bbC^d
\eeq
The same relation holds for  $\tilde\xi(p,\la) $ and hence none of the factors on the RHS of
\beq  \left\norm U_{\nu} P(p,\la) U_{-\nu} \right\norm_{\scrB(\scrG)} =  \left\norm \big\str U_\nu \xi(p,\la)\big\rangle \big\langle  U_\nu \tilde\xi(p,\la) \big \str \right\norm_{\scrB(\scrG)} = \norm U_\nu \xi(p,\la) \norm_{\scrG}   \, \norm  U_\nu \tilde\xi(p,\la) \norm_{\scrG}   
\eeq
can become small as $\nu$ varies. The lemma now follows from the uniform boundedness of $ U_{\nu} P(p,\la) U_{-\nu} $. 
\end{proof}

For $p=0$, the vectors $\xi(p,\la)$ and $\tilde\xi(p,\la)$ play a distinguished role, and we rename them as
\beq
\xi^{eq}= \xi_{\la}^{eq}:= \xi(p=0,\la)   , \qquad   \tilde \xi^{eq}= \tilde \xi_{\la}^{eq}:=  \tilde\xi(p=0,\la),
\eeq
Note that  $\xi^{eq}$ was already referred to in Theorem \ref{thm: equipartition}.

By exploiting symmetry and positivity properties of the reduced evolution $\caZ_t$, we can infer some further properties of the function $f(p,\la)$ and the operator $P(p,\la)$.

\begin{proposition} \label{prop: symmetries of f}
The function $f(p,\la)$, defined for all $p$ with $(p,0) \in \frD^{low}$,  has a negative real part, $ \Re f(p,\la) \leq 0 $, and satisfies
 the following properties
 \beq  f(p=0,\la)=0, \qquad  \textrm{and} \qquad    \nabla_p f(p,\la) \big\str_{p=0}=0   \label{eq: equalities f} \eeq 

\beq  \textrm{The Hessian} \quad D_\la:= (\nabla_p)^2 f(p,\la)\big\str_{p=0}  \quad \textrm{ has real entries and is strictly positive}  \eeq

The functions  $ \xi^{eq}  $ and $ \tilde \xi^{eq} $ can be chosen such that  
\beq \label{eq: normalization invariant state}
     \tilde \xi^{eq} =1, \qquad  \xi^{eq} (k) \geq 0, \qquad   \int_{\tor} \d k  \Tr_{\scrS} \left[\xi^{eq} (k) \right]=   \langle 1, \xi^{eq}  \rangle=1
\eeq
where $1 \in \scrG$ is the constant function on $\tor$ with value $1 \in \scrB_2(\scrS)$.
Moreover, it satisfies the space inversion symmetry $\left(\xi^{eq} \right)(k)=\left(\xi^{eq} \right)(-k)$.

\end{proposition}

The fact that  $f(p=0,\la)=0$, $\tilde \xi^{eq}=1$ and \eqref{eq: normalization invariant state} follow in a straightforward way from \eqref{eq: normalization zero fiber} and the asymptotic form \eqref{eq: asymptotic expression caZ}. The symmetry property $\xi^{eq} (k)=\xi^{eq} (-k)$ and $ \nabla_p f(p,\la) \big\str_{p=0}=0$  follow from \eqref{eq: space inversion zero fiber} and \eqref{eq: asymptotic expression caZ}.
The fact that $D_\la$ has real entries follows from $f(p,\la)=\overline{f(-p,\la)}$ which in turn follows from the reality of the probabilities $\mu_t(x)$ and the convergence \eqref{eq: diffusion}.

To derive the strict positivity of $D_{\la}$, we use the claim (in Proposition \ref{prop: properties of m}) that $D_{rw}$, the Hessian of $f_{rw}(p)$ at $p=0$, is strictly positive. By the convergence \eqref{eq: f and projector close to random walk1}  and the analyticity of $f_{rw}(p)$, it follows that $ \str D_{\la} - \la^2 D_{rw} \str  \searrow 0$ as $\la \searrow 0$. Indeed,  if a sequence of analytic functions  is uniformly  bounded on some open set and converges pointwise on that set, then all derivatives converge as well.

\subsubsection{Diffusion}\label{sec: diffusion connection}

We outline the derivation of Theorem \ref{thm: diffusion}.

Let $p$ be such that $(p,0) \in \frD^{low}$. We  calculate the logarithm of the characteristic function:
\baq
\log  \sum_x    \e^{- \i p x } \mu_t(x)   &=&   \log  \sum_{x}   \e^{ -\i p x }  \Tr_{\scrS} \rho_t(x,x)  \nonumber \\ [2mm]
   &=&  \log  \langle 1, \left\{\rho_t\right\}_{p} \rangle  \nonumber\\ [2mm]
 &=&      \log  \langle 1,  \left\{\caZ_t\right\}_{p}  \left\{\rho_0\right\}_{p} \rangle  \nonumber\\ [2mm]
     &=&      \log      \left( \e^{f(p,\la)t}  \langle 1,P(p,\la) \left\{\rho_0\right\}_{p}  \rangle +      \e^{-\la^2 g^{low}t}   \langle 1, R^{low}(t,p,\la)  \left\{\rho_0\right\}_{p} \rangle  \right) \nonumber\\ [2mm]
        &=&      \log   \e^{   f(p,\la)t}     \left(  \langle 1,P(p,\la)\left\{\rho_0\right\}_{p}  \rangle  +      \e^{-(\la^2 g^{low}- f(p,\la))  t } C \norm 1 \norm \norm \left\{\rho_0\right\}_{p} \norm     \right)  \nonumber \\ [2mm]
                &=&     f(p,\la) t +   \log  \left(  \langle 1,P(p,\la) \left\{\rho_0\right\}_{p} \rangle  +       \e^{-(\la^2 g^{low}- f(p,\la))  t }C   \norm 1 \norm \norm \left\{\rho_0\right\}_{p} \norm     \right)  
\eaq
where the scalar product $\langle \cdot, \cdot \rangle$ and $\norm \cdot \norm$ refer to the Hilbert space $\scrG$.
The second equality follows from Lemma \ref{lem: fiber decomposition}, the fourth  from \eqref{eq: asymptotic expression caZ} and the fifth from \eqref{eq: bound Rlow}. 
The second term between brackets in the last line vanishes as $t \nearrow \infty $, for $\str p \str$ small enough, such that $\la^2 g^{low}- f(p,\la) >0$.
To conclude the calculation, we need to check that the expression in $\log\left( \cdot\right)$ does not vanish.  We note that
\beq
 \langle 1,P(p=0,\la) \left\{\rho_0\right\}_{0}   \rangle =  \langle 1, \xi^{eq} \rangle    \langle \tilde \xi^{eq},  \left\{\rho_0\right\}_{0}   \rangle  =1
\eeq 
as follows from the fact that $\tilde \xi^{eq}=1$ and the normalization of $\xi^{eq}$ in \eqref{eq: normalization invariant state}.
Hence,  for  $p$ in a complex neighborhood of $0$, the expression $\langle 1,P(p,\la)  \left\{\rho_0\right\}_{p}  \rangle$ is bounded away from $0$ by analyticity in $p$. Consequently,
\beq \label{eq: convergence moment generating function}
\lim_{t \nearrow \infty}      \frac{1}{t}\log  \sum_x    \e^{-\i  p x } \mu_t(x)   =f(p,\la).
\eeq
Next, we remark that, for $\i p $ real, the LHS of \eqref{eq: convergence moment generating function} is a large deviation generating function for the family of probability densities $(\mu_t(\cdot))_{t \in \bbR^+}$. A classical result \cite{bryc} in large deviation theory states  that the analyticity of the large deviation generating function in a neighborhood of $0$ implies a central limit theorem for the variable $\frac{x}{\sqrt{t}}$, both in distribution, see \eqref{eq: diffusion},  as in the sense of moments, see \eqref{eq: convergence of moments}.

\subsubsection{Equipartition} \label{sec: equipartition connection}

To derive the result on equipartition in Theorem \ref{thm: equipartition}, we consider $F,A$ as in \eqref{eq: class of observables}. 
Since $\rho_t (F(P) \otimes A )$ is a trace-class operator, Lemma \ref{lem: fiber decomposition} implies that
\beq
\Tr \left[ (F(P) \otimes A ) \rho_{t} \right]=     \langle 1,  \left\{  (F(P) \otimes A ) \rho_{t} \right\}_0 \rangle_{\scrG}  =  \langle F \otimes A, \{\rho_{t} \}_0 \rangle_{\scrG}
\eeq
where, as in \eqref{eq: mapping observables inside fibers}, $F \otimes A$ stands for the function $k \mapsto F(k) A$ in $L^2(\tor, \scrB_2(\scrS))$.

Using Theorem \ref{thm: main} for the fiber $p=0$, we obtain
\baq
  \langle F \otimes A, \{ \rho_{t}\}_0 \rangle &=  &   \e^{ f(0,\la) t}  \langle F \otimes A,  P(p=0,\la) \{ \rho_0 \}_0 \rangle  +  \e^{-(\la^2 g^{low} )t} \langle F \otimes A,  R^{low}(t,p=0,\la) \{ \rho_0 \}_0 \rangle   \nonumber   \\[1mm]
   &=  &       \langle F \otimes A,  \xi^{eq}  \rangle+    C  \e^{-(\la^2 g^{low} )t}   \norm F \otimes A \norm_{\scrG}   \norm \{ \rho_{0} \}_0 \norm_{\scrG }    \label{eq: proof of equiparitition}
\eaq
To obtain the second equality, we have used the uniform boundedness of the operators $R^{low}(t,p=0,\la)$ (Statement 1) of Theorem \ref{thm: main}),  the fact that $f(p=0,\la)=0$ (Proposition \ref{prop: symmetries of f}) and the identities
\beq
 P(p=0,\la) \{ \rho_0 \}_0 =      \langle \tilde\xi^{eq}, \{ \rho_0 \}_0 \rangle     \xi^{eq}  =    \langle 1, \{ \rho_0 \}_0 \rangle  \xi^{eq}  =  \xi^{eq} 
 \eeq
 
Hence, from \eqref{eq: proof of equiparitition},  we obtain the asymptotic expression  \eqref{eq: equiparition} by choosing $g \leq g^{low}$.

\subsubsection{Decoherence} \label{sec: decoherence connection}

In this section, we derive  the bound \eqref{eq: decoherence} in Theorem \ref{thm: equipartition}. 
We decompose $\rho_{t}$ as follows, using Theorem \ref{thm: main}:
\baq
\rho_{t}&:=&  \mathop{\int}\limits_{\tor}^{\oplus} \d p  \,  \left\{ \rho_{t} \right\}_p  \\[2mm]
&=&  \underbrace{\int_{\str p \str \leq p^*}^{\oplus} \d p \,    \e^{\la^2 f(p,\la) t}  P(p,\la)  \left\{ \rho_0  \right\}_p}_{ =: A_1}    +   \e^{-\la^2 g^{low}t }  \underbrace{\int_{\str p \str \leq p^*}^{\oplus} \d p  \,    R^{low}(t,p,\la)  \left\{ \rho_0  \right\}_p}_{ =: A_2}      \\[2mm]
&+&    \e^{-\la^2 g^{high}t } \underbrace{\int_{\str p \str > p^*}^{\oplus} \d p   \,   R^{high}(t,p,\la)  \left\{ \rho_0  \right\}_p }_{ =: A_3}     \\
\eaq
The terms $A_2$ and $A_3$ are bounded  by 
\beq
\norm A_{2,3} \norm_2^2 \leq
C \int_{\tor} \d p  \,  \norm \left\{ \rho_0  \right\}_p \norm^{2}_{\scrG}   = C \norm \rho_0 \norm_2^2    \leq C \norm \rho_0 \norm_1^2   
\eeq
where the first inequality follows from the bounds \eqref{eq: bound Rlow} and \eqref{eq: bound Rhigh}.
Hence, for our purposes, it suffices to consider the first term $A_1$. 
To calculate the operator $A_1$ in position representation, we use the kernel expression \eqref{eq: rank one as kernel} for $P(p,\la)$ to obtain
\beq
A_1(x_{\links},x_{\rechts})= \int_{\str p \str \leq p^*} \d p     \e^{\i \frac{p}{2} \cdot (x_{\links}+x_{\rechts}) }     \e^{ f(p,\la) t}    \langle \tilde\xi(p,\la), \left\{ \rho_0  \right\}_p \rangle \int_{\tor}    \d k  \,   \xi(p,\la)(k) \e^{\i k \cdot (x_{\links}-x_{\rechts})}   
\eeq

We now shift the path of integration (in $k$) into the complex plane, using that the function $\xi(p,\la)(\cdot)$ is bounded-analytic in a strip of width $\delta$.
This yields exponential decay in $(x_{\links}-x_{\rechts})$.  Using also that $\Re f(p,\la) \leq 0$, for $\str p \str \leq p^* $ (see Proposition \ref{prop: symmetries of f}), we obtain the bound
\beq
\norm A_1(x^{}_{\links},x^{}_{\rechts}) \norm_{\scrB_2(\scrS)} \leq C  \e^{-\ga \str x^{}_{\links}-x^{}_{\rechts} \str }, \qquad  \textrm{for} \, \,  \ga < \de
\eeq
Combining the bounds on $A_1$ and $A_2,A_3$, we obtain 
\beq
\norm \rho_t(x^{}_{\links},x^{}_{\rechts}) \norm_{\scrB_2(\scrS)} \leq C  \e^{-\ga \str x^{}_{\links}-x^{}_{\rechts} \str }  + C' \e^{-(\la^2 g )t}, \qquad  \textrm{for} \, \, \ga < \de
\eeq
with $g:= \min\left( g^{low},g^{high} \right)$.
The fact that this bound is valid for any $\ga < \de$,  confirms the claim that the inverse decoherence length  $\ga_{dch}$ can be chosen uniformly in $\la$ as $\la \searrow 0$.

%
%
%
%

 \section{The Markov approximation}  \label{sec: lindblad}
 
  For small coupling strength $\la$ and times of order $\la^{-2}$, one can approximate the reduced evolution $\caZ_t$ by a  "quantum Markov semigroup" $\La_t$ which is of the form 
 \beq \label{eq: first appearance caM}
 \La_t = \e^{ t \left(-\i \adjoint(Y) +  \la^2 \caM \right) }
 \eeq
where $Y=1\otimes Y$ is the Hamiltonian of the internal degrees of freedom, and  $\caM$  is a Lindblad generator, see e.g.\ \cite{alickifannesbook}.   Lindblad generators, and especially the semigroups they generate, have received a lot of attention lately in quantum information theory. 
The operator $\caM$ has the additional property of being translation-invariant.  Translation-invariant Lindbladians have been classified in \cite{holevonote} and, recently,  studied in a physical context;  see \cite{hornbergervacchinireview} for a review.
  In Section \ref{sec: construction of semigroup}, we construct $\caM$ and we state its relation with $\caZ_t$.  We also describe heuristically how $\caM$ emerges from time-dependent perturbation theory in $\la$ as a lowest order approximation to $\caZ_t$.
   In Section \ref{sec: momentum representation of caM}, we discuss the momentum representation of $\caM$ (the derivation of this representation is however deferred to Appendix \ref{app: caM}), and we recognise that the evolution equation generated by $\caM$ is a mixture of a linear Boltzmann equation for the translational degrees of freedom and a Pauli master equation for the internal degrees of freedom. In Section \ref{sec: properties of semigroup}, we discuss spectral properties of $\caM$, which are largely proven in Appendix \ref{app: caM}.  Finally, in Section \ref{sec: bound on markov in position representation}, we derive  bounds on the long-time behavior of $\La_t\rho$, for any density matrix $\rho \in \scrB_1(\scrH_\sys)$.

 \subsection{Construction of the semigroup} \label{sec: construction of semigroup}

First, we define the operator $\hat\caL^{\star}(t)$ on $\scrB(\scrH_\sys)$:
\beq
 \hat  \caL^{\star}(t)(S) =   -    \initialres  \left( \adjoint (H_{\sys\res})  \,   \e^{\i t \adjoint (Y+H_\res)} \,  \adjoint (H_{\sys\res}) (S \otimes 1)    \right)   \\  
\eeq
This definition makes sense since  the conditional expectation $ \initialres$ is applied to an element of $\scrB(\scrH_\sys) \otimes \frC$, see Section \ref{sec: dynamics of coupled system}.
Then we consider the Laplace transform of $ \hat  \caL^{\star}(t)$, i.e.\ 
\beq
 \caL^{\star}(z)= \mathop{\int}\limits_{\bbR^+} \d t \, \e^{-t z}  \,  \hat  \caL^{\star}(t), \qquad  \Re z >0,
 \label{eq: first definition caL}
\eeq
%
and, finally, we let $ \caL(z)$ be the dual operator to $\caL^{\star}(z)$, acting on $ \scrB_1(\scrH_\sys)$, see \eqref{eq: duality}. Then the operator  $\caM$ is obtained  from $\caL$ by ``spectral averaging"  and adding the ``Hamiltonian" term $-\i \adjoint(\ve(P))$: 
 \beq \label{eq: def of M}
  \caM:= - \i \adjoint(\ve(P)) + \sum_{a \in \sp(\adjoint (Y))}   1_a (\adjoint(Y)) \caL(-\i a) 1_a (\adjoint(Y)) 
 \eeq
 For now, this definition is formal, since it involves  \eqref{eq: first definition caL}  with $\Re z=0$.

The following proposition provides a careful definition of $\caM$ and collects some basic properties of the semigroup evolution
$\La_t$. 

\begin{proposition}\label{prop: wc1} Assume Assumptions \ref{ass: analytic dispersion}, \ref{ass: space} and \ref{ass: analytic form factor}. Then, the operators $\caL(z)$, defined above, can be continued from $\Re z>0$ to a continuous function in the region $\Re z \geq 0$ and
\beq  \label{eq: caL defined on real axis}
\sup_{\Re z \geq 0}     \norm \caJ_{\ka}    \caL(z)  \caJ_{-\ka}  \norm   <\infty, \qquad  \textrm{for} \, \,   \ka \in \bbC^d \times \bbC^d
\eeq
(In fact, $\caJ_{\ka}    \caL(z)  \caJ_{-\ka} = \caL(z)$).

 The operator $\caM$, as defined in \eqref{eq: def of M}, is bounded both  on $\scrB_1(\scrH_\sys)$ and $\scrB_2(\scrH_\sys)$.
 Recall the constants $q_{\ve}(\ga), \ga >0$, defined in Assumption \ref{ass: analytic dispersion}. Then 
 \beq  \label{eq: caM has only propagation in hamiltonian}
\norm  \caJ_{\ka} \caM \caJ_{-\ka} -\caM  \norm    \leq   q_{\ve}(\str \Im \ka^{}_{\links} \str ) + q_{\ve}(\str \Im\ka^{}_{\rechts} \str ) ,  \qquad \str \Im\ka^{}_{\links,\rechts} \str \leq \de_{\ve}
 \eeq
 where the norm $\norm \cdot \norm$ refers to the operator norm on $\scrB(\scrB_2(\scrH_\sys))$.
 
The family of operators $\La_t$, defined in \eqref{eq: first appearance caM},
\beq
\La_t = \e^{t (-\i \adjoint(Y) +\la^2\caM)}, \qquad  t \in \bbR^+
\eeq
is a  ``quantum dynamical semigroup". This means\footnote{Most authors include "complete positivity" as a  property of quantum dynamical semigroups, see e.g.\  \cite{alickifannesbook}. Although the operators $\La_t$ satisfy complete positivity,  we do not stress this since it is not important for our analysis.}:
\beq
\left. \begin{array}  {llcr}   i)&   \La_{t_1}\La_{t_2}=\La_{t_1+t_2}  \qquad    & \textrm{for all} \, t_1,t_2 \geq 0  &     \textrm{(semigroup property)}     \\[1mm]  
ii) &  \La_{t}\rho \geq 0  \qquad  \qquad & \textrm{for any} \,  0 \leq\rho \in \scrB_1(\scrH_\sys)    & \textrm{(positivity preservation)}  \\[1mm]  
iii) & \Tr \La_{t}\rho=\Tr \rho  \qquad& \textrm{for any} \,  0 \leq\rho \in \scrB_1(\scrH_\sys)  &  \textrm{(trace preservation)}
 \end{array}  \right.
\eeq
\end{proposition}
We postpone the proof of this proposition to Appendix \ref{app: caM}.  

\subsubsection{Motivation of the semigroup $\La_t$}

The connection of the semigroup $\La_t$ with the reduced evolution $\caZ_t$ is that, for any $T<\infty$,
\beq\label{eq: davies result}
\mathop{\sup}_{ 0 <t < \la^{-2} T } \norm \caZ_{t}  - \La_t  \norm_{\scrB(\scrB_2(\scrH_\sys))} = o(\la^0),  \qquad   \la \searrow 0
\eeq

Results in the spirit of \eqref{eq: davies result} have been advocated by \cite{vanhove} and first proven, for confined (i.e.\ with no translational degrees of freedom) systems, in \cite{davies1}. They go under the name ``weak coupling limit" and they have given rise to extended mathematical studies, see e.g. \cite{lebowitzspohn1,derezinskifruboesreview}.   In our model, \eqref{eq: davies result} will be implied by our proofs but we will not state it explicitly in the form given above.  In fact, statements like \eqref{eq: davies result} can be proven under much weaker assumptions than those in our model; see 
 \cite{deroeckspehner} for a proof which holds in all dimensions $d>1$. 
 
 Here, we restrict ourselves to a short and heuristic sketch of the way the Lindblad generator $\caM$ emerges from the full dynamics. 
 First, we consider the Lie-Schwinger series  \eqref{eq: lie schwinger} in the interaction picture with respect to the free internal degrees of freedom, i.e.\ we consider $ \caZ_t^{\star}\e^{-\i  t \adjoint( Y)}$ instead of $ \caZ_t^{\star}$. 
 Keeping only terms  up to second order in $\la$ in  \eqref{eq: lie schwinger} and substituting our definition for $\hat\caL^{\star}(t)$ we obtain
 \baq  \label{eq: three terms of perturbation}
   \caZ_t^{\star} \e^{-\i  t \adjoint( Y)} &= &    1+ \,    \i \la^2 \, t   \adjoint( \varepsilon(P) ) + \,  \la^2 \mathop{\int}\limits_{0 <t_1 < t_2 <t}  \d t_1 \d t_2        \,     \e^{ \i (t-t_2)  \adjoint( Y)}  \hat \caL^{\star}(t_2-t_1)      \e^{-\i(t-  t_1) \adjoint( Y)}   +O(\la^4)
 \eaq
where we have also used $[ \varepsilon(P) , Y]=0$ to simplify the second term on the RHS.  It is useful to rewrite the third term by splitting $t_2=t_1+ (t_2-t_1)$ and to insert the spectral decomposition of unity corresponding to the operator $\adjoint(Y)$: we get
\beq
\la^2 \mathop{\sum}\limits_{a,a' \in \sp (\adjoint(Y))}  \int_{0}^{t} \d t_1  \,  \e^{\i (t-t_1)  (a-a')}   \, 1_a(\adjoint(Y)) \left( \int_0^{t-t_1} \d u  \,  \e^{-i u  a} \hat\caL^{\star}(u )  \right) 1_{a'}(\adjoint(Y))    \label{eq: third term spectral averaging}
\eeq

Next, we analyze the RHS of \eqref{eq: three terms of perturbation} for long times; we choose $t =\la^{-2}\frt$, and we  argue that, in the limit $\la \searrow 0$,  it reduces to 
\beq
1 + \i \adjoint( \varepsilon(P) ) \frt  \,  +  \left( \sum_a 1_a(\adjoint(Y))\caL^{\star}(\i a ) 1_{a}(\adjoint(Y))  \right)  \frt   \label{eq: limiting form three terms}
\eeq
This limit can be straightforwardly justified if the function $t \to \caL^*(t)$ is (norm-)integrable, which will follow from Assumptions \ref{ass: space} and \ref{ass: analytic form factor} in our case. Indeed, if  $t \to \caL^*(t)$ is integrable, then the integral $\int_0^{\la^{-2}\frt-t_1} \d u \ldots$ in  \eqref{eq: third term spectral averaging} converges to $\int_0^{\infty} \d u \ldots$ for fixed $t_1$, as $\la \to 0$. This yields the Laplace transform $\caL^{\star}(\i a )$.  The restriction $a=a'$ appears then because 
\beq
\la^2  \int_{0}^{\la^{-2 } \frt} \d t_1  \,  \e^{ \i (t - t_1)  (a-a')}  \qquad \mathop{\rightarrow}\limits_{\la \to 0}  \qquad \delta_{a,a'} \, \frt  
\eeq
and one finishes the argument by invoking dominated convergence. Comparing with \eqref{eq: def of M} and using that $(1_a(\adjoint(Y)))^{\star} =1_{-a}(\adjoint(Y))$, one checks that  \eqref{eq: limiting form three terms} is  equal to $1+\frt \caM^{\star}$, with $\caM^{\star}$  the dual to $\caM$.  This is the beginning of a series defining the semigroup $\e^{\frt \caM^{\star}}$ (we got only the first two terms because we kept only terms of order $\la^0$ and $\la^2$ in the original Lie-Schwinger series).

 \subsection{Momentum space representation of  $\caM$} \label{sec: momentum representation of caM}
 
 In this section, we give an explicit and intuitive expression for the operator $\caM$. As  $\caM$ is translation covariant, i.e., $\caT_{z} \caM \caT_{-z}=\caM$, as in   \eqref{translation inv of Z}, we have the fiber decomposition, 
\beq \label{eq: decomposition of caM}
\caM = \mathop{\int}\limits_{\tor}^{\oplus} \d p  \, \caM_p
\eeq
where the notation is as introduced in Section \ref{sec: translation invariance and fiber decomposition}.
We describe $\caM_p$ explicitly as an operator on  $ \scrG$. It is of the form
 \beq \label{def: caMp}
( \caM_p \xi)(k) = -\i [\Upsilon, \xi(k)] - \i (\varepsilon(k+\frac{p}{2})-\varepsilon(k-\frac{p}{2})) \xi(k) + (\caN \xi)(k), \qquad \xi \in \scrG
 \eeq
 where $\ve$ is the dispersion law of the particle, see Section \ref{sec: particle}, and $\Upsilon$ is a self-adjoint matrix in $\scrB(\scrS)$ whose only relevant property is that it commutes with $Y$, i.e., $[Y,\Upsilon]=0$.
 Physically, it describes the Lamb-shift of the internal degrees of freedom due to the coupling to the reservoir and its explicit form is given in Appendix \ref{app: caM}. The operator $\caN$ is given, for $\xi \in \caC(\tor, \scrB_2(\scrS))$, by
 \beq \label{def: caN}
 (\caN \xi)(k) =  \sum_{a \in \sp (\adjoint(Y))}  \int \d k'  \left( r_a(k',k)   W_{a}   \xi(k')  W^*_{a} -  \frac{1}{2}r_a(k,k') \left( \xi(k) W^*_{a} W_{a}+ W^*_{a} W_{a} \xi(k)   \right)  \right)  
 \eeq
 with the (singular) jump rates 
 \beq 
 r_{a}(k,k') :=2 \pi \mathop{\int}\limits_{\bbR^d} \d q  \,  \str\phi(q)\str^2   \,   \left\{ \begin{array}{lr}   \frac{ 1}{1-\e^{-\be \om(q)}}      \delta(\om(q)-a) \delta_{\tor}(k-k'-q  )  & \quad  a \geq 0   \\[3mm]
                     \frac{1}{\e^{\be \om(q)}-1}      \delta(\om(q) +a) \delta_{\tor}( k-k'+q  )  & \quad  a < 0  \end{array} \right.    \label{def: singular jump rates}
 \eeq
where $\phi$ is the form-factor, see Section \ref{sec: reservoir}, and $\delta_{\tor}(\cdot)$ is  a sum of  Dirac delta distributions on the torus;
\beq
\delta_{\tor}(\cdot) := \sum_{q_0=0 + (2\pi \bbZ)^d} \delta(\cdot-q_0), \qquad 
\eeq
 Note that  $ r_{a}(\cdot,\cdot) $ vanishes at $a=0$, due to the fact that the 'effictive squared form factor' $\hat \psi(\cdot)$ vanishes at $0$, see Assumption \ref{ass: analytic form factor}.

Equation \eqref{def: caMp} is most easily checked starting from the expressions for $\caM$ in Section \ref{sec: construction of caM}. In particular, the three terms in \eqref{def: caN} correspond to the fiber decompositions of the operators $\Phi(\rho),-  \frac{1}{2} \Phi^{\star}(1) \rho,  -  \frac{1}{2}\rho \Phi^{\star}(1)$ in \eqref{eq: general form generator}, and the first two terms on the RHS of \eqref{def: caMp} correspond to the commutator with $\Upsilon$ and $\varepsilon(P)$ in \eqref{eq: general form generator}.

We already stated that $\caM$ is translation-invariant, hence it commutes with $\adjoint(P)$. However, the operator $\caM$ also commutes with  $\adjoint (Y)$, as can be easily checked starting from the expressions \eqref{def: caMp} and \eqref{def: caN}  and employing the definitions of $W_a$ in \eqref{def: Wa} and the fact that $[Y, \Upsilon]=0$.  

 We can therefore  construct the double decomposition
 \beq
 \caM= \mathop{\oplus}\limits_{a \in \sp( \adjoint(Y))}  \int_{\tor}^{\oplus}\d p\,    \caM_{p,a}     
 \eeq
 where 
 \beq
\caM_{p,a}     :=     1_{a}(\adjoint(Y))   \caM_{p}     1_{a}(\adjoint(Y))
 \eeq
To proceed, we make use of our strong nondegeneracy condition in Assumption \ref{ass: fermi golden rule}.
Indeed, the operators $ \caM_{p,a}$ act on functions $\xi \in \scrG $ that satisfy the constraint
\beq  \label{eq: spectral constraint a}
\xi(k) = 1_a(\adjoint(Y)) \xi(k) =  \sum_{e,e' \in \sp Y, e-e'=a}  1_e(Y)   \xi(k)     1_{e'}(Y), \qquad   \xi \in \scrG \sim  L^2(\tor, \scrB_2(\scrS)) 
\eeq
Due to the  non-degeneracy assumption, the sum on the RHS contains only one non-zero term for $a \neq 0$,  i.e., there are unique eigenvalues $e,e'$ such that $a=e-e'$.
Let us denote the eigenvector in the space $\scrS$ of  the operator $Y$ with eigenvalue $e$ by $e$ as well (cfr.\ the discussion following Assumption \ref{ass: fermi golden rule}) , then this unique term in \eqref{eq: spectral constraint a} can be written as
\beq
1_e(Y)   \xi(k)     1_{e'}(Y)  =   \left(\langle e, \xi(k) e' \rangle\right) \,    \str e\rangle \langle e' \str, \qquad    e-e'=a
  \eeq
It follows that  the matrix valued function $\xi(k)$ satisfying \eqref{eq: spectral constraint a} can be identified with the $\bbC$-valued function
\beq \varphi(k)\equiv \langle e, \xi(k) e' \rangle_{\scrS} 
\label{def: varphi}
 \eeq 
For $a=0$,  a function $\xi(k)$ satisfying \eqref{eq: spectral constraint a} is necessarily diagonal in the basis of eigenvectors of $Y$. In that case, we can identify $\xi$ with
 \beq
\varphi(k,e) \equiv  \langle e, \rho(k,k) e \rangle_{\scrS}  \label{eq: identification with prob measure}
 \eeq
Hence, we  can identify
 $\caM_{p,a \neq 0}$ with an operator on $L^2(\tor)$ and $\caM_{p,0}    $ with an operator on $L^2(\tor \times \sp Y)$. 
 
 A careful analysis of these operators is performed in Appendix \ref{app: caM}. Here, we discuss the operator $\caM_{0,0}$ because it is crucial for understanding our model.
 
 \subsubsection{The Markov generator $\caM_{0,0}    $}
Let us choose $\varphi \in \caC(\tor \times \sp Y)$. Then, by the formulas given above, the operator $\caM_{0,0}    $ acts as
\beq \label{eq: explicit caM00}
\caM_{0,0} \varphi(k,e)   :=   \int_{\tor} \d k' \sum_{e' \in \sp Y}    \left(   r(k',e';k,e)    \varphi(k',e')  -      r(k,e;k',e')  \varphi(k,e)  \right)
\eeq
where $r(k,e;k',e')$ are (singular) \emph{transition rates}  given explicitly by 
\beq
r(k,e;k',e')   :=    r_{e-e'}(k,k')\str \langle e', W e \rangle \str^2
\eeq 
In formula \eqref{eq: explicit caM00}, one recognizes the structure of a Markov generator, acting on  densities of absolutely continuous probability measures (hence $L^1$-functions) on $\tor \times \sp Y$.
The numbers
\beq
j(e,k) := \int_{\tor} \d k' \sum_{e' \in \sp Y}     r(k,e;k',e')   \label{def: escape rates}
\eeq
are called \emph{escape rates} in the context of Markov processes. 
Let $\norm \varphi \norm_1= \int_{\tor} \d k \sum_{e\in \sp Y} \str \varphi(k,e) \str $. Using that $r(k,e;k',e')  \d k' $ is a positive measure, we get
\beq
\norm \caM_{0,0} \varphi \norm_1  \leq   2  \left(\sup_{e,k}  j(e,k)  \right)    \norm  \varphi \norm_1, 
\eeq
which implies that $\caM_{0,0}$ is bounded on $L^1(\tor \times \sp Y)$. In particular, this means that $\caM_{0,0}$  is a bonafide Markov generator (i.e.\ it generates a strongly continuous (in our case even norm-continuous) semigroup)  and 
 $\e^{t \caM_{0,0}}\varphi$ is a probability density for all $t\geq 0$.
 Physically speaking, the probability density $\varphi$  is read off from the diagonal part of the density matrix $\rho$, see \eqref{eq: identification with prob measure}.

We note that the transition rates $r(k,e;k',e') $ satisfy the 'detailed balance property' at inverse temperature $\be$ for the internal energy levels $e,e'$, and at infinite temperature for the momenta $k,k'$;
\beq \label{eq: detailed balance}
r(k,e;k',e')   =    \e^{\be(e-e')} r(k,'e';k,e) 
\eeq 
Physically, we would expect overall detailed balance at inverse temperature $\be$, i.e.\ 
\beq
   r(k,e;k',e')   =     \e^{\be( E(e ,k)-E(e' ,k')) }  r(k,'e';k,e) 
\eeq
where the energy $E(k,e)$ should depend on both $e$ and $k$. 
To understand why $E$ does not depend on $k$ in \eqref{eq: detailed balance}, we recall that the kinetic energy of the particle is assumed to be of order $\la^2$;  hence, the total energy is $ e +  \la^2\ve(k)$ which reduces to $e$ in  zeroth order in $\lambda$. 

One can  associate an intuitive picture with the operator $\caM_{0,0}$.  It describes the stochastic evolution of a particle with momentum $k$ and energy $e$. The state of the particle changes from $(k,e)$ to $(k',e')$ by emitting and absorbing reservoir particles with momentum $q$ and energy $\om(q)$, such that total momentum and total energy (which does not include any contribution from $k,k'$) are conserved, see Figure \ref{fig: markov}.

\begin{figure}[h!] 
\vspace{0.5cm}
  \centering
\psfrag{inmomentum}{ $k$}
\psfrag{outmomentum}{ $k'$}
\psfrag{emittedmomentum}{ $q$}
\psfrag{absorbedmomentum}{ $q$}
\psfrag{inlevel}{ $e$}
\psfrag{outlevel}{ $e'$}
  \subfloat[Emission of a boson, $k=k'+q$ and $e=\om(q)+e'$.] {\label{fig: markov 1} \includegraphics[width = 7cm, height=3.5cm]{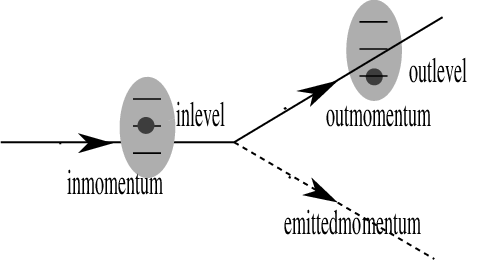}     } \hspace{0.5cm}
 \subfloat[Absorption of a boson, $k+q=k'$ and $\om(q)+ e=e'$.]   {\label{fig: markov 2} \includegraphics[width = 7cm, height=3.5cm]{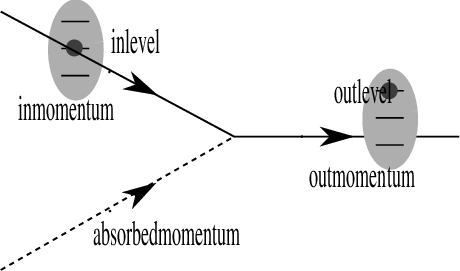}    }
\caption{ \footnotesize{The processes contributing to the gain term (the first term on the RHS in \eqref{eq: explicit caM00} ) of the operator $\caM_{0,0}$. Emission corresponds to $e>e'$ and absorption to $e<e'$. }  }
\label{fig: markov}

\end{figure}

It is clear from the collision rules in Figure \ref{fig: markov} that,  in the absence of internal degrees of freedom, the particle can only emit or absorb bosons with momentum $q=0$, and hence it cannot change its momentum.  This means that without the internal degrees of freedom, the semigroup $\La_t$ would not exhibit any diffusive motion. This is indeed the reason why we introduced these internal degrees of freedom.

\subsection{Asymptotic properties of the semigroup} \label{sec: properties of semigroup}
The following Proposition \ref{prop: properties of m} states some spectral results on the Lindblad operator $\caM$ and its restriction to momentum fibers $\caM_p  \in \scrB(\scrG)$.  These results are stated in a way that mirrors, as closely as possible, the statements of Theorem \ref{thm: main}.  

 These results are useful for two purposes. First of all, they show that our main physical  results, Theorems \ref{thm: diffusion} and \ref{thm: equipartition},  hold true if one replaces the reduced evolution $\caZ_t$ by the semigroup $\La_t$ (see the remark following Proposition \ref{prop: properties of m}).   
Second, a bound which follows directly from  Proposition \ref{prop: properties of m} will be a crucial ingredient in the proof of our main result Theorem \ref{thm: main}.  This bound is stated in \eqref{eq: bounds on markov dynamics} in Section \ref{sec: bound on markov in position representation}.   

We introduce the following sets (cf.\ (\ref{def: domain low}-\ref{def: domain high})) 
\baq
 \frD^{\footnotesize{low}}_{rw}  :=     \left\{ p \in \tor+\i \tor, \nu  \in \tor+\i \tor \Big\str \, \str  \Re p \str < p^*_{rw} ,\,  \str \Im p \str  \leq \de_{rw} , \,    \str \Im \nu \str \leq \de_{rw}      \right\}    \label{def: set D low markov}\\ [2mm]
 \frD^{\footnotesize{high}}_{rw} :=   \left\{ p \in \tor+\i \tor, \nu  \in \tor+\i \tor \Big\str \, \str  \Re p \str >  \frac{1}{2}p^*_{rw} ,\,  \str \Im p \str  \leq \de_{rw} , \,    \str \Im \nu \str \leq \de_{rw}      \right\}  \label{def: set D high markov},
\eaq
depending on positive parameters $p^*_{rw}>0$ and $\de_{rw}  >0$. The subscript  `$rw$' stands for `random walk' and it will always be used for objects related to $\La_t$.
  
 \begin{proposition}\label{prop: properties of m}
Assume Assumptions \ref{ass: analytic dispersion}, \ref{ass: space}, \ref{ass: analytic form factor} and \ref{ass: fermi golden rule}. There are positive constants $p^*_{rw} >0$ and $\de_{rw}>0$, determining  $ \frD^{\footnotesize{low}}_{rw} , \frD^{\footnotesize{high}}_{rw} $ above, such that the following properties hold 
\ben
\item{ For small fibers $p$, i.e., such that  $(p,\nu) \in  \frD^{\footnotesize{low}}_{rw} $, the  operator 
$
U_\nu\caM_p U_{-\nu} 
$ is bounded and 
has a simple eigenvalue $f_{rw}(p)$, independent of $\nu$, 
 \beq
     \sp ( U_\nu\caM_p U_{-\nu})= \{f_{rw}(p) \} \cup \Om_{p,\nu} 
 \eeq

The eigenvalue $f_{rw}(p)$ is elevated above the rest of the spectrum, uniformly in $p$, i.e., 
there is a positive $g_{rw} ^{\footnotesize{low}}  >0$ such that 
 \beq
\mathop{\sup}\limits_{(p,\nu) \in  \frD^{\footnotesize{low}}_{rw}}  \Re  \Om_{p,\nu}    <-   g_{rw} ^{\footnotesize{low}}  <   \mathop{\inf}\limits_{(p,\nu) \in  \frD^{\footnotesize{low}}_{rw}}  \Re f_{rw}(p) \leq 0
 \eeq
The one-dimensional spectral projector $U_\nu P_{rw}(p) U_{-\nu}$  corresponding to the eigenvalue $f_{rw}(p)$, is uniformly bounded:
\beq  \label{eq: bounds on projector markov}
\mathop{\sup}\limits_{ (p,\nu) \in  \frD^{\footnotesize{low}}}  \norm  U_{\nu}P_{rw}(p) U_{-\nu} \norm \leq C
\eeq
}
\item{ For large fibers $p$, i.e., such that  $(p,0) \in  \frD^{\footnotesize{high}}_{rw}$, the operator  $U_{\nu} \caM_p U_{-\nu} $  is bounded and its spectrum lies entirely below the real axis, i.e.,
\beq
\mathop{\sup}\limits_{(p, \nu) \in \frD^{\footnotesize{high}}_{rw}   }  \Re \sp \left(U_{\nu} \caM_p U_{-\nu} \right) < -g^{\footnotesize{high}}_{rw}, \qquad \textrm{for some}\, \,  g^{\footnotesize{high}}_{rw}>0
\eeq

}
\item{The function $f_{rw}(p)$, defined for all $p$ such that $(p,0) \in \frD_{rw}^{low}$, has a negative real part,  $ \Re f_{rw}(p) \leq 0$, and  satisfies 
 \beq  f_{rw}(p=0)=0, \qquad  \textrm{and} \qquad    \nabla_p f_{rw}(p)) \big\str_{p=0}=0   \label{eq: equalities f rw} \eeq 
\beq  \textrm{The Hessian} \quad D_{rw}:= (\nabla_p)^2 f_{rw}(p)\big\str_{p=0}  \quad \textrm{has real entries and is strictly positive}  \eeq

The spectral projector $P_{rw}(p=0)$ is given by
\beq
 P_{rw}(p=0) = \str \tilde\xi^{eq}_{rw} \rangle  \langle  \xi^{eq}_{rw} \str ,
\eeq
with
 \beq
\tilde\xi^{eq}_{rw}(k)=1_{\scrB(\scrS)}, \qquad \textrm{and} \qquad \xi^{eq}_{rw} (k)=  \frac{1}{(2\pi)^d} \frac{\e^{-\be Y} }{\Tr(\e^{-\be Y})}, \qquad \qquad  k \in \tor
 \eeq
 }
\een
\end{proposition}
The  conclusion of Proposition \ref{prop: properties of m}  is sketched in Figure \ref{fig: spectrum of M}. The proof of this proposition is very analogous to the proof in \cite{clarkderoeckmaes} (which, however, does not consider internal degrees of freedom). For completeness, we reproduce the proof in Appendix \ref{app: caM}.
\begin{figure}[h!] 
\vspace{0.5cm}
  \centering
\psfrag{gaplow}{ $g_{rw}^{low}$}
\psfrag{gaphigh}{ $g_{rw}^{high}$}
\psfrag{functionf}{ $f_{rw}(p)$}
\psfrag{spectrumM}{spectrum of $\caM_p$ (real part)}
\psfrag{momentump}{fiber $p$}
\psfrag{pcritical}{ $p_{rw}^*$}
 \includegraphics[width = 12cm, height=8cm]{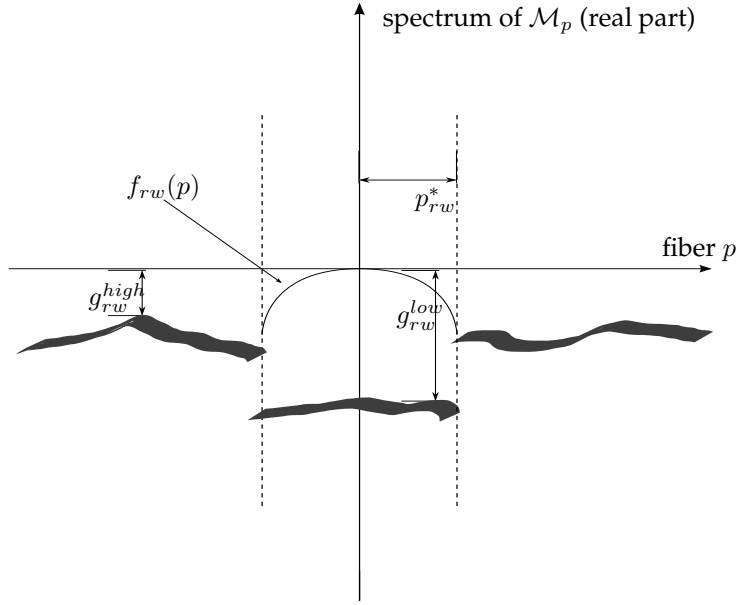}    
\caption{ \footnotesize{The spectrum of $\caM_p$ as a function of the fiber momentum $p$. Above the irregular black line, the only spectrum consists of the isolated eigenvalue $f_{rw}(p)$, in every small fiber $p$. Below the irregular black lines, we have no control. }  }
\label{fig: spectrum of M}
\end{figure}

 From Proposition \ref{prop: properties of m}, one can derive that the semigroup $\e^{t (-\i \adjoint(Y)+\caM)}$ exhibits diffusion, decoherence and equipartition. This follows by analogous reasoning as   in Sections
\ref{sec: diffusion connection}, \ref{sec: equipartition connection} and \ref{sec: decoherence connection}, but starting from Proposition \ref{prop: properties of m} instead of Theorem  \ref{thm: main}. The matrix $D_{rw}$ is the diffusion constant, the inverse decoherence length has to be chosen smaller than $\de_{rw}$ and the function $\xi^{eq}_{rw}$ is the 'equilibrium state'. 
We do not state these properties explicitly as they are not necessary for the proof of our main results.

\subsubsection{Bound on $\La_t$ in position representation}   \label{sec: bound on markov in position representation}
 By virtue of Proposition \ref{prop: properties of m}, we can write
 \baq
\left\{\La_t\right\}_p &=&     P_{rw}(p)   \e^{\la^2 f_{rw}(p)t }   +       R^{\footnotesize{low}}_{rw} (t,p)    \e^{-\la^2 g^{\footnotesize{low}}_{rw}  t }, \qquad   (p,\nu) \in \frD_{rw} ^{\footnotesize{low}}     \label{eq: expressions for Lambda1} \\[2mm]
 \left\{\La_t\right\}_p   &=&    R^{\footnotesize{high}}_{rw}(p,t)    \e^{-\la^2 g^{\footnotesize{high}}_{rw}  t }, \qquad  \qquad  \qquad  (p,\nu) \in \frD_{rw} ^{\footnotesize{high}},     \label{eq: expressions for Lambda2}  
 \eaq
 with $P_{rw}(p)$ as defined above, satisfying \eqref{eq: bounds on projector markov},  and the  operators $R^{\footnotesize{low}}_{rw} , R^{\footnotesize{high}}_{rw} $ satisfying
 \baq
\mathop{ \sup}\limits_{ (p,\nu) \in \frD_{rw} ^{low}}  \mathop{ \sup}\limits_{t \geq 0}   \norm   U_{\nu}R^{\footnotesize{low}}_{rw}(t,p)  U_{-\nu}  \norm < C   \label{eq: bound on Rlow markov}  \\[2mm]
\mathop{ \sup}\limits_{(p,\nu) \in \frD_{rw} ^{high}}  \mathop{ \sup}\limits_{t \geq 0}   \norm   U_{\nu}R^{\footnotesize{high}}_{rw}(t,p)  U_{-\nu}  \norm  <C \label{eq: bound on Rhigh markov} 
 \eaq
The appearance of the factor $\la^2$ is due to the fact that $\la^2 $ multiplies $\caM$ in the definition of the semigroup $\La_t$.  

Next, we derive estimates on $\La_t$ (see e.g.\  the bound \eqref{eq: bounds on markov dynamics} below) starting from (\ref{eq: expressions for Lambda1}-\ref{eq: expressions for Lambda2}) and (\ref{eq: bound on Rlow markov}-\ref{eq: bound on Rhigh markov}), without using explicitly the semigroup property of $\La_t$. This is important since in the proof of Lemma \ref{lem: bounds on cutoff dynamics}, we will  carry out an analogous derivation for objects which are not semigroups.

We choose $\ka =(\ka^{}_{\links}, \ka^{}_{\rechts}) \in \bbC^d \times \bbC^d $ such that $\Re\ka^{}_{\links}= \Re \ka^{}_{\rechts}=0 $ and we calculate, using  relation \eqref{eq: relation kappa and fibers}, 
\beq
 \caJ_{\ka}\La_t \caJ_{-\ka} =  \mathop{\int}\limits_{\tor }^\oplus \d p \, U_{\nu} \left\{\La_t\right\}_{p+\Delta p} U_{-\nu}, \qquad \textrm{with}\, \,   \Delta p:= \frac{\ka^{}_{\links}-\ka^{}_{\rechts}}{2}   \quad  \nu:= \frac{\ka^{}_{\links}+\ka^{}_{\rechts}}{4}
\eeq
 where we use the analyticity in $(p,\nu)$, see \eqref{eq: bounds on projector markov} and (\ref{eq: bound on Rlow markov}-\ref{eq: bound on Rhigh markov}). 
 Recall that $\{ \La\}_p$ acts on $\scrG^p \sim \scrG \sim L^2(\tor, \scrB(\scrS))$.  Our choice for $\ka$ ensures that $\Delta p$ and $\nu$ are purely imaginary.
 
 Next, we split the integration over $p \in \tor$ into small fibers ($\str p\str < p_{rw}^*$)  and large fibers ($\str p\str \geq p_{rw}^*$) by defining
 \beq
 I^{low}  :=  \{ p+ \Delta p \, \big\str    p \in \tor,  \str p\str < p_{rw}^*   \},  \qquad    I^{high}  :=   \{ p+ \Delta p \, \big\str    p \in \tor,  \str p\str \geq p_{rw}^*   \} \eeq 
 
Using the relations \eqref{eq: expressions for Lambda1} and \eqref{eq: expressions for Lambda2}, we obtain
 \baq
 \caJ_{\ka}\La_t \caJ_{-\ka} &=&   \underbrace{\mathop{\int}\limits_{I^{low} }^\oplus  \d p \,   \e^{ \la^2 f_{rw}(p) t } \,  U_{\nu}P_{rw}(p)  U_{-\nu}  }_{=:B_1} \quad +\quad   \e^{-\la^2 g^{low}_{rw}   t } \underbrace{\mathop{\int}
\limits_{I^{low} }^\oplus \d p \,   U_{\nu}R^{low} _{rw} (p,t)  U_{-\nu}  }_{=:B_2}  \label{def: B1 and B2}  \\[2mm]
& +&  \quad \e^{-\la^2 g^{high}_{rw}  t } \underbrace{\mathop{\int}\limits_{I^{high} }^\oplus \d p \, U_{\nu}R^{high} _{rw} (p,t)   U_{-\nu}   }_{=:B_3} \nonumber
 \eaq
 
We  establish decay properties of the operators $B_{1,2,3}$ in position representation.  For $B_2$ and $B_3$ we proceed as follows. Recall the duality (\ref{eq: conjugate1}-\ref{eq: conjugate2}). 
By varying $p$ and $\nu$, and using the bounds (\ref{eq: bound on Rlow markov}-\ref{eq: bound on Rhigh markov}), we obtain
\beq
 \norm (B_{2,3})_{x^{}_{\links},x^{}_{\rechts};x'_{\links},x'_{\rechts}}\norm_{\scrS} \leq C \e^{-  \frac{\ga}{2} \left\str (x'_{\links}+x'_{\rechts})- (x_{\links}+x_{\rechts})\right\str}  \e^{- \ga \left\str (x_{\links}-x_{\rechts})- (x'_{\links}-x'_{\rechts})  \right\str }
\eeq
for any $\ga < \de_{rw}$. 

For $B_1$, we need a better bound, which is attained by exploiting the fact that $P_{rw}(p)$ is a rank-1 operator with a  kernel of the form (recall the notation of \eqref{def: ket bra notation})
\beq
P_{rw}(p)(k,k') =   \Big\str \left(\xi_{rw}(p)\right)(k) \Big\rangle \Big\langle (\tilde\xi_{rw}(p))(k') \Big\str,   \qquad \textrm{for some} \quad \xi_{rw}(p), \, \tilde \xi_{rw}(p)   \in \scrG  \label{eq: rank one as kernel markov}
\eeq
where both $ \xi_{rw}(p), \tilde \xi_{rw}(p)$ are bounded-analytic functions of $k,k'$, respectively, in a strip of width $\de_{rw} $.  This follows from boundedness and analyticity of $P_{rw}(p)$ by the same reasoning as in Lemma \ref{lem: boundedness eigenvectors}.
By the definition of $B_1$ in \eqref{def: B1 and B2}  and \eqref{eq: position representation of caWp}, \eqref{eq: relation kappa and fibers},
\baq
(B_1)_{x^{}_{\links},x^{}_{\rechts};x'_{\links},x'_{\rechts}} &=& \int_{I^{low}} \d p  \,     \e^{  \la^2 f_{rw}(p) t}   \e^{-\i  \frac{p}{2}((x'_{\links}+x'_{\rechts})- (x_{\links}+x_{\rechts}))}   \\[2mm]
&& \int_{\tor+\nu}\d k  \int_{\tor+\nu} \d k'      \e^{-\i k(x_{\links}-x_{\rechts} ) +\i k' (x'_{\links}-x'_{\rechts}) }    P_{rw}(p)(k,k')   
\eaq
By  the analyticity of $P_{rw}(p)(\cdot,\cdot) $ in $k,k',p$, we derive, for $\ga < \de_{rw}$,
\beq \label{eq: bound on B 1}
 \norm (B_1)_{x_{\links},x_{\rechts};x'_{\links},x'_{\rechts}}\norm_{\scrS} \leq  C  \e^{  r_{rw}(\ga,\la) t}    \e^{-  \frac{\ga}{2} \left\str (x'_{\links}+x'_{\rechts})- (x_{\links}+x_{\rechts})\right\str} \e^{- \ga \left\str x_{\links}-x_{\rechts}  \right\str }   \e^{- \ga \left\str x'_{\links}-x'_{\rechts}  \right\str }
\eeq
where the function $r_{rw}(\ga,\la) $ is given by
\beq
r_{rw}(\ga,\la) :=  \la^2 \mathop{\sup}\limits_{\str \Im p \str \leq \ga, \str \Re p \str \leq p_{rw} ^*} \max\left( \Re f_{rw}(p), 0 \right).
\eeq
Note that
\beq \label{eq: bound on r markov}
r_{rw}(\ga,\la) := O(\la^2)O(\ga^2), \qquad \la \searrow 0, \ga \searrow 0, 
\eeq
as follows from $\Re f_{rw}(p) \leq 0$.  The bound \eqref{eq: bound on r markov} will be used to argue that  the exponential blowup in time, given by $\e^{ r_{rw}(\ga,\la)  t} $ can be compensated by the decay $\e^{- \la^2 g_{rw} t}$ by choosing $\ga$ small enough, see Section \ref{sec: strategy}.

Putting the bounds on $B_{1,2,3}$ together, we arrive at 
\baq 
\norm (\La_t)_{x_{\links},x_{\rechts};x'_{\links},x'_{\rechts}} \norm &\leq&   
 C  \e^{  r_{rw}(\ga,\la) t}  \e^{- \frac{\ga}{2} \left\str (x'_{\links}+x'_{\rechts})- (x_{\links}+x_{\rechts})\right\str}   \e^{- \ga\left\str x_{\links}-x_{\rechts}  \right\str }   \e^{- \ga \left\str x'_{\links}-x'_{\rechts}  \right\str }   \nonumber\\[2mm]
  &+&      C'  \e^{-\la^2 g_{rw} t }   \e^{-  \frac{\ga}{2} \left\str (x'_{\links}+x'_{\rechts})- (x_{\links}+x_{\rechts})\right\str}  \e^{- \ga \left\str (x_{\links}-x_{\rechts})- (x'_{\links}-x'_{\rechts})  \right\str }
 \label{eq: bounds on markov dynamics}
\eaq
for  $\ga <\de_{rw} $ and with the rate $g_{rw} := \min{(g^{\footnotesize{low}}_{rw}  ,g^{\footnotesize{high}}_{rw}  )}$.    The bound \eqref{eq: bounds on markov dynamics} is the main result of the present section and it will be used in a crucial way in the proofs.    The importance of this bound is explained in Section \ref{sec: exp decay for dressed correlation}.

 For completeness, we note that a bound like 
 \beq
 \norm (\La_t)_{x_{\links},x_{\rechts};x'_{\links},x'_{\rechts}}  \norm \leq  
  \e^{2 \la^2 q_{\ve}(\ga)  t}   \e^{- \ga \left\str x'_{\links}- x_{\links}\right\str}  \e^{- \ga \left\str x'_{\rechts}-x_{\rechts}  \right\str }
 \eeq
  can be derived simply from the fact that $\caJ_{\ka} \caM \caJ_{\ka}$ is bounded for complex $\ka$, see \eqref{eq: caM has only propagation in hamiltonian},  since
  \baq
\ka^{}_{\links} \qquad &\textrm{is conjugate to} & \qquad (x'_\links- x^{}_\links)  \\
\ka^{}_{\rechts} \qquad &\textrm{is conjugate to} & \qquad (x'_\rechts- x^{}_\rechts),
\eaq

\section{Strategy of the proofs} \label{sec: strategy}

In this section, we outline our strategy for proving the results in Section \ref{sec: result}. We start by introducing and analyzing the space-time reservoir correlation function $\psi(x,t)$. Then we introduce a perturbation expansion for the reduced evolution $\caZ_t$ (which involves the reservoir correlation function).  Afterwards, we describe and motivate the temporal cutoff that we will put into the expansion. Finally, a plan of the proof is given.

\subsection{Reservoir correlation function} \label{sec: reservoir correlation function}

A quantity that will play an important role  in our analysis is the free reservoir correlation function $\psi(x,t)$, which we define next.
Let 
\beq
I_{\sys\res}^{x}:=   \int \d q    \left( \phi(q) \e^{\i q \cdot x} 1_x  \otimes  a_q+    \overline{\phi(q)}  \e^{-\i q \cdot x}  1_x \otimes a^*_q  \right)
\eeq
where $1_x=1_x(X)$ is the projection on $\scrH_\sys$, acting as $(1_x \varphi)(x')=\delta_{x,x'} \varphi(x)$ for $\varphi \in l^2(\bbZ^d, \scrS)$.  The operator $I_{\sys\res}^{x}$ is the part of the system-reservoir coupling that acts at site $x$ after setting the matrix $W\in \scrB(\scrS)$ equal to $1$ (recall that the matrix $W$ describes the coupling of the internal degrees of freedom to the reservoir). 
We also define the time-evolved interaction term, with the time-evolution given by the free reservoir dynamics
\beq
 I_{\sys\res}^{x}(t) := \e^{\i t H_\res}  I_{\sys\res}^{x} \e^{- \i t H_\res}  
\eeq
  The reservoir correlation function $\psi$ is then defined as 
\baq
\psi(x,t) &:=&    \initialres  \left[    I_{\sys\res}^{x}(t)  I_{\sys\res}^{0}(0)   \right], \qquad  \qquad    \nonumber  \\
&=&   \langle  \phi^x , T_\beta \e^{\i t \omega } \phi \rangle_{\frh}   +      \langle\phi^x , (1+T_\beta) \e^{- \i t \omega } \phi\rangle_{\frh}   \nonumber  \\
&=&   \mathop{\int}\limits_\bbR \d \om \hat \psi(\om)     \e^{\i \om t}    \mathop{\int}\limits_{\bbS^{d-1}} \d s \,  \e^{\i \om s \cdot x}    \label{def: correlation function}
\eaq
where $(\phi^x)(q):= \e^{\i q\cdot x}\phi(q)$ and $\bbS^{d-1}$ is the $d-1$-dimensional hypersphere of unit radius. The `effective squared form factor' $\hat \psi$ was introduced in \eqref{def: psi}, and the density operator $T_\be$ in Section \ref{sec: thermal states}.

Assumptions \ref{ass: space} and  \ref{ass: analytic form factor} imply certain properties of the correlation function that will be primary ingredients of the proofs. We state these properties as lemmata. In fact, one could treat these properties as the very assumptions of our paper, since, in practice,  Assumptions \ref{ass: space} and  \ref{ass: analytic form factor}   will only be used to guarantee these properties, Lemmata \ref{lem: exponential decay}  and \ref{lem: integrability}. The straightforward proofs of Lemmata \ref{lem: exponential decay}  and \ref{lem: integrability} are postponed to Appendix \ref{app: reservoirs}.
 
The following lemma states that the free reservoir has exponential decay in $t$ whenever  $\str x\str/t$ is smaller than some speed $v_*$.
\begin{lemma}[Exponential decay at `subluminal' speed] \label{lem: exponential decay}
Assume Assumptions \ref{ass: space} and \ref{ass: analytic form factor}. Then there are positive constants $v_*>0, g_{\res}>0$ such that 
\beq \label{eq: decay at subluminal speed}
\str \psi(x,t)  \str   \leq   C \exp{ ( - g_\res \str t \str ) }  , \quad \textrm{if} \quad  \frac{\str x \str}{t}  \leq   v_*, \qquad  \textrm{for some constant}\, C.
\eeq	
\end{lemma}
 Property \eqref{eq: decay at subluminal speed} is  satisfied if the reservoir is 'relativistic', i.e., if the dispersion law $\om(q)$ of the reservoir particles is linear in the momentum $\str q\str$, temperature $\be^{-1}$ is positive and the form factor $\phi$ satisfies the infrared regularity condition  that $k \mapsto \str\phi(k) \str^2  \str k \str$ is analytic in a strip around the real axis.
The speed $v^*$ has to be chosen strictly smaller than the propagation speed of the reservoir modes given by the slope of $\om$.  In fact, the decay rate $g_{\res}$ vanishes when  $v^*$ approaches the propagation speed of the reservoir modes.  Lemma \ref{lem: exponential decay} does \textbf{not} depend on the fact that the dimension $d \geq 4$.
  
\begin{lemma}[Time-integrable correlations]  \label{lem: integrability}
Assume Assumptions \ref{ass: space} and \ref{ass: analytic form factor}. Then
\beq
   \int_{\bbR^+} \d t  \sup_{x \in \lat}\str \psi(x,t)  \str    < \infty
   \eeq
\end{lemma}
 This property is  satisfied for non-relativistic reservoirs, with $\om(q) \propto \str q\str^2$, in $d\geq 3$ and for relativistic reservoirs, with $\om(q) \propto \str q \str$, in $d\geq 4$, provided that we choose the coupling to be sufficiently regular in the infrared. 

\qed

\subsection{The Dyson expansion}\label{sec: dyson expansion}

In this section, we set up a convenient notation to handle the Dyson expansion introduced in Lemma \ref{lem: definition dynamics}.

We define the group $\caU_{t}  $   on  $\scrB(\scrH_\sys)$ by
\beq \label{def: caU}
\caU_{t}S := \e^{-\i t H_\sys} S \e^{\i t H_\sys} , \qquad  S \in \scrB(\scrH_\sys),
\eeq
and the operators $\caI_{x,l}$, with $x \in \lat $ and $l \in \{\links, \rechts \}$ ($\links,\rechts$ stand for '\emph{left}' and '\emph{right}'), as
\beq \label{def: operatorI}
\caI_{x,l}S:= \left\{ \begin{array}{rll} -\i & (1_x \otimes W)  S & \qquad \textrm{if} \quad l= \links \\[1mm]  \i&
  S   (1_x  \otimes W)  & \qquad \textrm{if} \quad l = \rechts \\   \end{array} \right. \qquad   S \in \scrB(\scrH_\sys).
\eeq
where the operators $1_x \equiv 1_x (X)$ are projections on a lattice site $x \in \lat$, as  used in Section \ref{sec: reservoir correlation function}.

We write $(t_i,x_i,l_i), i=1,\ldots, 2n$ to denote  $2n$ triples in $\bbR \times \lat \times \{ \links, \rechts \}$ and we assume them to be ordered by the time coordinates, i.e., $t_i <t_{i+1}$. 
 We evaluate the Lie-Schwinger series \eqref{eq: lie schwinger}  using the properties (\ref{def: gauge invariance}-\ref{def: thermal state canonical}-\ref{eq: gaussian property}), and we arrive at
\beq\label{eq: Z with pairings}
 \caZ_{t}=     \sum_{n \in \bbZ^+}    \sum_{\pi \in \caP_n}   \,  \mathop{\int}\limits_{0 < t_1 < \ldots < t_{2n} <t } \d t_1 \ldots \d t_{2n}    \zeta(\pi,  (t_i,x_i,l_i)_{i=1}^{2n}   )  \caV_t(  (t_i,x_i,l_i)_{i=1}^{2n} )
\eeq
where $\pi \in \caP_n$ are pairings, as in (\ref{eq: gaussian property}), and we define
\beq \label{def: cav}
\caV_t( (t_i,x_i,l_i)_{i=1}^{2n}):=       \caU_{t-t_{2n}}     \caI_{x_{2n},l_{2n} }    \ldots \caI_{x_2,l_2}   \caU_{t_2-t_1} \caI_{x_1,l_1}  \caU_{t_1}
\eeq
with $\caU_t$ as in \eqref{def: caU} and 
\beq \label{def: zeta}
\zeta(\pi,  (t_i,x_i,l_i)_{i=1}^{2n})  := \prod_{(r,s) \in \pi}  \la^2    \left\{  \begin{array}{cr}   \psi (x_s-x_r, t_s-t_r )   & \qquad  l_{r}=l_s= \links \\ [1mm]  
\overline{\psi}(x_s-x_r, t_s-t_r )    & l_{r}=l_s= \rechts \\ [1mm]  
    \overline{ \psi} (x_s-x_r, t_s-t_r )  &    l_{r}= \links, l_s= \rechts \\ [1mm]  
\psi (x_s-x_r, t_s-t_r )   &    l_{r}= \rechts, l_s= \links \\ [1mm]  
  \end{array}\right.
\eeq
with the correlation function $\psi$ as defined in \eqref{def: correlation function}. We recall the convention $r<s$ for each element of a pairing $\pi$. \\
For $n=0$, the integral in \eqref{eq: Z with pairings} is meant to be equal to $\caU_t$.  
In Section \ref{sec: expansions},  we will introduce some combinatorial concepts to deal with the pairings $\pi \in \caP_n$ that are used in \eqref{eq: Z with pairings}. For convenience, we will replace the variables $(\pi, (t_i,x_i,l_i)_{i=1}^{2n}) \in \caP_n \times ([0,t]  \times \lat \times \{\links, \rechts \})^{2n}$ by a single variable $\si $ that carries the same information.

\begin{figure}[h!] 
\vspace{0.5cm}
\begin{center}
\psfrag{links}{$\links$ (left)}
\psfrag{rechts}{$\rechts$ (right)}
\psfrag{label0}{$0$}
\psfrag{labelt}{$t$}
\psfrag{label1}{  $\scriptstyle{(t_1,x_1, l_1)}$}
\psfrag{label2}{ $\scriptstyle{(t_2,x_2,l_2)}$}
\psfrag{label3}{ $\scriptstyle{(t_3,x_3, l_3)}$}
\psfrag{label4}{ $\scriptstyle{(t_4,x_4, l_4) }$}
\psfrag{label5}{ $\scriptstyle{(t_5,x_5, l_5) }$}
\psfrag{label6}{ $\scriptstyle{(t_6,x_6, l_6) }$}
\psfrag{label7}{ $\scriptstyle{(t_7,x_7, l_7) }$}
\psfrag{label8}{ $\scriptstyle{(t_8,x_8, l_8) }$}
\psfrag{label9}{ $\scriptstyle{(t_9,x_9, l_9) }$}
\psfrag{label10} { $\scriptstyle{(t_{10}, x_{10},, l_{10} )}$ }
\psfrag{label11} { $\scriptstyle{(t_{11}, x_{11}, l_{11} )}$ }
\psfrag{label12} { $\scriptstyle{(t_{12}, x_{12}, l_{12} )}$ }
\includegraphics[width = 18cm, height=6cm]{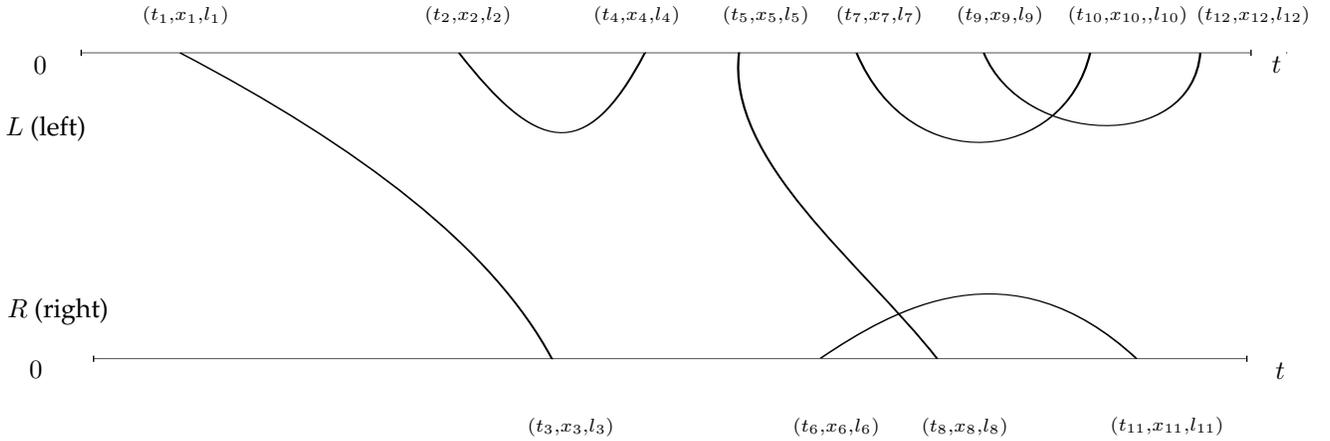}
\end{center}
\caption{ \footnotesize{Graphical representation of a  term  contributing to the RHS of \eqref{eq: Z with pairings} with $\pi=\{ (1,3), (2,4), (5,8 ), (6,10), (7,11), (9,12) \} \in \caP_6$. The times $t_i$ correspond to the position of the points on the horizontal axis.} \label{fig: explanation V} }
   \vspace{2mm}
     \begin{minipage}{18cm}
   \footnotesize{
 Starting from this graphical representation, we  can reconstruct the corresponding term in \eqref{eq: Z with pairings} - an operator on $\scrB_2(\scrH_\sys)$)-  as follows
\begin{itemize}
\item To each straight line between the points $(t_i,x_i,l_i)$ and $(t_{j},x_{j},l_{j})$, one associates the operators  $\e^{  \pm  \i     (t_{j}-t_i) H_\sys }$, with $\pm$ being $-$ for $l_{i}=l_{j}=\links$ and $+$ for $l_{i}=l_{j}=\rechts$. 
\item To each point $(t_i,x_i,l_i)$, one associates the operator $ \la^2 \caI_{x_i,l_i}$, defined in \eqref{def: operatorI}.
\item To each curved line between  the points $(t_r,x_r,l_r)$ and $(t_s,x_s,l_s)$, with $r<s$, we associate the factor $    \psi^{\#} (x_s-x_r, t_s-t_r)  $
with  $\psi^{\#}$ being $\psi$  or $\overline{\psi}$, depending on $l_r,l_s$, as prescribed  in \eqref{def: zeta}.
\end{itemize}
Rules like these are commonly called "Feynman rules" by physicists.
}
 \end{minipage}

\end{figure}

\subsection{The cut-off model} \label{sec: cut off model}
In our model, the space-time correlation function $\psi(x,t)$  does not decay exponentially in time, uniformly in space, i.e,  
\beq \label{eq: no exp decay correlation function}
\textrm{\it{there is no}} \quad g>0 \quad \textrm{ \it{such that}} \qquad   \sup_{x \in \lat}  \str \psi(x,t)   \str  \leq C  \e^{-g \str t\str } 
\eeq
\normalfont
The impossibility of choosing the form factor $\phi$ or any other model parameter such that one has exponential decay is a fundamental   consequence of  local momentum conservation, as explained in Section \ref{sec: related results}.

However, if the correlation function $\psi(x,t) $ did decay exponentially, we could set up a perturbation expansion for $\caZ_t$ around the Markovian limit $\La_t$. Such a scheme was implemented in  \cite{deroeckfrohlichpizzo}, building on an expansion introduced in \cite{deroeck}.

In the present section, we modify our model by introducing a cutoff time $\tau$ into the correlation function $\psi(x,t)$. More concretely, we modify the perturbation expansion  for $\caZ_t$ by replacing
\beq
\psi(x,t), \qquad   \longrightarrow \qquad     1_{\str t \str \leq \tau } \psi(x,t)
\eeq
The cutoff time $\tau$ will be chosen as a function of $\la$ satisfying 
\beq  \tau(\la) \to \infty,  \qquad  \tau(\la) \la  \to 0, \qquad  \textrm{as}\, \,  \la \searrow 0  \label{def: condition on tau} \eeq
However, we will take care to keep $\tau$ explicit in the estimates, until Section \ref{sec: bounds on long pairings} where the $\tau$-dependence will often be hidden in generic constants $\ctwo, \cthree$. 
 With the cut-off in place,  the  correlation function $\psi(x,t)$  decays exponentially, uniformly in $x \in \lat$, i.e., obviously,
\beq \label{eq: exp decay cutoff function}
\sup_{x \in \lat}    1_{\str t \str \leq \tau}  \str \psi(x,t) \str    \leq  C  \e^{-\frac{\str t  \str }{\tau}}.  \qquad  
\eeq
The modified reduced dynamics obtained in this way will be called  $\caZ^{\tau}_{t}$. 


 That is;  
\beq\label{eq: Z cutoff with pairings}  
\caZ^{\tau}_{t}=     \sum_{n \in \bbZ^+}    \sum_{\pi \in \caP_n}   \,  \mathop{\int}\limits_{0 < t_1 < \ldots < t_{2n} <t } \d t_1 \ldots \d t_{2n} \,  \zeta_\tau(\pi, (t_i,x_i,l_i)_{i=1}^{2n} )  \,  \caV_t((t_i,x_i,l_i)_{i=1}^{2n}  )
\eeq
with
\beq \label{def: zeta cutoff}
\zeta_{\tau}(\pi, (t_i,x_i,l_i)_{i=1}^{2n} )  := \left(\prod_{(r,s) \in \pi} 1_{\str t_s -t_r \str \leq \tau } \right)     \zeta(\pi, (t_i,x_i,l_i)_{i=1}^{2n} )
  \eeq

If $\tau$ is chosen to be  independent of $\la$ then one can analyze $\caZ^{\tau}_{t}$ by the technique deployed in  \cite{deroeckfrohlichpizzo}.  It turns out that for a $\la$-dependent $\tau$, one can still analyze the cutoff model by the same techniques  as long as $\la^{2} \tau(\la) \searrow 0 $ as $\la \searrow 0$, which is satisfied by our choice \eqref{def: condition on tau}. The analysis of $\caZ_t^\tau$ is outlined in Lemma \ref{lem: polymer model cutoff},  in Section \ref{sec: polymer model for cutoff dynamics}.  The main conclusion of the treatment of the cutoff model is that
\beq
\textrm{\it{The cutoff reduced dynamics}} \, \,  \caZ_t^\tau \, \,\textrm{\it{is 'close' to}} \, \, \textrm{\it{the semigroup}}  \, \,\La_t
\eeq
\normalfont
This conclusion is partially embodied in Lemma \ref{lem: bounds on cutoff dynamics}.  For the sake of this explanatory chapter, one can identify $\caZ_t^\tau$ with $\La_t$.

The reason why it is useful to treat the cutoff model  first, is that  we will perform a renormalization step, effectively replacing the free evolution $\caU_t$ in the expansion \eqref{eq: Z with pairings} by the cutoff reduced dynamics $\caZ_t^\tau$. The benefit of such a replacement is explained in Section \ref{sec: exp decay for dressed correlation}.

\subsection{Exponential decay for the renormalized correlation function} \label{sec: exp decay for dressed correlation}

\subsubsection{The joint system-reservoir correlation function} \label{sec: joint system reservoir correlation}
We recall that the free reservoir correlation function $\psi(x,t)$ does not decay exponentially in $t$, uniformly in $x$.   This was mentioned already in Section \ref{sec: cut off model} and it motivated the introduction of the temporal cutoff  $\tau$. 

In the perturbation expansion for the reduced evolution $\caZ_t$, the correlation function $\psi(x,t) $ models the propagation of reservoir modes over a space-time 'distance' $(x,t)$ and it occurs  together with terms describing the propagation of the particle.  Let us look at the lowest-order terms in the expansion of $\caZ_t$, introduced in Section \ref{sec: dyson expansion} above;
\beq
\caZ_t =  \la^2 \mathop{\int}\limits_{0<t_1<t_2 <t} \d t_2 \d t_1 \mathop{\sum}\limits_{x_1,x_2,l_1,l_2}  \,  \psi^{\#}(x_2-x_1,t_2-t_1)  \quad  \caU_{t-t_2} \caI_{x_2,l_2}   \caU_{t_2-t_1}  \caI_{x_1,l_1} \caU_{t_1} + \textrm{higher orders in $\la$}
\eeq 
with $ \psi^{\#}$ being $\psi$  or $\overline{\psi}$, as prescribed by the rules in \eqref{def: zeta}. 
It  is natural to ask whether the 'joint correlation function' 
\beq\label{eq: correlation with bare liouville}
 \psi^{\#}(x_2-x_1,t_2-t_1)  \quad   \caI_{x_2,l_2}   \caU_{t_2-t_1}  \caI_{x_1,l_1}
\eeq
has better decay properties than $\psi(x,t)$ by itself.  In particular, we ask whether \eqref{eq: correlation with bare liouville}
is exponentially decaying in $t_2-t_1$, uniformly in $x_2-x_1$. This turns out to be the case only if $l_1=l_2$ since in that case, the question essentially amounts to bounding
\beq \label{eq: correlation with bare propagator}
\left\str \psi(x_2-x_1,t_2-t_1) \right\str \times \left\norm \left(\e^{\pm \i (t_2-t_1) \la^2  \ve(P)} \right) (x_1,x_2)\right\norm
\eeq
The expression \eqref{eq: correlation with bare propagator} has exponential decay in time because 
\begin{itemize}
\item  For speed $\left\str \frac{x_2-x_1}{t_2-t_1} \right\str$ greater than some $v>0$, we estimate
\beq   \label{eq: decay for propagator at high speed}
 \left \norm  \left(\e^{\pm \i (t_2-t_1) \la^2  \ve(P)} \right) (x_1,x_2) \right\norm \leq  \e^{- (\ga v -\la^2 q_{\ve}(\ga)   )  \str t_2-t_1 \str },  \qquad \textrm{for} \, 0< \ga \leq \delta_{\ve}
\eeq 
with $\delta_{\ve}, q_{\ve}(\cdot)$ as in \eqref{eq: combes thomas} and Assumption \ref{ass: analytic dispersion}.  Hence, for fixed $v$, one can choose $\ga$ so as to make the exponent on the RHS of \eqref{eq: decay for propagator at high speed} negative, for $\la$ small enough.
\item For speeds $\left\str \frac{x_2-x_1}{t_2-t_1} \right\str$ smaller than $v^*>0$, the reservoir correlation function $\psi(x_2-x_1,t_2-t_1)$  decays with rate $g_\res$, as asserted in Lemma  \ref{lem: exponential decay} with  $v^*$ as defined therein.
\end{itemize}

%




%

When $l_1 \neq l_2$ in \eqref{eq: correlation with bare liouville},  there is no decay at all from $\caU_{t_2-t_1} $, in other words, the 'matrix element' 
\beq
\left(\caU_{t_2-t_1} \right)_{x^{}_{\links},x^{}_{\rechts};x'_{\links},x'_{\rechts}}
\eeq
is obviously not  decaying in the variables $x^{}_{\links}-x'_{\rechts}$ or $x^{}_{\rechts}-x'_{\links}$, since it is a function of $x^{}_{\links}-x'_{\links}$ and $x^{}_{\rechts}-x'_{\rechts}$ only. Hence, for $l_2 \neq l_1$, the joint correlation function \eqref{eq: correlation with bare liouville} has as poor decay properties as the reservoir correlation function $\psi(x,t)$. 

The situation is summarized in the following table
\begin{center}
\begin{tabular}{|c|c|c|}
\hline 
 \multicolumn{3}{|c|}{  joint $\sys-\res$ correlation fct.}  \\
 \hline 
  \multicolumn{2}{|c|}{  $\str x\str /t > v^*$}  &     $\str x\str/t \leq v^*$ \\
  \hline
   $ l_1 \neq l_2$ & $ l_1 = l_2$ &  $l_1, l_2 $ arbitrary  \\
\hline
   \parbox{3cm}{ No exp.\ decay  } &   \parbox{3cm}{\vspace{1mm} exp.\ decay from \\   $\left\norm \e^{\pm \i t H_\sys }(0,x)\right\norm $ \vspace{1mm} }&     \parbox{3cm}{  exp.\ decay from $\psi(x,t)$ } \\  
   \hline
\end{tabular}
\end{center}

\subsubsection{Renormalized joint correlation function}\label{sec: renormalized joint system reservoir correlation}

The bad decay property of the joint correlation function \eqref{eq: correlation with bare liouville} suggests
to perform a renormalization step, replacing the free propagator $\caU_t$ by the cutoff reduced dynamics $\caZ_t^\tau$, for which   \eqref{eq: correlation with bare liouville} has exponential decay when $l_1 \neq l_2$. The cutoff reduced dynamics $\caZ_t^\tau$ was introduced in Section \ref{sec: cut off model}, where we argued that it is well approximated by  the Markov semigroup $\La_t$. 
Hence, 
we   replace the group $\caU_t$ by the semigroup $\La_t$  in \eqref{eq: correlation with bare liouville}, thus obtaining a 'renormalized joint system-reservoir correlation function'. We then check that the so-defined renormalized correlation function has exponential decay in time, uniformly in space: For $\la$ small enough,
\beq \label{eq: bound for off-diagonal decay}
\left \str \psi (x'_{l_2}-x^{}_{l_1},t_2-t_1)\right\str  \times  \left\norm   \left( \La_{t_2-t_1} \right)_{x^{}_{\links},x^{}_{\rechts};x'_{\links},x'_{\rechts}}  \right\norm   \leq   \e^{- t \la^2   g_{rw}     } , \qquad  \textrm{for}  \qquad l_1,l_2 \in \left\{ \links, \rechts \right\}     \eeq
with the decay rate $g_{rw}$ as in \eqref{eq: bounds on markov dynamics}. 
To verify \eqref{eq: bound for off-diagonal decay}, we assume for concreteness  that $l_1=\links$ and $l_2=\rechts$, and we estimate by the triangle inequality
\beq
\left\str x'_\rechts -  x^{}_\links \right\str  \leq   \frac{1}{2} \left \str x'_\rechts -  x'_\links  \right\str+  \frac{1}{2}  \left\str( x'_\rechts +   x'_\links )-( x^{}_\rechts +   x^{}_\links )    \right\str  +\frac{1}{2} \left \str x^{}_\rechts -  x^{}_\links  \right\str \label{eq: triangle inequality distances}
\eeq
We note that the three terms on the RHS of \eqref{eq: triangle inequality distances} correspond (up to factors $ \frac{1}{2} $) to the three spatial arguments multiplying $\ga$ in the first line of \eqref{eq: bounds on markov dynamics}. By  \eqref{eq: triangle inequality distances}, at least one of these terms is larger than $\frac{1}{3} \left\str x'_\rechts -  x^{}_\links \right\str $. Hence we dominate  \eqref{eq: bounds on markov dynamics} by replacing that particular term by $\frac{1}{3} \left\str x'_\rechts -  x^{}_\links \right\str $. Setting all other spatial arguments in \eqref{eq: bounds on markov dynamics} equal to zero, we obtain
\beq \label{eq: bound markov left right}
\norm (\La_t)_{x_{\links},x_{\rechts};x'_{\links},x'_{\rechts}} \norm        \leq   C  \e^{  r_{rw}(\ga,\la) t}  \e^{- \frac{\ga}{6}   \left\str x'_\rechts -  x^{}_\links \right\str} +C' \e^{- (\la^2 g_{rw}) t}
\eeq
Assuming that $ \left\str x'_\rechts -  x^{}_\links \right\str \geq v^* \str t_2-t_1 \str$ and using that $ r_{rw}(\ga,\la)=O(\ga^2)O(\la^2)$, see \eqref{eq: bound on r markov}, we  choose $\ga$ such that the first term of \eqref{eq: bound markov left right} decays exponentially in $t_2-t_1$ with a rate or order $1$. Hence, at high speed ($\geq v^*$) \eqref{eq: bound for off-diagonal decay} is satisfied. 
 At low speed ($\leq v^*$), \eqref{eq: bound for off-diagonal decay} holds by the exponential decay of $\psi$ and the bound $\norm \La_t \norm \leq C\e^{ O(\la^2) t}$, which is easily derived from \eqref{eq: bounds on markov dynamics}.
   
For $l_1=l_2$, we can  apply the same reasoning, and hence \eqref{eq: bound for off-diagonal decay} is proven in general.   However, in the case $l_1=l_2$, the proof is actually simpler. We can follow the same strategy as used for bounding \eqref{eq: correlation with bare propagator}, but replacing the propagation estimate \eqref{eq: combes thomas} for $\caU_t$ by the analogous estimate \eqref{eq: caM has only propagation in hamiltonian} for $\La_t$.  
 Indeed, the  exponential decay in the case  $l_1=l_2$ was already present without the coupling to the reservoir, as explained in Section \ref{sec: joint system reservoir correlation}, whereas the decay in the case $l_1 \neq l_2$ is a nontrivial consequence of the decoherence induced by the reservoir.

\begin{center}
\begin{tabular}{|c|c|c|}
\hline
 \multicolumn{3}{|c|}{ renormalized $\sys-\res$ correlation fct.}  \\
 \hline
  \multicolumn{2}{|c|}{$\str x\str /t > v^*$}  &   $\str x\str/t \leq v^*$ \\
  \hline
   $ l_1 \neq l_2$ & $ l_1 = l_2$ &  $l_1, l_2 $ arbitrary  \\
\hline
 \parbox{3cm}{\vspace{1mm} exp.\ decay from \\  decoherence  of $ \La_t $ \vspace{1mm} } &   \parbox{3cm}{\vspace{1mm} exp.\ decay from   $\left\norm \e^{\pm \i t H_\sys }(0,x)\right\norm $ \vspace{1mm} }&     \parbox{3cm}{  exp. decay from $\psi(x,t)$ } \\  
   \hline
\end{tabular}
\end{center}

Along the same line, we note that the decay rate in \eqref{eq: bound for off-diagonal decay} cannot be made greater than $O(\la^2) $, since the effect of the reservoir manifests itself only after a time $O(\la^{-2})$.
This should be contrasted with the decay rate for  \eqref{eq: correlation with bare propagator}, which can be chosen to be independent of $\la$.

\subsection{The renormalized model} \label{sec: the renormalized model}
 
 We have argued in the previous section  that it makes sense to evaluate the perturbation expansion \eqref{eq: Z with pairings} in two steps by introducing a cutoff $\tau$ for the temporal arguments of the correlation function $\zeta$. 
 The  resulting cutoff reduced evolution $\caZ_t^\tau$ was described in Section \ref{sec: cut off model}. 
By reordering the perturbation expansion, we are able to rewrite the reduced evolution $\caZ_t$ approximatively as
 \beq \label{eq: renormalized dyson series}
\caZ_t \approx \la^2 \mathop{\int}\limits_{\footnotesize{\left.\begin{array}{c} 0< t_1<t_2 <t \\  \str t_2 -t_1 \str > \tau \end{array}\right.}  } \d t_2 \d t_1 \mathop{\sum}\limits_{x_1,x_2,l_1,l_2}  \,  \psi^{\#}(x_2-x_1,t_2-t_1)  \quad   \caZ^\tau_{t-t_2} \caI_{x_2,l_2}   \caZ^\tau_{t_2-t_1}  \caI_{x_1,l_1}  \caZ^\tau_{t_1} \, + \,  \textrm{higher orders in $\la$}
\eeq 
where the restriction that $t_2-t_1 > \tau$ reflects the fact that the short diagrams have been resummed.  Note that it is somewhat misleading to call the remainder of the perturbation series 'higher order in $\la$', since $\tau$ will be $\la$-dependent, too.

The main tools in dealing with the renormalized model are
\ben
\item{ The exponential decay of the renormalized joint correlation function, as outlined in Section \ref{sec: exp decay for dressed correlation}.  This property holds true thanks to the decoherence in the Markov semigroup $\La_t$ and the exponential decay for low (`subluminal') speed of the bare reservoir correlation function. The latter  is a consequence of the fact that the dispersion law of the reservoir modes is linear (see Lemma \ref{lem: exponential decay}).  The necessity of the exponential decay of the renormalized joint correlation function for the final analysis will become apparent in Lemma \ref{lem: bound min irr}.}
\item{ The integrability in time of the correlation function,  uniformly in space, as stated in Lemma \ref{lem: integrability}.  This property allows us to sum up all subleading diagrams in the renormalized model. This will be made more explicit in Section \ref{sec: from irreducible long to minimal}, in particular in Lemma \ref{lem: bound irr by min irr}.
}
 \een
 The most convenient  description of the renormalized model will be reached at the end of Section \ref{sec: bounds on long pairings} and the beginning of Section \ref{sec: model only long}, where a representation  in the spirit of \eqref{eq: renormalized dyson series} is discussed. The treatment of the renormalized model is contained in Section \ref{sec: model only long}.

 \subsection{Plan of the proofs}

 In Section  \ref{sec: polymer models}, we present the analysis of the cutoff reduced dynamics $\caZ_t^\tau$ and the full reduced dynamics $\caZ_t$, starting from bounds that are obtained in later sections. The main ingredient of this analysis is spectral perturbation theory, contained in Appendix \ref{app: spectral}. 
 
 In Section \ref{sec: expansions}, we introduce Feynman diagrams and we use them to derive convenient expressions for  the cutoff reduced dynamics $\caZ_t^\tau$ and the full reduced dynamics $\caZ_t$. We will distinguish between \emph{long} and \emph{short} diagrams. The cutoff reduced dynamics contains only short diagrams.
 
 Section \ref{sec: bounds on long pairings} contains the analysis of the short diagrams. In particular, we prove the bounds on $\caZ_t^\tau$, which were used in Section \ref{sec: polymer models}.
 
In Section \ref{sec: model only long}, we deal with the long diagrams. In particular, we prove the bounds on $\caZ_t$ from Section \ref{sec: polymer models}.
At the end of the paper, in Section \ref{sec: list of parameters}, we collect the most important constants and parameters of our analysis.
  A flow chart of the proofs is presented in Figure \ref{fig: planofproof1}.
 \begin{figure}[h!] 
\vspace{0.5cm}
\begin{center}
\psfrag{spectral}{\begin{tabular}{c}  \parbox{3cm}{Spectral perturbation theory:  Lemma \ref{lem: spectral abstract}   } \end{tabular} }
\psfrag{analysisLindblad}{\begin{tabular}{c}  \parbox{3cm}{ Spectral analysis $\caM$: Proposition \ref{prop: properties of m}  } \end{tabular} }
\psfrag{controlexcutoff}{\begin{tabular}{c}  \parbox{3cm}{ Control of short diagrams: Lemma \ref{lem: polymer model cutoff}  } \end{tabular} }
\psfrag{analysiscutoff}{\begin{tabular}{c}  \parbox{3cm}{Analysis of the cutoff model: Lemma \ref{lem: bounds on cutoff dynamics}  } \end{tabular} }
\psfrag{controlexful}{\begin{tabular}{c}  \parbox{3cm}{ Control of long diagrams: Lemma \ref{lem: polymer model full} } \end{tabular} }
\psfrag{analysisfull}{\begin{tabular}{c}  \parbox{3cm}{Analysis of the full model: Theorem \ref{thm: main} } \end{tabular} }
\psfrag{mainresult}{\begin{tabular}{c}  \parbox{3cm}{Main results: Theorems \ref{thm: diffusion} and \ref{thm: equipartition} } \end{tabular} }
\psfrag{integrability}{\begin{tabular}{c}  \parbox{3cm}{Properties of $\psi(x,t)$: Lemmata \ref{lem: integrability} and \ref{lem: exponential decay}} \end{tabular} }
\includegraphics[width = 16cm, height=12cm]{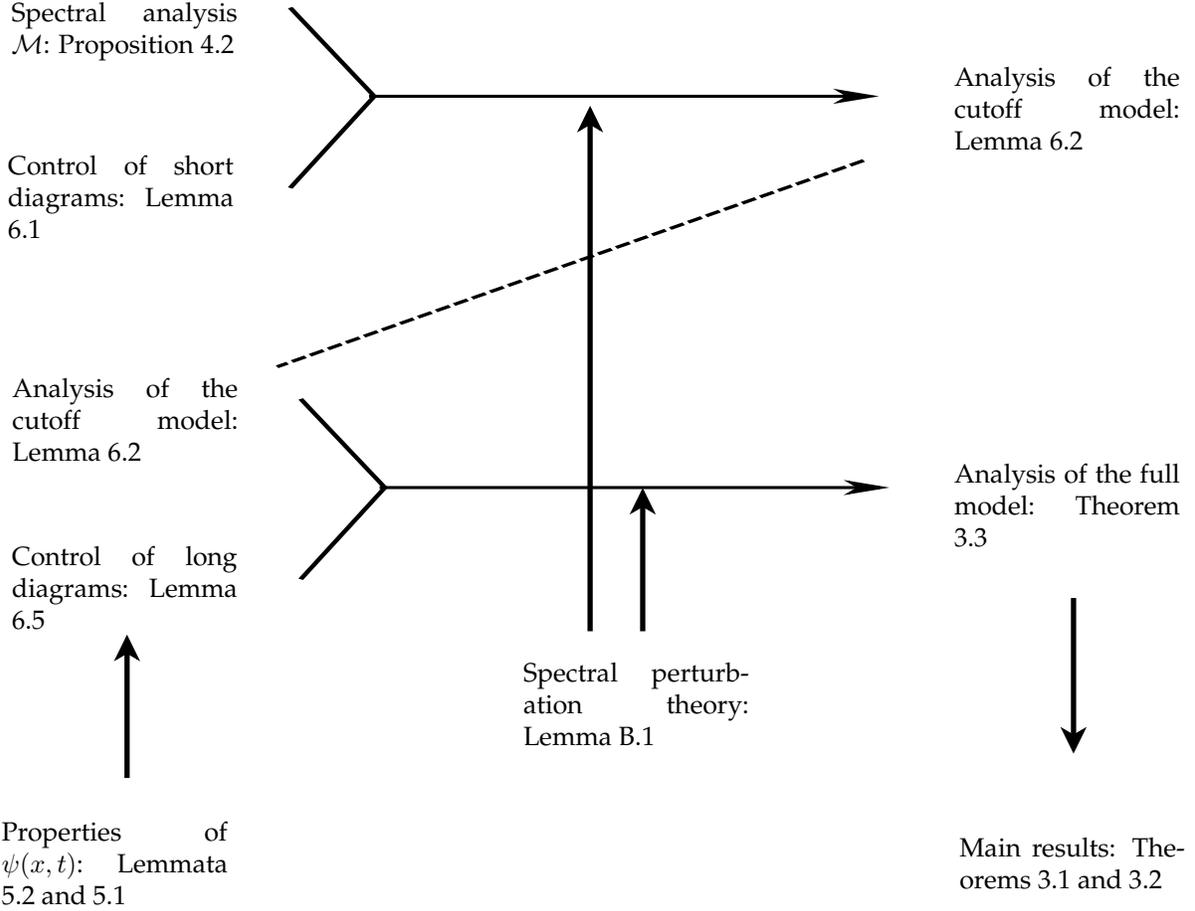}
\vspace{0.4cm}
\caption{A flow chart of the proofs. The arrows mean ``implies". The arrows pointing to arrows specify the proof of the implication.} \label{fig: planofproof1} 
\end{center}

\end{figure}

\section{Large time analysis of the reduced evolution $\caZ_t$ and the cutoff reduced evolution $\caZ_t^\tau$} \label{sec: polymer models}

In this section, we analyze the evolution operators  $\caZ_t$ and $\caZ_t^\tau$ starting from bounds on their Laplace transforms 
\beq
    \caR^{\tau}(z) := \mathop{\int}\limits_{\bbR^+} \d t  \, \e^{-tz }\caZ^{\tau}_{t}   \label{def: first laplace trafo cutoff}
\eeq
and
\beq
    \caR(z) := \mathop{\int}\limits_{\bbR^+} \d t  \, \e^{-tz }\caZ_{t}.    \label{def: first laplace trafo}
\eeq
These bounds are proven by diagrammatic expansions  in Sections  \ref{sec: expansions}, \ref{sec: bounds on long pairings} and \ref{sec: model only long}.
However, the present section is written in such a way that one can ignore these diagrammatic expansions and consider the bounds on $\caR^\tau(z)$ and $\caR(z)$ as an abstract starting point. Our results, Lemma \ref{lem: bounds on cutoff dynamics}, and Theorem \ref{thm: main}, follow from these bounds by an application of the inverse Laplace transform and  spectral perturbation theory.   For convenience, these tools are collected in  Lemma \ref{lem: spectral abstract}  in Appendix \ref{app: spectral}.

  \subsection{Analysis of $\caZ_t^\tau$} \label{sec: polymer model for cutoff dynamics}
  
  Our main tool in the study of $    \caR^{\tau}(z) $ is Lemma \ref{lem: polymer model cutoff} below.  Loosely speaking, the important consequence of this lemma is the fact  that we can represent the Laplace transform $\caR^{\tau}$, defined in \eqref{def: first laplace trafo cutoff}, as 
\beq  \label{eq: essence polymer cutoff}
 \caR^{\tau}(z) = (z - (- \i \adjoint(Y) +\la^2\caM+ \caA^\tau(z)))^{-1},
\eeq
where  the operator $ \caA^{\tau}(z)   $ is  ``small" wrt.\ $\la^2 \caM$, in a sense specified by the theorem.  Note that if we set $\caA^\tau(z)=0$, then the RHS of \eqref{eq: essence polymer cutoff} is the Laplace transform of the Markov semigroup 
$\La_t$.  This is consistent with the claim that $\caZ_t^\tau$ is 'close to' $\La_t$.

The subscripts 'ld' and 'ex', introduced below, stand for "ladder" and "excitations", respectively. These subscripts will acquire an intuitive meaning in  Section \ref{sec: expansions} when the diagrammatic representation of the expansion is introduced. The (sub)superscript $\tau$ indicates the dependence on the cutoff $\tau$, but sometimes we will also use the (sub)superscript $c$. This will be done for quantities that are designed for the cut-off model but that do not necessarily change when $\tau$ is varied.     Lemma \ref{lem: polymer model cutoff} can be stated for any $\tau$, but, as announced, it will be used for a $\la$-dependent $\tau$.

\begin{lemma}\label{lem: polymer model cutoff}
For $\la$ small enough, there are operators $ \caR_{\mathrm{ex}}^{\tau}(z)$ and  $\caR_{\mathrm{ld}}^{\tau}(z)$ in $\scrB(\scrB_2(\scrH_\sys))$, depending on $\la$ and $\tau$, satisfying the following properties:
\ben
\item{ For   $\Re z$ sufficiently large, the integral in \eqref{def: first laplace trafo cutoff} converges absolutely in $\scrB(\scrB_2(\scrH_\sys))$ and
\beq \label{eq: expansion for resolvent}
 \caR^{\tau}(z) = ( z -(-  \i \adjoint(H_\sys) + \caR_{\mathrm{ld}}^{\tau}(z) +\caR_{\mathrm{ex}}^{\tau}(z) ) )^{-1}.
\eeq 
}
\item{ The operators  $\caR_{\mathrm{ld}}^{\tau}(z),  \caR_{\mathrm{ex}}^{\tau}(z) $  are analytic in $z$ in the domain $\Re z > -\frac{1}{2 \tau}$. Moreover, there is a positive constant $\delta_{1} >0$ such that
\beq
 \mathop{\sup}\limits_{ \str \Im\ka^{}_{\links,\rechts} \str \leq \delta_1, \Re z >-\frac{1}{2\tau} }  \left\{ \begin{array}{clll}    
  \norm    \caJ_{\ka}  \caR_{\mathrm{ex}}^{\tau}(z) \caJ_{-\ka}  \norm   & = &O(\la^2) O(\la^2  \tau), & \qquad \la^2  \tau \searrow 0, \la \searrow 0   \\[3mm]
 \norm  \caJ_{\ka}\caR_{\mathrm{ld}}^{\tau}(z)\caJ_{-\ka} \norm   &  \leq &  \la^2 C  &      \end{array} \right.
     \label{eq: bounds on ladders and excitations cutoff} 
    \eeq
}

\item{ Recall the operator $\caL(z)$, introduced in Section \ref{sec: construction of semigroup}.  It satisfies 
\beq
 \mathop{\sup}\limits_{ \str \Im\ka_{\links,\rechts} \str \leq \delta_1, \Re z \geq 0 }  \left \norm  \caJ_{\ka} \left( \caR_{\mathrm{ld}}^{\tau}(z )  - \la^2 \caL(z) \right) \caJ_{-\ka} \right \norm   \leq   \la^2 C   \int_{\tau}^{+\infty} \d t  \sup_{x}  \str\psi(x,t)\str+ \la^4 \tau C'  \label{eq: closeness of ladders and caL} 
\eeq
 }

\een

\end{lemma}
The proof of this lemma is given in Section \ref{sec: bounds on long pairings}.

From Lemma \ref{lem: polymer model cutoff}, one  can deduce, by spectral methods, that $ \caZ^{\tau}_{t}$ inherits some of the   properties of the Markovian dynamics $\La_t$.
Instead of stating explicitly all possible results about $\caZ_t^{\tau}$, we restrict our attention to  Lemma \ref{lem: bounds on cutoff dynamics}, in particular, to the bound  \eqref{eq: bounds on cutoff dynamics}. This bound is the analogue of the bound \eqref{eq: bounds on markov dynamics} for the semigroup dynamics $\La_t$, and it will be used heavily in the analysis of $\caZ_t$ in Section \ref{sec: bounds on long pairings}. 

\begin{lemma}\label{lem: bounds on cutoff dynamics}
Let the cutoff reduced evolution $\caZ_t^\tau$ be as defined in Section \ref{sec: cut off model}, with the cutoff time $\tau=\tau(\la)$ satisfying \eqref{def: condition on tau}.  Then there
 are positive numbers $\delta_{c}>0$, $\la_c>0$ and $g_{c}>0$ such that, for $0< \str\la\str < \la_c$ and  $\ga \leq \delta_{c}$, 
\baq 
 \left\norm\left( \caZ^{\tau}_{t} \right)_{x^{}_{\links},x^{}_{\rechts};x'_{\links},x'_{\rechts}} \right\norm_{\scrB_2(\scrS)} &\leq&   c_{\caZ}^1 \,       \e^{ r_{\tau}(\ga,\la)  t}    \e^{- \frac{\ga}{2} \left\str (x'_{\links}+x'_{\rechts})- (x_{\links}+x_{\rechts})\right\str}   \e^{- \ga \left\str x^{}_{\links}-x^{}_{\rechts}  \right\str }   \e^{- \ga \left\str x'_{\links}-x'_{\rechts}  \right\str }  \nonumber  \\
 & + &      c_{\caZ}^2 \,   \e^{- \la^2 g_{c} t }    \e^{- \frac{\ga}{2} \left\str (x'_{\links}+x'_{\rechts})- (x^{}_{\links}+x^{}_{\rechts})\right\str}   \e^{- \ga \left\str (x^{}_{\links}-x^{}_{\rechts} )-( x'_{\links}-x'_{\rechts}  )  \right\str },   \label{eq: bounds on cutoff dynamics}
\eaq 
for constants $c_{\caZ}^1,c_{\caZ}^2 >0$,  and with
\beq
  r_\tau(\ga,\la) = O(\la^2)O(\ga^2) + o(\la^2) \qquad   \la \searrow 0, \ga  \searrow 0 \label{eq: bound on r tau}
\eeq
where the bound $o(\la^2)$ is uniform for $\ga \leq \delta_c$. 
\end{lemma}
The constants $\delta_{c}>0$ and decay rate $g_{c}>0$ are in general smaller than the analogues $\delta_{rw}$ and $g_{rw}$ in the bound \eqref{eq: bounds on markov dynamics}. 

\begin{proof}

We  apply Lemma \ref{lem: spectral abstract} in Appendix B with $\epsilon:=\la^2$ and
\baq
V(t, \epsilon) & :=&  U_{\nu} \left\{\caZ_t^{\tau}\right\}_p    U_{-\nu}     \label{def: identification cutoff 1} \\[2mm]
 A_{1}(z, \epsilon)   & :=&    U_{\nu} \left( \left\{  \caR^\tau_{ex}(z) \right\}_p +   \left\{  \caR^\tau_{ld}(z) \right\}_p    -  \la^2 \i  \left\{ \adjoint(\varepsilon(P)) \right\}_p  \right) U_{-\nu}     \label{def: identification cutoff 2}  \\[2mm]
 N       & :=&      U_{\nu} \left\{\caM \right\}_p  U_{-\nu}   \label{def: identification cutoff 3}   \\[2mm]
  B     & :=&     -U_{\nu}  \left\{\adjoint (Y) \right\}_p    U_{-\nu}    \simeq  -\adjoint (Y)   \label{def: identification cutoff 4}
\eaq
and $(p,\nu)  \in   \frD^{low}_{c}$ with 
\beq
 \frD^{low}_{c} :=  \frD^{low}_{rw} \cap \left \{  \str\Im p \str  \leq \min(\delta_1,\de_{\varepsilon}),   \str \Im \nu \str \leq \frac{1}{2} \min(\delta_1,\de_{\varepsilon})     \right\}  
\eeq  The set $\frD^{low}_{rw}$ has been defined before Proposition \ref{prop: properties of m}, the bound on $p, \nu$ involving $\delta_1$ ensures that we can convert  the domain of analyticity in the variable $\ka$ in Lemma \ref{lem: polymer model cutoff} into a domain of analyticity in the variables $(p,\nu)$, via the relation \eqref{eq: relation kappa and fibers}.
 Similarly, the bound on $p, \nu$ involving $\delta_{\varepsilon}$ ensures that 
 \beq  \label{eq: free liouvillian bounded}
\sup_{(p,\nu) \in \frD_{c}^{low}}\norm U_{\nu} \left\{ \adjoint(\varepsilon(P)) \right\}_p   U_{-\nu}  \norm \leq C  
\eeq
as a consequence of the bound on $\caJ_{\ka}\adjoint(\varepsilon(P))\caJ_{-\ka}$ provided by Assumption \ref{ass: analytic dispersion} and eq.\ \eqref{eq: combes thomas}.

We now check, step by step, the conditions of Lemma  \ref{lem: spectral abstract}. First, the continuity of $V(t,\ep)$ and the bound \eqref{eq: bound on V} follow from Lemma \ref{lem: definition dynamics} and Statement 1) of Lemma \ref{lem: polymer model cutoff}.
The relation \eqref{def: expression A} is Statement 1) of Lemma \ref{lem: polymer model cutoff} \newline
\noindent \textbf{Condition 1)} of Lemma \ref{lem: spectral abstract} is trivially satisfied since $Y$ is a Hermitian matrix on a finite-dimensional space.  \newline

\noindent To check \textbf{Condition 2)} of Lemma \ref{lem: spectral abstract}, we choose $g_{A}$ as $g_A= 2 g_{rw}^{low}  $ and we 
will actually show that the bound \eqref{eq: conditions A2}, which is required in the region $\Re z > -\la^2 g_A$, holds in the region $\Re z > -1/(2\tau)$, as long as $\la^2 \tau$ is small enough. By Cauchy's formula, this implies  that 
\beq \frac{\partial}{\partial z} A_{1}(z,\la)= \frac{1}{2\pi i} \oint_{\caC_z} \d z'  \frac{A_{1}(z',\la)}{(z-z')^2} =   O(\la^2 \tau), \qquad \textrm{for}\, \Re z > -\la^2 g_A,  \label{eq: first application cauchy}  \eeq
with $\caC_z$ a circle of radius $O(1/\tau)$ centered at $z$. Hence \eqref{eq: conditions A3} follows. 

To check  \eqref{eq: conditions A2}, we use the bound \eqref{eq: free liouvillian bounded} for $\adjoint(\varepsilon(P))$.  The boundedness of  the other terms in $ A_{1}(z, \la)$ follows immediately from \eqref{eq: bounds on ladders and excitations cutoff} .

\noindent\textbf{Condition 3)} contains conditions on the spectrum of $\caM_p$ that are satisfied thanks to Proposition \ref{prop: properties of m}.  It remains to check \eqref{eq: app closeness A and N}. By the bound on $\caR_{ex}^\tau(z)$ in \eqref{eq: bounds on ladders and excitations cutoff}, it suffices to check that, for  any $a \in \sp (\adjoint(Y))$; 
\beq
 1_a(\adjoint(Y)) \caJ_{\ka}   \Big( \la^2 \caM -(  - \la^2 \i \adjoint(\varepsilon(P)) +  \caR^\tau_{ld}(-\i  a)   )   \Big) \caJ_{-\ka}    1_a(\adjoint(Y)) =  o(\la^2), \qquad \textrm{as}\,  \la  \to 0
\eeq
This follows by the estimate in \eqref{eq: closeness of ladders and caL} and the relation between $\caM$ and $\caL$ in \eqref{eq: def of M}. Note that we used  that $\tau(\la) \to \infty$ as $\la \to 0$ to get $o(\la^2)$ from the estimate \eqref{eq: closeness of ladders and caL}.

Hence, we can apply Lemma \ref{lem: spectral abstract} and we obtain 
 a number $f_\tau(p,\la)$, a rank-one projector $P_\tau(p,\la)$ and a family of operators $R_\tau^{low}(t,p,\la)$ such that 
 \beq
 \left\{   \caZ_t^\tau\right\}_p  =   \e^{ f_\tau(p,\la) t } P_\tau(p,\la)  +  \e^{-(\la^2 g^{\footnotesize{low}}_{c} ) t }  R^{low}_\tau(t,p,\la)
 \eeq 
 for some $g^{\footnotesize{low}}_{c}>0$ (which can be chosen arbitrarily close to $g^{\footnotesize{low}}_{rw}$  by taking $\str\la\str$ small enough),
 and such that    
 \baq
&\mathop{ \sup}\limits_{(p,\nu) \in \frD_{c} ^{low}}     \norm   U_{\nu}P_\tau(p,\la)  U_{-\nu}  \norm \leq C& \\[2mm]
& \mathop{ \sup}\limits_{(p,\nu) \in \frD_{c} ^{low}}   \mathop{ \sup}\limits_{t \geq 0}   \norm   U_{\nu}R^{\footnotesize{low}}_{\tau}(p,t)  U_{-\nu}  \norm  \leq C & 
 \eaq
  
 The above reasoning applies to small fibers, since we use the spectral analysis of Proposition \ref{prop: properties of m}.  We now establish a simpler result about the cut-off reduced evolution $\left\{\caZ_t^\tau\right\}_{p}$, for large fibers. 
 Let 
 \beq
  \frD^{high}_{c} :=  \frD^{high}_{rw} \cap \left \{  \str\Im p \str  \leq \min(\delta_1,\de_{\varepsilon}),   \str \Im \nu \str \leq \frac{1}{2} \min(\delta_1,\de_{\varepsilon})     \right\}  
 \eeq
 Although for $(p,\nu) \in   \frD^{high}_{c}$, we cannot apply Lemma \ref{lem: spectral abstract}, we can still apply Lemma \ref{lem: spectrum change small o} to conclude that, for $\la$ small enough, the singularities of $\left\{\caR^\tau(z) \right\}_p$ in the domain, say,  $\Re z >  -2\la^2 g^{\footnotesize{high}}_{rw}$ lie at a distance $o(\la^2)$ from $\sp \caM_p$. One can then easily prove that $\left\{\caR^\tau(z) \right\}_p$ is bounded-analytic in a domain of the form $\Re z >- \la^2 g^{\footnotesize{high}}_{rw}+  o(\la^2)$ and hence
 \beq
\left\{\caZ_t^{\tau}\right\}_p      =    R^{\footnotesize{high}}_{\tau}(p,  t)    \e^{-(\la^2 g^{\footnotesize{high}}_{c} ) t } 
 \eeq
 with a rate  $g^{\footnotesize{high}}_{c}>0$ (which can be chosen arbitrarily close to $g^{\footnotesize{high}}_{rw}$ by making $\la$ small enough) and
 \beq
 \sup_{(p,\nu) \in \frD_{c}^{high}}  \sup_{t \geq 0} \norm  U_{\nu}R^{\footnotesize{high}}_{\tau}(t,p, \la)    U_{-\nu}  \norm \leq   C
 \eeq
 
 Finally, we note that one can easily find a constant $\de_c$ such that $\frD_{c}^{low}$ and $\frD_{c}^{high}$ are of the form \eqref{def: set D low markov} and \eqref{def: set D high markov} with the parameters  $\delta_c$ instead of $\delta_{rw}$ (the parameter $p_{rw}^*$ does not need to be readjusted).

With the information on $\caZ_t^{\tau}$ obtained above, we are able to prove  the bound \eqref{eq: bounds on cutoff dynamics}  by the same reasoning as we employed in the lines following Proposition \ref{prop: properties of m} to derive the bound \eqref{eq: bounds on markov dynamics}.  

The function $r_\tau(\ga,\la)$ in the statement of Lemma \ref{lem: bounds on cutoff dynamics} is determined as 
\beq
r_\tau(\ga,\la) :=  \sup_{p \in \tor, \str p \str \leq \ga} \max\left( \Re f_{\tau}(p,\la), 0 \right)   
 \eeq
 and the bound \eqref{eq: bound on r tau} follows by \eqref{eq: bound on r markov} and 
 \beq 
 f_\tau(p,\la)-\la^2 f(p)=o(\la^2), \qquad  \la \searrow 0,
 \eeq
  which follows from \eqref{eq:  lemma bounds1 abstract}  in Lemma \ref{lem: spectral abstract}. The decay rate $g_c$ is chosen as  $g_c:= \min (g_{c}^{low},g_{c}^{high}) $.
 This concludes the proof of Lemma \ref{lem: bounds on cutoff dynamics}. 

 \end{proof}

We close this section with two remarks which are however not necessary for an understanding of the further stages of the proofs.

\begin{remark}
 As apparent from the bound \eqref{eq: closeness of ladders and caL}, one cannot take $\tau \equiv const$, since in that case, this bound becomes $O(\la^2)$ instead of $o(\la^2)$. This would mean that there is a 
 difference of $O(\la^2)$ between   $\caZ_t^{\tau}$ and $\La_t$, whereas the important terms in $\La_t$ are themselves of $O(\la^2)$.  This is however not an essential point: as one can see from the classification of diagrams in Section \ref{sec: expansions}, one could easily modify the definition of the cutoff model such that $\caZ_t^{\tau}$ is close to $\La_t$  even at $\tau \equiv const$. This can be achieved by  performing the cutoff on the \emph{non-ladder} diagrams only, which is a notion introduced in Section  \ref{sec: expansions}. The true reason why $\tau$ must diverge to $\infty$ when $\la \searrow 0$ will become clear in the proof of Lemma \ref{lem: polymer model full}, in Section \ref{sec: summing min irr pairings}.
 \end{remark}
 
 \begin{remark}
 One is tempted to say that any claim that is made about $\caZ_t$ in  Section \ref{sec: result} could be stated for $\caZ_t^{\tau}$ as well.  While this is correct for Theorem \ref{thm: main}, it fails for Proposition \ref{prop: symmetries of f}. The reason is that the identity $f(p=0,\la)=0$ follows from the fact that $\caZ_t$ conserves the trace of density matrices, as it is the reduced dynamics of a unitary evolution. 
 This is not true for  $\caZ_t^{\tau}$, and hence we cannot prove (or even expect) that $f_\tau(p=0,\la)=0$.
  \end{remark}

  \subsection{Spectral analysis of $\caZ_t$} \label{sec: polymer model full dynamics}

In this section, we state Lemma \ref{lem: polymer model full},  the $\tau=\infty$ analogue of Lemma \ref{lem: polymer model cutoff}.  This Lemma leads to our main result, Theorem \ref{thm: main},  via reasoning that is almost identical to the one that led from Lemma \ref{lem: polymer model cutoff} to Lemma \ref{lem: bounds on cutoff dynamics}. 

Essentially (and analogously to Lemma \ref{lem: polymer model cutoff}), Lemma \ref{lem: polymer model full} states that the Laplace transform $\caR(z)$, defined in \eqref{def: first laplace trafo}, is of the form
\beq
\caR(z) = (z -(-  \i \adjoint(Y) + \la^2 \caM + \caA(z)) )^{-1} 
\eeq
where $\caA(z) $ is  `small' w.r.t.\, $\la^2 \caM$.

\begin{lemma}\label{lem: polymer model full}
   There is an operator $ \caR_{\mathrm{ex}}(z)\in \scrB(\scrB_2(\scrH_\sys))$, depending on $\la$ and satisfying the following properties, for $\la$ small enough:
\ben
\item{ For   $\Re z$ sufficiently large, the integral in  \eqref{def: first laplace trafo} converges absolutely in $\scrB(\scrB_2(\scrH_\sys))$ and
\beq \label{eq: expansion for resolvent full model}
 \caR(z) = ( \caR^{\tau}(z)^{-1} - \caR_{\mathrm{ex}}(z)  )^{-1}
\eeq 
where $\caR^{\tau}(z)$ was introduced in  \eqref{def: first laplace trafo cutoff} and $\tau =\tau(\la)$ was defined in \eqref{def: condition on tau}.
}

\item{ There are positive constants $\delta_{ex},  g_{ex} $ such that the operator $\caR_{ex}(z)$ is analytic in $z$ in the domain $\Re z > -\la^2 g_{ex}$ and  
\beq
   \mathop{\sup}\limits_{ \str \Im\ka^{}_{\links,\rechts} \str \leq \delta_{ex}, \Re z > -\la^2 g_{ex}} \norm    \caJ_{\ka}\caR_{\mathrm{ex}}(z) \caJ_{-\ka}  \norm    =o(\la^2), \qquad \la \searrow 0. \label{eq: bound on caR excitations}
\eeq
}

\een
\end{lemma}

The proof of Lemma \ref{lem: polymer model full} is contained in Section \ref{sec: bounds on long pairings}. Starting from  Lemma \ref{lem: polymer model full}, we can prove our main result, Theorem \ref{thm: main}, by the spectral analysis outlined in Appendix \ref{app: spectral}.\\
\vspace{0.2cm}

\noindent\emph{Proof of Theorem \ref{thm: main}}  We  apply Lemma \ref{lem: spectral abstract} with $\epsilon:= \la^2$ and
\baq
V(t, \epsilon) & :=&  U_{\nu} \left\{\caZ_t\right\}_p    U_{-\nu}      \label{def: identification 1} \\[2mm]
 A_{1}(z, \epsilon)   & :=&    U_{\nu} \left( \left\{     \caR_{ex}(z) \right\}_p+  \left\{     \caR^\tau_{ex}(z) \right\}_p +   \left\{  \caR^\tau_{ld}(z) \right\}_p    -  \la^2 \i  \left\{ \adjoint(\varepsilon(P)) \right\}_p  \right) U_{-\nu}   \label{def: identification 2}  \\[2mm]
 N       & :=&      U_{\nu} \left\{\caM \right\}_p  U_{-\nu}     \label{def: identification 3} \\[2mm]
  B     & :=&  -   U_{\nu}  \left\{\adjoint (Y) \right\}_p    U_{-\nu}    = - \adjoint (Y)   \label{def: identification 4}
\eaq
and 
\beq \label{eq: choice of domain in second step low}
(p,\nu) \in  \frD^{low}_{rw} \cap \left \{  \str\Im p \str  \leq \min(\delta_{ex},\de_{\varepsilon}),   \str \Im \nu \str \leq \frac{1}{2} \min(\delta_{ex},\de_{\varepsilon})     \right\}  
\eeq 

Hence, the only difference with the relations (\ref{def: identification cutoff 1}-\ref{def: identification cutoff 2}-\ref{def: identification cutoff 3}-\ref{def: identification cutoff 4}) is that we have added the term 
$
\left\{     \caR_{ex}(z) \right\}_p
$
in \eqref{def: identification 2}, we consider $\caZ_t$ instead of $\caZ_t^{\tau}$ in \eqref{def: identification 1}, and we replace $\de_1$ by $\de_{ex}$ in \eqref{eq: choice of domain in second step low}.
 This means that  we can copy  the proof of Lemma \ref{lem: bounds on cutoff dynamics}, except that, in addition,  we have to check  the bounds \eqref{eq: conditions A2} and  \eqref{eq: conditions A3} for the term $  \caR_{ex}(z)$.
We choose $g_A:= \frac{1}{2} g_{ex}$. Then  the bound  \eqref{eq: conditions A2} follows from \eqref{eq: bound on caR excitations}, and  \eqref{eq: conditions A3}  follows since, by the Cauchy integral formula and \eqref{eq: bound on caR excitations}, 
 \beq
  \sup_{\Re z \geq -\frac{1}{2} \la^2 g_{ex}}  \norm \frac{\partial}{ \partial z}   U_{\nu} \left\{\caR_{ex}(z) \right\}_p U_{-\nu} \norm  =  \left\str \Re z-(- \la^2 g_{ex} )\right \str^{-1}  o(\la^2)  =o(\str\la\str^{0}), \qquad  \la \searrow 0
  \eeq
  where we use the same argument as in \eqref{eq: first application cauchy}, but with a circle radius  of the order of $\left\str \Re z-(- \la^2 g_{ex} )\right \str$.
This application of the Cauchy integral formula is the reason for the factor 
$\frac{1}{2} $ into the definition of  $g_A $. 
Lemma \ref{lem: spectral abstract} yields the function $f(p,\la)$, the rank-one projector $P(p,\la)$ and the operator $R^{low}(t,p,\la)$ required in the small fiber statements of Theorem \ref{thm: main}.

For 
\beq  \label{eq: choice of domain in second step high}
(p,\nu) \in  \frD^{high}_{rw} \cap \left \{  \str\Im p \str  \leq \min(\delta_{ex},\de_{\varepsilon}),   \str \Im \nu \str \leq \frac{1}{2} \min(\delta_{ex},\de_{\varepsilon})     \right\},
\eeq
we can again apply Lemma \ref{lem: spectrum change small o} to derive the large fiber statements of Theorem \ref{thm: main}. 
As in the proof of Lemma \ref{lem: bounds on cutoff dynamics}, we can again choose parameters $\delta, p^*$ such that domains $\frD^{low},\frD^{high} $ as defined in (\ref{def: domain low}-\ref{def: domain high}), are included in the domains for $(p,\nu)$ specified by \eqref{eq: choice of domain in second step low} and \eqref{eq: choice of domain in second step high}.
\qed

\section{Feynman Diagrams} \label{sec: expansions}

In this section, we introduce  the expansion of the reduced evolution $\caZ_t$ and the cutoff reduced evolution $\caZ_t^\tau$ in amplitudes labelled by Feynman diagrams. These expansions will be the  main tool in the proofs of Lemmata \ref{lem: polymer model cutoff} and \ref{lem: polymer model full}.
We start by introducing a notation for the Dyson expansion of $\caZ_t$ which is more convenient than that of  Section \ref{sec: dyson expansion}.
\subsection{Diagrams $\si$}  \label{sec: diagrams}
Consider a pair of elements in $I \times \lat \times \{\links, \rechts \}$ with $I \subset \bbR^+$ a closed interval whose elements should be thought of as  times. The smaller time of the pair is called $u$ and the larger time is called $v$, and we require  that $u  \neq v$, i.e.\ $u<v$.  The set of pairs satisfying this constraint is called $\Si^1_{I}$. 
We define $\Si^n_I$ as the set of $n$ pairs of elements in $I \times \lat \times \{\links, \rechts \}$ such that no two times coincide. 
That is, each $\si \in \Si^n_I$ consists of $n$ pairs whose time-coordinates are parametrized by $(u_i,v_i)$, for $i=1,\ldots,n$, and with the convention that 
$u_i < v_i$ and $u_i < u_{i+1}$.  The elements $\si$ are called \emph{diagrams}.
As announced in Section \ref{sec: dyson expansion}, there is a one-to-one mapping between, on the one hand, a set of triples $(t_i(\dsi), x_i(\dsi), l_i(\dsi))_{i=1}^{2n}  $ with $t_i < t_{i+1}$ and $t_i \in  I$,  together with a pairing $\pi \in \caP_n$, and, on the other hand,  a diagram
 $\dsi \in \Si^n_{I} $ as defined above.\\
   To construct this mapping, proceed as follows: Choose from the pairing $\pi$ the pair $(r,s)$ for which $r=1$ and set $u_1=t_{r}, v_1=t_s$. The pair $((t_r,x_r,l_r), (t_s,x_s,l_s))$ becomes  the first pair in the diagram $\si$.  Then choose the pair $(r',s') \in \pi$ such that 
 \beq
 r' = \min \{ \{1,2, \ldots, 2n  \} \setminus \{r,s \} \}
 \eeq
 Set $u_2= t_{r'}, v_2=t_{s'}$. The pair  $((t_{r'},x_{r'},l_{r'}), (t_{s'},x_{s'},l_{s'}))$ becomes the second pair of $\si$.  Repeat this until one has $n$ pairs, each time picking the pair whose $r$ is the smallest of the remaining integers. The mapping is easily visualized in a picture, see Fig.\ \ref{fig: introdiagrams1}.\\
 
We also  use the notation $\underline{t}(\dsi),\underline{x}(\dsi),\underline{l}(\dsi)$ to denote the 'coordinates' of the diagram $\dsi$. Here, 
$\underline{t}(\dsi),\underline{x}(\dsi),\underline{l}(\dsi)$ are $2n $-tuples of elements in $I, \bbZ^d, \left\{ \links, \rechts \right\}$, respectively, and such that the $i$'th components of these $2n$-tuples constitutes the $i$'th triple $(t_i(\dsi),x_i(\dsi),l_i(\dsi))$. 

Note that the time-coordinates $ \underline{t} \equiv t_1(\dsi), \ldots, t_{2n}(\dsi)$ can also be defined as the ordered set of times containing the elements $\left\{ u_i,v_i, i=1,\ldots,n\right\}$. Evidently, the triples $(t_i(\dsi), x_i(\dsi), l_i(\dsi))_{i=1}^{2n} $ do not fix a diagram uniquely since the combinatorial structure that is encoded in $\pi$ is missing. That combinatorial structure is now encoded in the way the time coordinates $\underline{t}(\dsi)$ are partitioned into pairs $(u_i,v_i)$, see also Figure \ref{fig: introdiagrams1}.

We drop the superscript $n$ to denote the union over all $n \geq 1$, i.e.\
\beq
\Si_I  := \bigcup_{n \geq 1}  \Si^n_I
\eeq
and we write $\str \dsi\str =n $ to denote that $\si \in \Si_I^n$. 

 We define the \emph{domain} of a diagram as
\beq  \label{def: union over all sizes}
\Dom \dsi := \mathop{\bigcup}\limits_{i =1}^n [ u_i,v_i] , \qquad  \textrm{for} \, \,  \si \in \Si_I^n 
 \eeq


 We call a diagram $\si \in \Si_I$  
  irreducible (notation: $\mathrm{ir}$) whenever its domain is a connected set (hence an interval). In other words, $\si$ is irreducible whenever there are no two (sub)diagrams $\si_1, \si_2  \in  \Si_I$ such that  
  \beq
  \si=\si_1 \cup \si_2, \qquad  \textrm{and} \qquad  \Dom\dsi_1 \cap \Dom\dsi_2 =\emptyset
  \eeq
 where the union refers to a union of pairs of elements in $I \times \lat \times \{\links, \rechts \}$.
  For any  $\si \in \Si_I$ that is not irreducible, we can thus find   a unique (up to the order) sequence  of (sub) diagrams $\si_1,\ldots, \si_m$  such that 
  \beq \si_{1},\ldots, \si_m \quad  \textrm{ are irreducible and } \quad    \si=\si_1 \cup \ldots \cup \si_m   \eeq  
  We fix the order of $ \si_{1},\ldots, \si_m $ by requiring that $\max \Dom \dsi_{i} \leq \min \Dom \dsi_{i+1} $ and 
  we call the sequence $(\si_1,\ldots,\si_m)$ obtained in this way the  decomposition of $\dsi$ into irreducible components.

 We let $\Si^{n}_{I}(\mathrm{ir}) \subset \Si_I^n$ stand for the set of irreducible diagrams $\dsi$ (with $n$ pairs) that satisfy $\Dom \dsi=I$, that is, $u_1=t_1(\dsi)=\min I$ and $\max_i  t_i(\dsi)= \max_i v_i=\max I$.\\

   A diagram $\si \in \Si_{I}(\mathrm{ir}) $  is called  \emph{minimally irreducible} in the interval $I$ whenever it has the following property: For any  subdiagram $\dsi' \subset \dsi$, the diagram $\dsi \setminus
  \dsi'$ does not  belong to $\Si_I({\mathrm{ir}})$. Intuitively, this means that either the diagram $ \dsi'$ contains a boundary point of $I$ as one of its time-coordinates, or the diagram $\dsi \setminus                       
   \dsi'$ is not irreducible. The set of minimally irreducible diagrams (with $n$ pairs) is denoted by $ \Si^{n}_{I}(\mathrm{mir})$.   See pictures \ref{fig: introdiagrams1} and \ref{fig: introdiagrams2}  for a graphical representation of the diagrams.   Since, up to now, most definitions depend solely on the time-coordinates, we only indicate  the time-coordinates in the pictures. In the terminology introduced below, we draw equivalence classes of diagrams $[\dsi]$ rather than  the diagrams $\dsi$ themselves.
    
\begin{figure}[h!] 
\vspace{0.5cm}
\begin{center}
\psfrag{initialtime2} {  $u_2$ }
\psfrag{initialtime1} { $u_1$}
\psfrag{finaltime2} {  $v_2$ }
\psfrag{finaltime1} { $v_1$}
\psfrag{initialtime3} {  $u_3$ }
\psfrag{initialtime4} { $u_4$}
\psfrag{finaltime3} {  $v_3$ }
\psfrag{finaltime4} { $v_4$}
\psfrag{initialtime5} { $u_5$}
\psfrag{finaltime5} {  $v_5$ }
\psfrag{finaltime} {  $I_+$ }
\psfrag{initialtime} { $I_-$}
\psfrag{time10}{  $t_{10}$}
\psfrag{time9}{  $t_9$}
\psfrag{time8} {  $t_8$ }
\psfrag{time7} {  $t_7$ }
\psfrag{time6} { $t_6$}
\psfrag{time5}{  $t_5$}
\psfrag{time4}{  $t_4$}
\psfrag{time3} {  $t_3$ }
\psfrag{time2} {  $t_2$ }
\psfrag{time1} { $t_1$}
\includegraphics[width = 12cm, height=4cm]{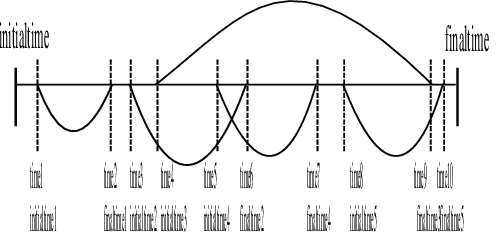}
\caption{ \footnotesize{ A diagram $\dsi \in \Si_I$ with $\str \dsi \str=5$. Its time coordinates are shown explicitly. Note that the parametrization by $u_i,v_i$ encodes the combinatorial structure (the way the times are connected by pairings), whereas the $t_i$ are ordered.  We consistenly draw the long pairings (see later) above the time-axis and the short ones below.} \label{fig: introdiagrams1} }
\includegraphics[width = 18cm, height=4cm]{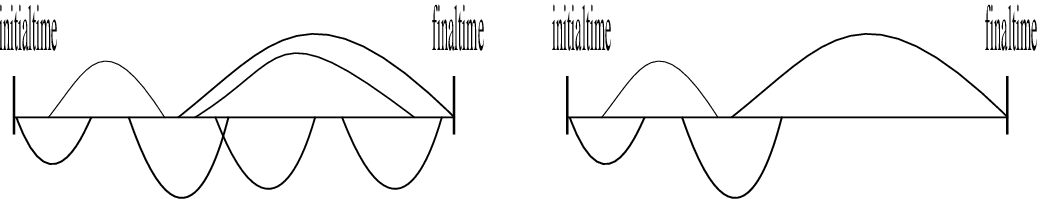}
\caption{ \footnotesize{The left figure shows an irreducible diagram $\dsi$ in the interval $I=[I_-,I_+]$ with $\str \dsi \str=7$. This diagram is not minimally irreducible. The right figure shows a minimally irreducible subdiagram. In this case, there is only one such minimally irreducible subdiagram, but this need not always be the case.} \label{fig: introdiagrams2} }
\end{center}

\end{figure}

  A diagram $\dsi$ in $\Si_I$  for which each pair of time coordinates $(u,v)$ satisfies $\str v-u\str \geq \tau $, or $\str v-u\str \leq \tau $, is called 
long, or short, respectively. The set of all long/small diagrams with $n$ pairs is denoted by $\Si_I^n(> \tau )$ / $\Si_I^n({< \tau} )$. 
Note that $\Si_I^n(> \tau ) \cup \Si_I^n({< \tau} )$ is strictly smaller than $\Si_I^n$ whenever $n>1$.

In addition to the sets  $\Si_I^n( \mathrm{ir}), \Si_I^n( \mathrm{mir}), \Si_I^n(> \tau)$, we will sometimes
use more than one specification ($\mathrm{adj}$) to denote a subset of $\Si_I$ or $\Si_I^n$, and we will drop the superscript $n$ to denote the union over all $\str \dsi \str$, as in \eqref{def: union over all sizes}, for example, 
\beq
\Si_I^n(< \tau, \mathrm{ir}), \qquad      \Si_I(> \tau,  \mathrm{mir})
\eeq
are the sets of short irreducible diagrams with $\str\dsi\str=n$ and long minimally irreducible diagrams, respectively.

  On  the set $\Si^{n}_{I}$, we define the ``Lebesgue measure" $\d \dsi$ by
    \beq \label{def: measure}
    \mathop{\int}\limits_{\Si^{n}_{I}}\d \dsi F(\dsi) :=    \mathop{\int}\limits_{I_{-} < u_1 < \ldots < u_n < I_{+}} \d u_1 \ldots \d u_n  \mathop{\int}\limits_{u_i <v_i } \d v_1 \ldots \d v_n   \sum_{\underline{x}(\dsi), \underline{l}(\dsi)}   F(\dsi)
  \eeq
  where $I=[I_-,I_+]$.
  
  Since $\Si^{n}_{I}(\mathrm{ir})$ is a zero-measure subset of $\Si_I^n$, the definition of the measure $\d \dsi$ on $\Si_I^n({\mathrm{ir}})$ has to modified in an obvious way: For all continuous (in the time coordinates $\underline{t}(\dsi)$) functions $F$ on $\Si_I^n$, we set
 \beq \label{def: measure irreducible}
\mathop{ \int}\limits_{\Si^{n}_{I}(\mathrm{ir})} \d \dsi F(\dsi)  =     \mathop{ \int}\limits_{\Si^{n}_{I}} \,  \d \dsi  \delta(\max{\underline{t}(\dsi)} -I_{+}) \delta(\min{\underline{t}(\dsi)} -I_{-}) F(\dsi)
 \eeq
 where the Dirac distributions $\delta(I_+-\cdot)$ and $\delta(I_--\cdot)$ are a priori ambiguous since $I_-, I_+$ are the boundary points of the interval $I$. They   are defined as
 \beq \label{def: delta functions on boundary}
 \delta(\cdot-I_+) := \lim_{s \mathop{\rightarrow}\limits_{<}I_+ }  \delta( \cdot-s), \qquad    \delta( \cdot-I_-) := \lim_{s \mathop{\rightarrow}\limits_{>}I_- }  \delta( \cdot-s)
 \eeq 
 We extend the definition of the measure $\d \dsi$ also to  $\Si_I$ and the various $\Si_I(\mathrm{adj})$ by setting
  \beq
\mathop{ \int}\limits_{\Si_I(\mathrm{adj})} \d \dsi F(\dsi)  :=   \mathop{\sum}_{ n \geq 1}  \mathop{ \int}\limits_{\Si^n_I(\mathrm{adj})} \,  \d \dsi F_n(\dsi)
 \eeq
 where $F_n$ is the restriction to $\Si^n_I(\mathrm{adj})$ of  a function $F$ on $\Si_I(\mathrm{adj})$.

We will often  encounter functions of $\dsi$ that are independent of the coordinates $\underline{x}(\dsi),\underline{l}(\dsi)$ and that must be integrated only over $\underline{t}(\dsi)$ and summed over $\str \dsi \str$.
To deal elegantly with such situations, we let $[\dsi]$ stand for an equivalence class of diagrams that is obtained by dropping the $\underline{x},\underline{l}$-coordinates. That is
\beq
[\dsi] =[\dsi']  \Leftrightarrow \left\{\begin{array} {c} \str \dsi \str = \str \dsi' \str \, \,    \\[2mm] \,  u_i(\dsi)=u_{i}(\dsi'),  v_i(\dsi)=v_{i}(\dsi'), \,\,  \textrm{for all} \, i=1,\ldots, \str \dsi \str    \end{array} \right.
\eeq
The set of such equivalence classes is denoted by  $\Pi_T \Si_I$ (the symbol $\Pi_T$ can be thought of as a projection onto the time coordinates) and we naturally extend the definition to $\Pi_T \Si_I(\mathrm{adj})$
 where $adj$ can again stand for $\mathrm{ir}, \mathrm{mir}, {>\tau}, {< \tau}$ . 
The integration over equivalence classes of diagrams 
is defined as above in \eqref{def: measure} and \eqref{def: measure irreducible}, but with $\sum_{\underline{x}(\dsi),\underline{l}(\dsi)}$ omitted, i.e., such that for all functions $\tilde F$ on $\Si_I(\mathrm{adj})$: 
\beq
  \mathop{\int}\limits_{\Pi_T\Si_I (\mathrm{adj} )  }  \d [ \dsi]   F([\dsi])     =   \mathop{\int}\limits_{\Si_I(\mathrm{adj} ) }  \d  \dsi   \tilde F(\dsi) , \qquad   \textrm{with} \, \,   F([\dsi])  =    \sum_{\underline{x}(\dsi),\underline{l}(\dsi)} \tilde F(\dsi)
\eeq

 Lemma \ref{lem: F and G} contains the main application of this construction.  It is in fact a simple $L^1-L^{\infty}$-bound. 
\begin{lemma}\label{lem: F and G}
 Let $F$ and $G$ be  positive, continuous functions  on $\Si_I$. Then
\beq
\mathop{\int}\limits_{\Si_I(\mathrm{adj})} \d \dsi  F(\dsi) G(\dsi) 
 \leq   \mathop{\int}\limits_{\Pi_T \Si_I(\mathrm{adj})}  \d [\dsi]  \left[ \sup_{\underline{x}(\dsi), \underline{l}(\dsi)} G(\dsi) \right]  \left[  \sup_{\underline{t}(\dsi)} \sum_{\underline{x}(\dsi), \underline{l}(\dsi)} F(\dsi) \right]    \label{eq: general inequality}
\eeq
where it is understood that the sum and sup over $\underline{x}(\dsi), \underline{l}(\dsi)$ are performed while keeping $\str \dsi \str$ and $\underline{t}(\dsi) $ fixed. 
\end{lemma}

In \eqref{eq: general inequality}, the sum/$\sup$, over ${\underline{x}(\dsi), \underline{l}(\dsi)}$  is in fact  a shorthand notation for the sum/$\sup$ over all $\dsi'$ such that $[\dsi']=[\dsi]$ for a given $\dsi$. Hence, $ \sup_{\underline{x}(\dsi), \underline{l}(\dsi)}   G(\dsi)$ is a function of $[\dsi]$ only, as required.  The supremum  $\sup_{\underline{t}(\dsi)}$ is  over $I^{2\str \dsi\str }$,  with  $\str \dsi \str$ fixed.   Hence, the second factor on the RHS of \eqref{eq: general inequality} is in fact a function of $\str \dsi \str$ only.

\begin{proof}
We start from the explicit expressions in \eqref{def: measure} or \eqref{def: measure irreducible}, and we use a $L^1-L^{\infty}$ inequality: first for the sum over $\underline{x}(\dsi), \underline{l}(\dsi)$ and then for the integration over $u_i,v_i$.
 \end{proof}

\subsubsection{Representation of the reduced evolution $\caZ_t$} \label{sec: representation of caZ}

Recall the operators $\caV_t((t_i,x_i,l_i)_{i=1}^{2n} )$ defined in  \eqref{def: cav}. Since,  by the above discussion,  there is a one-to-one correspondence between a diagram $\dsi$  and $(\pi, (t_i,x_i,l_i)_{i=1}^{2\str \dsi \str} )$  where $\pi \in \caP_{\str \dsi \str}$ and $t_i < t_{i+1}$,  we can write $\caV_t(\si)$ instead of  $\caV_t( (t_i,x_i,l_i)_{i=1}^{2\str \dsi \str} )$ and $\zeta(\dsi)$ instead of $\zeta(\pi, (t_i,x_i,l_i)_{i=1}^{2\str \dsi \str}  )$, i.e.\ 
\beq \label{def: cavdsi}
\caV_t(\dsi) :=   \caV_t(  (t_i(\dsi),x_i(\dsi),l_i(\dsi))_{i=1}^{2\str \dsi \str} )   
\eeq
and
\beq \label{def: zeta2}
 \zeta(\dsi)  := \prod_{ ((u,x,l), (v,x',l')) \in \dsi}   \la^2    \left\{  \begin{array}{cr}   \psi (x'-x,v-u )   & \qquad  l=l'= \links \\ [1mm]  
\overline{\psi}  (x'-x,v-u )    & l=l'= \rechts \\
    \overline{\psi}   (x'-x,v-u )   &    l= \links, l'= \rechts \\
\psi  (x'-x,v-u ) &    l= \rechts, l'= \links \\
  \end{array}\right.
\eeq
As a slight generalization of the operators 
$\caV_{t}(\dsi)$, we also define $\caV_{I}(\dsi)$ for a closed interval $I:=[I_{-},I_{+}]$ by
\beq \label{def: cavI}
\caV_{I}(\dsi) :=      \caU_{I_{+}-t_{2n}}     \caI_{x_{2n},l_{2n} }    \ldots \caI_{x_2,l_2}   \caU_{t_2-t_1} \caI_{x_1,l_1}  \caU_{t_1-I_{-}}, \qquad \textrm{for}\,\, \dsi \, \,  \textrm{such that}\,\,  \Dom \dsi \subset I
\eeq
The only difference with $\caV_{t}(\dsi)$ is in the time-arguments '$t$' of $\caU_{t}$ at the beginning and the end of the expression.
With this new notation, $\caV_{t}(\dsi)= \caV_{[0,t]}(\dsi)$. 
Next, we state the representation of the reduced evolution $\caZ_t$ as an integral over diagrams
\beq \label{eq: caZ as integral over diagrams}
\caZ_t = \caU_t +  \int_{\Si_{[0,t]}} \d \dsi \zeta(\dsi) \caV_{[0,t]}(\dsi)
\eeq
Similarly, the cutoff dynamics $\caZ^{\tau}_{t}$ is represented as
\beq \label{eq: caZ  cutoff as integral over diagrams}
\caZ^{\tau}_{t}  =  \caU_t+   \mathop{\int}\limits_{\Sigma_{[0,t]}({< \tau})}  \d \dsi   \zeta(\dsi)   \caV_{[0,t]}(\dsi)
\eeq
Formulas \eqref{eq: caZ as integral over diagrams} and \eqref{eq: caZ cutoff as integral over diagrams} are immediate consequences of \eqref{eq: Z with pairings} and \eqref{eq: Z cutoff with pairings}, respectively.

We use the notion of irreducible diagrams $\dsi$ to decompose the operators $\caV_{[0,t]}(\si)$ into products and to derive a new representation, \eqref{eq: z as polymer model}, for $\caZ_t$ and $\caZ_t^\tau$

Let  $(\dsi_1, \ldots, \dsi_p)$ be the decomposition of a diagram $\dsi \in \Si_{[0,t]}$ into irreducible components. Define the times $s_1, \ldots, s_{2p}$ to be the boundaries of the domains of the irreducible components, i.e.,  $[s_{2i-1}, s_{2i}]= \Dom \dsi_i$, for $i=1,\ldots, p$.   Then
\beq \label{factorization}
\caV_{I}(\dsi)=  \caU_{I_{+}-s_{2p}} \,  \caV_{[s_{2p-1},s_{2p}]} (\dsi_p) \, \caU_{s_{2p-1}-s_{2p-2}}   \ldots        \caU_{s_3-s_2} \,  \caV_{[s_1,s_2]}(\dsi_1) \,  \caU_{s_1- I_{-}} ,
\eeq
as can be checked from (\ref{def: zeta2}-\ref{def: cavI}).  Here, the essential observation is that all time coordinates of $\dsi_i$ are smaller than those of $\dsi_{i+1}$. We introduce 
\beq \label{def: irreducible evolutions}
\caZ^{\mathrm{ir}}_{t}:=    \mathop{\int}\limits_{{\Si}_{[0,t]}(\mathrm{ir}) }  \d \dsi   \zeta(\dsi) \caV_{[0,t]}(\dsi),  \qquad    \caZ^{\tau, \mathrm{ir}}_{t}:=    \mathop{\int}\limits_{{\Si}_{[0,t]}(< \tau,\mathrm{ir}) }  \d \dsi   \zeta(\dsi)\caV_{[0,t]}(\dsi)
\eeq
and we remark that the definitions in \eqref{def: irreducible evolutions}   allow for a shift of time on the RHS, that is 
\beq
\caZ^{\mathrm{ir}}_{t} =  \mathop{\int}\limits_{{\Si}_{I}(\mathrm{ir}) }  \d \dsi  \zeta(\dsi)  \caV_{I}(\dsi), \qquad \textrm{for any} \, I=[s,s+t], \quad s \in \bbR
\eeq
and similarly for $\caZ^{\tau,\mathrm{ir}}_{t} $.
By this  time-translation invariance, the  factorization property   \eqref{factorization} and the factorization property of the correlation function in \eqref{def: zeta2}, i.e., 
\beq
\zeta (\dsi_1 \cup \ldots  \cup \dsi_p)  = \prod_{i=1}^p  \zeta(\dsi_i),
\eeq
 we can rewrite the expression  \eqref{eq: caZ as integral over diagrams} as
\beq \label{eq: z as polymer model}
 \caZ_{t} =   \sum_{m \in 2 \bbZ^+}  \mathop{\int}\limits_{0 \leq s_1 \leq \ldots \leq s_{m} \leq t} \d s_1 \ldots \d s_{m}      \, \left(    \caU_{t-s_{m}}  \caZ^{\mathrm{ir}}_{s_{m}-s_{m-1}}   \ldots        \caU_{s_3-s_2} \caZ^{\mathrm{ir}}_{s_2-s_1} \caU_{s_1} \right).
\eeq
where the term on the RHS corresponding to $m=0$ is understood to be equal to $\caU_t$. 
The idea behind \eqref{eq: z as polymer model} is that, instead of summing over all diagrams, we  sum over all sequences of irreducible diagrams. 
An analogous formula holds with $ \caZ_{t} $ and  $\caZ^{\mathrm{ir}}_{t} $ replaced by $ \caZ_{t}^{\tau} $ and $\caZ^{\tau,\mathrm{ir}}_{t} $.

\subsection{Ladder diagrams and excitations} \label{sec: ladders and excitations}
We are  ready to identify the operators $\caR^{\tau}_{ex}(z)$ and $\caR_{\mathrm{ld}}^{\tau}(z)$, whose existence was postulated in Lemma \ref{lem: polymer model cutoff} and the operator $\caR_{ex}(z)$, which was postulated in Lemma \ref{lem: polymer model full}.

The Laplace transform,  $\caR(z)$, of $\caZ_t$ has been introduced in \eqref{def: first laplace trafo}.  We calculate $\caR(z)$ starting from  \eqref{eq: z as polymer model}
\baq
\caR(z) &=&  \int_{\bbR^+} \d t \,  \e^{-t z}    \caZ_{t}   \\[2mm]
  &=&  \sum_{m \geq 0}    \left[   (z + \i \ad(H_\sys))^{-1}     \int_{\bbR^+} \d t  \,  \e^{-t z}  \caZ^{\mathrm{ir}}_t   \right] ^m  (z+ \i \ad(H_\sys))^{-1}    \\[2mm]
  &=&   (z + \i \ad(H_\sys) - \caR_{\mathrm{ir}}(z))^{-1}, \qquad  \textrm{with}  \, \,    \caR_{\mathrm{ir}}(z) :=    \int_{\bbR^+} \d t  \,  \e^{-t z}  \caZ^{\mathrm{ir}}_t 
\eaq
The second equality follows by $ \int_{\bbR^+} \d t  \e^{-t z}    \caU_{t} = (z+ \i \ad(H_\sys))^{-1}  $ for $\Re z>0$. The third equality follows by summing the geometric series.

An identical computation yields
\beq \label{eq: laplace cutoff}
\caR^{\tau}(z) =   (z + \i \ad(H_\sys) - \caR^{\tau}_{\mathrm{ir}}(z))^{-1}, \qquad  \textrm{with}  \, \,    \caR^{\tau}_{\mathrm{ir}}(z) :=     \int_{\bbR^+} \d t \,   \e^{-t z}  \caZ^{\tau,\mathrm{ir}}_t 
\eeq
 The definition of $\caR^{\tau}_{ex}(z)$ and $\caR_{\mathrm{ld}}^{\tau}(z)$ relies on the following splitting of $\caR^{\tau}_{\mathrm{ir}}(z)$
\baq
\caR_{\mathrm{ld}}^{\tau}(z)  &: =&  \mathop{\int}\limits_{\bbR^+}  \d t \,  \e^{-t z}  \mathop{\int}\limits_{\Si_{[0,t]}( {< \tau}, \mathrm{ir} )} \zeta(\dsi)  1_{\str \dsi \str =1}   \caV_{[0,t]}(\dsi)    \label{def: car ladders}   \\
 \caR^\tau_{ex}(z) &: = &  \mathop{\int}\limits_{\bbR^+}  \d t  \,  \e^{-t z}    \mathop{\int}\limits_{\Si_{[0,t]}( {< \tau}, \mathrm{ir})} \zeta(\dsi) 1_{\str \dsi \str  \geq 2} \caV_{[0,t]}(\dsi)
\eaq
The subscripts refer to `'ladder'- and `excitation'-diagrams. The name `ladder' originates from the graphical representation of diagrams whose irreducible components consist of one pair (it is standard in condensed matter theory). 
Since obviously $\caR^{\tau}_{ex}(z)+\caR_{\mathrm{ld}}^{\tau}(z)= \caR^{\tau}_{\mathrm{ir}}(z)$, the relation \eqref{eq: laplace cutoff} implies Statement (1) of Lemma \ref{lem: polymer model cutoff}.

In the model without cutoff, we do not  disentangle ladder and excitation diagrams, since every diagram that contains a long pairing, is considered an excitation. We can thus define
\beq  \label{def: laplace excitations}
 \caR_{ex}(z) :=   \caR_{\mathrm{ir}}(z)-  \caR^{\tau}_{\mathrm{ir}}(z)
\eeq
We will come up with a more constructive representation of $ \caR_{ex}(z)$  in formula \eqref{eq: representation of laplace ex with C}.

\subsubsection{The reduced evolution as a a double integral over long and short diagrams} \label{sec: representation of caZ as double integral}

We develop a new representation of $\caZ_t^{\mathrm{ir}}$ by fixing the long diagrams, i.e., those in $\Si_{[0,t]}(>\tau)$, and integrating  the short ones.

We define the \emph{conditional cutoff dynamics}, $\caC_t(\dsil)$, depending on a long diagram $\dsil \in  \Si_{[0,t]}(>\tau)$, as follows:
\beq \label{def: C}
\caC_t({\dsi_l})  =  1_{\dsi_l \in \Si_{[0,t]}(>\tau, \mathrm{ir})}  \caV_{[0,t]}(\dsi_l) +    \mathop{\int}\limits_{\footnotesize{\left.\begin{array}{c} \Sigma_{[0,t]}({< \tau} ) \\   \dsi_l \cup \dsi  \in \Sigma_{[0,t]}(\mathrm{ir})    \end{array}\right.  }  } \hspace{-0.3cm}  \d \dsi \,    \zeta(\dsi)    \caV_{[0,t]}(\dsi \cup \dsi_l)
\eeq
In words, $\caC_t({\dsi_l})$ contains contributions of  short diagrams $\dsi \in  \Sigma_t({< \tau} )$ such that $\dsi_l \cup \dsi$ is irreducible in the interval $[0,t]$.   Hence, if $\dsi_l$ is itself irreducible in the interval $[0,t]$, then there is a term without any short diagrams; this is the first term in \eqref{def: C}.  In general, $\dsi_l$ need not be irreducible.
Note that the constraint on $\dsi$ (in the domain of the integral) in the second term of \eqref{def: C} depends crucially on the nature of $\dsi_l$. In particular, if $\Dom\dsi_l$ does not contain the boundary points $0$ or $t$, then  $\dsi$ has to contain $0$ or $t$, and this introduces one or two delta functions into the constraint on $\dsi$.  
To relate $\caC_t(\dsi_l)$ to $\caZ_t^{\mathrm{ir}}$, we must explicitly add those $\dsi_l$ that contain one or both of the times $0$ and $t$. This is visible in the  following formula, which follows 
from \eqref{def: C} and the definition of $\caZ_t^{\mathrm{ir}}$ in \eqref{def: irreducible evolutions}.
\beq  \label{eq: difference of caZ and caZ cutoff} 
\caZ_t^{\mathrm{ir}}-\caZ_t^{\tau, \mathrm{ir}} =     \mathop{\int}\limits_{\Sigma_{[0,t]}(> \tau)} \d \dsi_l   \zeta(\dsi_l)   \caC_t( \dsi_l)  \,  \left[1+\delta(t_1(\dsil) \right]  \, \left[1+\delta(t_{2\str \dsil\str}(\dsil)-t) \right]
\eeq

We must subtract $\caZ_t^{\tau,\mathrm{ir}} $ on the LHS, since all contributions to the RHS involve at least one long diagram. The $\delta$-functions on the RHS are defined as in \eqref{def: delta functions on boundary}.

The following formula is an obvious consequence of  \eqref{def: laplace excitations}  and \eqref{eq: difference of caZ and caZ cutoff}:
\beq \label{eq: representation of laplace ex with C}
 \caR_{ex}(z)   =   \mathop{\int}\limits_{\bbR^+} \d t  \, \e^{-tz}  \mathop{\int}\limits_{\Sigma_{[0,t]}(> \tau)}   \d \dsil    \zeta(\dsil)   \caC_t(\dsil)  \,  \left[1+\delta(t_1(\dsil) \right]  \, \left[1+\delta(t_{2\str \dsil\str}(\dsil)-t) \right]
\eeq
All of Section \ref{sec: bounds on long pairings}  will be devoted to proving good bounds on $ \caR_{ex}(z) $, as claimed in  Lemma \ref{lem: polymer model full}.

\subsection{Decomposition of the conditional cutoff dynamics $\caC_t(\dsi_l) $} \label{sec: decomposition of caC}
Our next step is to decompose the conditional cutoff dynamics $\caC_t(\dsi_l) $, as defined in \eqref{def: C}, into components.   Since $\caC_t(\dsi_l) $ is defined as an integral over short diagrams $\dsi$, we can achieve this by classifying the short diagrams $\dsi$ that contribute to this integral.  The idea is to look at the irreducible components of $\dsi$ whose domain contains one or more of the time-coordinates of $\dsi_l$ (In our final formula, \eqref{formula caC}, these domains correspond to the intervals $[s^{\mathrm{i}}_k,s^{\mathrm{f}}_k]$). The irreducible components whose domain does not contain any of the time coordinates of $\dsi_l$ can be resummed right away, and they do not play a role in our classification (this corresponds to the operators $\caZ_{t}^{\tau}$ in \eqref{formula caC}).
We outline  the abstract decomposition procedure in Section \ref{sec: vertices and vertex partitions}, and we present an example (with figures) in Section \ref{sec: examples vertex partition}.

\subsubsection{Vertices and vertex partitions} \label{sec: vertices and vertex partitions}

Consider a long diagram $\dsi_l \in \Sigma_{[0,t]}(> \tau)$ with $\str \dsi_l \str=n$ and time-coordinates $\underline{t}(\dsi_l)= (t_1,\ldots, t_{2n})$. With this diagram, we will associate different 
 \emph{vertex partitions} $\frL$. First, we define \emph{vertices}.
A vertex  $\frl$  is determined by a label, \emph{bare} or \emph{dressed},  and a vertex set $S(\frl)$, given by 
\beq
S(\frl) =\{ t_{j},t_{j+1} \ldots, t_{j+m-1}  \}, \qquad \textrm{for some} \, 1\leq j < j+m-1 \leq 2n
\eeq
Hence, the vertex set is a a subset of the times $\{t_1(\dsil), \ldots, t_{2n}(\dsil) \}$.
Moreover, a vertex $\frl$ with $\str S(\frl) \str >1$ is always dressed. Hence,  if $\frl$ is bare then  $S(\frl)$ is necessarily a singleton, i.e., $ S(\frl) =\{ t_j\}$ for some $j$.  

A vertex partition $\frL$ compatible with $\dsi_l$ (Notation: $\frL \sim \dsil$) is a collection of vertices $\frl_1,\ldots,\frl_m$ such that  
\begin{itemize}
\item
 The vertex sets $S(\frl_1), \ldots, S(\frl_p)$ form a partition of $\{t_1(\dsil), \ldots, t_{2n}(\dsil) \}$. By convention, we always number the vertices in a vertex partition such that the elements of $S(\frl_k)$ are smaller than those of $S(\frl_{k+1})$. The number, $p$, of vertices in a vertex partition is called the cardinality of the vertex partition and  is denoted by $\str\frL\str$.
\item
Any two consecutive times $t_j,t_{j+1}$ such that  $[t_j, t_{j+1}]  \not\subset \Dom \dsil$, belong to the vertex set $S(\frl_k)$ of one of the vertices $\frl_k$. Such a vertex $\frl_k$ is  necessarily dressed since its vertex set contains at least two elements.
\item If $t_1=0$, then $S(\frl_1)=\{ t_1 \}$ and $\frl_1$ is bare. If  $t_1 >0$, then $S(\frl_1) \ni t_1$ and $\frl_1$ is dressed.
\item If $t_{2n}=t$, then $S(\frl_m)=\{ t_{2n} \}$ and $\frl_m$ is bare. If  $t_{2n} < t$, then $S(\frl_m) \ni t_{2n}$ and $\frl_m$ is dressed.
\end{itemize}

The idea is to split
\beq \label{eq: decomposition of cac into pi}
\caC_t(\dsi_l)= \sum_{ \frL \sim \dsil} \caC_t(\dsi_l, \frL)
\eeq
where  the sum is over all $\frL$ compatible with $\dsil$ and
$\caC_t(\dsi_l, \frL)$ contains the contributions of all short pairings $\dsi$ that \emph{match} the vertex partition $\frL$; 

\beq
\dsi\,\,  \textrm{matches} \,   \, \frL \quad \Leftrightarrow  \quad \left\{\begin{array}{l} 
 \forall \,\,  \textrm{dressed}\,\, \frl_k :  \exists! \,  \textrm{irr. component}  \, \dsi_j \subset \dsi  \, \, \textrm{such that} \, \,  \, S(\frl_k) \subset \Dom \dsi_j   \\ \textrm{and} \, \, S(\frl_{k'}) \cap  \Dom \dsi_j =\emptyset \, \,   \textrm{for all} \,   k' \neq k
\\[2mm]
\forall  \,\,  \textrm{bare}\,\, \frl_k : S(\frl_k) \cap \Dom \dsi = \emptyset
\end{array} \right.
\eeq 

For the sake of completeness, we define  the operators $\caC_t(\dsi_l, \frL) $, below, but, in  Section \ref{sec: abstract definition of vertex operators}, we will provide a more constructive expression for them. 
First, assume that the vertex partition $\frL$  contains at least one dressed vertex. Then $\caC_t(\dsi_l, \frL) $ is defined by restricting the second integral in  \eqref{def: C} to those $\dsi$ that match $\frL$;
\beq  \label{def: caC for vertex partition}
\caC_t(\dsi_l, \frL) :=
 \mathop{\int}\limits_{\footnotesize{\left.\begin{array}{c} \Sigma_{[0,t]}({< \tau} ) \\   \dsi_l \cup \dsi  \in \Sigma_{[0,t]}(\mathrm{ir})    \end{array}\right.  }  }
\d \dsi  \zeta(\dsi) \caV_{[0,t]}(\dsi\cup \dsil)   
1_{\dsi\, \textrm{matches}  \, \frL} 
\eeq
Next, we assume that the vertex partition $\frL$ contains only bare vertices. If $\dsi_l$ is irreducible in the interval $[0,t]$, i.e., $\dsil \in \Si_{[0,t]}(>\tau, \mathrm{ir})$, then the vertex partition  with only bare vertices is compatible with $\dsil$. If $\dsil \notin \Si_{[0,t]}(>\tau, \mathrm{ir})$, then this vertex partition is  not compatible with $\dsil$.  Hence, we assume that $\dsil \in \Si_{[0,t]}(>\tau, \mathrm{ir})$ and we define
\beq \label{def: caC for bare vertex partition}
\caC_t(\dsi_l, \frL) :=  \caV_{[0,t]}(\dsi_l)+  \mathop{\int}\limits_{\Si_{[0,t]}(<\tau)} \d \dsi \, \zeta(\dsi) \caV_{[0,t]}(\dsi\cup \dsil) \, \times \,  1_{ \left\{  \{ t_1,\ldots, t_{2 \str \dsi_l\str}   \}   \cap  \Dom \dsi =\emptyset  \right\}}
\eeq
The second term is the same as in \eqref{def: caC for vertex partition} (but specialized to the partition with only bare vertices) and the first term is a contribution without any short diagrams.  This first term equals the first term (on the RHS) of \eqref{def: C}. 
%

In Section \ref{sec: examples vertex partition}, we give examples of vertex partitions that are intended to render the above concepts more intuitive. 

%
%
 
\subsubsection{Examples of vertex partitions} \label{sec: examples vertex partition}
 We choose a long diagram $\dsi_l \in \Si^3_{[0,t]}(>\tau)$ which consists of three pairs such that the time coordinates $(u_i,v_i)_{i=1}^3$ are ordered as
 \beq
  \left.\begin{array}{rccccccccccccccl}    &&  t_1 & &   t_2 &&   t_3 &&  t_4  &&   t_5  &&  t_6    \\
0 &<&  u_1 &<&   u_2 &<& v_1 &<& u_3 &<& v_2 &< & v_3  &=& t      \end{array}\right. \label{eq: times of a diagram}
 \eeq
 Hence, $\dsi_l$ is irreducible in the interval $[u_1,t]$, but not in the interval $[0,t]$, at least not if $t_1 \neq 0$.  
 
 \begin{figure}[h!] 
\vspace{0.5cm}
\begin{center}
\psfrag{t0}{ $\scriptstyle{0}$}
\psfrag{t1}{ $\scriptstyle{t_1}$}
\psfrag{t2}{ $\scriptstyle{t_2}$}
\psfrag{t3}{ $\scriptstyle{t_3}$}
\psfrag{t4}{ $\scriptstyle{t_4}$}
\psfrag{t5}{ $\scriptstyle{t_5}$}
\psfrag{t6}{ $\scriptstyle{t_6}$}
\psfrag{b1}{ $\scriptstyle{\frl_1}$}
\psfrag{b2}{ $\scriptstyle{\frl_2}$}
\psfrag{b3}{ $\scriptstyle{\frl_3}$}
\psfrag{b4}{ $\scriptstyle{\frl_4}$}
\psfrag{b5}{ $\scriptstyle{\frl_5}$}
\psfrag{b6}{ $\scriptstyle{\frl_6}$}
\psfrag{init1}{ $\scriptstyle{s_1}$}
\psfrag{init2}{ $\scriptstyle{s_2}$}
\psfrag{init3}{ $\scriptstyle{s_3}$}
\psfrag{init4}{ $\scriptstyle{s_4}$}
\psfrag{2init1}{ $\scriptstyle{s_1}$}
\psfrag{2init2}{ $\scriptstyle{s_2}$}
\psfrag{3init1}{ $\scriptstyle{s_1}$}
\psfrag{3init2}{ $\scriptstyle{s_2}$}
\psfrag{fint1}{ $\scriptstyle{s_1}$}
\psfrag{fint2}{ $\scriptstyle{s_2}$}
\psfrag{fint3}{ $\scriptstyle{s_3}$}
\psfrag{fint4}{ $\scriptstyle{s_4}$}
\psfrag{2fint1}{ $\scriptstyle{s_1}$}
\psfrag{2fint2}{ $\scriptstyle{s_2}$}
\psfrag{3fint1}{ $\scriptstyle{s_1}$}
\psfrag{3fint2}{ $\scriptstyle{s_2}$}
\includegraphics[width = 18cm, height=3.5cm]{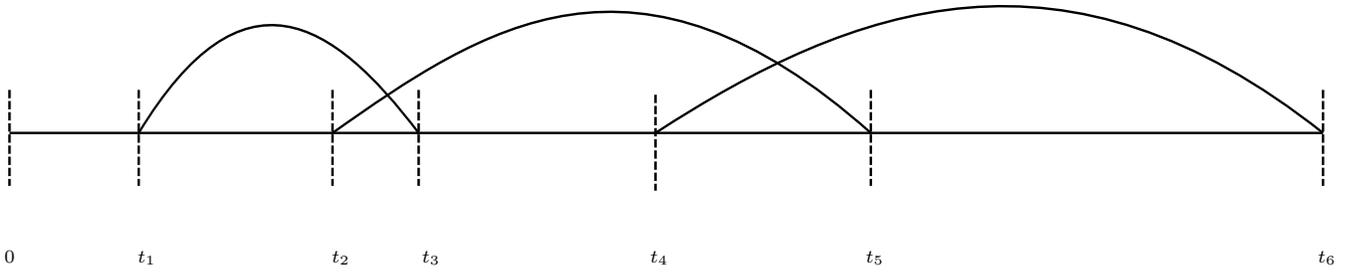}
\caption{ \footnotesize{A long diagram with time coordinates as in \eqref{eq: times of a diagram}} \label{fig: explanation 0} }
\end{center}
 \end{figure}

 Below we display three diagrams $ \dsi \in \Si_{[0,t]}({< \tau}) $ satisfying the condition $\dsi_l \cup \dsi \in \Si_{[0,t]}(\mathrm{ir} )$. 
  To assign to each of those diagrams a vertex partition, we proceed as follows. 
 Starting on the left, we look at the  time-coordinates $\underline{t}(\dsi_l)$ and we check whether these times are 'bridged' by a short pairing, i.e.,  whether they belong to the domain of a short diagram. If this is the case then such a time belongs to the vertex set of a dressed vertex. The vertex set of this vertex is the set of all time-coordinates that are connected to this point by short pairings. 
If this is not the case, i.e., if a time-coordinate of $\dsi_l$ is not 'bridged' by any short pairing,  than such a point constitutes a bare vertex, whose vertex set is just this one point.

 Actually, for the first time-coordinate (in our case $u_1$), this is particularly simple. Either the first time-coordinate  is \textbf{not} equal to $0$, in which case it has to be 'connected' by short pairings to $0$ (indeed, if this were not the case, then $\dsi_l \cup \dsi$ cannot be irreducible in $[0,t]$), or the first time-coordinate \textbf{is} equal to $0$, in which case it cannot be connected by short pairings to the second coordinate, because then the first time-coordinate of the short diagram would have to be $0$ as well, which is a  zero measure event (for this reason, we have excluded this case in the definition of the diagrams in Section \ref{sec: diagrams}).  In our example $u_1 \neq 0$, and  one checks that, in all three choices of $\dsi $, there are short diagrams connecting $u_1$ and $0$.

\begin{figure}[h!] 
\vspace{0.5cm}
\begin{center}
\psfrag{t0}{ $\scriptstyle{0}$}
\psfrag{t1}{ $\scriptstyle{t_1}$}
\psfrag{t2}{ $\scriptstyle{t_2}$}
\psfrag{t3}{ $\scriptstyle{t_3}$}
\psfrag{t4}{ $\scriptstyle{t_4}$}
\psfrag{t5}{ $\scriptstyle{t_5}$}
\psfrag{t6}{ $\scriptstyle{t_6}$}
\psfrag{b1}{ $\scriptstyle{\frl_1}$}
\psfrag{b2}{ $\scriptstyle{\frl_2}$}
\psfrag{b3}{ $\scriptstyle{\frl_3}$}
\psfrag{b4}{ $\scriptstyle{\frl_4}$}
\psfrag{b5}{ $\scriptstyle{\frl_5}$}
\psfrag{b6}{ $\scriptstyle{\frl_6}$}
\psfrag{init1}{ $\scriptstyle{s^{\mathrm{i}}_1}$}
\psfrag{init2}{ $\scriptstyle{s^{\mathrm{i}}_2}$}
\psfrag{init3}{ $\scriptstyle{s^{\mathrm{i}}_3}$}
\psfrag{init4}{ $\scriptstyle{s^{\mathrm{i}}_4}$}
\psfrag{2init1}{ $\scriptstyle{s^{\mathrm{i}}_1}$}
\psfrag{2init2}{ $\scriptstyle{s^{\mathrm{i}}_2}$}
\psfrag{3init1}{ $\scriptstyle{s^{\mathrm{i}}_1}$}
\psfrag{3init2}{ $\scriptstyle{s^{\mathrm{i}}_2}$}
\psfrag{fint1}{ $\scriptstyle{s^{\mathrm{f}}_1}$}
\psfrag{fint2}{ $\scriptstyle{s^{\mathrm{f}}_2}$}
\psfrag{fint3}{ $\scriptstyle{s^{\mathrm{f}}_3}$}
\psfrag{fint4}{ $\scriptstyle{s^{\mathrm{f}}_4}$}
\psfrag{2fint1}{ $\scriptstyle{s^{\mathrm{f}}_1}$}
\psfrag{2fint2}{ $\scriptstyle{s^{\mathrm{f}}_2}$}
\psfrag{3fint1}{ $\scriptstyle{s^{\mathrm{f}}_1}$}
\psfrag{3fint2}{ $\scriptstyle{s^{\mathrm{f}}_2}$}
\includegraphics[width = 18cm, height=4cm]{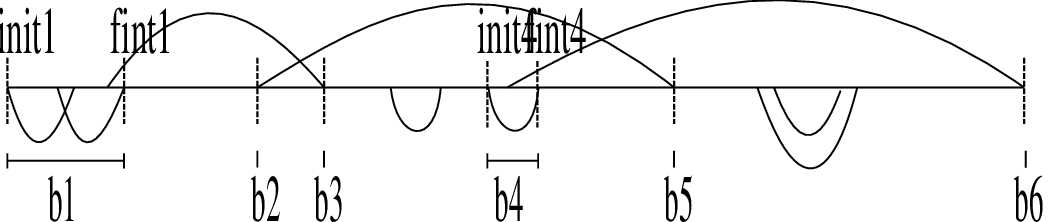} \\
\vspace{1cm}
\includegraphics[width = 18cm, height=4cm]{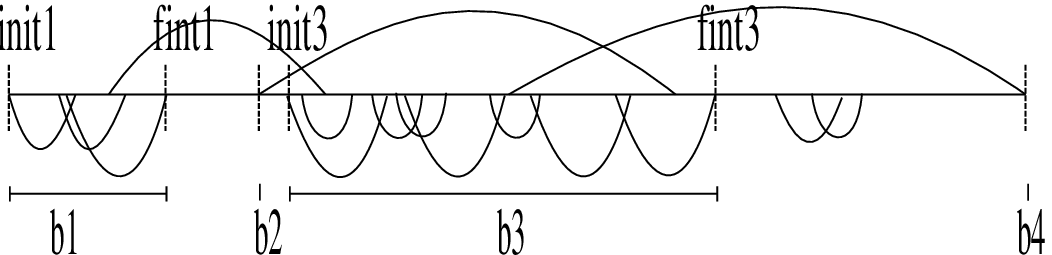}\\
\vspace{1cm}
\includegraphics[width = 18cm, height=4cm]{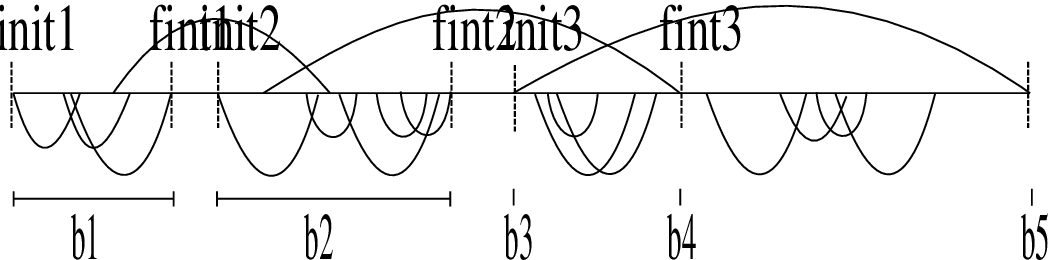}
\caption{ \footnotesize{The picture shows three different choices of short diagrams $ \dsi \in \Si_{[0,t]}({< \tau}) $.  Recall that short diagrams are drawn below the horizontal (time) axis. In each picture, we show the resulting vertex partition by listing the vertices $\frl_1,\frl_2, \ldots$. The dressed vertices are denoted by a horizontal bar whose endpoints represent the vertex time-coordinates $s^{\mathrm{i}}_k,s^{\mathrm{f}}_k$. The  bare vertices are denoted by a short vertical line whose position represents the (dummy) vertex time coordinates $s^{\mathrm{i}}_k =s^{\mathrm{f}}_k=t_j$. The time-coordinates of the bare vertices are not shown since they coincide with time-coordinates of long pairings.
 For example, in the bottom picture, $\frl_1,\frl_2$ are dressed and $\frl_3,\frl_4,\frl_5$ are bare. } \label{fig: vertexpartitions} }
\end{center}

\end{figure}

 Let us determine the vertices in the three displayed figures
 
 \begin{center}
\begin{tabular}{ccc}    vertex partition $1$  &     vertex partition $2$   &   vertex partition $3$   \\[2mm]
 \begin{tabular}{|c|c|c|}\hline 
      $\frl_1$&  $\{ t_1 \}$     &  dressed         \\[1mm] 
      \hline
       $\frl_2$&  $\{ t_2 \}$     &  bare       \\[1mm]         \hline
        $\frl_3$&  $\{ t_3 \}$     &  bare         \\[1mm]         \hline
          $\frl_4$&  $\{ t_4 \}$     &  dressed         \\[1mm] 
      \hline
       $\frl_5$&  $\{ t_5 \}$     &  bare        \\[1mm]       \hline
        $\frl_6$&  $\{ t_6 \}$     &  bare         \\[1mm] 
 \hline 
\end{tabular}   \qquad
&
 \begin{tabular}{|c|c|c|}\hline 
      $\frl_1$&  $\{ t_1 \}$     &  dressed         \\[1mm] 
      \hline
       $\frl_2$&  $\{ t_2 \}$     &  bare         \\[1mm]         \hline
        $\frl_3$&  $\{ t_3,  t_4, t_5 \}$     &  dressed         \\[1mm]         \hline
          $\frl_4$&  $\{ t_6 \}$     &  bare       \\[1mm] 
      \hline
\end{tabular}   \qquad
&
 \begin{tabular}{|c|c|c|}\hline 
      $\frl_1$&  $\{ t_1 \}$     &  dressed         \\[1mm] 
      \hline
       $\frl_2$&  $\{ t_2, t_3 \}$     &  dressed       \\[1mm]         \hline
        $\frl_3$&  $\{ t_4 \}$     &  bare         \\[1mm]         \hline
          $\frl_4$&  $\{ t_5 \}$     & bare       \\[1mm] 
      \hline
       $\frl_5$&  $\{ t_6 \}$     &  bare        \\[1mm]       \hline
\end{tabular} 
 \end{tabular}   \\ [2mm]
 \end{center}

 In the example displayed above, it is also very easy to determine which vertex partitions $\frL$  are compatible with $\dsi_l$ ($\frL \sim \dsi_l$). 
 Apart from the fact that the vertex sets $S(\frl_k)$ of the vertices in $\frL$ have to form a partition of $\{t_1, \ldots, t_6\}$, we need that $\frl_1$ is dressed and  $\frl_{\str\frL\str}$ (the last vertex in the partition) is bare. 
 
  To each vertex $\frl_k$ in the above examples, we can associate time coordinates $s^{\mathrm{i}}_k$ and $s^{\mathrm{f}}_k$ as the boundary times of the domains of irreducible diagrams bridging the times in the vertex. 
  Eventually, we intend to fix a vertex partition and associated time coordinates $\underline{s}^{i}$ and $\underline{s}^{f}$ and to integrate over all short diagrams that are irreducible in the interval $[{s}^{i}, {s}^{f}]$. This integration gives rise to the vertex operators, see eq.\ \eqref{def: vertex operator}. 
  To illustrate this, we zoom in on a part of a long diagram, shown in Figure \ref{fig: vertexoperator}. A formal definition is given in the next section.
 
    \begin{figure}[h!] 
\vspace{0.5cm}
\begin{center}
\psfrag{t5}{ $\scriptstyle{t_{j+4} }$}
\psfrag{t4}{ $\scriptstyle{ t_{j+3} }$}
\psfrag{t3}{ $\scriptstyle{ t_{j+2} }$}
\psfrag{t2}{ $\scriptstyle{ t_{j+1}  }$}
\psfrag{t1}{ $\scriptstyle{ t_{j}  }$}
    \psfrag{initialvertextime}{$s^{\mathrm{i}}$}   
       \psfrag{finalvertextime}{$s^{\mathrm{f}}$}   
          \psfrag{shortdiagrams}{Short diagrams $\dsi$}   
      \psfrag{operator}{$\caB(\frl, s^{\mathrm{i}},s^{\mathrm{f}}) =      \mathop{\int}\limits_{ \Si_{[s^{\mathrm{i}},s^{\mathrm{f}}]} (<\tau, \mathrm{ir} )    } \!  \! \d \dsi     $}   
 \includegraphics[width = 12cm, height=3cm]{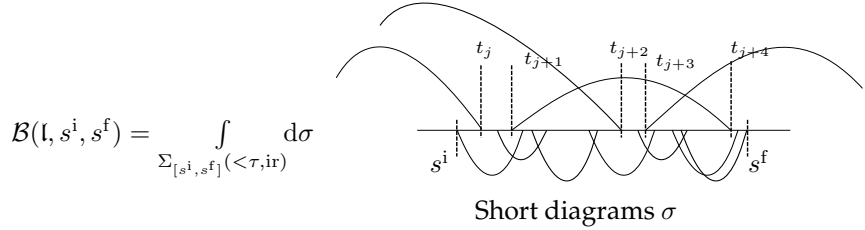}  
\caption{ \footnotesize{A part of a long diagram $\dsi_l \in \Si_{[0,t]} (> \tau)$ is shown, suggesting a dressed vertex $\frl$ with vertex set $ S(\frl)=\{ t_{j},\ldots, t_{j+4}\}$. The end points of the pairings that are 'floating' in the air are immaterial to this vertex, as long as they land on the time-axis outside the interval $[s^{\mathrm{i}},s^{\mathrm{f}}]$. 
The vertex operator $\caB(\frl, s^{\mathrm{i}},s^{\mathrm{f}})$ is obtained by integrating  all  diagrams in  $\Si_{[{s}^{i}, {s}^{f}]}({< \tau},\mathrm{ir} )$.
 } \label{fig: vertexoperator} }
\end{center}
\end{figure}

 \subsubsection{Abstract definition of the vertex operator} \label{sec: abstract definition of vertex operators}

   Let $\frL$ be a vertex partition compatible with $\dsil$, with vertices $\frl_k, k=1, \ldots, \str \frL \str$.
   In what follows, we focus on one particular vertex $\frl_k$ which we assume first to be dressed.    
   The vertex $\frl_k$ is assumed to have a vertex set $S(\frl_k) = \{t_{j}, t_{j+1}, \ldots, t_{j+m-1}\} $.  This means in particular that the time-coordinate $t_{j-1}$ belongs to the vertex set of the vertex $\frl_{k-1}$ (unless $j=1$) and the time-coordinate $t_{j+m}$ belongs to the vertex set of the vertex $\frl_{k+1}$ (unless $j+m-1=2\str \dsil\str$).
   We fix an initial time $s^{\mathrm{i}}_k$  and final time $s^{\mathrm{f}}_k$ such that 
   \beq
 t_{j-1}  \leq     s^{\mathrm{i}}_k \leq t_j    \leq   t_{j+m-1}  \leq     s^{\mathrm{f}}_k \leq t_{j+m}
   \eeq
where it is understood that $ t_{j-1}  =0$ if $j=1$ and $t_{j+m}=t$ if $j+m-1=2\str \dsil\str$. 
   The vertex operator  $\caB(\frl_k, s^{\mathrm{i}}_k,s^{\mathrm{f}}_k)  $ is defined by summing the contributions of all $\dsi \in {\Si_{[ s^{\mathrm{i}}_k,s^{\mathrm{f}}_k]}  ({< \tau},\mathrm{ir}) } $

 To write a formula for the vertex operator $\caB(\frl_k, s^{\mathrm{i}}_k,s^{\mathrm{f}}_k)  $, we  need to relabel the time-coordinates of $\dsi_l$ and  $\dsi \in {\Si_{[ s^{\mathrm{i}}_k,s^{\mathrm{f}}_k]}  ({< \tau},\mathrm{ir} ) } $.
 
Consider the  $m$ triples $(t_i(\dsil), x_i(\dsil),l_i(\dsil))$, for $i=j,\ldots,j+ m-1$, i.e.,  a subset of the $2 \str \dsil\str$ triples determined by the long diagram $\dsil$, and the $2 \str \dsi \str$ triples $t_i(\dsi), x_i(\dsi), l_i(\dsi)$ with $i=1,\ldots, 2 \str \dsi \str$ determined by  $\dsi \in {\Si_{[ s^{\mathrm{i}}_k,s^{\mathrm{f}}_k]}  ({< \tau},\mathrm{ir} ) } $.
We now define the triples
 $(t''_i, x''_i,l''_i)_{i}^{m+ 2 \str \dsi \str}$ by time-ordering (i.e.\ such that  $t''_i \leq t''_{i+1}$) of  the union of  triples 
 \beq (t_i(\dsi), x_i(\dsi), l_i(\dsi))_{i=1}^{2 \str \dsi\str} \quad  \textrm{and} \quad  (t_i(\dsil), x_i(\dsil),l_i(\dsil))_{i=j}^{j+m-1}. \eeq

The vertex operator $\caB(\frl_k, s^{\mathrm{i}}_k,s^{\mathrm{f}}_k) $ is then defined as follows 
\beq \label{def: vertex operator}
\caB(\frl_k, s^{\mathrm{i}}_k,s^{\mathrm{f}}_k)  := \mathop{\int}\limits_{\Si_{[ s^{\mathrm{i}}_k,s^{\mathrm{f}}_k]}  ({< \tau}, \mathrm{ir} ) }  \d \dsi  \,   \zeta(\dsi)   \caV_{[ s^{\mathrm{i}}_k,s^{\mathrm{f}}_k]}  \left((t''_i, x''_i,l''_i)_{i=1}^{m+ 2 \str \dsi \str}   \right)\eeq
where the dependence of the integrand on $\dsi$ is implicit in the above definition of the triples  $(t''_i, x''_i,l''_i)$.
The double primes in the coordinates $(t''_i, x''_i,l''_i)$ are supposed to render the comparison with later formulas easier.

We now treat the simple case in which the vertex  $\frl_k$ is bare. In that case, there is a $j$ such that $S(\frl_k)=\{ t_j\}$ and the vertex operator is simply defined as
\beq
 \caB(\frl_k, s^{\mathrm{i}}_k,s^{\mathrm{f}}_k)  :=   \caI_{x_j,l_j}  , \qquad     s^{\mathrm{i}}_k=s^{\mathrm{f}}_k=t_j
\eeq
Hence, in this case,  the vertex time-coordinates $s^{\mathrm{i}}_k,s^{\mathrm{f}}_k$ are dummy coordinates, see also Figure \ref{fig: vertexpartitions}.

\subsubsection{The operator $\caC_t(\dsil, \frL)$ as an integral over time-coordinates of vertex operators}\label{sec: integral over vertex operators}

 We are  ready to give a constructive formula for $\caC_t(\dsil, \frL)$, as announced in Section \ref{sec: vertices and vertex partitions}. 
First, we define the integration measure over the vertex time-coordinates $s^{\mathrm{i}}_k,s^{\mathrm{f}}_k$;
   \beq  \label{eq: integral over vertex operators}
\caD \underline{s}^{\mathrm{i}} \caD  \underline{s}^{\mathrm{f}} :=   \mathop{\prod}\limits_{\footnotesize{ \left.\begin{array}{c}   k=1,\ldots,\str \frL \str \\  \frl_k \, \, \textrm{dressed}     \end{array} \right. }    }  \d s^{\mathrm{i}}_k\d s^{\mathrm{f}}_k   
  \quad    \left\{ \begin{array}{ll}  \delta(s^{\mathrm{i}}_1)  &   \frl_1\, \, \textrm{dressed} \\  1  &   \frl_1 \,\, \textrm{bare}      \end{array} \right\}  
   \times 
    \left\{ \begin{array}{ll} \delta(s^{\mathrm{f}}_{\str \frL \str}-t )  &   \frl_{\str \frL \str} \,\,  \textrm{dressed} \\  1  &  \frl_{\str \frL \str} \, \, \textrm{bare}      \end{array} \right\}  
   \eeq


To understand this formula, we observe that only non-dummy vertex time coordinates need to be integrated over.  A dummy vertex time coordinate is a time coordinate whose value is a-priori fixed by  $\dsi_l$ and the vertex partition $\frL$.
The non-dummy times are the time coordinates of the dressed vertices, except at the temporal boundaries $0,t$, where such a time coordinate is also a dummy coordinate. 
The terms between $\{ \cdot \}$-brackets in formula \eqref{eq: integral over vertex operators} take care of this .
Finally, the formula for $   \caC_t(\dsil, \frL) $ is 
 \beq\label{formula caC}
   \caC_t(\dsil, \frL)    =   \mathop{\int}\limits_{\footnotesize{\left.\begin{array}{c}  0 <s^{\mathrm{i}}_k < s^{\mathrm{f}}_k < t  \\[1mm] 
   s^{\mathrm{f}}_k  < s^{\mathrm{i}}_{k'}    \,  \textrm{for} \,  k' >k   \end{array}\right.  }  }
    \caD  \underline{s}^{\mathrm{i}} \caD  \underline{s}^{\mathrm{f}}  \,\,    \caB(\frl_{\str \frL\str }, s^{\mathrm{i}}_{\str \frL\str },s^{\mathrm{f}}_{\str \frL\str }) 
       \caZ^{\tau}_{s^{\mathrm{i}}_{\str \frL\str }-s^{\mathrm{f}}_{{\str \frL\str }-1} }   \ldots     \caZ^{\tau}_{s^{\mathrm{i}}_3-s^{\mathrm{f}}_2}   \caB(\frl_2, s^{\mathrm{i}}_2,s^{\mathrm{f}}_2)  \caZ^{\tau}_{s^{\mathrm{i}}_2-s^{\mathrm{f}}_1}   \caB(\frl_1, s^{\mathrm{i}}_1,s^{\mathrm{f}}_1) 
 \eeq
where the indices $k,k'$ correspond to vertices $\frl_k,\frl_{k'}$: only the time-coordinates of dressed vertices are integrated over, even though all vertices appear on the RHS. 
This formula can be checked from the definition \eqref{def: caC for vertex partition} and the explicit expressions for the vertex operators $\caB(\cdot;\cdot,\cdot)$ above.  The cutoff reduced dynamics $\caZ_t^\tau$ in \eqref{formula caC} appears by summing the small diagrams between the vertices, using  formula \eqref{eq: caZ  cutoff as integral over diagrams}.

\section{The sum over "small" diagrams} \label{sec: bounds on long pairings}

In this section, we establish two results.  First, we analyze the cutoff-dynamics $\caZ_t^\tau$. The main bound is stated in Lemma \ref{lem: bound on small pairings}, and a proof of Lemma \ref{lem: polymer model cutoff} (concerning the Laplace transform of $\caZ^\tau_t$) is outlined immediately after Lemma \ref{lem: bound on small pairings}. 
Second, we resum the small subdiagrams within a general irreducible diagram: Recall that the conditional cutoff dynamics $\caC_t(\dsi_l)$ is defined as the sum over all irreducible  diagrams in $[0,t]$ containing the long diagram $\dsi_l$. In Lemma \ref{lem: irreducible bounds with caE}, we obtain a description of $\caC_t(\dsi_l)$ that does not involve any small diagrams.  In this sense, we have performed a \emph{blocking procedure}, getting rid of information on time-scales smaller than $\tau$.

Since this section uses parameters and constants that were introduced earlier in the paper, we encourage the reader to consult the overview tables in Section \ref{sec: list of parameters}.

\subsection{Generic constants} \label{sec: generic constants}

In Sections \ref{sec: bounds on long pairings} and \ref{sec: model only long}, we will state bounds that will depend in a crucial way on the parameters $\la, \ga$ and $\tau$. The parameter $\ga$ is a momentum-like variable used to bound matrix elements in position representation, see below in Section \ref{sec: bounding operators}. It appeared first in Section \ref{sec: bound on markov in position representation}.  To simplify the presentation, we introduce the following notation and conventions. 
\begin{itemize}
\item We write $\cone$ for functions of $\ga \geq 0$ with the property that $\cone$ is decreasing as $\ga \nearrow \infty$, and $\cone$ is finite, except, possibly, at $\ga=0$. It is understood that $\cone$  is independent of $\la$.
\item We write $\ctwo$ for functions of $\ga \geq 0$ and $\la \in \bbR $  that have the asymptotics
\beq
\ctwo = o(\ga^0)O(\la^2)+  o(\la^2), \qquad    \ga \to 0,\,  \la \to 0
\eeq
\item We write $\cthree$ for functions of $\ga \geq 0$ and $\la \in \bbR $  that have the asymptotics
\beq
\cthree =  o(\ga^0)O(\la^2)+   \cone  o(\la^2), \qquad    \ga \to 0,\,  \la \to 0
\eeq
\item The cutoff time $\tau=\tau(\la)$ is treated as an implicit function of $\la$, satisfying \eqref{def: condition on tau}.  In particular, $\ctwo$ and $\cthree$ can depend on $\tau$. 
\end{itemize}

\subsection{Bounds in the sense of matrix elements}\label{sec: bounds in the sense of me}

In Section \ref{sec: translation invariance and fiber decomposition} , we introduced  the kernel notation $\caA_{x^{}_\links,x^{}_{\rechts};x'_\links,x'_\rechts}$, for operators $\caA$ on $\caB_2(l^2(\lat,\scrS))$;  $\caA_{x^{}_\links,x^{}_{\rechts};x'_\links,x'_\rechts}$ is an element of $\scrB(\scrB_2(\scrS))$ such that
\beq
    \langle S, \caA S' \rangle_{}= \mathop{\sum}\limits_{x_\links,x_{\rechts};x'_\links,x'_\rechts}   \langle S(x_\links,x_{\rechts}), \caA_{x^{}_\links,x^{}_{\rechts};x'_\links,x'_\rechts}S'(x'_\links,x'_\rechts) \rangle^{}_{\scrB_2(\scrS)}
    \eeq
First, we introduce a notion that allows us to bound operators $\caA$ by their `matrix elements' $\caA_{x^{}_\links,x^{}_{\rechts};x'_\links,x'_\rechts}$. 
\begin{definition}\label{def: bounds me}
Let $\caA$ and $\tilde \caA$ be operators on $\caB_2(l^2(\lat,\scrS))$ and   $\caB_2(l^2(\lat))$, respectively.  We say that $\tilde \caA$ dominates $\caA$ in `the sense of matrix elements', denoted by
\beq
\caA \mathop{\leq}\limits_{m.e.} \tilde \caA,
\eeq
iff
\beq
\norm \caA_{x^{}_\links,x^{}_{\rechts};x'_\links,x'_\rechts}  \norm_{\scrB(\scrB_2(\scrS))} \leq  \tilde\caA_{x^{}_\links,x^{}_{\rechts};x'_\links,x'_\rechts}
\eeq
\end{definition}
Note that, if $\caA$ is an operator on $\caB_2(l^2(\lat))$,  the inequality $\caA \mathop{\leq}\limits_{m.e.} \tilde \caA$ literally means that the absolute values of the matrix elements of $\caA$ are smaller than the matrix elements of $\tilde \caA$.
We will need the following implication 
\beq \label{eq: meleq implies norm}
\caA \mathop{\leq}\limits_{m.e.} \tilde \caA   \qquad  \Rightarrow \qquad \norm \caA \norm  \leq  \norm \tilde \caA \norm 
\eeq
Indeed, for any $S \in \scrB_2(l^2(\lat) \otimes \scrS) \sim   l^2(\lat \times \lat,\scrB_2(\scrS))$, we construct 
\beq
\tilde S(x_\links, x_\rechts) :=   \norm S(x_\links, x_\rechts) \norm_{\scrB_2(\scrS)}
\eeq
such that $ \norm \tilde S \norm_{l^2(\lat \times \lat )} = \norm S \norm_{l^2(\lat \times \lat, \scrB_2(\scrS) )}$  and
\beq
         \str  \langle S ,  \caA S'  \rangle  \str          \leq     \langle \tilde S ,  \tilde \caA   \tilde S'  \rangle
\eeq 
from which \eqref{eq: meleq implies norm} follows. 

\subsection{Bounding operators} \label{sec: bounding operators}
We introduce operators on  $\scrB_2(l^2(\lat))$  that  will be used as upper bounds `in the sense of  matrix elements', as defined above.
These bounding operators will depend on the coupling constant $\la$, the conjugation parameter $\ga>0$ and the cutoff time $\tau =\tau(\la)$.  Let  the function $r_\tau(\ga,\la)$ and the constants $c_{\caZ}^1, c_{\caZ}^2$ be as defined in Lemma \ref{lem: bounds on cutoff dynamics} and, in addition, let
\beq
r_{\ve}(\ga,\la):= 2\la^2 q_{\ve}(2\ga), \qquad \textrm{for}\, \, 2 \ga \leq \de_{\ve},
\eeq
with $q_{\ve}(\cdot)$ and $\de_{\ve}$ as in Assumption \ref{ass: analytic dispersion}.
We define
\baq
(\tilde{\caI}_{x,l})_{x^{}_\links,x^{}_{\rechts};x'_\links,x'_\rechts}  &:=&   \delta_{x^{}_\links,x^{'}_\links}    \delta_{x^{}_{\rechts},x^{'}_{\rechts}} \left(    \delta_{l=\links}  \delta_{x^{}_\links=x} +  \delta_{l=\rechts}  \delta_{x^{}_\rechts=x}    \right)    \\[2mm]
(\tilde{\caZ}^{\tau,\ga}_t)_{x^{}_\links,x^{}_{\rechts};x'_\links,x'_\rechts} &:=&     c_{\caZ}^1 \,       \e^{ r_\tau(\ga,\la)  t}    \e^{- \frac{\ga}{2} \left\str (x'_{\links}+x'_{\rechts})- (x_{\links}+x_{\rechts})\right\str}   \e^{- \ga \left\str x_{\links}-x_{\rechts}  \right\str }   \e^{- \ga \left\str x'_{\links}-x'_{\rechts}  \right\str }   \nonumber   \\[2mm]
 & + &      c_{\caZ}^2 \,   \e^{- \la^2 g_{c}t }    \e^{- \frac{\ga}{2} \left\str (x'_{\links}+x'_{\rechts})- (x_{\links}+x_{\rechts})\right\str}   \e^{- \ga \left\str (x_{\links}-x_{\rechts} )-( x'_{\links}-x'_{\rechts}  )  \right\str } \label{def: tilde caZ} \\[2mm]
  (\tilde{\caU}^{\ga}_t)_{x^{}_\links,x^{}_{\rechts};x'_\links,x'_\rechts} &:=&    \e^{ r_{\ve}(\ga,\la)   t}   \e^{- \frac{\ga}{2} \left\str (x'_{\links}+x'_{\rechts})- (x_{\links}+x_{\rechts})\right\str}   \e^{- \ga \left\str (x_{\links}-x_{\rechts} )-( x'_{\links}-x'_{\rechts}  ) \right\str  }  \label{def: bounding op cau}
  \eaq
  In order for  definitions (\ref{def: tilde caZ}, \ref{def: bounding op cau}) to make sense, $\la$ and $\ga >0$ have to be sufficiently small, such that the functions  $r_{\ve}(\ga,\la)$ and  $r_{\tau}(\ga,\la)$ are well-defined.  In particular, we need  conditions  on $\la$ and $\ga$ such that  Lemma \ref{lem: bounds on cutoff dynamics} applies.

The operators $\tilde{\caI}_{z,l},\tilde{\caZ}^{\tau,\ga}_t, \tilde{\caU}^{\ga}_t$ inherit their notation from the operators they are designed to bound, as we have the following inequalities, for $\la,\ga$ small enough:\baq
\caI_{x,l}  & \mathop{\leq}\limits_{m.e.}&     \tilde{\caI}_{x,l}     \\
\caZ^{\tau}_{t}  & \mathop{\leq}\limits_{m.e.}&    \tilde{\caZ}^{\tau,\ga}_t    \\
\caU_{t}  & \mathop{\leq}\limits_{m.e.}&     \tilde{\caU}^{\ga}_t   
\eaq
The first inequality is obvious from the definition of $\caI_{x,l}$ in \eqref{def: operatorI} and the fact that $\norm W \norm_{\scrB(\scrS)} \leq 1$. Indeed, $ \tilde{\caI}_{x,l} $  can be obtained from ${\caI}_{x,l} $ by replacing $W$ by $1$.
The second inequality is the result of Lemma \ref{lem: bounds on cutoff dynamics} and  the third inequality follows from the bounds following Assumption  \ref{ass: analytic dispersion}.

  We start by stating obvious rules to multiply the operators $\tilde{\caZ}^{\tau,\ga}_t$ and $\tilde{\caU}^{\ga}_t$.
\begin{lemma} \label{lem: bounds with tilde}
For $\la,\ga$ small enough, the following bounds hold (with $\cone$ and $\ctwo$ as defined in Section \ref{sec: generic constants})
\begin{itemize}
\item
For all sequences of times $s_1,\ldots,s_n$ with $t=\sum_{i=1}^n s_i$,
\baq
\tilde\caU^{\ga}_{s_n} \ldots \tilde\caU^{\ga}_{s_2}\tilde\caU^{\ga}_{s_1} & \meleq&     [\cone]^{n-1} \,   \e^{ \ctwo t}    \,    \tilde \caU^{\frac{\ga}{2}}_t   \label{eq: inequality caucau cau}\\[2mm]
\tilde\caZ^{\tau,\ga}_{s_n}  \ldots \tilde\caZ^{\tau,\ga}_{s_2} \tilde\caZ^{\tau,\ga}_{s_1} 
 & \meleq&  [\cone]^{n-1} \, \e^{ \ctwo t}  \,   \tilde\caZ^{\tau,\frac{\ga}{2}}_{t}    \label{eq: inequality caecae cae}
 \eaq
 \item
 For all times $s <t$,
 \baq
\tilde \caZ^{\tau, \ga}_{ t-s}  \tilde\caU^{\ga}_{ s}  & \meleq& \,  \cone \e^{\frac{s}{2 \tau}}   \e^{ \ctwo t}   \tilde\caZ^{\tau,\frac{\ga}{2}}_{ t}    \label{eq: inequality caecau cae} \\[2mm]
 \tilde \caU^{\ga}_{ t-s}  \tilde\caZ^{\tau, \ga}_{ s}    & \meleq& \,  \cone   \e^{\frac{t-s}{2 \tau}} \e^{ \ctwo t}   \tilde\caZ^{\tau,\frac{\ga}{2}}_{ t}    \label{eq: inequality caucae cae}
\eaq
\end{itemize}
\end{lemma}
\begin{proof}
Inequalities \eqref{eq: inequality caucau cau} and \eqref{eq: inequality caecae cae} are immediate consequences of the fact that
\beq \label{eq: template gamma bounds}
\sum_{x \in \lat} \e^{- \ga \str x-x_1\str}   \e^{- \ga \str x-x_2\str}  \leq e^{-\frac{\ga}{2} \str x_1-x_2\str} \sum_{x\in \lat} \e^{-\frac{\ga}{2} \str x\str}, \qquad \textrm{for any}\, \ga >0
\eeq
To derive inequalities \eqref{eq: inequality caecau cae}  and \eqref{eq: inequality caucae cae},  we use \eqref{eq: template gamma bounds} and  we dominate exponential factors $\e^{O(\la^2)t}$ on the RHS by $\e^{\frac{t}{2\tau}}$, using that $\tau \la^2 \to 0$ as $\la \searrow 0$. 

\end{proof}

 Lemma \ref{lem: integrating out pairings}, below, shows how the bounds of  Lemma \ref{lem: bounds with tilde} are used to integrate over diagrams.  This lemma  will be used repeatedly in the next sections, and,  since it is a crucial step, we treat the following  simple example in detail:
We attempt to bound the  expression
\beq
\caF:=  \mathop{\int}\limits_{s^{\mathrm{i}} < t_1 <t_2 < s^{\mathrm{f}}} \d t_1 \d t_2 \mathop{\sum}\limits_{x_1,x_2,l_1,l_2}  \zeta(\dsi)   \underbrace{ \tilde \caU^{\ga}_{s^{\mathrm{f}}-t_2 } \tilde \caI_{x_2,l_2}  \tilde \caU^{\ga}_{t_2-t_1}  \tilde \caI_{x_1,l_1} \tilde\caU^{\ga}_{t_1-s^{\mathrm{i}} } }_{=:\tilde\caA}
\eeq
in ``the sense of matrix elements", with $\dsi$ being the diagram in $\Si^1_{[s^{\mathrm{i}},s^{\mathrm{f}}]}$ consisting of the ordered pair $((t_1,x_1,l_1),(t_2,x_2,l_2) )$. 
We proceed as follows: 
\ben
\item{
We bound $\zeta(\dsi)$ by $\sup_{x_1,x_2,l_1,l_2} \str \zeta(\dsi) \str$. Note that the latter expression is a function of $t_2-t_1$ only.}
\item{Since the only dependence on $x_1,x_2,l_1,l_2$ is in the operators $\tilde\caI_{x_i,l_i}$, we perform the sum $\sum_{x_i,l_i}\tilde\caI_{x_i,l_i}=1$, for $i=1,2$.}
\item{Since the operators $\tilde\caI_{x_i,l_i}$ have disappeared, we can bound
\beq \label{eq: example bound caucau}
 \tilde \caU^{\ga}_{s^{\mathrm{f}}-t_2 }   \tilde \caU^{\ga}_{t_2-t_1} \tilde \caU^{\ga}_{t_1-s^{\mathrm{i}} }  \meleq    [\cone]^2 \e^{ \ctwo \str s^{\mathrm{f}}-s^{\mathrm{i}}  \str} \hspace{2mm} \tilde \caU^{\frac{\ga}{2}}_{s^{\mathrm{f}}-s^{\mathrm{i}} } 
\eeq
using Lemma \ref{lem: bounds with tilde}. 
}
\een
Thus
\beq \label{eq: toy bound final}
 \caF \meleq   \tilde \caU^{\frac{\ga}{2}}_{s^{\mathrm{f}}-s^{\mathrm{i}} } \, \e^{ \ctwo \str s^{\mathrm{f}}-s^{\mathrm{i}}  \str}  \,   \mathop{\int}\limits_{s^{\mathrm{i}} < t_1 <t_2 < s^{\mathrm{f}}} \d t_1 \d t_2  \,  [\cone]^2  \sup_{x_1,x_2,l_1,l_2} \str \zeta(\dsi) \str
\eeq
Note that $ \sup_{x_1,x_2,l_1,l_2} \str \zeta(\dsi) \str= \sup_x \str \psi(x,t_2-t_1)\str$ because $\str \dsi \str=1$.
The short derivation above can be considered to be an application of Lemma \ref{lem: F and G}, as we illustrate by writing
\beq
\caF_{x^{}_{\links},x^{}_{\rechts}; x'_{\links},x^{'}_{\rechts} }= \mathop{\int}\limits_{\Si^1_{[s^{\mathrm{i}},s^{\mathrm{f}}] }} \d \dsi   G(\dsi) F(\dsi),  \qquad \textrm{with} \, G(\dsi)=\zeta(\dsi), \,  F(\dsi):= \tilde \caA_{x^{}_{\links},x^{}_{\rechts}; x'_{\links},x^{'}_{\rechts} }
\eeq
and hence \eqref{eq: toy bound final} follows from Lemma \ref{lem: F and G} after applying \eqref{eq: example bound caucau}.

Lemma \ref{lem: integrating out pairings}  is a generalization of the bound  \eqref{eq: toy bound final} above.

\begin{lemma}\label{lem: integrating out pairings}
Fix an interval $I= [s^{\mathrm{i}},s^{\mathrm{f}}]$ and a set of $m$ triples $(t'_i,x'_i,l'_i)_{i=1}^m$ such that $t'_i \in I$ and $t'_i < t'_{i+1} $.
 For any $\dsi \in \Si_{I}(\mathrm{ir})$,  we define the  set of $n:= m+2\str \dsi\str$ triples 
$(t''_i,x''_i,l''_i)_{i=1}^{m+2 \str \dsi \str}$ by  time-ordering (i.e., such that $t''_i \leq t''_{i+1}$) of the union of  triples
\beq
(t'_i,x'_i,l'_i)_{i=1}^m, \qquad \textrm{and} \qquad  (t_i(\dsi),x_i(\dsi),l_i(\dsi))_{ i=1}^{2 \str \dsi \str}
\eeq
Then
\baq
&& \mathop{\int}\limits_{\Si_I(\mathrm{ir})} \d \dsi \,  \str\zeta(\dsi)\str    \, \,   \tilde\caI_{x''_n, l''_n}    \tilde\caU^{\ga}_{t''_n-t''_{n-1}}                \ldots  \tilde\caU^{\ga}_{t''_3-t''_2}  \tilde\caI_{x''_2, l''_2} \tilde\caU^{\ga}_{t''_2-t''_1}  \tilde\caI_{x''_1, l''_1}    \label{eq: integrating out pairings} \\[2mm]
& \meleq&   \left( \e^{\ctwo\str I \str}  \mathop{\int}\limits_{ \Pi_T \Si_I(\mathrm{ir})}  \, \d [\dsi]   \,   [\cone]^{2\str \dsi \str}  \mathop{\sup}\limits_{\underline{x}(\dsi), \underline{l}(\dsi)}  \str \zeta(\dsi)\str \right)   
\quad \times \quad  \tilde \caU^{\frac{\ga}{2}}_{s^{\mathrm{f}}-t'_m}  \tilde \caI_{x'_{m}, l'_{m}}       \tilde \caU^{\frac{\ga}{2}}_{t'_{m}-t'_{m-1}} \ldots   \tilde \caI_{x'_{2}, l'_{2}}  \tilde \caU^{\frac{\ga}{2}}_{t'_2-t'_1} \tilde \caI_{x'_{1}, l'_{1}}  \tilde \caU^{\frac{\ga}{2}}_{t_1'-s^{\mathrm{i}}}     \nonumber  
\eaq
Moreover, the statement remain true if one replaces  $\tilde \caU^{\ga}_t \rightarrow \tilde \caZ^{\tau,\ga}_t$ on the LHS and $\tilde \caU^{\frac{\ga}{2}}_t \rightarrow \tilde \caZ^{\tau,\frac{\ga}{2}}_t$ on  the RHS of \eqref{eq: integrating out pairings}.
 \end{lemma}

\begin{proof}
The proof is a copy of the proof of the the bound \eqref{eq: toy bound final}. The steps are
\ben
\item{Dominate $\str\zeta(\dsi)\str$ by $\sup_{\underline{x}(\dsi),\underline{l}(\dsi)} \str\zeta(\dsi)\str$}
\item{Sum over  $\underline{x}(\dsi),\underline{l}(\dsi) $ by using $\sum_{x_i,l_i } \tilde\caI_{x_i,l_i}=1$}
 \item{Multiply the operators  $\tilde\caU^{\ga}_t$ or $\tilde \caZ_t^{\tau,\ga}$, using the  bound \eqref{eq: inequality caucau cau} or  \eqref{eq: inequality caecae cae}.}
  \item{Interpret the remaining sum over $\str \dsi \str$ and integration over $\underline{t}(\dsi)$ as an integration over equivalence classes $[\dsi]$.}
\een

\end{proof}

\subsection{Bound on short pairings and proof of Lemma \ref{lem: polymer model cutoff}}

We recall that the crucial result  in Lemma \ref{lem: polymer model cutoff}(see Statement 2 therein) is the bound 
\beq \label{eq: repetition crucial bound lemma polymer cutoff}
 \caJ_{\ka}\caR^{\tau}_{ex}(z)   \caJ_{-\ka} = \int_{\bbR^+} \d t  \, \e^{-t z} \mathop{\int}\limits_{\Si_{[0,t]}({< \tau}, \mathrm{ir})} \d\dsi\,  1_{\str \dsi \str \geq 2}   \caJ_{\ka}\caV_{[0,t]}(\dsi)  \caJ_{-\ka}= O(\la^2)O(\la^2\tau),
\eeq
uniformly for $\Re z \geq -\frac{1}{2\tau}$ and for $\str \Im \ka \str$ small enough. 

In the \textbf{first} step of the proof of \eqref{eq: repetition crucial bound lemma polymer cutoff}, we sum over the $\underline{x}(\dsi), \underline{l}(\dsi)$- coordinates of the diagrams in \eqref{eq: repetition crucial bound lemma polymer cutoff}. The strategy for doing this has been outlined in Sections  \ref{sec: bounds in the sense of me} and \ref{sec: bounding operators}.

\begin{lemma}\label{lem: bound on small pairings} For $\la,\ga$ smalll enough,
\beq  \label{bound tildevu}
\mathop{\int}\limits_{\Si_{[0,t]}({< \tau}, \mathrm{ir})}   \d\dsi\,  1_{\str \dsi \str \geq 2} \zeta(\dsi) \caV_{[0,t]}(\dsi)   \quad \meleq  \quad   \e^{ \ctwo  t}  \,\,  \tilde\caU^{\ga}_{t}  \,    \mathop{\int}\limits_{\Pi_T \Si_{[0,t]}({< \tau}, \mathrm{ir})}  \d [\dsi] \,  \cone^{\str \dsi \str}   \,  1_{\str \dsi \str \geq 2}       \sup_{\underline{x}(\dsi),\underline{l}(\dsi)} \str\zeta(\dsi) \str    \,     
\eeq
\end{lemma}
\begin{proof}
In the definition of $\caV_I(\dsi)$, see e.g. \eqref{def: cavI}, we bound $\caI_{x,l}$ by $\tilde \caI_{x,l} $ and $\caU_{t}$ by $\tilde\caU^{\ga}_{t}$. Then, we use 
 the bound \eqref{eq: integrating out pairings} with $m=0$ to obtain \eqref{bound tildevu}.  Note that, since $m=0$,  the set of triples $(t''_i,x''_i,l''_i)_{i=1}^{m + 2 \str \dsi \str}$ is equal to the set of triples $(t_i(\dsi),x_i(\dsi),l_i(\dsi))_{i=1}^{2 \str \dsi \str}$. Note also that we use  \eqref{eq: integrating out pairings} with $\Si_I(<\tau,\mathrm{ir})$ instead of $\Si_I(\mathrm{ir})$, and with the restriction to $\str \dsi \str \geq 2$. However, this does not change the validity of \eqref{eq: integrating out pairings}, as one  easily checks. 
\end{proof}
To appreciate how Lemma \ref{lem: bound on small pairings} relates to the bound  \eqref{eq: repetition crucial bound lemma polymer cutoff}, we note already that the bound \eqref{bound tildevu} remains true if one puts left and right hand sides between $\caJ_\ka \cdot \caJ_{-\ka}$ for  purely imaginary $\ka$ (for general $\ka$ the matrix elements can become negative, which is not allowed by our definition of $\meleq$).  \\

In the \textbf{second} step of the proof, we estimate the Laplace transform of the integral over equivalence classes $[\dsi]$ appearing on the RHS of  \eqref{bound tildevu}. This estimate uses three important facts
\ben
\item{The correlation functions in \eqref{bound tildevu} decay exponentially with rate $1/\tau$, due to the cutoff.}
\item{The diagrams are restricted to $\str \dsi \str  \geq 2$, they are therefore  subleading with respect to a diagram with $\str \dsi \str=1$.}
\item{We allow the estimate to depend on $\ga$ in a non-uniform way. Indeed, $\ga$ will be fixed in the last step of the argument.}
\een

Concretely, we  show that, for $ 0 < a \leq \frac{1}{\tau}$ and for $\la$ small enough (depending on $\ga$) 
\beq \label{eq: bound combinatorics diagram cutoff}
\mathop{\int}\limits_{\bbR^+} \d t \,  \e^{a t} \mathop{\int}\limits_{\Si_{[0,t]}({< \tau}, \mathrm{ir})}   \d [\dsi] \,   1_{\str \dsi \str \geq 2}      \left( [\cone]^{2\str \dsi \str}  \sup_{\underline{x}(\dsi),\underline{l}(\dsi)} \str\zeta(\dsi) \str \right)  =   O( \la^2)  O(\la^2 \tau) \cone  , \qquad  \la \searrow 0, \la^2 \tau \searrow 0
\eeq
To verify \eqref{eq: bound combinatorics diagram cutoff}, we set
 \beq k(t):=  \la^2  \cone \,  1_{\str t \str \leq \tau} \sup_{x} \str\psi(x,t)\str  
 \eeq
and we calculate, by exploiting the cutoff $\tau$ in the definition of $k(\cdot)$, 
\beq
\norm  \e^{ \frac{1}{\tau} t } k\norm_1   <   \la^2 \cone, \qquad    \norm t \e^{ (\frac{1}{\tau}+\norm k \norm_1) t }k \norm_1 =  \tau O(\la^2)\cone    
\eeq
The norm $\norm \cdot \norm_1$  refers to the variable $t$, i.e.,  $\norm h \norm_1 = \int_{0}^{\infty} \d t \str h(t) \str $.
Hence \eqref{eq: bound combinatorics diagram cutoff} follows from  the bound \eqref{eq: bound irreducible}  in Lemma \ref{lem: combinatorics of irreducible diagrams}, in Appendix \ref{app: combinatorics}, after using that $\tau O(\la^2) <C$, as $\la \searrow 0$, and choosing $\la$ small enough.\\ 

 In the \textbf{third} step of the proof, we fix $\ga$. 
By using the explicit form \eqref{def: bounding op cau} and the relation \eqref{eq: relation kappa and fibers}, we check that
\beq \label{eq: proof cutoff polymer 1}
 \left \norm \left(\caJ_{\ka} \tilde \caU^{\ga}_t \caJ_{-\ka} \right)_{x^{}_\links,x^{}_{\rechts};x'_\links,x'_\rechts}  \right\norm \leq    \e^{ \ctwo t}, \qquad  \textrm{for any} \qquad \str\Im\ka^{}_{\links,\rechts}\str < \ga.
  \eeq
Next, we make use of the following general fact that can be easily checked (e.g., by the Cauchy-Schwarz inequality): If, for some $\ga>0$ and $C <\infty$,
\beq   \norm \left(\caJ_{\ka}\caA \caJ_{-\ka}\right)_{x^{}_\links,x^{}_{\rechts};x'_\links,x'_\rechts} \norm \leq C, \qquad \textrm{uniformly for $\ka^{}_{\links,\rechts}$ \, s.t.} \,\,  \str \Im\ka^{}_{\links,\rechts} \str \leq \ga,
 \eeq then 
 \beq \label{eq: conversion exponential bound norm bound 1}
  \norm \caA  \norm  \leq  c(\ga)
  \eeq 
  where the norms refer to the operator norm  on $\scrB(\scrB_2(l^2(\bbZ^d,\scrS)))$, as in Definition \ref{def: bounds me}.
    
   Hence, from \eqref{eq: proof cutoff polymer 1} we get
  \beq  \label{eq: banality proof polymer 1}
 \mathop{\sup}\limits_{\str \Im \ka^{}_{\links,\rechts}\str  < \ga/2}  \norm \caJ_{-\ka}  \tilde \caU^{ \ga}_t \caJ_{\ka}  \norm   \leq    \cone    \e^{ \ctwo t}.
  \eeq 
By the first equality in \eqref{eq: repetition crucial bound lemma polymer cutoff} and Lemma \ref{lem: bound on small pairings}, 
\baq
 \caJ_{\ka}\caR_{ex}(z)   \caJ_{-\ka} \meleq    \int_{\bbR^+} \d t \,  \e^{-t \Re z}     \caJ_{\ka} \tilde\caU^{\ga}_t  \caJ_{-\ka}   \e^{\ctwo t}   \mathop{\int}\limits_{\Si_{[0,t]}({< \tau}, \mathrm{ir})}\d [\dsi] \,  1_{\str \dsi \str \geq 2}       [\cone]^{2\str \dsi \str}  \sup_{\underline{x}(\dsi),\underline{l}(\dsi)} \str\zeta(\dsi) \str  \label{eq: start of statement 2}
\eaq
We combine  \eqref{eq: start of statement 2} and \eqref{eq: banality proof polymer 1}  with \eqref{eq: bound combinatorics diagram cutoff}, setting  
\beq 
a \equiv  \max{(-\Re z,0)} + \ctwo \eeq
 for $\la$ small enough such that $\ctwo  \leq \frac{1}{2\tau}$ and such that \eqref{eq: bound combinatorics diagram cutoff} applies. At this point, the parameter $\ga$ has been fixed and this choice determines the maximal value of $\str \Im \ka \str$.  This concludes the proof of  the bound in \eqref{eq: repetition crucial bound lemma polymer cutoff}.    The other statements of  Lemma \ref{lem: polymer model cutoff} are proven below. \\
\vspace{0.3cm} 

\noindent\emph{Proof of Lemma \ref{lem: polymer model cutoff} }
The claim about $\caR^{\tau}_{ld}(z)$ (Statement 2)  follows by a drastically simplified version of the above argument for $\caR^{\tau}_{ex}(z)$.\\

\noindent To establish the convergence claim in Statement 1) of Lemma \ref{lem: polymer model cutoff}, it suffices, by \eqref{eq: laplace cutoff}, to check that $\norm \caZ^{\tau,\mathrm{ir}}_t \norm \leq \e^{C t}$ for some constant $C$. This has been established in the proof of Statement 2),  above, since $\caR^{\tau}_{ld}(z)+\caR^{\tau}_{ex}(z) $ is the Laplace transform of $\caZ^{\tau,\mathrm{ir}}_t$. The identity \eqref{eq: expansion for resolvent}
was established in Section \ref{sec: ladders and excitations}. \vspace{0.3cm}

\noindent To check Statement 3), we employ  expression \eqref{eq: caL as perturbation} for $\caL(z)$ and \eqref{def: car ladders}  for $\caR^{\tau}_{ld}(z)$. The latter differs from  $\caL(z)$ in that it is the Laplace transform of a quanitity with a cutoff at $t=\tau$ and in the fact that it includes  the propagator $\caU_t = \e^{\i (\ad (Y) + \la^2\ad(\varepsilon)) t}$ whereas $\caL(z)$ includes only $\e^{\i (\ad (Y))t }$.
We observe that
\baq
 \left \norm \caJ_{\ka} \left( \caR_{\mathrm{ld}}^{\tau}(z )  - \la^2 \caL(z) \right) \caJ_{-\ka}   \right \norm &\leq&4\la^2 \int_{\tau}^{\infty} \d t      \sup_x\str \psi(x,t)\str  \left \norm \caJ_{\ka}  \e^{-\i \adjoint(Y)t}  \caJ_{-\ka}  \right \norm   \\
& + & 4\la^2  \int_{0}^{\tau} \d t      \,    \sup_x\str \psi(x,t)\str  \left \norm  \caJ_{\ka}   \e^{-\i \adjoint(Y)t} \left(    \e^{\i \la^2 \adjoint(\varepsilon(P))t}- 1    \right)  \caJ_{-\ka}   \right \norm
\eaq
where the factors '$4$' originate from the sum over $l_1,l_2$ and we use that $\Re z \geq 0$.  In the first term on the RHS,  $\left \norm \caJ_{\ka}  \e^{-\i \adjoint(Y)t}  \caJ_{-\ka}  \right \norm=1$ since $Y$ commutes with the position operator $X$.
The second term is bounded by 
\beq
  4 \la^2\int_0^\tau \d s  \, \sup_x\str \psi(x,s)\str  \times 
   \sup_{t \leq \tau} \left(   \la^2 t   \left\norm  \caJ_{\ka}  \adjoint(\varepsilon(P))   \e^{\i \la^2 \adjoint(\varepsilon(P))t} \caJ_{-\ka}   \right \norm     \right) \leq   \tau \la^4C 
   \eeq
where we have used  Lemma \ref{lem: integrability} and the bound \eqref{eq: combes thomas}. 
\qed

\subsection{Bound on the vertex operators $\caB(\frl, s^{\mathrm{i}},s^{\mathrm{f}}) $}

In this section, we  prove a  bound on the 'dressed vertex operators', which were introduced in Section \ref{sec: abstract definition of vertex operators}. 
Since such 'dressed vertex operators' contain an irreducible short diagram in the interval $[s^{\mathrm{i}},s^{\mathrm{f}}]$, we obtain a bound that is exponentially decaying in $\str s^{\mathrm{f}} -s^{\mathrm{i}} \str$.
In \eqref{bound tildebag}, this exponential decay resides in the function $w(\cdot)$ and it is made explicit through the calculation in \eqref{eq: fast decay for function w}.

   The proof of the next lemma parallels the proof of Lemma  \ref{lem: bound on small pairings} above. Consider $m$ triples $(t'_i,x'_i,l'_i)_{i=1}^m$ and let $\frl$ be a (dressed) vertex with vertex set $S(\frl)= \{ t'_1, \ldots, t'_m \}$. Let $s^{\mathrm{i}},s^{\mathrm{f}}$ be vertex time-coordinates associated to $\frl$, i.e., such that $s^{\mathrm{i}} < t'_1$ and $s^{\mathrm{f}} > t'_m$. 

\begin{lemma}\label{lem: bound on vertex operators}
For $\la,\ga$ small enough, the following bound holds:
\beq \label{bound tildebag}
\caB(\frl, s^{\mathrm{i}},s^{\mathrm{f}})   \meleq   w(s^{\mathrm{f}}-s^{\mathrm{i}}) \,    \tilde \caU^{\ga}_{s^{\mathrm{f}}-t'_{m}} \tilde{\caI}_{x'_{m},l'_{m}}    \tilde\caU^{\ga}_{t'_{m}-t'_{m-1}} \ldots        \tilde\caU^{\ga}_{t'_2-t'_1}  \tilde{\caI}_{x'_1,l'_1}   \tilde \caU^{\ga}_{t'_1-s^{\mathrm{i}}}  
\eeq
where
\beq
 w(  s^{\mathrm{f}}- s^{\mathrm{i}}) :=   \e^{  \la^2 C'  \str s^{\mathrm{f}}- s^{\mathrm{i}} \str } \mathop{\int}\limits_{\Pi_T\Si_{[s^{\mathrm{i}} , s^{\mathrm{f}}]}( {< \tau}, \mathrm{ir})}  \d [\dsi] \,   C^{\str \dsi \str}  \sup_{\underline{x}(\dsi),\underline{l}(\dsi)} \str\zeta(\dsi) \str   \label{def: w of an interval}
\eeq
The RHS of \eqref{def: w of an interval}  indeed depends only on $s^{\mathrm{f}}-s^{\mathrm{i}}$, since the correlation function $\zeta(\dsi)$ depends only on differences of the time-coordinates of $\dsi$.   The function $w(\cdot)$ depends on the coupling strength $\la$ via the correlation function $\zeta(\dsi)$, see \eqref{def: zeta2}.
\end{lemma}
\begin{proof}
Starting from the definition of the vertex operator $\caB(\frl, s^{\mathrm{i}},s^{\mathrm{f}}) $  given in \eqref{def: vertex operator}, we bound the operators $\caI_{x,l}, \caU_t$ by $\tilde \caI_{x,l}, \tilde\caU^{\ga}_t$
and we apply Lemma \ref{lem: integrating out pairings} to obtain
\baq  
\caB(\frl, s^{\mathrm{i}},s^{\mathrm{f}})  & \meleq&   \e^{ \ctwo \str s^{\mathrm{f}}- s^{\mathrm{i}} \str  }  \mathop{\int}\limits_{\Pi_T\Si_{[s^{\mathrm{i}} , s^{\mathrm{f}}]}({< \tau}, \mathrm{ir})}  \d [\dsi] \,   [\cone]^{2\str \dsi \str}  \sup_{\underline{x}(\dsi),\underline{l}(\dsi)} \str\zeta(\dsi) \str      \nonumber   \\
&&  \qquad \hspace{2cm}    \tilde\caU^{\frac{\ga}{2}}_{s^{\mathrm{f}}-t'_{m}} \tilde{\caI}_{x'_{m},l'_{m}}  \tilde\caU^{\frac{\ga}{2}}_{t'_{m}-t'_{m-1}} \ldots        \tilde\caU^{\frac{\ga}{2}}_{t'_2-t'_1}  \tilde{\caI}_{x'_1,l'_1}   \tilde\caU^{\frac{\ga}{2}}_{t'_1-s^{\mathrm{i}}}       \label{eq: application of integration to vertex}
\eaq
 From the definition of $\tilde \caU^{\ga}_t$ in \eqref{def: bounding op cau}, we see that
\beq \label{eq: relation between caU for different gamma}
\tilde \caU^{\ga_1}_t  \leq   \e^{ \ctwo t }   \,   \tilde \caU^{\ga_2}_t , \qquad  \textrm{for}\,\,  \ga_2 < \ga_1
\eeq
We dominate the RHS of \eqref{eq: application of integration to vertex} by fixing $\ga/2= \ga_1$   and  applying \eqref{eq: relation between caU for different gamma}  for any 
 $\ga_2 \leq \ga_1$.  This yields \eqref{bound tildebag}, with the constant $C$ in \eqref{def: w of an interval} given by fixing $\ga =\ga_1$ in $\cone$. One sees that the maximal value we can choose for $\ga_1$ is $\ga_1=\frac{1}{4} \de_{\ve}$, with $\de_{\ve}$ as in Assumption \ref{ass: analytic dispersion}.
\end{proof} 
For later use, we note here that, for $\la$ sufficiently small and with $w(t)$ defined in  \eqref{def: w of an interval};
\baq 
\int_{\bbR^+} \d t \,  \str t\str  w(t) \e^{ \frac{t}{ 2\tau} } &   \leq   &  \tau C  \int_{\bbR^+} \d t   w(t) \e^{ \frac{t}{ \tau}}   \nonumber \qquad  \\[2mm]
&\leq&    O(\la^2   \tau^2), \qquad  \la \searrow 0, \la^2 \tau \searrow 0    \label{eq: fast decay for function w}
\eaq
where the second inequality follows by the bound \eqref{eq: bound irreducible classes all} in Lemma \ref{lem: combinatorics of irreducible diagrams}, with 
\beq
k(t):=    \la^2  C  \sup_{x}    \str \psi(x,t) \str  1_{t \leq \tau} \quad   \qquad   \textrm{and} \qquad  a := \frac{1}{\tau}
\eeq
for $\la$ such that $\la^2C' < 1/\tau$ with $C'$ as in the exponent of \eqref{def: w of an interval}.

\subsection{Bound on the conditional cutoff dynamics $\caC_t(\dsi_l) $} \label{sec: bounds on caC}

In this section, we state bounds on $   \caC_t(\dsil, \frL) $ and $\caC_t(\dsi_l) $, defined in Sections \ref{sec: representation of caZ as double integral} and \ref{sec: vertices and vertex partitions}, respectively.
Our bounds will follow in a straightforward way from Lemma \ref{lem: bound on vertex operators} and  formula \eqref{formula caC}, which we repeat here for convenience
 \beq\label{eq: formula caC repeated}
   \caC_t(\dsil, \frL)    =    \mathop{\int}\limits_{\footnotesize{\left.\begin{array}{c}  0 <s^{\mathrm{i}}_k < s^{\mathrm{f}}_k < t  \\[1mm] 
   s^{\mathrm{f}}_k  < s^{\mathrm{i}}_{k'}    \,  \textrm{for} \,  k' >k   \end{array}\right.  }  } \caD  \underline{s}^{\mathrm{i}} \caD  \underline{s}^{\mathrm{f}}  \,\,    \caB(\frl_{\str \frL\str }, s^{\mathrm{i}}_{\str \frL\str },s^{\mathrm{f}}_{\str \frL\str }) 
       \caZ^{\tau}_{s^{\mathrm{f}}_{\str \frL\str }-s^{\mathrm{i}}_{{\str \frL\str }-1} }   \ldots     \caZ^{\tau}_{s^{\mathrm{f}}_3-s^{\mathrm{i}}_2}   \caB(\frl_2, s^{\mathrm{i}}_2,s^{\mathrm{f}}_2)  \caZ^{\tau}_{s^{\mathrm{f}}_2-s^{\mathrm{i}}_1}   \caB(\frl_1, s^{\mathrm{i}}_1,s^{\mathrm{f}}_1) 
 \eeq

By inserting the bound from Lemma \ref{lem: bound on vertex operators} in \eqref{eq: formula caC repeated}, we obtain a bound on $\caC_t(\dsil, \frL)$ depending on the vertex time-coordinates $s^{\mathrm{i}},s^{\mathrm{f}}$. In the next bound, Lemma \ref{lem: bounds with cah}, we simply integrate out these coordinates.   To describe the result, it is convenient to introduce some taylor-made notation. Let  the times $(t_1, \ldots, t_{2n})$ be the time-coordinates of $\dsil$. We will now specify the effective dynamics between each of those times, depending on the vertex partition $\frL$.
 \begin{itemize}
 \item  If the times $t_i$ and $t_{i+1}$ belong to the  vertex set of the same vertex, then
 \beq \label{def: H1}
\tilde\caH^{\ga}_{t_{i+1}, t_i}:=  \tilde\caG^{\ga}_{t_{i+1}- t_i}, \qquad  \textrm{with} \, \,  \tilde \caG^{\ga}_t  := \e^{- \frac{t}{3 \tau} } \,  \tilde \caU^{\ga}_t  
 \eeq
 \item If the times $t_i$ and $t_{i+1}$ belong to different  vertices, then 
  \beq \label{def: H2}
\tilde\caH^{\ga}_{t_{i+1}, t_i}:=   \tilde\caZ_{t_{i+1}- t_i}^{\tau, \ga}
\eeq
 \end{itemize}
 The idea of this distinction is clear: Within a dressed vertex, we get additional decay from the short diagrams; this is the origin of the exponential decay $\e^{- \frac{t}{3 \tau}}$ in $\tilde\caG^{\ga}_t$.  Between the vertices, we encounter the cutoff reduced evolution $ \caZ^{\tau}_t$, as already visible in \eqref{eq: formula caC repeated}.  Moreover, we get an additional small factor for each dressed vertex. To make this explicit, we define
 \beq
 \str \frL \str_{\textrm{dressed}} := \# \{\textrm{dressed} \,  \frl_k  \}  \qquad  (= \textrm{number of dressed vertices in the vertex partition $\frL$})
 \eeq 
 
 \begin{lemma}\label{lem: bounds with cah}
Let the operators $\tilde\caH^{\ga}_{t_i,t_{i+1}}$ be defined as above, depending on the diagram $\dsi_l$ and the vertex partition $\frL$. Then,  for $\la,\ga$ small enough, 
\beq \label{eq: bound with cah}
 \caC_t(\dsil, \frL)  \meleq     \left[ (\str\la\str \tau)^2 \cone \right]^{\str \frL \str_{\textrm{dressed}}}     \tilde\caG^{\ga}_{t-t_{2n}}     \tilde\caI_{x_{2n},l_{2n}}  \tilde\caH^{\ga}_{t_{2n},t_{2n-1}}         \ldots    \tilde\caH^{\ga}_{t_3-t_2}  \tilde\caI_{x_2,l_2}       \tilde \caH^{\ga}_{t_2-t_1}
\tilde\caI_{x_1,l_1}         \tilde\caG^{\ga}_{t_1} 
\eeq
\end{lemma}
Note that between the times $0$ and $t_1$, we always (for each vertex partition) put $\tilde \caG^{\ga}_{t_1}$. This is because either $t_1=0$, in which case $\tilde \caG^{\ga}_{t_1}=1$, or $t_1$ belongs to a dressed vertex   whose initial time coordinate, $s^{\mathrm{i}}_1$, is fixed to be $s^{\mathrm{i}}_1=0$. The same remark applies between the times $t_n$ and $t$.

\begin{figure}[ht!] 
\vspace{0.5cm}
\begin{center}
\psfrag{time1}{$t_1$}
\psfrag{time2}{$t_2$}
\psfrag{time3}{$t_3$}
\psfrag{time4}{$t_4$}
\psfrag{time5}{$t_5$}
\psfrag{time6}{$t_6$}
\psfrag{time7}{$t_7$}
\psfrag{time8}{$t_8$}
\psfrag{timet}{$t$}
\psfrag{b1}{ $\scriptstyle{\frl_1}$}
\psfrag{b2}{ $\scriptstyle{\frl_2}$}
\psfrag{b3}{ $\scriptstyle{\frl_3}$}
\psfrag{b4}{ $\scriptstyle{\frl_4}$}
\psfrag{b5}{ $\scriptstyle{\frl_5}$}
\psfrag{b6}{ $\scriptstyle{\frl_6}$}
\psfrag{G}{ $\scriptstyle{\tilde\caG^{\ga}}$}
\psfrag{Z}{ $\scriptstyle{\tilde\caZ^{\tau,\ga}}$}
\psfrag{I}{ $\scriptstyle{\tilde\caI}$}
\psfrag{transfo}{Lemma \ref{lem: bounds with cah}}
\includegraphics[width = 12cm]{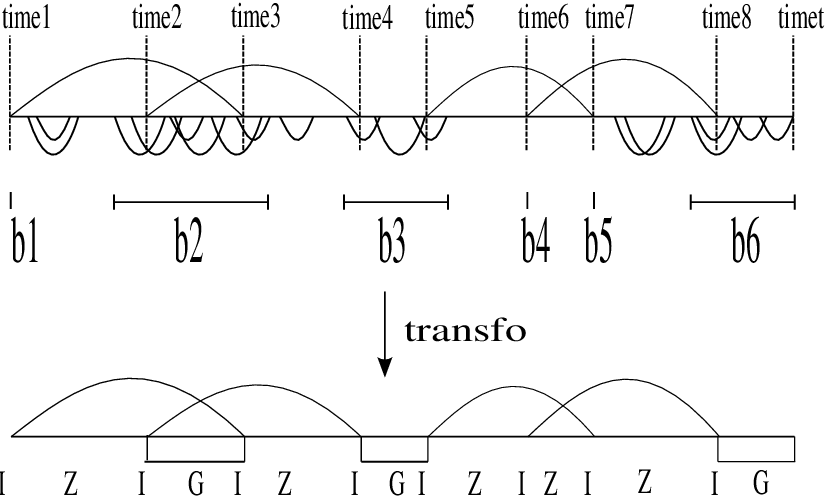}
\end{center}
\caption{\footnotesize{Consider the long diagram $\dsi_l \in \Si_{[0,t]}(\scriptstyle{>\tau})$ with $\str \dsi \str=4$ shown above.  In the upper figure, we show a short diagram $\dsi$ such that $\dsi_l \cup \dsi$ is irreducible in $[0,t]$. The corresponding vertex  partition $\frL= \{ \frl_1, \ldots, \frl_6 \}$ is indicated by vertical lines for the bare vertices $\frl_1,\frl_4,\frl_5$ and horizontal bars for the dressed vertices $\frl_2, \frl_3, \frl_6$.   In the picture below we suggest the representation that emerges after applying Lemma \ref{lem: bounds with cah}: There are no vertex time coordinates any more.   The time-coordinates of the long diagrams correspond to operators $\tilde \caI$. The intervals between time-coordinates of the long diagrams correspond to operators $\tilde\caZ^{\tau,\ga}$ or $\tilde\caG^{\ga}$.  The intervals corresponding to $\tilde\caG^{\ga}$ are those which in the upper picture belong entirely to the domain of a short diagram.
 } \label{fig: fromvertexoperatorstovertices1}
} 
\end{figure}

\begin{proof}
The proof starts from the representation of $\caC_t(\dsil,\frL)$ in \eqref{eq: formula caC repeated} and the bound for the vertex operators $ \caB(\frl_k, s^{\mathrm{i}}_k,s^{\mathrm{f}}_k) $ given in Lemma \ref{lem: bound on vertex operators}. Then we integrate out the $s^{\mathrm{i}}_k,s^{\mathrm{f}}_k$-coordinates for the dressed vertices $\frl_k$. The main tool in doing so is the fast decay of the function $w(\cdot)$, as follows from  \eqref{eq: fast decay for function w}. 

We consider a simple example. Take $t_1=0$ and $t_{2n}=t$ and let $\str \frL \str=1$, i.e.\ there is one vertex $\frl$. It follows that $\frl$ is dressed and  $S(\frl)=\{ t_2, \ldots, t_{2n-1} \}$.
In this case, formula  \eqref{formula caC} reads 
\beq
\caC_t(\dsil,\frL) =       \mathop{\int}\limits_{\footnotesize{ \left.\begin{array} {c} 0< s^{\mathrm{i}}  < t_2  \\  t_{2n-1} < s^{\mathrm{f}} < t    \end{array} \right. }} \d s^{\mathrm{i}}  \d s^{\mathrm{f}}    \caI_{x_{2n},l_{2n}}   \caZ^{\tau}_{ t-s^{\mathrm{f}} }   \caB(\frl, s^{\mathrm{i}},s^{\mathrm{f}})  \caZ^{\tau}_{ s^{\mathrm{i}} -t_{1}}  \caI_{x_{1},l_{1}} 
\eeq
and the bound in Lemma \ref{lem: bound on vertex operators} is
\beq \label{eq: specific vertex bound}
\caB(\frl, s^{\mathrm{i}},s^{\mathrm{f}})   \meleq     w(s^{\mathrm{f}}-s^{\mathrm{i}})  \,  \tilde\caU^{\ga}_{ s^{\mathrm{f}} -t_{2n-1}}  \times  \underbrace{\tilde\caI_{x_{2n-1},l_{2n-1}}   \tilde\caU^{\ga}_{t_{2n-1}-t_{2n-2}} \ldots  \tilde\caU^{\ga}_{t_3-t_2}\tilde\caI_{x_{2},l_{2}}}_{\tilde \caA^{\ga}} \times 
 \,    \tilde \caU^{\ga}_{t_2- s^{\mathrm{i}}  }\eeq 
where the operator $\tilde\caA^{\ga}$ is defined as  the `interior part' of the vertex operator. The sole property of  $\tilde\caA^{\ga}$ that is relevant for the present argument is that
\beq\label{eq: only relevant property}
\tilde \caA^{\ga} \meleq      \e^{\ctwo t} \tilde\caA^{\frac{\ga}{2}} 
\eeq
as follows from the definition of $\tilde\caU^{\ga}_t$ and the bound \eqref{eq: inequality caucau cau}.
From \eqref{eq: specific vertex bound}, \eqref{eq: only relevant property}  and (\ref{eq: inequality caecau cae}, \ref{eq: inequality caucae cae}) ,  we obtain
\baq
\caC_t(\dsil,\frL) & \meleq&   \e^{-\frac{(t_{2n-1}-t_2)}{3\tau} }      \mathop{\int}\limits_{\footnotesize{ \left.\begin{array} {c} 0< s^{\mathrm{i}}  < t_2  \\  t_{2n-1} < s^{\mathrm{f}} < t   \end{array} \right. }} \d s^{\mathrm{i}}  \d s^{\mathrm{f}}   \,   w(s^{\mathrm{f}}-s^{\mathrm{i}})  \,  \e^{\frac{s^{\mathrm{f}}-s^{\mathrm{i}}}{2\tau} }          (\cone)^2      \\
&& \hspace{6cm}      \tilde   \caI_{x_{2n},l_{2n}}   \tilde\caZ^{\tau,\frac{\ga}{2}}_{ t_{2n}-t_{2n-1}}  \tilde\caA^{\frac{\ga}{2}}      \tilde\caZ^{\tau,\frac{\ga}{2}}_{ t_{2}-t_{1}}   \tilde\caI_{x_1,l_1}   \nonumber
\eaq
where we have used the decomposition $s^{\mathrm{f}}-s^{\mathrm{i}}=  (s^{\mathrm{f}}-t_{2n-1})+   (t_{2n-1}-t_2) + (t_2-s^{\mathrm{i}})  $ and we have chosen $\la,\ga$ small enough such that $\ctwo < 1/(6 \tau)$ in \eqref{eq: only relevant property}. 
By  a change of integration  variables, we find that
\beq
 \mathop{\int}\limits_{\footnotesize{ \left.\begin{array} {c} 0< s^{\mathrm{i}}  < t_2  \\  t_{2n-1} < s^{\mathrm{f}} < t   \end{array} \right. }} \d s^{\mathrm{i}}  \d s^{\mathrm{f}}    w(s^{\mathrm{f}}-s^{\mathrm{i}})  \e^{\frac{s^{\mathrm{f}}-s^{\mathrm{i}}}{2 \tau}  }   \leq 
   \mathop{\int}\limits_{\bbR^+} \d t \, \str t\str   w(t)  \e^{\frac{t}{2 \tau}}      \label{eq: change of integration w}
\eeq
and we note that this bound remains valid if, in the integration domain on the LHS, we  replaced $0$ by a smaller number, or $t_{2n}$ by a larger number.
Hence, by the bound \eqref{eq: fast decay for function w}, we  obtain
\beq \label{eq: specific bound 2}
\caC_t(\dsil,\frL) \meleq       [\cone]^2  ( \str\la\str \tau)^2     \e^{-\frac{(t_{2n-1}-t_2)}{3 \tau} }    \caI_{x_{2n},l_{2n}}  \tilde\caZ^{\tau,\ga}_{ t_{2n}-t_{2n-1}}  \tilde\caA^{\ga}     \tilde\caZ^{\tau,\ga}_{ t_{2}-t_{1}}  \tilde\caI_{x_1,l_1} 
\eeq
where the constant that originates from the RHS of the bound \eqref{eq: fast decay for function w} has been absorbed in $\cone$.
The bound \eqref{eq: specific bound 2}  is indeed \eqref{eq: bound with cah} for our special choice of $\frL$ in which $\str \frL \str_{\textrm{dressed}}=1$. To obtain the general bound, one repeats the above calculation for each dressed vertex. These calculations can be performed completely independently of  each other, as is visible from the remark below \eqref{eq: change of integration w}. 
 
\end{proof} 
 
 In Lemma \ref{lem: bounds with cah}, the bound depends on $\frL$ through $\tilde \caH^{\ga}$, see \eqref{def: H1} and \eqref{def: H2}.  The next step is to sum over $\frL$. 
 First, we weaken our bound in \eqref{eq: bound with cah} to be valid for all $\frL$, such that the sum over $\frL$ amounts to counting all possible $\frL \sim \dsi_l$. 
 By ``weakening the bound", we mean that we bound some of the operators $\tilde\caG^{\ga}_t$ by $\tilde \caZ^{\tau,\ga}_t$. This can always be done, since, for $\la$ small enough, 
\beq
\tilde\caG^{\ga}_t \meleq \tilde \caZ^{\tau,\ga}_t
\eeq
with $\caG^{\ga}_t$ as in \eqref{def: H1} (in fact,  $\tilde \caG^{\ga}_t$ is  smaller than the second term of $ \tilde \caZ^{\tau,\ga}_t$, see  \eqref{def: tilde caZ}).
  Let  $\dsi_1, \ldots, \dsi_{m}$ be the decomposition of $\dsi_l$ into irreducible components and let $s_{2i-1},s_{2i}$ be the boundaries of the domain of $\dsi_i$. These times $s_i$ should not be confused with the vertex time-coordinates $\underline{s}^{\mathrm{i}}, \underline{s}^{\mathrm{f}} $ that were employed in an earlier stage of our analysis. In particular, the times $s_{2i-1},s_{2i}, i=1,\ldots m$, are a subset of the times $t_i, i=1,\ldots,2n$.  The central remark is that
 \begin{center}
 \emph{ For any $i$, the times $s_{2i}, s_{2i+1}$ belong to the same vertex \emph{for all} vertex partitions $\frL \sim \dsi_l$. }
 \end{center}
 Indeed, since the interval $[s_{2i}, s_{2i+1}]$ is not in the domain of $\dsi_l$, it must be in the domain of any short diagram contributing to $\caC_t(\dsi_l)$, or, in other words, any vertex partition $\frL \sim \dsi_l$ must contain a vertex whose vertex set contains both  $s_{2i}, s_{2i+1}$.
 Consequently, the operators $\tilde \caH^{\ga}_{s_{2i}, s_{2i+1}}$ in \eqref{eq: bound with cah} are always (i.e., for each compatible vertex partition) equal to $\tilde\caG^{\ga}_{s_{2i}, s_{2i+1}}$, and we will not replace them.  However, we replace all other $\tilde\caH^{\ga}_{t_j,t_{j+1}}$, i.e.\ those with the property that  the times $t_j,t_{j+1} $ are in the domain of the same irreducible component of $\dsi_l$, by $\tilde \caZ^{\tau,\ga}_{t_{j+1}-t_j}$.
 
 This procedure  is illustrated in Figure \ref{fig: frommanyCtooneC}.

\begin{figure}[ht!] 
\vspace{0.5cm}
\begin{center}
\psfrag{s1}{$s_1$}
\psfrag{s2}{$s_2$}
\psfrag{s3}{$s_3$}
\psfrag{s4}{$s_4$}
\psfrag{G}{ $\scriptstyle{\tilde\caG^{\ga}}$}
\psfrag{Z}{ $\scriptstyle{\tilde\caZ^{\tau,\ga}}$}
\psfrag{I}{ $\scriptstyle{\tilde\caI}$}
\includegraphics[width = 14cm]{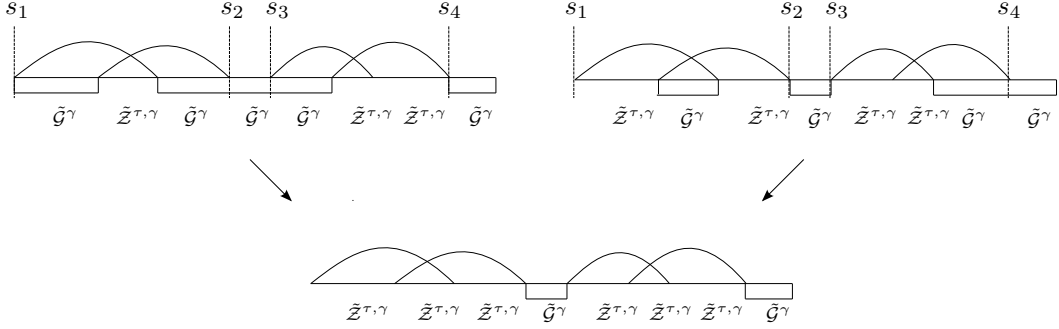}
\end{center}
\caption{\footnotesize{ Consider the long diagram $\dsi_l \in \Si_{[0,t]}(>\tau)$ with $\str \dsi \str=4$ shown above. It has two irreducible components with domains $[s_1,s_2]$ and $[s_3,s_4]$.   In the upper figures, two different vertex partitions (compatible with $\dsi_l$) are shown together with their respective bounds, obtained in Lemma \ref{lem: bounds with cah}.  These bounds are represented by the operators $\tilde\caG^{\ga}$ and $\tilde\caZ^{\tau,\ga}$, as in Figure \ref{fig: fromvertexoperatorstovertices1}, except for the fact that we omit the operators $\tilde \caI$ corresponding to the time coordinates of $\dsi_l$.  In the lower figure, we show the (weaker) bound that gives rise to Lemma \ref{lem: irreducible bounds with caE}. To establish this weaker bound, we replace the $\tilde\caG^{\ga}$ that are 'bridged' by the long diagram by $\tilde\caZ^{\tau,\ga}$.   }
} \label{fig: frommanyCtooneC}
\end{figure}
  
  After this replacement, the operator part of the resulting expression is  independent of $\frL$, and we can perform the sum over $\frL \sim \dsil$ by estimating
   \beq \label{eq: sum over vertex partitions}
   \sum_{\frL \sim \dsil}   \left[ (\str\la\str \tau)^2 \cone\right]^{\str \frL \str_{\textrm{dressed}}}  \leq   (\str\la\str \tau)^{ 2 v(\dsil)}   \cone^{ \str \dsi_l \str}, \qquad \textrm{for}\, \,    \str\la\str \tau \leq 1
 \eeq
with  
 \beq
 v(\dsil):=  \min_{ \frL \sim \dsi_l }  \str \frL \str_{\textrm{dressed}} 
 \eeq
 To obtain \eqref{eq: sum over vertex partitions}, one uses that
 \beq \label{eq: counting the vertex partitions}
\#    \{ \frL \sim \dsi_l \}   \leq   4^{2 \str \dsi_l \str-1}
 \eeq
 Indeed, $2^{2\str\dsil\str-1}$ is the number of ways to partition the time-coordinates into vertex sets. The extra factor $2$ for each vertex takes into account the choice bare/dressed.

 We have thus arrived at the following lemma

\begin{lemma}\label{lem: irreducible bounds with caE}
Let $s_{2i-2},s_{2i}$ be the boundaries of the domain of $\dsi_i$, the $i'$th irreducible component of $\dsi_l$.  Then, for $\la,\ga$ small enough, 
\beq \label{eq: bound with caE}
 \caC_t(\dsil)  \meleq  \,       (\str\la\str \tau)^{ 2 v(\dsil)}   \tilde  \caG^{\ga}_{t-s_{2m}}    \tilde\caE^{\ga}(\dsi_m) \tilde\caG^{\ga}_{s_{2m-1}-s_{2m-2}}  \tilde \caE^{\ga}(\dsi_{m-1})  \ldots   \tilde\caE^{\ga}(\dsi_1)    \tilde \caG^{\ga}_{s_1}     
  \eeq
  where, for  an irreducible diagram $\dsi$ with $\str \dsi \str=p$, 
  \beq \label{def: caE}
\tilde \caE^{\ga}(\dsi) :=    [\cone]^{\str \dsi\str}  \tilde \caI_{x_{2p}(\dsi),l_{2p}(\dsi)}  \tilde\caZ^{\tau,\ga}_{t_{2p}(\dsi)-t_{2p-1}(\dsi) }   \ldots   \tilde\caZ^{\tau,\ga}_{t_2(\dsi)-t_{1}(\dsi) }    \tilde \caI_{x_1(\dsi),l_1(\dsi)}  
\eeq 
with  $v(\dsi_l)$ as defined above. 
\end{lemma}

Note that $v(\dsil)$ is actually the number of factors $\tilde\caG^{\ga}_u$ in the expression \eqref{eq: bound with caE} for which $u \neq 0$ ($u$ can be zero only for the rightmost and leftmost $\tilde\caG^{\ga}_u$). Or, alternatively,
\beq
 v(\dsil)= \#\{ \textrm{irreducible components in}\, \dsil \}-1 +  1_{ t_{2n}  \neq t} + 1_{ t_1 \neq 0} 
\eeq

\subsection{Bounds on $\caR_{ex}(z)$ in terms of $\tilde\caE^{\ga}(\dsi)$}  \label{sec: laplace of renormalized diagrams}

To realize why the bound \eqref{eq: bound with caE} in Lemma  \ref{lem: irreducible bounds with caE} is useful, we recall that our aim is to calculate $\caR_{ex}(z)$, given by (see  \eqref{eq: representation of laplace ex with C}) 
\beq \label{eq: representation of laplace ex with C repetition}
 \caR_{ex}(z)   =   \int_{\bbR^+} \d t  \, \e^{-tz}  \mathop{\int}\limits_{\Si_{[0,t]}(>\tau)} \d \dsi    \zeta(\dsi)   \caC_t(\dsi)  \,  \left[1+\delta(t_1(\dsil) \right]  \, \left[1+\delta(t_{2\str \dsil\str}(\dsil)-t) \right]
\eeq
We calculate $ \caR_{ex}(z) $ by replacing the integral over diagrams by an integral over sequences of irreducible diagrams, as we  did in \eqref{eq: z as polymer model}, i.e.,
\beq
\mathop{\int}\limits_{\Si_{[0,t]}(>\tau)} \d \dsi   \ldots  =   \sum_{n \geq 1} \mathop{\int}\limits_{0 \leq s_1< \ldots < s_{2n}  \leq  t}  \prod_{j=1}^n  \left(\mathop{\int}\limits_{\Si_{[s_{2j-1}, s_{2j}]}(>\tau, \mathrm{ir})  } \d \dsi_j \right)  \ldots
\eeq
Using the bound \eqref{eq: bound with caE},  we obtain, with the shorthand $z_r:= \Re z$, 
\baq
\caR_{ex}(z) &\meleq&      \sum_{n \geq 0}   (1+    \caR_{\tilde\caG}(z_r))  \caR_{\tilde\caE}(z_r) \left(  \caR_{\tilde\caG}(z_r)      \caR_{\tilde\caE}(z_r)\right)^n (1+    \caR_{\tilde\caG}(z_r)) \nonumber  \\[2mm]
&= &      (1+    \caR_{\tilde\caG}(z_r))    \caR_{\tilde\caE}(z_r)  \left(1-   \caR_{\tilde\caG}(z_r)  \caR_{\tilde\caE}(z_r)  \right)^{-1}  (1+    \caR_{\tilde\caG}(z_r))    \label{eq: caR as a function of caRcaE and caRcaG}
\eaq
where, for $\Re z$ large enough,  
\baq
 \caR_{\tilde\caG}(z)  &:=&  (\tau \str\la\str)^2    \int_{\bbR^+}   \d t \,  \e^{-t z}  \tilde\caG^{\ga}_t  \\[2mm]
 \caR_{\tilde\caE}(z)  &:=& \int_{\bbR^+} \d t  \, \e^{-t z}  \mathop{\int}\limits_{\Si_{[0,t]}(>\tau, \mathrm{ir})} \d \dsi    \zeta(\dsi)  \tilde\caE^{\ga}(\dsi)   \label{def: caRcaE}
\eaq
Since $\tilde\caG^{\ga}_t$ is known explicitly, the only task that remains is to study $ \caR_{\tilde\caE}(z) $. This study is undertaken in Section \ref{sec: model only long}.  


\section{The renormalized model} \label{sec: model only long}

In this  section, we prove  Lemma \ref{lem: polymer model full}, thereby concluding the proof of our main result, Theorem \ref{thm: main}.  We briefly recall the logic of our proof.
As announced in Section \ref{sec: joint system reservoir correlation}, we analyze  $\caZ_t^{\mathrm{ir}}$ and $\caZ_t $ through a renormalized perturbation series, where the  short diagrams have already been resummed. 
However, we do not study the  Laplace transform of irreducible diagrams (defined in \eqref{def: irreducible evolutions})
\beq  \label{eq: repetition caZ irr}
    \caR^{\mathrm{ir}}(z)   = \mathop{\int}\limits_{\bbR^+} \d t  \,   \e^{- t z}   \caZ^{\mathrm{ir}}_t = \mathop{\int}\limits_{\bbR^+} \d t \,  \e^{- t z}   \mathop{\int}\limits_{\Si_{[0,t]}(\mathrm{ir})} \d \dsi \,   \zeta(\dsi) \caV_t(\dsi),
\eeq
directly, 
but rather the Laplace transform of irreducible renormalized diagrams (defined in Lemma \ref{lem: irreducible bounds with caE} and Section \ref{sec: laplace of renormalized diagrams})
\beq \label{eq: renormalized model}
  \caR_{\tilde \caE}(z)    =    \mathop{\int}\limits_{\bbR^+} \d t  \,   \e^{- t z}     \mathop{\int}\limits_{\Si_{[0,t]}(>\tau, \mathrm{ir})}  \d \dsi  \zeta(\dsi)  \tilde \caE^{\ga}(\dsi).  
\eeq
 Although the quantities  \eqref{eq: renormalized model} and \eqref{eq: repetition caZ irr} are not equal, we will argue below (in the proof of Lemma \ref{lem: polymer model full} starting from Lemma \ref{lem: exponential decay long excitations})   that  good bounds on $\caR_{\tilde \caE}(z)$ yield good bounds on $  \caR^{\mathrm{ir}}(z) $ and hence also on $\caR(z)$. 
 The reason that the expression \eqref{eq: repetition caZ irr} itself cannot be bounded by an integral over long irreducible diagrams, is the fact that an irreducible diagram in the interval $[0,t]$ does not necessarily contain an irreducible \emph{long} subdiagram in the interval $[0,t]$.  Indeed, Lemma \ref{lem: irreducible bounds with caE} decomposes the domain of an irreducible diagram into domains of long, irreducible subdiagrams and intermediate intervals. These remaining intervals  give rise to  operators $\tilde\caG_t^{\ga}$, which are easily dealt with, as we will see  below, in the proof of Lemma \ref{lem: polymer model full}, since they originate from short diagrams and therefore have good decay properties. 

Nevertheless, we clearly see the similarity between \eqref{eq: repetition caZ irr} and \eqref{eq: renormalized model}. 
To highlight this similarity, we write the inverse Laplace transform of  $\caR_{\tilde \caE}(z)$: For $r>0$ large enough, we have 
\beq
 \frac{1}{2\pi i} \mathop{\int}\limits_{r+ \i \bbR} \d z \, \e^{tz} \caR_{\tilde \caE}(z) =    \mathop{\int}\limits_{\Si_{[0,t]}(>\tau, \mathrm{ir})}  \d \dsi   \, \cone^{\str \dsi\str} \zeta(\dsi)  \,   \tilde \caI_{x_{2n},l_{2n}}\tilde \caZ^{\tau,\ga}_{t_{2n}-t_{2n-1}} \ldots      \tilde \caI_{x_{2},l_{2}} \tilde \caZ^{\tau,\ga}_{t_{2}-t_{1}}     \tilde \caI_{x_{1},l_{1}}  \label{eq: inverse laplace of renormalized series}
\eeq
where $\underline{t},\underline{x},\underline{l}$ are the coordinates of $\dsi$ and, since $\dsi$ is irreducible in $[0,t]$, $t_1=0$ and $t_{2n}=t$.  
The inverse Laplace transform of $ \caR^{\mathrm{ir}}(z)$, i.e.\ $\caZ_{t}^{\mathrm{ir}}$, is
\beq
\caZ_{t}^{\mathrm{ir}}=   \mathop{\int}\limits_{\Si_{[0,t]}( \mathrm{ir})}  \d \dsi  \,  \zeta(\dsi)   \caI_{x_{2n},l_{2n}}   \caU^{}_{t_{2n}-t_{2n-1}} \ldots   \caI_{x_{2},l_{2}}  \caU^{}_{t_{2}-t_{1}}  \caI_{x_{1},l_{1}}
\label{eq: inverse laplace of series}
\eeq
where  $\underline{t},\underline{x},\underline{l}$ have the same meaning as above. 
Thus, the perturbation series in \eqref{eq: inverse laplace of renormalized series} is indeed a renormalized version of \eqref{eq: inverse laplace of series}. The diagrams are constrained to be long, and the short diagrams have been absorbed into the 'dressed free propagator' $\tilde \caZ^{\tau,\ga}_t$. This point of view has also been stressed in Section \ref{sec: the renormalized model}.
Observe, however, that $\caZ_t^{\tau, \ga}$ depends on the positive parameter $\ga$, whereas there is no such dependence in \eqref{eq: repetition caZ irr}. 

The following lemma is our main result on $  \caR_{\tilde \caE}(z)  $.
\begin{lemma} \label{lem: exponential decay long excitations}
Recall  that $\caR_{\tilde\caE}(z)$ depends on $\ga$, because $\tilde\caE^{\ga}(\cdot)$ does.  One can choose $\ga$  such that there are
 positive constants $g_{ex}>0$ and $\delta_{ex}>0$, such that
\beq
\mathop{\sup}\limits_{\str \Im\ka^{}_{\rechts, \links}\str  \leq \de_{ex},  \Re z \geq -\la^2 g_{ex}} 
 \norm    \caJ_{\ka}   \caR_{\tilde\caE}(z)   \caJ_{-\ka} \norm  = o(\la^2),   
  \qquad  \textrm{as}\,\,  \la \searrow 0  \label{eq: inequality caF}
\eeq
\end{lemma}
The  main tools in the proof of Lemma \ref{lem: exponential decay long excitations}  will be the  exponential decay of the 'renormalized correlation function', which follows from the  bounds on $\caZ_t^\tau$ stated in Lemma \ref{lem: bounds on cutoff dynamics}, and the strategy for integrating over diagrams presented in Lemma \ref{lem: integrating out pairings}.  
With Lemma  \ref{lem: exponential decay long excitations} at hand, the proof of Lemma \ref{lem: polymer model full} is immediate. \\

\noindent\emph{ Proof of Lemma \ref{lem: polymer model full}}   We only need  to prove Statement 2) since Statement 1) will follow by a remark analogous to that in the proof of Statement 1) of Lemma \ref{lem: polymer model cutoff}.
Clearly,  for $\la$ small enough, 
\beq \label{eq: bound on tilde caG}
\sup_{\str \Im\ka^{}_{\rechts, \links}\str  \leq \ga ,  \Re z \geq - \frac{1}{4\tau}  }      \left\norm   \caJ_{\ka}  \caR_{\tilde \caG}(z)    \caJ_{-\ka}  \right\norm    \leq   O(\la),   \qquad  \textrm{as}\,\,  \la \searrow 0 
\eeq
 This follows from the properties of $\tilde \caU^{\ga}_t$, see e.g.\  the proof of Lemma \ref{lem: polymer model cutoff}, and the definition of $\tilde\caG_t^{\ga}$, see \eqref{def: H1}.
 Next, we remark that
 \beq
 \norm   \caJ_{\ka}  \caR_{ex}(z)   \caJ_{-\ka}  \norm   \leq    \norm 1+      \caJ_{\ka}  \caR_{\tilde\caG}(z_r)   \caJ_{-\ka}  \norm^2    \norm   \caJ_{\ka} \caR_{\tilde\caE}(z_r)   \caJ_{-\ka}  \norm   \norm \left(1-     \caJ_{\ka}  \caR_{\tilde\caG}(z_r)  \caR_{\tilde\caE}(z_r)   \caJ_{-\ka}   \right)^{-1} \norm     \label{eq: caR as a function of caRcaE and caRcaG norm}
\eeq
with $z_r =\Re z$.  This follows from the bound \eqref{eq: caR as a function of caRcaE and caRcaG}, the fact that $\caJ_{\ka} \caJ_{-\ka}=1$, and the implication \eqref{eq: meleq implies norm} (which allows to pass from `$\meleq$' to an inequality between norms). 

Hence, Statement 2) follows by plugging the bounds of Lemma \ref{lem: exponential decay long excitations} and \eqref{eq: bound on tilde caG} into the the RHS of \eqref{eq: caR as a function of caRcaE and caRcaG norm}.
\qed

  \subsection{Bound on the renormalized correlation function}\label{sec: bound on the renormalized correlation function}

In this section, we prove  Lemma \ref{lem: renormalized propagator}, which establishes (as its first claim) the property  \eqref{eq: bound for off-diagonal decay} with $\La_t$ replaced by $\caZ_t^{\tau}$. Indeed, in Section \ref{sec: polymer model for cutoff dynamics}, we argued that $\caZ_t^{\tau}$ is very close to $\La_t$, and this was made explicit in Lemma \ref{lem: bounds on cutoff dynamics}.  
Let 
\beq \label{def: h}
h(t) :=  \la^2  c_h  \sup_{x \in \lat} \left\{ \begin{array}{ll}    \e^{- (1/2)g_\res  t}    &  \str x\str/t \leq v^*       \\[3mm]     \str\psi (x,t)\str   &  \str x\str/t \geq v^*        \end{array}\right.
\eeq
with the velocity $v^*$ and decay rate $g_\res$ as in Lemma  \ref{lem: exponential decay} and the constant $c_h$  chosen such that
  \beq
\la^2   \sup_{x} \left\str \psi(x,t)  \right\str \leq  h(t), \qquad  \textrm{for} \, \, t >\tau  \label{eq: h dominates correlation}
  \eeq
Lemma \ref{lem: exponential decay} ensures that such a choice is possible.

\begin{lemma}\label{lem: renormalized propagator}
There  are positive constants $\delta_{r}>0$ and $g_{r}>0$ such that, for all $\ga < \delta_r$, $\la$ small enough,  and $\ka\equiv(\ka_{\links}, \ka^{}_{\rechts})$ satisfying $\str \Im\ka^{}_{\rechts, \links}\str  \leq \frac{\delta_{r}}{2}$, 
\beq\label{eq: renormalized with correlation bound}
\la^2 \str \psi(x'_{l_2}-x^{}_{l_1},t)\str \times   \left\norm (\caJ_{\ka} \tilde\caZ^{\tau,\ga}_{t}  \caJ_{-\ka})_{x^{}_{\links},x^{}_{\rechts};x^{'}_{\links},x^{'}_{\rechts}}    \right\norm  \leq h(t) \e^{- \la^2 g_{r}t }, \qquad  \textrm{for}  \qquad l_1,l_2 \in \left\{ \links, \rechts \right\} 
  \eeq
and
 \beq \label{eq: renormalized apriori bound}
  \left\norm (\caJ_{\ka} \tilde\caZ^{\tau,\ga}_{t}  \caJ_{-\ka})_{x^{}_{\links},x^{}_{\rechts};x^{'}_{\links},x^{'}_{\rechts}}     \right\norm   \leq C \e^{ \ctwo t }
\eeq
\end{lemma}
This lemma is derived from the bound \eqref{eq: bounds on cutoff dynamics} in Lemma  \ref{lem: bounds on cutoff dynamics} in a way that is completely analogous to the proof  of \eqref{eq: bound for off-diagonal decay} starting from  \eqref{eq: bounds on markov dynamics}, as outlined in Section \ref{sec: joint system reservoir correlation}.
The only difference is that in Lemma  \ref{lem: renormalized propagator}, we allow for a small blowup in space given by the multiplication operator $\caJ_\ka$. 

For future use, we also define 
\beq \label{def: htau}
h_{\tau}(t) :=   1_{ \str t\str \geq  \tau }   h(t)
\eeq
and we note that
\beq
\norm h_\tau \norm_1 := \int_{\bbR^+}  h_{\tau}(t)  =o(\la^2), \qquad \textrm{as}\,  \la \searrow 0   \label{eq: htau small}
\eeq
since $\norm h \norm_1 < \infty$ and $\tau(\la) \to \infty$ as $\la \searrow 0$.

\subsection{Sum over non-minimally irreducible diagrams}\label{sec: from irreducible long to minimal}

In a first step towards performing the integral in \eqref{eq: renormalized model}, we reduce the integral over irreducible diagrams to an integral over minimally irreducible diagrams.   Indeed, since any diagram that is irreducible in $I$ has a minimally irreducible (in $I$) subdiagram, we have, for  any  positive function  $F$,
\beq  \label{eq: irr bounded by min irr}
 \mathop{\int}\limits_{\Sigma_I(\mathrm{ir}, > \tau)}   \d \dsi F(\dsi)  \leq    \mathop{\int}\limits_{\Sigma_I(\mathrm{mir}, > \tau)}  \d \dsi    \left( F(\dsi)+    \mathop{\int}\limits_{\Sigma_I( > \tau)}  \d \dsi' F(\dsi \cup \dsi')  \right)
\eeq
The first term between  brackets on the RHS corresponds to the minimally irreducible diagrams on the LHS. The second term contains the integration over 'additional' diagrams $\dsi'$. The integration over these diagrams is unconstrained since $\dsi \cup \dsi'$ is irreducible in $I$ for any $\dsi'$, provided that $\dsi$ is irreducible in $I$.
  This is also explained and used in Appendix \ref{app: combinatorics}: see  \eqref{eq: reduction to min irr 1} and \eqref{eq: reduction to min irr 2}. 
  

Lemma \ref{lem: bound irr by min irr}
   shows that such an integration over  unconstrained  long diagrams yields a factor  $\exp\{ \cthree \str I \str \}$, with the generic constant $\cthree$ as introduced in Section \ref{sec: generic constants}.

  \begin{lemma} \label{lem: bound irr by min irr}
For $\la,\ga$ small enough, 
\beq \label{eq: lem bound by minimal}
 \mathop{\int}\limits_{\Sigma_{[0,t]}( > \tau, \mathrm{ir})}   \d \dsi  \,  \tilde\caE^{\ga}(\dsi) \str\zeta(\dsi) \str  \\
 \meleq      \e^{ \cthree t }    \mathop{\int}\limits_{\Sigma_{[0,t]}( > \tau,\mathrm{mir})}  \d \dsi \, \str\zeta(\dsi) \str  \tilde\caE^{\frac{\ga}{2}}(\dsi)
 \eeq
 \end{lemma}
\begin{proof}
By  formula \eqref{eq: irr bounded by min irr} (applied in the case where $F(\dsi)$ is a matrix element of the operator $ \str\zeta(\dsi)\str \tilde\caE^{\ga}(\dsi)$), we  have that
\beq \label{eq: bound by minimal practical}
 \mathop{\int}\limits_{\Sigma_{[0,t]}( > \tau, \mathrm{ir})}   \d \dsi  \,  \tilde\caE^{\ga}(\dsi) \str\zeta(\dsi) \str \meleq  \mathop{\int}\limits_{\Sigma_{[0,t]}( > \tau, \mathrm{mir})}   \d \dsi  \,  \tilde\caE^{\ga}(\dsi) \str\zeta(\dsi) \str +    \mathop{\int}\limits_{\Sigma_{[0,t]}( > \tau, \mathrm{mir})}   \d \dsi        \mathop{\int}\limits_{\Sigma_{[0,t]}( > \tau)}  \d \dsi'  \,   \tilde\caE^{\ga}(\dsi \cup \dsi') \str\zeta(\dsi \cup \dsi') \str  
\eeq
First, we bound
\beq   \label{eq: lem bound by minimal2}
    \mathop{\int}\limits_{\Sigma_{[0,t]}( > \tau)}  \d \dsi'  \,   \tilde\caE^{\ga}(\dsi \cup \dsi') \str\zeta(\dsi \cup \dsi') \str  
\eeq
with $\dsi$ fixed. 
To perform the integral over $\dsi'$ in \eqref{eq: lem bound by minimal2}, we recall that $\tilde \caE^{\ga}(\cdot)$ consists of products of the operators $\tilde \caI_{x_i,l_i}$ and $\caZ^{\tau,\ga}_{t_{i+1}-t_i}$. 
Hence, by Lemma
\ref{lem: integrating out pairings} with $\tilde \caU^{\ga}_t$ replaced by $\tilde \caZ^{\tau,\ga}_t$,  we can sum over the $\underline{x},\underline{l}$-coordinates of $\dsi'$ and multiply the $\tilde\caZ^{\tau,\ga}_{t_{i+1}-t_i}$ operators using the bound \eqref{eq: inequality caecae cae}. This yields 
\beq
\textrm{\eqref{eq: lem bound by minimal2}}  \, \meleq   \,    \tilde\caE^{\frac{\ga}{2}}_t(\dsi ) \str\zeta(\dsi) \str  \,   \e^{ \ctwo t}       \mathop{\int}\limits_{\Pi_T \Sigma_t( > \tau)}  \d [\dsi']   c(\ga)^{\str \dsi' \str}   \sup_{\underline{x}(\dsi'), \underline{l}(\dsi')}    \str\zeta(\dsi') \str.         \label{eq: lem bound by minimal3}
\eeq

The integral on the RHS of  \eqref{eq: lem bound by minimal3} is estimated as
 \beq \label{eq: integrating correlation functions of non-minimal}
   \mathop{\int}\limits_{\Pi_T \Sigma_t( > \tau)}  \d [\dsi']   c(\ga)^{\str \dsi' \str}   \sup_{\underline{x}(\dsi'), \underline{l}(\dsi')}    \str\zeta(\dsi')\str \quad   \leq \quad  \e^{  c(\ga)  \norm h_{\tau} \norm_{1} t} -1
  \eeq
with  $h_\tau$ as defined in \eqref{def: htau}. This follows from the bound \eqref{eq: unconstrained integration diagrams} (integral over unconstrained diagrams)  in Appendix \ref{app: combinatorics}.
To bound the  first term in \eqref{eq: bound by minimal practical}, we dominate
\beq \label{eq: bound on first term}
\tilde\caE^{\ga}(\dsi) \meleq   \e^{ \cthree t}  \,  \tilde\caE^{\frac{\ga}{2}} (\dsi) 
\eeq
The lemma follows by inserting the bounds  \eqref{eq: bound on first term} and  (\ref{eq: lem bound by minimal3}, \ref{eq: integrating correlation functions of non-minimal})  in \eqref{eq: bound by minimal practical} and using that $\cone\norm h_\tau \norm_1 =\cthree$ since  $\norm h_\tau \norm_1=o(\la^2) $.
\end{proof}

 \subsection{Sum over minimally irreducible diagrams} \label{sec: summing min irr pairings}
 
 In this section, we perform the integral
 \beq
 \mathop{\int}\limits_{\Sigma_t(> \tau, \mathrm{mir} )}  \d \dsi \, \str\zeta(\dsi) \str  \tilde\caE^{\ga}(\dsi)
 \eeq
 that appears on the RHS of the bound in Lemma \ref{lem: bound irr by min irr} (upon replacing $\ga \to \ga/2$), 
 and we prove that it is exponentially decaying in time with decay rate $O(\la^2)$, for well-chosen $\ga$ and $\la$ small enough, depending on $\ga$. 
 It is in this place that we  use the decay property of the renormalized correlation function that was stated in Lemma \ref{lem: renormalized propagator}.
 
 The key idea is the following. If $\dsi$ is a long diagram  with $\str \dsi \str=1$ consisting of the two triples $(t_i,x_i,l_i)_{i=1}^2$,  then
 \beq\label{eq: what if caC has one pair}
 \str\zeta(\dsi) \str  \tilde\caE^{\ga}(\dsi) = \cone \str  \psi^{\#}(x_2-x_1,t_2-t_1) \str \quad   \tilde\caI_{x_2,l_2}  \tilde\caZ^{\tau,\ga}_{t_2-t_1} \tilde \caI_{x_1,l_1}
 \eeq
 In this case, we can obviously use Lemma \ref{lem: renormalized propagator} to deduce exponential decay in $t_2-t_1$ of \eqref{eq: what if caC has one pair}, uniformly in $x_1,x_2,l_1,l_2$. 
  In general, there is of course more than one pairing in an irreducible diagram and so one has to `split' the decay coming from $ \tilde\caI_{x_2,l_2}  \tilde\caZ^{\tau,\ga}_{t_2-t_1} \tilde \caI_{x_1,l_1}$ between the different pairings, thus weakening the decay by a factor which can be as high as $\str \dsi \str$.
 
   However, since we are considering minimally irreducible pairings, there are at most two pairings bridging any given time $t'$, see Figure \ref{fig: fromirtomir1}. Hence, one can attempt to split the decay from  $ \tilde\caI_{x_2,l_2}  \tilde\caZ^{\tau,\ga}_{t_2-t_1} \tilde \caI_{x_1,l_1}$  in half. This can  be done and it is described in Lemma \ref{lem: bound min irr}.

For  $\dsi \in  \Si_{[0,t]}$, we define the function
 \beq
 H_{\tau}(\dsi) =\prod_{j=1}^{ \str\dsi\str} h_{\tau}(v_j-u_j)
 \eeq
 with $h_{\tau}$ as  in \eqref{def: htau}.  Note that  $H_{\tau}(\dsi)$ depends only on the equivalence class $[\dsi]$, and hence we can write $H_{\tau}([\dsi]):=H_{\tau}(\dsi)$.
 
  \begin{lemma} \label{lem: bound min irr}
Let the positive constants $g_{r}$ and $\de_r$ be defined as in Lemma \ref{lem: renormalized propagator} and choose $\ga < \delta_r$. Let $\ka=(\ka_{\links},\ka^{}_{\rechts} )$ such that $\str \Im\ka^{}_{\links,\rechts} \str \leq \ga/8 $ and fix a long irreducible diagram class  $ [\dsi] \in \Pi_T\Si_{[0,t]}( > \tau, \mathrm{mir})$. Then 
 \beq 
  \left  \norm  \sum_{\underline{x}(\dsi),\underline{l}(\dsi)} \str\zeta(\dsi) \str      \caJ_{\ka}   \tilde\caE^{\ga}(\dsi)  \caJ_{-\ka} \right\norm  \leq   c(\ga)^{\str \dsi \str}     \e^{ \ctwo t}  \e^{-\frac{1}{2} \la^2 g_r t} H_{\tau}([\dsi])   \label{eq: min irr}  
 \eeq
  Note that the operator between $\norm \cdot \norm$ depends on the equivalence class $[\dsi]$ only, due to the sum over
$\underline{x}(\dsi),\underline{l}(\dsi)$.
 \end{lemma}
 \begin{proof}
 For concreteness, we assume that $\str \dsi \str = n$ is even, the argument for $\str \dsi \str$ odd is analogous. 
We can find $\dsi_1$ and $\dsi_2$ such that $\dsi_1 \cup \dsi_2= \dsi, \str \dsi_1\str =\str \dsi_2\str=n/2 $ and  $\dsi_1$ and $\dsi_2$ are ladder diagrams, i.e.\ their decompositions into irreducible components consists of singletons. To be more concrete,  the time-pairs of $\dsi_1$ are
\beq
(t_1,t_3), (t_4,t_7), (t_8,t_{11}) \ldots   (t_{2n-4}, t_{{2n-1}})
\eeq
and those of $\dsi_2$ are
\beq
(t_2,t_5), (t_6,t_9), (t_{10}, t_{13}) \ldots   (t_{2n-2}, t_{{2n}})
\eeq
The possibility of making such a decomposition is a consequence of the structure of minimally irreducible diagrams, as illustrated in Figure \ref{fig: fromirtomir1}. 
\begin{figure}[ht!]
\vspace{0.5cm}
\begin{center}
\psfrag{dsi}{$\dsi$}
\psfrag{dsi1}{$\dsi_1$}
\psfrag{dsi2}{$\dsi_2$}
\includegraphics[width = 14cm]{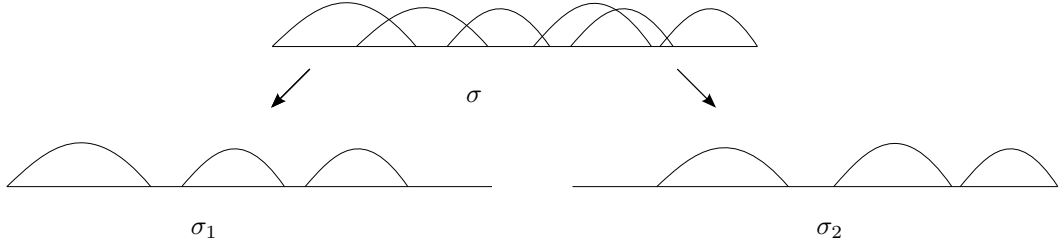}
\end{center}
\caption{\footnotesize{The decomposition of a minimally irreducible diagram $\dsi$ into two `ladder diagrams' $\dsi_1$ and $\dsi_2$. In the upper figure, one can easily check that any point on the (horizontal) time-axis is bridged by at most two pairings.}
}   \label{fig: fromirtomir1}
\end{figure}

 We now estimate the LHS of \eqref{eq: min irr}   in two ways. In our first estimate, we take the supremum over the $\underline{x},\underline{l}$-coordinates of  $\dsi_2$ and we keep those of $\dsi_1$.  In the second estimate, the roles of $\dsi_1$ and $\dsi_2$ are reversed. 

  We estimate
 \baq
 &&\sum_{ x(\dsi),l(\dsi)}  \str\zeta(\dsi) \str   \tilde\caE^{\ga}(\dsi)     \label{eq: bound by splitting min irr}  \\[2mm]
 &= &  [\cone]^{\str \dsi \str}    \sum_{ x(\dsi_2),l(\dsi_2)}   \str\zeta(\dsi_2) \str \sum_{ x(\dsi_1),l(\dsi_1)}  \str\zeta(\dsi_1) \str \left( \tilde \caI_{x_{2n},l_{2n}}  \caZ^{\tau,\ga}_{t_{2n}-t_{2n-1} }   \ldots   \tilde \caI_{x_2,l_2} \caZ^{\tau,\ga}_{t_2-t_1 }    \tilde \caI_{x_1,l_1} \right)    \nonumber \\[2mm]
 &\meleq &  [\cone]^{\str \dsi \str}       \e^{\ctwo t}  [\cone]^{\str \dsi_2 \str}  \left(  \sup_{\underline{x}(\dsi_2), \underline{l}(\dsi_2)}      \str\zeta(\dsi_2) \str  \right) 
  \nonumber \\[2mm]
   &&\qquad    \sum_{ x(\dsi_1),l(\dsi_1)}    \str \zeta(\dsi_1) \str     \tilde \caZ^{\tau,\frac{\ga}{2}}_{t_{2n}-t_{2n-1} } \tilde \caI_{x_{2n-1},l_{2n-1}}  \tilde \caZ^{\tau,\frac{\ga}{2}}_{t_{2n-1}-t_{2n-4} }   \ldots  \tilde \caI_{x_3,l_3} \tilde \caZ^{\tau,\frac{\ga}{2}}_{t_3-t_1 }    \tilde \caI_{x_1,l_1)}   \nonumber 
 \eaq
 The equality follows from the definition of $\tilde \caE^{\ga}(\dsi)$. To obtain the inequality on the third line, we perform the sum over  $\underline{x}(\dsi_2),\underline{l}(\dsi_2)$ by the same procedure that was used to obtain   \eqref{eq: lem bound by minimal3}, i.e., by using  Lemma \ref{lem: integrating out pairings}.  
 Since $\str\zeta (\dsi_1)\str $ factorizes into a function of the pairs in $\dsi_1$, the operator part in the last line of \eqref{eq: bound by splitting min irr}  is a product of two types of terms, namely; 
 \beq
\la^2 \mathop{\sum}
\limits_{\footnotesize{\left.\begin{array} {c}
 x_{4i},x_{4i+3}  \\
l_{4i},l_{4i+3}  \\
 \end{array} \right.} }
\left\str \psi(x_{4i+3}-x_{4i}, t_{4i+3}-t_{4i})\right\str   \tilde\caI_{x_{4i+3},l_{4i+3}}   \tilde\caZ^{\tau,\frac{\ga}{2}}_{t_{4i+3}-t_{4i}}  \tilde\caI_{x_{4i},l_{4i}}   \label{eq: bound i odd}  \eeq
for $i=1, \ldots, n/2-1$ (and an analogous term where we replace $4i \to 1$ and $4i+3 \to 3$, corresponding to the first pair in $\dsi_1$)
 and 
 \beq
  \tilde\caZ^{\tau,\frac{\ga}{2}}_{t_{4i}-t_{4i-1}}     \label{eq: bound i even}
 \eeq
 for $i=1, \ldots, n/2$.
 
We note that Lemma \ref{lem: renormalized propagator} provides a  bound on the matrix elements of these expressions. In particular, we use \eqref{eq: renormalized with correlation bound} to bound \eqref{eq: bound i odd}  and  \eqref{eq: renormalized apriori bound} to bound \eqref{eq: bound i even}. We obtain, for $\str\Im \ka^{}_{\links,\rechts} \str \leq \frac{\ga}{4}$,
   \baq
 \left( \caJ_{\ka}    \left(   \textrm{\ref{eq: bound i odd}} \right) \caJ_{-\ka}\right)_{ x^{}_{\links},x^{}_{\rechts}; x'_{\links},x^{'}_{\rechts} }    &\leq&   h_{\tau}\left(  t_{4i+3}-t_{4i}   \right)   \e^{-\la^2 g_r  \left(   t_{4i+3}-t_{4i}  \right)}
     \\[2mm]
   \left(  \caJ_{\ka}    \left(   \textrm{\ref{eq: bound i even}} \right) \caJ_{-\ka}  \right)_{ x^{}_{\links},x^{}_{\rechts}; x'_{\links},x^{'}_{\rechts} } & \leq &     e^{ \ctwo \left(  t_{4i}-t_{4i-1}   \right)}
   \eaq
   By the relation  stated in \eqref{eq: conversion exponential bound norm bound 1} and the line following it, we can convert these bounds on the kernels into bounds on the operator norms, yielding, for 
   $\str\Im \ka^{}_{\links,\rechts} \str \leq \frac{1}{8} \ga$
      \baq
 \left\norm    \caJ_{\ka}   \left(   \textrm{\ref{eq: bound i odd}} \right) \caJ_{-\ka}\right\norm  & \leq &  \cone    h_{\tau}\left(  t_{2i+1}-t_{2i-1}   \right)   \e^{-\la^2 g_r  \left(  t_{4i+3}-t_{4i}    \right)}
     \\[2mm]
 \left\norm   \caJ_{\ka}   \left(   \textrm{\ref{eq: bound i even}} \right) \caJ_{-\ka}  \right\norm & \leq&  \cone    \e^{ \ctwo \left(  t_{4i}-t_{4i-1}   \right)}
   \eaq
and hence, by multiplying these bounds for the operators appearing in \eqref{eq: bound by splitting min irr} and using that
\beq
h_\tau( t_{3}-t_{1})  \mathop{\prod}\limits_{i=1 }^{n/2-1} h_\tau( t_{4i+3}-t_{4i})   \leq H_\tau(\dsi_1), \qquad    \sup_{\underline{x}(\dsi_2), \underline{l}(\dsi_2)}      \str\zeta(\dsi_2) \str  \leq H_\tau(\dsi_2),
\eeq
 (see \eqref{eq: h dominates correlation}), we arrive at
 \beq \label{eq: product of bounds}
 \left\norm \caJ_{\ka}  \left(  \textrm{  \ref{eq: bound by splitting min irr}} \right)   \caJ_{-\ka} \right\norm  \leq  c(\ga)^{\str \dsi \str}    H_{\tau}(\dsi)     \e^{-  \la^2 g_{r}\str \Dom \dsi_1 \str} \e^{ \ctwo t} \eeq
  The claim of the lemma now follows by applying the same bound with the roles of $\dsi_1$ and $\dsi_2$ swapped, taking the geometric mean of the two bounds and noting that  \beq
 [0,t] \leq \str  \Dom \dsi_1 \str + \str \Dom \dsi_2 \str \eeq
\end{proof}

Next, we use  Lemmata \ref{lem: bound irr by min irr} and  \ref{lem: bound min irr} to prove Lemma \ref{lem: exponential decay long excitations}. By these two lemmas, the integral over renormalized irreducible diagrams is reduced to an integral over minimally irreducible equivalence classes $[\dsi]$.  Each equivalence class $[\dsi]$ essentially contributes  $ c(\ga)^{\str \dsi \str}  H_\tau(\dsi)$ to the integral.  Since $H_\tau(\dsi)$ is not exponentially decaying in $\Dom \dsi$, the Laplace transform of $H_\tau(\dsi)$ cannot be continued to negative $\Re z$. 
However, the factor $\e^{-\la^2 \frac{1}{2}g_r t }$ in Lemma  \ref{lem: bound min irr} enables  us to do such a continuation since the factors  $\cthree, \ctwo$ from Lemmata  \ref{lem: bound irr by min irr} and \ref{lem: bound min irr} can be made smaller than  $\la^2 \frac{1}{2}g_r$ by first choosing $\ga$ small enough, and then adjusting $\la$.
\vspace{0.2cm}

\noindent\emph{Proof of Lemma \ref{lem: exponential decay long excitations}}
We choose $\ga$ small enough, as required in the conditions of Lemmata \ref{lem: bound irr by min irr} and \ref{lem: bound min irr}, and we estimate, for $\str\Im \ka^{}_{\links,\rechts} \str \leq \ga/8$,
\baq
 \left\norm \mathop{\int}\limits_{\Si_{[0,t]}(>\tau, \mathrm{ir})} \d \dsi    \str\zeta(\dsi)\str    \caJ_{\ka}   \tilde\caE^{\ga}(\dsi)  \caJ_{-\ka} \right\norm   
 &\leq &    \e^{\cthree t}      \left\norm \mathop{\int}\limits_{\Si_{[0,t]}(>\tau, \mathrm{mir})} \d \dsi    \str\zeta(\dsi)\str    \caJ_{\ka}   \tilde\caE^{\frac{\ga}{2}}(\dsi)  \caJ_{-\ka} \right\norm   \label{eq: bound cac irr 4}  \\[2mm]
 &\leq&   \e^{\cthree t}   
 \mathop{\int}\limits_{\Pi_T\Si_{[0,t]}(>\tau, \mathrm{mir})} 
 \d [ \dsi]         \left \norm  \sum_{\underline{x}(\dsi), \underline{l}(\dsi)}     \str \zeta(\dsi) \str \,   \caJ_{\ka}  \tilde\caE^{\frac{\ga}{2}}(\dsi) \caJ_{-\ka}  \right\norm  \nonumber  \\[2mm]
  & \leq &    \e^{(\cthree -\frac{1}{4} \la^2 g_r)  t}       \mathop{\int}\limits_{\Pi_T \Si_{[0,t]}(>\tau, \mathrm{mir})}  \,  \d [ \dsi]   \,     c(\ga)^{\str \dsi \str}  \e^{-\frac{1}{4} \la^2 g_r t}  H_{\tau}(\dsi)   \nonumber
\eaq
The first inequality is Lemma \ref{lem: bound irr by min irr}, the second inequality uses the definition of the measure $\d \dsi$, and the third inequality 
 follows from   Lemma \ref{lem: bound min irr}.   

We will now estimate the Laplace transform of the integral in the last line of \eqref{eq: bound cac irr 4}. To prove   Lemma \ref{lem: exponential decay long excitations}, we fix $g_{ex} :=\frac{1}{8}  g_r $ and we show that one can choose $\ga$ such that, for $\la$ small enough and $\Re z  \geq - \la^2 g_{ex}$, 
\beq \label{eq: laplace of last proof}
\int_{\bbR^+} \d t  \, \e^{-tz}  \e^{(\cthree -\frac{1}{4} \la^2 g_r)  t} \, \mathop{\int}\limits_{\Pi_T \Si_{[0,t]}(>\tau, \mathrm{mir} )} \d [\dsi]     c(\ga)^{\str \dsi \str}    \e^{-\frac{1}{4} \la^2 g_r t}  H_{\tau}(\dsi)  =o(\la^2), \qquad \textrm{as} \, \, \la \searrow 0
\eeq
Of course, the choice of $\ga$ will have to depend on the specific value of $\cthree$.
To show  \eqref{eq: laplace of last proof}, we first choose $\ga$ such that, for $\la$ small enough, we have
\beq
\cthree -\frac{1}{4} \la^2 g_r \leq - \frac{1}{8} \la^2 g_r
\eeq
with $\cthree$ as in \eqref{eq: laplace of last proof}.

Consequently, we can dominate the factor $ \e^{-z t} \e^{(\cthree -\frac{1}{4} \la^2 g_r)  t}$ by $1$ in \eqref{eq: laplace of last proof}. 
Next, we note that, if $\dsi$ is minimally irreducible in the interval $[0,t]$, then  \beq  \sum_{i=1}^{\str \dsi \str} \str v_i-u_i \str \leq 2 t \eeq where $(u_i,v_i)$ are the pairs of time-coordinates associated to $\dsi$. This follows from the observation that each point  in the interval $[0,t]$ is bridged by at most two pairings of $\dsi$, see also Figure \ref{fig: fromirtomir1}. Consequently, we find that
\beq
  H_{\tau}(\dsi)   \e^{-\frac{1}{4} \la^2 g_r t}   \leq   \prod_{j=1}^{ \str\dsi\str} h_{\tau}(v_j-u_j)   \e^{-\frac{1}{8} \la^2 g_r  \str v_j-u_j \str} 
\eeq

We  estimate the LHS of \eqref{eq: laplace of last proof},  with $ \e^{-z t}  \e^{(\cthree -\frac{1}{4} \la^2 g_r)  t} $ replaced by $1$, by  invoking \eqref{eq: bound minimally irreducible classes all} in Lemma \ref{lem: combinatorics of irreducible diagrams}, with 
\beq
k(t):=     c(\ga)   \e^{-\frac{1}{8} \la^2 g_r  t } h_\tau(t) \qquad \textrm{and}\qquad a:=0
\eeq
Indeed, by using  $ \norm  h_{\tau} \norm_1=o(\la^2)$ and the exponential decay $ \e^{-\frac{1}{8} \la^2 g_r  t }$, we obtain that
\baq
\norm  k \norm_1 &= &   c(\ga)   o(\la^2)    \label{eq: finale 1}  \\[3mm]
\norm  t k \norm_1   &= & c(\ga)     o(\str \la \str^0), \qquad \textrm{as} \, \la \searrow 0   \label{eq: finale 2}
\eaq
Therefore,  the bound \eqref{eq: bound minimally irreducible classes all} yields  \eqref{eq: laplace of last proof}.
\qed

\subsection{List of important  parameters} \label{sec: list of parameters} 

We list some constants and functions that we use throughout the paper.
We start with the different decay rates; in the third column we indicate where the constant appears first. By ``full model", we mean: ``the model without cutoff".
 \begin{center}

 \begin{tabular}{|c|l|l|}\hline 
     $g_\res$  & bare reservoir correlation fct.\ (for subluminal speed)    &  Lemma \ref{lem: exponential decay}        \\[1mm] 
      \hline
       $\frac{1}{\tau}= \str\la\str^{1/2}$&   bare reservoir correlation fct.\  in the cut-off model   &   Section \ref{sec: cut off model}    \\[1mm]         \hline
        $ \la^2g_{r}$&  renormalized joint $\sys-\res$ correlation function    &  Lemma \ref{lem: renormalized propagator}        \\[1mm]         \hline
          $ \la^2 g_{rw}$&  Markov semigroup    &  Proposition  \ref{prop: properties of m}       \\[1mm] 
      \hline
       $\la^2g_{c}$&  cut-off model     &  Lemma \ref{lem: bounds on cutoff dynamics}      \\[1mm]       \hline
     $ \la^2g_{ex}$ &  excitations in the full model &   Lemma \ref{lem: exponential decay long excitations}    \\[1mm]       \hline
        $ \la^2g $&  full model       &   Theorem  \ref{thm: equipartition}    \\[1mm] 
 \hline 
\end{tabular}  
\end{center}

Additionally, the rates $g_{rw}, g_{c}, g$ come with a superscript $\scriptsize{low,high}$  indicating that the gap refers to small, large fibers $p$, respectively. \\ \vspace{0.1cm}

\noindent The  following constants restrict the values of complex deformation parameters, in particular the parameter $\ka$ in $\caJ_{\ka}$, as defined in \ref{def: caJ}.
 \begin{center}

 \begin{tabular}{|c|l|l|}\hline 
     $\delta_{\ve}$  &  particle dispersion law  & Assumption \ref{ass: analytic dispersion}      \\[1mm] 
      \hline
       $\delta_\res $&   reservoir dispersion law  &   Assumption  \ref{ass: analytic form factor}  \\[1mm]         \hline
        $\delta $&              full  model         &  Theorem \ref{thm: main}    \\[1mm]         \hline
          $\delta_{rw}$&   full  Markov semigroup    &  Proposition  \ref{prop: properties of m}       \\[1mm] 
      \hline
     $ \delta_{ex}$ &  excitations in the full  model &   Lemma \ref{lem: polymer model full}    \\[1mm]       \hline
        $\delta_{r}$&  renormalized  $\sys-\res$ correlation fct.\       &  Lemma \ref{lem: renormalized propagator}  \\[1mm] 
 \hline 
\end{tabular}  

\end{center}
\vspace{0.3cm}
%


\noindent  The following functions of $\ga,\la$ appear as blowup-rates in exponential bounds.
 \begin{center}
 \begin{tabular}{|c|l|l|}\hline 
     $r_{rw}(\ga,\la)$&   Markov semigroup    &  Lemma \ref{lem: bounds with tilde}    \\[1mm]       \hline \hspace{2mm}
     $r_{\ve}(\ga,\la)$  &  $ \tilde \caU^{\ga}_t$   & Section \ref{sec: bounding operators}     \\[1mm] 
      \hline
       $r_{\tau}(\ga,\la)$&   $\tilde \caZ^{\tau\ga}_t$ &   Lemma  \ref{lem: bounds on cutoff dynamics}  \\[1mm]         \hline
\end{tabular}  
\end{center}
In the final Sections \ref{sec: bounds on long pairings} and \ref{sec: model only long}, these rates are represented by the generic constants $\ctwo,\cthree$, introduced in Section \ref{sec: generic constants}. 
%
%

\appendix

 \section{Appendix: The reservoir correlation function}  
 \label{app: reservoirs}

In this appendix, we study the reservoir correlation function $\psi(x,t)$ and we prove Lemmata \ref{lem: exponential decay} and \ref{lem: integrability}.  Recall the definition of the ``effective squared form factor" $\hat\psi$ in \eqref{def: psi}. It is related to $\psi(x,t)$ by, see \eqref{def: correlation function},
\beq  \label{def: repetition psi}
 \psi(x,t) =     \mathop{\int}\limits_\bbR  \d \om  \mathop{\int}\limits_{\bbS^{d-1}}\d s \,    \hat\psi (\om) \e^{\i t \om}   \e^{\i  \om s \cdot x  } 
\eeq
From this expression, one understands that  $\psi(x,t)$ cannot  have exponential decay in $t$, uniformly in $x$.  One also sees that, for $x$ fixed, there is exponential decay provided that 
 $ \hat\psi (\cdot)$ is analytic in a strip around $\bbR$.

Consider $q(\cdot)$ such that
\beq \label{eq: psi as sphere integral}
 \psi(x,t)  =   \mathop{\int}\limits_{\bbS^{d-1}}\d s \,   q(t+ s \cdot x)
\eeq
 By Assumption \ref{ass: analytic form factor}, there is a $\delta_{\res}>0$ such that
$ q(t)$  decays  as $C \e^{ -\delta_{\res} \str t \str} $. Choosing $v^*=\frac{1}{2}$, we obtain, for $ \str x \str / \str t \str  \leq v^*$, 
\beq
\str \psi(x,t) \str  \leq    \e^{-g_\res \str t  \str}  C,  \qquad  \textrm{with} \,   g_\res := \frac{2}{3} \delta_\res
\eeq
which proves Lemma \ref{lem: exponential decay}.  From now on, we assume that $ \str x \str / \str t \str  > v^*$.  \vspace{3mm}

We remark that \eqref{eq: psi as sphere integral} can be rewritten, after an explicit calculation,  as
\beq
 \psi(x,t)  = \int_{-1}^{1} \d \eta \,     q(t+\eta \str x\str)  a(\eta), \qquad   a(\eta) :=  \mathrm{Area}(\bbS^{d-2})     \left( 1-\eta^2\right)^{\frac{d-3}{2}}  
\eeq
By Assumption \ref{ass: analytic form factor}, in particular  the condition $\hat \psi(0)=0$ and the analyticity of $\hat \psi$, we deduce that $\frac{\hat\psi(\om)}{\om}$ is analytic in a strip around $\bbR$, as well. Its Fourier transform, $Q$, is an exponentially decaying $C^1$-function (since $\hat\psi(\om) \in L^1$) whose derivative equals $q$. 

Hence,  by partial integration, the fact that $a(\eta)\big\str_{-1}= a(\eta)\big\str_{1}=0$ for $d >3$, and  the change of variables $\zeta= \str x \str \eta$, we obtain
\beq \label{eq: integral in zeta}
 \psi(x,t)  =- \frac{1}{\str x \str^{3/2}} \int_{-\str x \str }^{\str x \str} \d \zeta   \,   Q(t+\zeta) \,  \frac{1}{\str x \str^{1/2}}a'(\frac{\zeta}{\str x \str}), \qquad Q'=q
\eeq
Here, $Q'$ and $a'$ stand for the derivatives of $Q$ and $a$.  
We evaluate this integral by splitting it into the regions 
\beq 
 -\str x \str \leq  \zeta  \leq -\str x \str+ 1, \qquad  -\str x \str+ 1 \leq \zeta  \leq  \str x \str- 1, \qquad  \str x \str- 1 \leq \zeta  \leq \str x \str
 \eeq
  In the second region, we dominate the integral \eqref{eq: integral in zeta} by $ \norm Q \norm_1\, \times\,  \norm \frac{1}{\str x \str^{1/2}}a'(\frac{\zeta}{\str x \str}) \norm_{\infty}$ (we  assume here that $\str x \str \geq 1$, otherwise the decay in $t$ has been proven above). In the first and third region, we dominate  the integral \eqref{eq: integral in zeta} by $ \norm Q \norm_{\infty}\, \times\, \norm \frac{1}{\str x \str^{1/2}}a'(\frac{\zeta}{\str x \str}) \norm_{1}$. Using the explicit form of $a'$ and the fact that $ \str x \str / \str t \str  > v^*$,
we conclude
\beq
\sup_x \str \psi(x,t) \str \leq C (1+\str t \str)^{-3/2}, \qquad  \textrm{for}  \, \, d \geq 4   \label{eq: dispersive estimate}
\eeq
which implies Lemma \ref{lem: integrability}. 
Obviously, dispersive estimates  like \eqref{eq: dispersive estimate} can be derived in much greater generality, see e.g.\ \cite{evansbook}.

  \section{Appendix: Spectral perturbation theory}  
 \label{app: spectral}
 
Let $\epsilon \in \bbR$ be a small parameter and consider a continuous function $ \bbR^+ \ni t  \mapsto  V(t, \epsilon) $, taking values in  a Banach space, and  such that  
\beq \mathop{\sup}\limits_{t \geq 0 } \e^{-t m}\norm V(t, \epsilon) \norm < \infty, \qquad  \textrm{for some} \, \,  m>0.    \label{eq: bound on V} \eeq 
The Laplace transform 
\beq \label{def: abstract laplace}
 A(z, \epsilon) := \int_{\bbR^+} \d t   \e^{-t z } V(t, \epsilon) 
\eeq
is well-defined for $\Re z >m $ and it follows (by the inverse Laplace transform) that
\beq
V(t, \epsilon) =  \frac{1}{2 \pi \i} \mathop{\int}\limits_{\Ga^{\rightarrow}} \d z  \,   \e^{z t}  A(z,\epsilon), \qquad \qquad \textrm{with} \, \,  \Ga^{\rightarrow}:= m'+\i \bbR \,\, \textrm{for any} \, m'>m
\eeq
where the integral is meant in the sense of improper Riemann integrals.

We will state  assumptions that allow to continue $A(z,\epsilon) $ downwards in the complex plane, i.e., to $\Re z \leq m$ and to obtain bounds on $V(t, \epsilon) $. 

\begin{lemma}\label{lem: spectral abstract}
For $\Re z$ large enough, let
\beq \label{def: expression A}
 A(z, \epsilon)  :  =  (z- \i  B-  A_{1}(z,\epsilon))^{-1}
\eeq  
and assume the following conditions.
\ben
\item{$ B$ is  bounded and  its  spectrum consists of finitely many semisimple eigenvalues on the real axis, that is
\beq
B= \sum_{b \in \sp B} b 1_b(B),
\eeq
where $1_b(B)$ is the spectral projection corresponding to the eigenvalue $b$. For concreteness, we assume that $0 \in \sp B$.  }

\item{ For $\epsilon$ small enough, the operator-valued function $z \mapsto A_{1}(z,\epsilon)$  is analytic in the domain $\Re z>- \epsilon g_{A}$ and  
\baq
  \mathop{\sup}\limits_{\Re z >-   \epsilon g_{A} } \norm  A_{1}(z,\epsilon)\norm& =&O(\epsilon),\qquad  \epsilon \searrow 0,   \label{eq: conditions A2}  \\[2mm] 
    \mathop{\sup}\limits_{\Re z >-   \epsilon g_{A} } \norm  \frac{\partial}{\partial z} A_{1}(z,\epsilon)\norm &=&o(\str \epsilon\str^0), \qquad  \epsilon \searrow 0  \label{eq: conditions A3}
\eaq
}

\item{ There are bounded  operators $N_b$, for $b \in \sp B$, acting on the spectral subspaces $\Ran 1_b(B)$  and such that, for all $b \in \sp B$,
\beq
    \epsilon N_b - 1_b(B)  A_{1}(\i b,\epsilon) 1_b(B)= o(\epsilon), \qquad  \epsilon \searrow 0.  \label{eq: app closeness A and N}
\eeq
Consider the operator
\beq
 N:= \mathop{\oplus}\limits_{b \in \sp B}  N_b, \qquad \textrm{with} \, [B,N]=0,
 \eeq
and assume that $N$  has a simple eigenvalue $f_{N}$ such that 
\beq \label{eq: gap abstract}
 \sp N= \{ f_N \}  \cup \Om_N \qquad  \textrm{and}   \qquad  \sup  \Re \Om_N \leq   - g_{N}
\eeq 
for some gap $g_{N}>0$.  We also require that 
\beq \label{eq: eigenvalue smaller than gap}
\Re f_N   >  - g_{N}, \qquad     \Re f_N   >     -  g_{A}   
\eeq
The eigenvalue $f_N$ is necessarily an eigenvalue of  $N_b$ for some $b \in \sp B$. For concreteness (and to match with our applications), we assume that it is an eigenvalue of $N_0$
}

\een

Then, there is an $\epsilon_0 >0$ such that, for $\str \epsilon \str \leq \epsilon_0$, there is a  number  $  f(\epsilon)$, a rank-one operator  $P(\epsilon)$,  bounded operators $R(t,\epsilon)$ and a decay rate $g>0$, such that 
\beq \label{eq: exponential decay1 abstract}
V(t, \epsilon)    =    P(\epsilon) \e^{f(\epsilon) t}  +   R(t,\epsilon) \e^{- \epsilon g t } 
\eeq
with
\baq
 f(\epsilon) -\epsilon f_N &=& o(\epsilon)    \label{eq:  lemma bounds1 abstract} \\
 \norm    P(\epsilon)- 1_{f_{N}}(N)\norm &=& o(\str\epsilon\str^0)  \label{eq:  lemma bounds2 abstract}  \\ 
\sup_{t \in \bbR^+}  \norm R(t,\epsilon) \norm &=& O(\str\epsilon\str^0), \qquad  \textrm{as}\, \,  \str \epsilon \str \searrow 0  \label{eq:  lemma bounds3 abstract} 
\eaq
with $1_{f_{N}}(N)$ the spectral projection  of $N$ associated to the eigenvalue $f_{N}$. The decay rate $g$ can be chosen arbitrarily close to $\min\{g_N, g_A \}$ by making  $\epsilon_0$ small enough. In particular, one can choose $g$ and $\epsilon_0$ such that $ \Re f(\epsilon)  > - \epsilon g $ for all $\str\epsilon \str \leq \epsilon_0 $.

If, in addition $N$ and  $A_{1}$ depend analytically on a parameter $\al$ in a complex domain $ \caD \subset \bbC,$ such that  
\eqref{eq: conditions A2}-\eqref{eq: conditions A3}-\eqref{eq: app closeness A and N}-\eqref{eq: gap abstract}-\eqref{eq: eigenvalue smaller than gap} hold uniformly in $\al \in \caD$, then  \eqref{eq: exponential decay1 abstract} holds  with $f, P$ and $R$ analytic in $\al$ and the estimates \eqref{eq:  lemma bounds1 abstract}-\eqref{eq:  lemma bounds2 abstract}-\eqref{eq:  lemma bounds3 abstract}  are satisfied uniformly in $\al \in \caD$.

\end{lemma}

 Lemma \ref{lem: spectral abstract}  follows in a straightforward way from spectral perturbation theory of discrete spectra. For completeness, we give a proof below, using freely some well-known results that can be found in, e.g., \cite{katoperturbation}. 

\begin{lemma} \label{lem: spectrum change small o}
The singular points of  $A(z,\epsilon)$ in the domain $ \Re z \geq - \epsilon g_{A}$
lie within a distance of  $o(\epsilon)$ of the spectrum of $\i B +\epsilon N$ (provided that there are any singular points at all).
\end{lemma}
\begin{proof}

Standard perturbation theory implies that the spectrum of the operator
\beq \label{eq: operator to be perturbed}
\i  B+ A_{1}(z,\epsilon), \qquad \textrm{for} \, \,       \Re z \geq - \epsilon g_{A}
\eeq  
 lies at a distance  $O(\epsilon)$ from the spectrum of $\i B$. Here and in what follows, the estimates in powers of $\epsilon$ are uniform for $ \Re z \geq -\epsilon g_{A}$.  Let $1_b^0 \equiv 1_b(B)$ be the spectral projections of $B$ on the eigenvalue $b$.  As long as $\epsilon$ is small enough, there is an invertible operator $U\equiv U(\epsilon,z)$ satisfying $\norm U-1 \norm=O(\epsilon)$ and such that the projections
\beq
1_b := U  1^0_b  U^{-1}, \qquad b \in \sp B
\eeq
are spectral projections of the operator \eqref{eq: operator to be perturbed} associated to the spectral patch originating from the eigenvalue $b$ at $\ep=0$ (see Chapter 2.4 of \cite{katoperturbation} for an explicit construction of $U$).
It follows that the spectral problem for \eqref{eq: operator to be perturbed} is equivalent to the spectral problem for
\beq
 \sum_b   U^{-1}   1_b \left( \i  B+  A_{1}(z,\epsilon) \right)1_b  U   
=  \sum_b \left(  \i b 1^0_b  +\epsilon N_b +  A_{ex,b}(z,\epsilon)   \right) \label{eq: to what it is equivalent}
\eeq
where 
\beq
\left.\begin{array}{rclr}
A_{ex,b}(z,\epsilon) &:=&  1^0_b U^{-1}( \i B)  U 1^0_b - \i  b 1^0_b, \qquad \qquad &  (O(\epsilon^2))      \\[2mm]
& +&   \epsilon   1^0_b U^{-1} N_b  U 1^0_b - \epsilon N_b , \qquad \qquad &  (O(\epsilon^2))     \\[2mm]
& +&      1^0_b U^{-1}  \left( A_{1}(\i b,\epsilon)  - \epsilon N_b               \right)   U  1^0_b, \qquad \qquad & (o(\epsilon))    \\[2mm]
&+&      1^0_b U^{-1}  \left( A_1(z,\epsilon) -A_1(\i b,\epsilon)   \right)   U  1^0_b, \qquad \qquad &  ( \str z-\i b \str o(\str\epsilon\str^0)  )
\end{array}\right. \label{eq: bounds on tildeAex}
\eeq
The estimates in powers of $\epsilon$ are obtained by using  $U-1=O(\epsilon)$, the property $1_b  U =  U 1_b^0$ and the bounds \eqref{eq: conditions A2}-\eqref{eq: conditions A3}-\eqref{eq: app closeness A and N}.
When $z$ is chosen at a distance $O(\epsilon)$ from $\i b$, then all terms in \eqref{eq: bounds on tildeAex} are $o(\epsilon)$. 
The claim of Lemma \ref{lem: spectrum change small o} now follows by simple perturbation theory applied to the RHS of \eqref{eq: to what it is equivalent}.

\end{proof}

\begin{lemma}\label{lem: at least one singularity}
The function $A(z)$
has exactly one singularity at  a distance $o(\epsilon)$ from $\epsilon f_N$. This singularity is called $f \equiv f(\epsilon)$.   The corresponding residue  $P \equiv P(\ep)$  is a rank-one operator satisfying
\beq \label{ex: P close to projection}
\norm P- 1_{N}(f_N) \norm =o(\str\epsilon\str^0), \qquad     \epsilon \searrow 0
\eeq
\end{lemma}
\begin{proof}
By  Lemma \ref{lem: spectrum change small o}, there can be at most one singularity. We prove below that there is at least one. 
By the reasoning in the proof of Lemma \ref{lem: spectrum change small o} and the fact that the eigenvector corresponding to $f_N$ belongs to $\Ran 1^{0}_{b=0}$ (see condition 3) of Lemma \ref{lem: spectral abstract}), it suffices to study the singularities of the function
\beq \label{eq: function to have effectively at least one}
 z \mapsto \left(z- \epsilon N_0+ A_{ex,0}(z,\epsilon) \right)^{-1}   
\eeq
Let the contour $\Ga^{f} \equiv \Ga^{f}(\epsilon)$ be a  circle with center $\epsilon f_N$ and radius   $\epsilon r$ for some $r>0$.
Clearly, for $r$ small enough,  the entire spectrum of $\epsilon N_0$ lies outside  the contour $\Ga^f$, except for the eigenvalue $\epsilon f_N$.  The contour integral of $(z-\epsilon N)^{-1}$ along $\Ga^f$ equals the spectral projection corresponding to $f_N$.  We estimate
\baq
 && \oint_{\Ga^f} \d z \,   \left[(z-\epsilon N_0-A_{ex,0}(z,\epsilon))^{-1}-    (z-\epsilon N_0)^{-1}   \right] \\ [2mm]
  &=& \oint_{\Ga^f}  \d z \, (z-\epsilon N_0)^{-1}  A_{ex,0}(z,\epsilon)  (z-\epsilon N_0- A_{ex,0}(z,\epsilon) )^{-1}     \\ [2mm]
    &  =&   \oint_{\Ga^f}  \d z \,       ( \epsilon^{-2}c(r))^{2}   o(\epsilon), \qquad  \textrm{as}  \,\, \epsilon \searrow 0  \,\,  \textrm{with} \,\, c(r):= \sup_{\str z-f_N \str=r}  \norm (z-N_0)^{-1} \norm,          \label{eq: at least one4}
\eaq
The last estimate holds in norm  and it follows from the bound $ \norm A_{ex,0}(z,\epsilon) \norm =o(\epsilon)$, see \eqref{eq: bounds on tildeAex}.
The expression \eqref{eq: at least one4} is $o(1)$, as $\epsilon \searrow 0$,  since the circumference of the contour $\Ga^f$ is $2 \pi r \epsilon$. 
From the fact that the contour integral of  \eqref{eq: function to have effectively at least one} does not vanish, we conclude that $A(z)$
has at least one singularity inside $\Ga^f$. 

The claim about the residue is most easily seen in an abstract setting.  Let $F(z)$ be an operator-valued analytic function in some open domain containing $0$, and such that  $0 \in \sp F(0)$ is an isolated eigenvalue.
We have  the Taylor expansion
\beq
F(z) = \mathop{\sum}_{n \geq 0}  \frac{z^{n}}{n!}  F_{n} , \qquad F_{n} := F^{(n)}(0), \qquad  0 \in \sp F_0
\eeq
If $\norm F_1-1 \norm$ is small enough,  then also $F_1^{-1} F_0$  has $0$ as an isolated eigenvalue. We denote the corresponding spectral projection by $1_0(F_1^{-1} F_0)$ and we calculate the residue ($\mathrm{res}$) at $0$
\beq
\mathrm{res} (F(z)^{-1})  =\mathrm{res}  ( (F_0 + z F_1)^{-1} ) =  \left(  \mathrm{res} ( F_1^{-1} F_0 + z )^{-1}  \right)   F_1^{-1} = 1_0(F_1^{-1} F_0)    F_1^{-1}.
\eeq
The last expression is clearly a rank-one operator.  In the case at hand, $F_1^{-1}=1+o(\str \epsilon \str^0)$, as $\epsilon \searrow 0$, which yields \eqref{ex: P close to projection}.
\end{proof}

We proceed to the proof of Lemma \ref{lem: spectral abstract}.

First, we choose the rate $g$ such that  $ f_N <   g < \min\{ g_A, g_N \}$ and we fix  the  contours $\Ga^{f}$ and $\Ga_{\rightarrow}$ (see also Figure \ref{fig: spectral}); 
\begin{itemize}
\item The contour $\Ga^f$ is as described in Lemma \ref{lem: at least one singularity}, with $r <  \str g-f_N \str $. In particular,  for small $\epsilon$, it encircles the point $f$ but no other singular points of $A(z)$.
\item The contour $\Ga_{\rightarrow}$ is given by $\Ga_{\rightarrow}:= - \epsilon g  +\i \bbR$.
\end{itemize}

\begin{figure}[ht!]
\vspace{0.5cm}
\begin{center}
\psfrag{rez}{$\Re z$}
\psfrag{imz}{$\i \Im z$}
\psfrag{contourgamma}{$\Ga^{\rightarrow}$}
\psfrag{cgpriem}{$\Ga_{\rightarrow}$}
\psfrag{contourcf}{$\Ga^f$}
\psfrag{f}{$\epsilon f_N$}
\psfrag{contourczero}{$b=0$-patch excluding $\epsilon f_N$  }
\psfrag{ccright}{ $b=b_{-1}$-patch}
\psfrag{ccleft}{$b=b_{1}$-patch}
\psfrag{nonanalyticity}{ $A_1(z,\ep)$ nonanalytic}
\psfrag{gapg}{$\epsilon g$}
\psfrag{gapN}{$\epsilon g_N$}
\psfrag{gapgres}{$\epsilon g_A$}
\psfrag{small}{$o(\ep)$}
\includegraphics[width = 14cm]{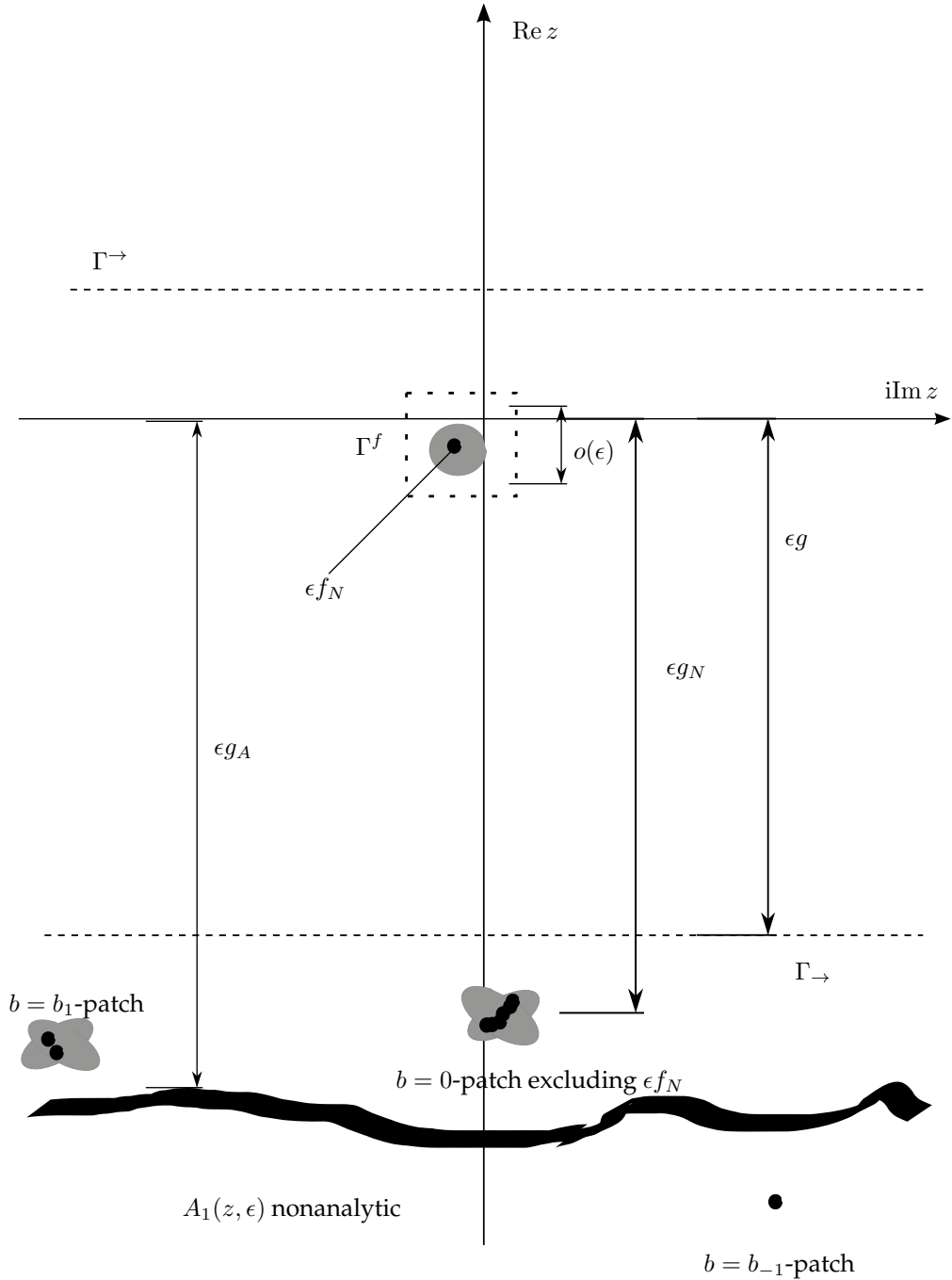}
\end{center}
\caption{\footnotesize{The (rotated) complex plane. The black dots  indicate the spectrum of $\i B+ \epsilon N$ (which need not be discrete). The upper dot is the eigenvalue $\ep f_N$. In the  picture, we have assumed that the spectrum of $B$ consists of 3 semisimple eigenvalues: $0,b_1,b_{-1}$. The gray patches contain the possible singularities of the function $A(z)$ above the irregular black line. These singularities lie at $o(\epsilon)$ from the spectrum of $\i B+\epsilon N$. Below the irregular black line, i.e., in the region $\Re z< -\ep g_A$, we have no control since $A_1(z,\ep)$ ceases to be analytic in that region (hence we have also not drawn a patch around $b_{-1}$).  The integration contours $\Ga^{\rightarrow},\Ga_{\rightarrow}$ and $\Ga^f$ are drawn as dashed lines.}
}   \label{fig: spectral}
\end{figure}

By Lemma \ref{lem: spectrum change small o}, we know that for small $\epsilon$, there are no singularities of $A(z)$  in the region $\Re z >- \epsilon g$ except for the point $z \equiv f$.  Hence,  we can deform contours as follows
\baq
V(t, \epsilon)& =&  \frac{1}{2 \pi \i} \mathop{\int}\limits_{\Ga^{\rightarrow}} \d z \,  \e^{z t}  A(z, \epsilon) \\
& =&  \frac{1}{2 \pi \i} \mathop{\oint}\limits_{\Ga^{f}} \d z \,    \e^{z t}  A(z, \epsilon)    +    \frac{1}{2 \pi \i} \mathop{\int}\limits_{\Ga_{\rightarrow}} \d z   \,  \e^{z t}  A(z, \epsilon)  \label{eq: deformation contours2}
\eaq
The first term in \eqref{eq: deformation contours2} yields $ \e^{t f} P$.
The second term of \eqref{eq: deformation contours2} is split as follows
\baq
&&  \int_{\Ga_{\rightarrow}} \d z   \,  \e^{z t}  A(z, \epsilon)     \\
 &=&     \int_{\Ga_{\rightarrow}} \d z \,    \e^{z t}   (z-\i B-\epsilon N)^{-1}  \label{eq: deformation contours3}  \\
&+&     \int_{\Ga_{\rightarrow}}   \d z   \, \e^{z t}   (z-\i B-\epsilon N)^{-1}    (A_{1}(z,\epsilon)-\epsilon N  )  A(z, \epsilon)    \label{eq: deformation contours4}  
\eaq
The term  \eqref{eq: deformation contours3} equals
\beq \label{eq: estimate contour semigroup}
  \e^{ t \left(  \i B+ \epsilon   1_{\Om_N}(N) N\right) } =O(\e^{-\epsilon g_N  t}),  \qquad  t \nearrow \infty
\eeq
since the contour $\Ga_{\rightarrow}$ can be closed in the lower half-plane to enclose the spectrum of  $\i B +\epsilon N$ minus the eigenvalue $\epsilon f_N$, i.e., the set $\epsilon\Om_N$.

The integrand of \eqref{eq: deformation contours4}  decays as $\str z \str^{-2}$ for $z \nearrow \infty$ , since for   a bounded operator  $M$
\beq
\norm (z-M)^{-1} \norm =O( \frac{1}{ \str z \str} ), \qquad  \str z \str \nearrow \infty
\eeq
Using that $A_1(z,\ep)=O(\ep)$, it is now easy to establish that the integral in \eqref{eq: deformation contours4} is $O(1)$, as $\epsilon \searrow 0$. 
One extracts   $\e^{t\Re z}$ from the integration \eqref{eq: deformation contours4} to get the bound $O(\e^{-\epsilon g t})$. Together with \eqref{eq: estimate contour semigroup}, this proves  Lemma \ref{lem: spectral abstract}.

 \section{Appendix: Construction and analysis of the Lindblad generator $\caM$ }  
\label{app: caM}

The operator $\caM$ was introduced at the beginning of Section \ref{sec: lindblad}. We provide a more explicit construction and we prove Propositions \ref{prop: wc1} and \ref{prop: properties of m}.
\subsection{Construction of $\caM$} \label{sec: construction of caM}

First, we note that  by using the notions introduced in Section  \ref{sec: dyson expansion}, the operator  $\caL(z)$, defined in Section  \ref{sec: construction of semigroup}, can be expressed as
\beq \label{eq: caL as perturbation}
\caL(z) =  \mathop{\int}\limits_{\bbR^+} \d t \,  \e^{-t z}    \mathop{\sum}\limits_{x_1,x_2,l_1,l_2}  \psi^{\#}(x_2-x_1,t)   \,    \caI_{x_2,l_2}  \e^{-\i \adjoint(Y)t}   \caI_{x_1,l_1}
\eeq
where $\psi^{\#}$ equals $\psi$ or $\overline{\psi}$, depending on $l_1,l_2$, according to the rules  in \eqref{def: zeta}. In words, $\la^2 \caL(z)$ contains the terms of order $\la^2$ in the Lie-Schwinger series of Lemma \ref{lem: definition dynamics}.

Next, we define some auxiliary objects.
\baq
\Upsilon &:=& \Im  \sum_{a \in \sp(\adjoint(Y))}    W_{a}  W^*_{a} \,  \int_{\bbR^+}    \d t  \,     \psi(0 ,t)    \e^{\i a t}        \\[2mm]
\Psi(\rho) &:=&  \sum_{x^{}_\links, x^{}_\rechts \in \lat}   \sum_{a \in \sp(\adjoint(Y))}     \left( \int_{\bbR} \d t \, \e^{\i a t} \psi(x^{}_\links- x^{}_\rechts,t) \right)\,  \times \, (1_{x^{}_\links} \otimes W_{a})   \rho         (1_{x^{}_\rechts} \otimes  W_{a})^{*} \label{eq: explicit Psi} 
\eaq

The operator  $\Upsilon=\Upsilon^* \in \scrB(\scrS)$ was already referred to in Section \ref{sec: lindblad}. 
From the above expression and the definition of $W_a$ in \eqref{def: Wa}, we check immediately that  $[Y,\Upsilon]=0$.
Further,  we can rewrite \eqref{eq: explicit Psi} as 
\beq
\Psi(\rho) =   2 \pi \mathop{ \sum}\limits_{a \in \sp(\adjoint(Y))}  \mathop{ \int} \limits_{\bbS^{d-1}} \d s   \,     \hat\psi(a)
V(s,a)   \rho  V^*(s,a)    \label{eq: explicit Psi kraus} 
\eeq
with $\hat\psi(\cdot)$ as in \eqref{def: psi} and \eqref{def: correlation function}, and 
\beq
V(s,a) :=    \sum_{x \in \lat} \e^{\i a s \cdot x} 1_x \otimes W_{a} 
\eeq
The expression \eqref{eq: explicit Psi kraus} is essentially the Kraus decomposition of $\Psi$, see \cite{alickifannesbook}, and hence it shows that $\Psi$ is a completely positive map.  This means in particular that for $\rho \geq 0$, $\norm \Psi(\rho)\norm_1=\Tr \Psi(\rho)$, and hence, by using \eqref{eq: explicit Psi}, 
 \beq
 \norm\Psi\rho\norm_1  = \sum_{x}    \Tr_{\scrS} \left[(\Psi\rho)(x, x) \right]  \leq     \left( \dim \scrS \int \d t \str \psi(0,t)\str  \norm W \norm^2  \right)     \norm\rho \norm_1 
 \eeq
 where the finiteness of the factor between brackets on the RHS  is  implied by  Lemma \ref{lem: integrability}.   Since any trace class operator can be written as a linear combination of four positive trace class operators, it follows that $\Psi$ is bounded on $\scrB_1(\scrH_\sys)$  and  by a similar calculation, one can check that $\Psi$ is also bounded on $\scrB_2(\scrH_\sys)$. 
 
 We are now ready to verify that
 \beq \label{eq: general form generator}
\caM(\rho) =  - \i [  \ve(P)+ \Upsilon, \rho] +    \Psi (\rho) - \frac{1}{2} (\Psi^*(1) \rho+\rho \Psi^*(1) ). 
\eeq
Indeed, this is checked most conveniently starting from \eqref{eq: def of M} and employing \eqref{eq: caL as perturbation}.  The terms with $l_1 \neq l_2$ give rise to $\Psi(\rho)$, while the terms with $l_1=l_2$ give rise to $-\i[\Upsilon, \rho]$ and $- \frac{1}{2} (\Psi^*(1) \rho+\rho \Psi^*(1) )$. Moreover, by the boundedness and complete positivity of $\Psi$ and the representation \eqref{eq: general form generator}, it follows that $\caM$ is of Lindblad type, see e.g.\ \cite{alickifannesbook}.
 Starting from \eqref{eq: explicit Psi kraus}, one can derive the momentum space representation of $\caM$ given in Section \ref{sec: momentum representation of caM}. For example, by expressing $V(s,a)$ in momentum representation, one obtains
\beq
\Psi(\rho) (k^{}_\links, k^{}_\rechts) =  2 \pi \mathop{ \sum}\limits_{a \in \sp(\adjoint(Y))}  \mathop{ \int} \limits_{\bbS^{d-1}} \d s   \,     \hat\psi(a)   W_{a} \rho(k^{}_\links +s a, k^{}_\rechts+s a)   W^*_{a} 
\eeq
which gives rise to the first term of \eqref{def: caN}.

\subsubsection{Proof of Proposition \ref{prop: wc1}} \label{sec: app proof of proposition}

By the integrability in time of the correlation function $\psi(x,t)$, as stated in Lemma \ref{lem: integrability}, the expression \eqref{eq: caL as perturbation} implies immediately that $\caL(z)$ can be continued continuously to $z \in \bbR$. This proves \eqref{eq: caL defined on real axis}. 
The boundedness  of $\caM$ on $\scrB_2(\scrH_\sys)$ and $\scrB_1(\scrH_\sys)$  follows from the boundedness of $\Psi$, which was explained above. 
 The complete positivity of the map $\Psi$ and the canonical form \eqref{eq: general form generator} imply that $\caM$ is a Lindblad generator, see e.g.\ \cite{alickifannesbook}.
 Consequently, $-\i \adjoint(Y)+ \la^2 \caM$ is also a Lindblad generator and the semigroup  $\La_t$ is positivity-preserving and trace-preserving.
To check \eqref{eq: caM has only propagation in hamiltonian}, we note that 
\beq
\caJ_{\ka}  \caM \caJ_{-\ka} -\caM =      -\i \left[ \caJ_{\ka}  \adjoint (\varepsilon(P))    \caJ_{-\ka} - \adjoint (\varepsilon(P)) \right]   
\eeq
and hence \eqref{eq: caM has only propagation in hamiltonian} follows immediately from Assumption \ref{ass: analytic dispersion}.

\subsection{Spectral analysis and proof of Proposition \ref{prop: properties of m}}
The  claims of Proposition \ref{prop: properties of m} require a spectral analysis which we present now. 
 We recall the decomposition 
\beq
\caM=  \int_{\tor }^\oplus \d p \,  \caM_p = \int_{\tor }^\oplus \d p \, \mathop{\oplus}\limits_{a  \in \sp(\adjoint(Y))}  \caM_{p,a} 
\eeq
and we keep in mind that Proposition \ref{prop: properties of m}  treats  $\caM_p$ as an operator on the Hilbert space $\scrG \sim L^2(\tor, \scrB_2(\scrS))$. 
\subsubsection{Explicit representation of $\caM_{p,0}$}

By exploiting the nondegeneracy condition in Assumption \ref{ass: fermi golden rule}, we can identify $\caM_{p,a=0} $, for each $p$ with an operator on $L^2(\tor \times \sp Y)$.  This was explained in Section \ref{sec: momentum representation of caM}.   
We introduce explicitly  \emph{gain}, \emph{loss}  and \emph{kinetic}  operators; $G, L$ and $K_{p}$, acting on $L^2(\tor \times \sp Y)$,  by
\baq
 G \varphi (k,e)  & := &\sum_{e' \in \sp Y}\int_{\tor} \d k'     r(k',e';k,e)   \varphi(k',e') \\[1mm]
L  \varphi (k,e)  & := &  -   \sum_{e' \in \sp Y}\int_{\tor} \d k'    r(k,e;k',e')  \varphi(k,e) \\[1mm]
K_{p}  \varphi (k,e)   & := &   \i( \varepsilon(k+\frac{p}{2} )-  \varepsilon(k-\frac{p}{2} ) ) \varphi(k,e), \qquad  \varphi \in L^2(\tor \times \sp Y)
\eaq
The kinetic operator $K_{p}$ models the free flight of the particle between collisions.  The operators $L$ and $K_p$ act by multiplication in the variables $k,e$. 
The expression for $\caM_{p,0}$, given in   Section \ref{sec: momentum representation of caM} (in particular in \eqref{def: caMp}), can  be rewritten as
\beq \label{eq: app formula for caMoo}
\caM_{p,0} =  G +L +K_{p}
\eeq
We define the similarity transformation 
\beq
A \mapsto \hat A :=     \e^{ \frac{1}{2} \be Y} A  \e^{- \frac{1}{2} \be Y}, \qquad \textrm{for} \,\,  A \in \scrB(L^2(\tor \times \sp Y))
\eeq
where we have slightly abused the notation by writing $Y$ to denote a multiplication  operator on $\sp Y$, i.e.,  $Y \varphi(k,e)= e \varphi(k,e)$.
Since $L$ and $K_{p}$ act by multiplication, we have $\hat L =L$ and $\hat K_{p}= K_{p} $. 
The usefulness of this  similarity transformation resides in the fact that $\hat G$, and hence also $\hat \caM_{0,0}$, are  self-adjoint on $L^2(\tor \times \sp Y)$.  

\subsubsection{Explicit representation of $\caM_{p,a \neq 0}$}

To write an explicit expression for $\caM_{p,a \neq 0}$, we first define the operators (acting on $\scrB(\scrS)$)
 \beq
\adjoint_a( \Upsilon) = 1_a(\adjoint(Y))  \adjoint(\Upsilon) 1_a(\adjoint(Y)), \qquad  a \in \sp (\adjoint(Y))   \eeq
which satisfy $\adjoint_0(\Upsilon)=0$  and $ \adjoint(\Upsilon)  = \oplus_{a}  \adjoint_a( \Upsilon)$ since $[\Upsilon, Y]=0$.  Due to the nondegeneracy condition in Assumption \ref{ass: fermi golden rule},  both $1_a(\adjoint(Y))$ and  $\adjoint_a( \Upsilon)$ are rank-one operators and we identify $\adjoint_a( \Upsilon)$ with a number $\Upsilon_a$, such that $\adjoint_a( \Upsilon) =  \Upsilon_a 1_a(\adjoint(Y))$.  
  In fact, as already remarked in Section \ref{sec: momentum representation of caM},  the operator $\caM_{p,a \neq 0}$  itself acts as the rank-one operator $1_a(\adjoint(Y))$  on $\scrB(\scrS)$ and hence we identify it with an operator on $L^2(\tor)$ (which is also called $\caM_{p,a \neq 0}$ here):
 \beq \label{eq: first expression caMpa}
\caM_{p,a \neq 0} \varphi(k)= - \i (\Upsilon_a+\varepsilon(k+\frac{p}{2} )-  \varepsilon(k-\frac{p}{2} ) )  \varphi(k)  - \frac{1}{2} \left(j(k,e)+j(k,e')\right)  \varphi(k),  \qquad  \varphi \in L^2(\tor) 
\eeq
where  $j(k,e)$ are the escape rates introduced in  \eqref{def: escape rates} and $e,e'$ are determined by $a=e-e'$.
To check \eqref{eq: first expression caMpa}, one starts from \eqref{def: caMp} and one uses
\begin{itemize}
\item  The fact that $r(k',e';k,e)$ vanishes for $e'=e$, as remarked following \eqref{def: singular jump rates}. 
\item The definition of the matrices $W_a$ in  \eqref{def: Wa} and the escape rates $j(\cdot,\cdot)$ in  \eqref{def: escape rates}.
\item The definition $ \varphi(k)\equiv \langle e, \xi(k) e' \rangle_{\scrS} $ in \eqref{def: varphi}.
\end{itemize}
In particular, the last term on the RHS of \eqref{eq: first expression caMpa} appears because 
\beq
 \int_{\tor}  \d k' r_{e'-e}(k,k')  \Big\langle e',  \left( W_a W_{a}^* \xi(k) +  \xi(k) W_a W_{a}^*  \right)e\Big\rangle_{\scrS} =    \left(j(k,e)+j(k,e')\right) \varphi(k)
\eeq

Hence, $\caM_{p,a \neq 0}$ acts by multiplication in the variable $k$.

\subsubsection{Analysis of $\caM_{0,0}$} \label{sec: analysis of caM00}

We already established that $\caM_{0,0}$ is a bounded  Markov generator  on $L^1(\tor \times \sp Y)$. 
This implies that 
\beq \label{eq: spectrum of markov}
 \Re \sp_{L^1} \caM_{0,0}   \leq 0. 
\eeq
The operator $\hat \caM_{0,0}$ is not longer a Markov generator, but its spectrum is  identical to $\hat \caM_{0,0}$,  since $  \e^{\pm \frac{1}{2} \be Y}$ is bounded and invertible. 
The \emph{loss} operator $L$ is a multiplication operator and its spectrum is found  to be (see \eqref{def: escape rates})
\beq
\sp(L) =  - \left\{  j(e,k) \, \big\str \,  e \in \sp Y, k \in \tor   \right\}    
\eeq
It is important to note that the escape rates $j(e,k)$ are bounded away from $0$; this is a consequence of Assumption \ref{ass: fermi golden rule} and more concretely, of the fact that for each $e$, there is a $e'$ such that \eqref{def: connectedness} holds. Hence, we have
\beq
\sp(L)  <0
\eeq
Next, we argue that $G$ is a compact operator on $L^2$.
 Indeed,  for fixed  $e,e'$, the kernel $ r(k',e';k,e)$ depends only on $\Delta k \equiv k-k'$  and, hence, its Fourier transform acts on $l^2(\lat)$ by multiplication with the function 
\beq \label{eq: multiplication operator}
\lat \ni x \mapsto  \int_{\tor} \d (\Delta k) \e^{\i \Delta k \cdot x}  r(0,e';\Delta k,e)
\eeq
From the explicit expression for $r(k',e';k,e)$, one checks that the function \eqref{eq: multiplication operator} decays at infinity if 
the dimension $d>1$ (recall that $d \geq 4$ by Assumption  \ref{ass: space}).  Hence $G$ is compact.

Given the compactness of $G$, Weyl's theorem ensures that the self-adjoint operators $\hat \caM_{0,0}$ and $\hat L$ have the same essential spectrum.

By inspection, we  check that $\hat \caM_{0,0}$ has an eigenvalue $0$, corresponding to the  eigenvector $\hat\varphi^{eq}(k,e) \equiv \e^{-\frac{\be}{2}e}$. Note that the corresponding right eigenvector of $\caM_{0,0}$ is the (unnormalized) Gibbs state $\varphi^{eq}(k,e) \equiv \e^{-\be e}$ and the corresponding left eigenvector is the constant function, since indeed  
\beq
\hat\varphi^{eq} =    \e^{\frac{\be}{2} Y} \varphi^{eq}, \qquad      \hat\varphi^{eq} =    \e^{-\frac{\be}{2} Y}  1_{\tor \times \sp Y}
\eeq
Since any eigenvalue of $\hat \caM_{0,0}$ on $L^2$ has to be an eigenvalue of  $\hat \caM_{0,0}$ on $L^1$ (Note that $L^2(\tor \times \sp Y) \subset  L^1(\tor \times \sp Y)  $), the relation \eqref{eq: spectrum of markov}
implies that there are no eigenvalues with strictly positive real part. \\

We now exploit a Perron-Frobenius  argument to argue that the eigenvalue $0$ is simple and that it is the only eigenvalue on the real axis.    
 See e.g.\ Theorem 13.3.6 in \cite{davieslinearoperators}  for a version of the Perron-Frobenius theorem that establishes this in our case, provided that the semigroup  $\e^{ t \hat\caM_{0,0}} $ is irreducible, i.e., that for any nonnegative functions $\varphi \in L^1(\tor\times\sp Y)$, the inclusion
 \beq
\supp\left( \e^{t \hat\caM_{0,0}} \varphi \right)  \subset \supp( \varphi), \qquad (\supp\, \textrm{stands for 'support'})    \eeq
implies that either $  \supp (\varphi)= \tor \times \sp Y$ or $\varphi=0$. 
This irreducibility criterion is easily checked starting from Assumption \ref{ass: fermi golden rule}, in particular its rephrasing in  terms of a connected graph.
 Theorem 13.3.6 yields that  the eigenvalue $1$ of  $\e^{ t \hat\caM_{0,0}} $ is simple, which implies that the eigenvalue $0$ of $\hat\caM_{0,0}$ is simple. To exclude purely imaginary eigenvalues  $\i b $ of $\hat\caM_{0,0}$, we apply this theorem for $t$ such that $\e^{\i b t}=1$.

\subsubsection{Analysis of $\caM_{p,0}$ and $\caM_{p,a}$ }   \label{sec: analysis of caMp0}
We investigate the spectrum of $\hat \caM_{p,0}$ as follows. 
By the same reasoning as in Section \ref{sec: analysis of caM00}, any spectrum with real part greater than (the negative number) $\sup\sp L$ consists of eigenvalues of finite multiplicity. 
Assume that  $\hat \caM_{p} $ has an eigenvalue $m_p$ with (right)  eigenvector) $\hat \varphi_p$. 
Then
\baq
 \Re  m_p \langle \hat \varphi_p, \hat \varphi_p \rangle  & =&  \Re   \langle \hat \varphi_p, \hat \caM_{p,0}  \hat \varphi_p \rangle  \\
& =&   \Re   \langle \hat \varphi_p,  K_{p}  \hat \varphi_p \rangle  +  \Re   \langle \hat \varphi_p, \hat \caM_{0,0}  \hat \varphi_p \rangle  \label{eq: app no other eigenvectors}
\eaq
The first term in \eqref{eq: app no other eigenvectors} vanishes because the multiplication operator $K_{p}$ is purely imaginary.  
The second term can only become positive if $\hat \varphi=\hat \varphi^{eq}$, with $\hat \varphi^{eq}$ the  eigenvector of $\hat \caM_{0,0}$ corresponding to the eigenvalue $0$.  This means that either the eigenvalue $m_p$ has strictly negative real part, or  the vector $\hat \varphi^{eq}$ is an eigenvector of  $\hat \caM_{0,0} $ with eigenvalue $0$.  In the latter case,  $\hat \varphi^{eq}$ must also be an eigenvector of   $K_{p} $ with eigenvalue $0$, which can only hold if $\varepsilon(k+\frac{p}{2} )-  \varepsilon(k-\frac{p}{2} )=0$ for all $k$.   This is however excluded by the condition \eqref{eq: no smaller periodicity} in Assumption \ref{ass: analytic dispersion}.

We conclude that for all $p \in \tor \setminus \{ 0 \}$, we have $
\Re \sp \caM_p  <0 $.
By compactness of $\tor$ and the lower semicontinuity of the spectrum, we deduce hence that 
\beq 
\sup_{\tor  \setminus I_0}   \Re \sp \caM_{p,0}    = c(I_0) <0, \qquad \textrm{for any neighborhood $I_0$ of $0$} 
\eeq


For $a \neq0$, the operator $\caM_{p,a \neq 0}$ is a multiplication operator in $k$ and
\beq \label{eq: spectrum caMpa bounded}
 \Re \sp  \caM_{p,a} \leq  - \frac{1}{2} \inf_{k,e} j(k,e) <0, \qquad  \textrm{independently of}\, p
\eeq
as follows by \eqref{eq: first expression caMpa} and the fact that $j(k,e) $ is bounded away from $0$.

\subsubsection{Proof of Proposition \ref{prop: properties of m}} \label{sec: proof of properties of m}

We summarize the results of Sections  \ref{sec: analysis of caM00} and \ref{sec: analysis of caMp0}. 
For $a \neq 0$, the real part of the spectrum of the operators $\caM_{p,a}$ is  strictly negative, uniformly in $p$, see \eqref{eq: spectrum caMpa bounded}.
The real part of the spectrum of $\caM_{p,0}$ is strictly negative, uniformly in $p$ except for a neighborhood of $0 \in \tor$.  

The operator $\caM_{0,0}$ has a simple eigenvalue at $0$ with corresponding eigenvector $\varphi^{eq}$, as defined in Section \ref{sec: analysis of caM00}.  The rest of the spectrum of  $\caM_{0,0}$ is separated from the eigenvalue by a gap. 

Since $\caM_{p} =\oplus_{a} \caM_{p,a}$, and using the uniform bound \eqref{eq: spectrum caMpa bounded}, we obtain immediately that the operator 
$\caM_{0}$ has a simple eigenvalue at $0$ with corresponding eigenvector $\xi^{eq}$
\beq
 \xi^{eq}  :=   \varphi^{eq}   \oplus  \underbrace{0 \oplus \ldots \oplus 0}_{a \neq 0},
\eeq
separated from the  rest of the spectrum of $\caM_{0} $ by a gap. 
By the analyticity in $\ka$, see \eqref{eq: caM has only propagation in hamiltonian}, and the correspondance between $\ka$ and $(p,\nu)$, as stated in \eqref{eq: relation kappa and fibers}, we can apply analytic perturbation theory in $p$ to the family of operators $\caM_p$.  We conclude that for $p$ in a neighborhood of $0$, the operator $\caM_p$ has a simple eigenvalue, which we call $f_{rw}(p)$, that is separated by a gap from the rest of the spectrum.  We also obtain that the corresponding eigenvector is analytic in $p$ and $\nu$.

In this way we have derived all claims of Proposition \ref{prop: properties of m}, except for the symmetry $\nabla_p f_{rw}(p)=0$ and the strict postive-definiteness of the matrix $(\nabla_p)^2 f_{rw}(p)$. These two claims will be proven in Section \ref{sec: strict positivity}.  
We note that the function $f_{rw}(p)$, which we defined above as the simple and isolated eigenvalue of $\caM_p$ with maximal real part,  is also a simple and isolated eigenvalue of $\caM_{p,0}$ with maximal real part.

\subsubsection{Strict positivity of the diffusion constant} \label{sec: strict positivity}

By the remark at the end of Section \ref{sec: proof of properties of m} and the fact that $\sp \hat\caM_{p,0}=\sp\caM_{p,0}$, 
we view $f_{rw}(p)$ as the eigenvalue of $\hat \caM_{p,0}$ that reduces to $0$ for $p=0$. 

We recall that $ \hat \caM_{p,0} = \hat \caM_{0,0}+K_{p}$ and we define the operator-valued vector $V:= \nabla_p K_{p}\big\str_{p=0}$ (note that $V$ is in fact a vector of operators). The first order shift of the eigenvalue is given by 
\beq \label{eq: vanishing first order}
\nabla_p f_{rw}(p)  :=  \frac{1}{\langle \hat \varphi^{eq},  \hat \varphi^{eq} \rangle  } \langle \hat \varphi^{eq},  V  \hat \varphi^{eq}  \rangle =0
\eeq
To check that \eqref{eq: vanishing first order} indeed vanishes, we  use that $\hat \varphi^{eq}$ is symmetric under the transformation $k \mapsto -k$ (in fact, it is independent of $k$) while $V$ is anti-symmetric under $k \mapsto -k$ (this follows from the symmetry $\ve(k)=\ve(-k)$ in Assumption \ref{ass: analytic dispersion}). 

The second order shift is then given by
\beq \label{eq: second order}
D_{rw}:= \left(\nabla_p \right)^{2} f_{rw}(p)  =   - \frac{1}{\langle \hat \varphi^{eq},  \hat \varphi^{eq} \rangle} \langle \hat \varphi^{eq}, V    \hat\caM_{0,0}^{-1} V  \hat \varphi^{eq}  \rangle  + \frac{1}{\langle \hat \varphi^{eq},  \hat \varphi^{eq} \rangle}     \langle \hat \varphi^{eq},  ( \nabla_p)^2 K_{p}    \hat \varphi^{eq}  \rangle 
\eeq
where the first term on the RHS of \eqref{eq: second order} is well-defined since  $  V \hat \varphi^{eq}  $ is orthogonal to the $0$-spectral subspace of $\caM_{0,0}$, by \eqref{eq: vanishing first order}.
The second term vanishes because $( \nabla_p)^2 K_{p}=0$, as can again be checked explicitly. 

Let $\upsilon \in \bbR^d$ and $V_{\upsilon}:=\upsilon \cdot V$ (recall that $V$ is a vector).
Then, by \eqref{eq: second order},
\beq
\upsilon \cdot D_{rw}  \upsilon = -  \frac{1}{\langle \hat \varphi^{eq},  \hat \varphi^{eq} \rangle}  \langle \hat \varphi^{eq}, V_\upsilon    \hat\caM_{0,0}^{-1} V_\upsilon  \hat \varphi^{eq}  \rangle
\eeq
Upon using the spectral theorem and the gap for the self-adjoint operator $\hat\caM_{0,0}$, we see that the RHS of the last expression is positive and it can only vanish if
\beq
0=\norm V_{\upsilon}  \hat \varphi^{eq}  \norm^2 = \left[\sum_{e \in \sp Y}\e^{-\be e}\right] \int \d k  \str\upsilon \cdot \nabla\ve(k)\str^2
\eeq
which is however excluded by Assumption \ref{ass: analytic dispersion}. 
The strict positive-definiteness of the diffusion constant $D_{rw}$ is hence proven.

  \section{Appendix: Combinatorics}  
\label{app: combinatorics}
 
 In this appendix, we show how to integrate over irreducible equivalence classes of diagrams. In other words, we assume that the $\underline{x}, \underline{l}$-coordinates have already been summed over (or a supremum over them has been taken) and we carry out the remaining integration over the time-coordinates $\underline{t}$ and the diagram size $\str \dsi \str$.  
 We first define a function of diagrams, $K(\dsi)$, that depends only on the equivalence class $[\dsi]$.
  Let $k$ be a positive function on $\bbR^+$ and put
\beq
K(\dsi) := \mathop{\prod}_{i=1}^{\str \dsi \str}  k(v_i-u_i)
\eeq
where $(u_i,v_i)$ are the pairs of times in the diagram $\dsi$.  In the applications, the function $k$ will be (a multiple of) $\sup_x\str \psi(x,t)\str $, sometimes restricted to $t<\tau$ or $t > \tau$.

  \begin{lemma}\label{lem: combinatorics of irreducible diagrams}
Let $a \geq 0$ and assume that $\norm t\e^{ a t } k\norm_1= \int_{\bbR^+} \d t \, t\e^{ a t } k(t) <1 $,  then 
\beq \label{eq: bound minimally irreducible classes all}
 \mathop{\int}\limits_{\bbR^+} \d t \,  \e^{ a t }   \mathop{\int}\limits_{\Pi_T \Si_{[0,t]}(\mathrm{mir}) }  \d [ \dsi]  K(\dsi)  \leq   \norm \e^{ a t } k\norm_1 \frac{1}{1-\norm t\e^{ a t } k\norm_1}
\eeq
If in addition, $\norm t\e^{ {\tilde a} t } k\norm_1 <1$ with $ {\tilde a}:= a +\norm k\norm_1 $, then 
\baq  
 \mathop{\int}\limits_{\bbR^+}  \d t \,  \e^{ a t }   \mathop{\int}\limits_{\Pi_T \Si_{[0,t]}(\mathrm{ir}) }  \d [ \dsi]   \, K(\dsi)  &\leq&   2 \norm \e^{ {\tilde a} t } k\norm_1 \frac{1 }{1-\norm t\e^{ {\tilde a} t } k\norm_1}   \label{eq: bound irreducible classes all}  \\
 \mathop{\int}\limits_{\bbR^+}  \d t \,  \e^{ a t }   \mathop{\int}\limits_{\Pi_T \Si_{[0,t]} (\mathrm{ir}) }  \d [ \dsi]  1_{\str \dsi \str \geq 2} \, K(\dsi) & \leq &  2 \norm \e^{ {\tilde a} t } k\norm_1 \frac{\norm t\e^{ {\tilde a} t } k\norm_1 }{1-\norm t\e^{ {\tilde a} t } k\norm_1}  \label{eq: bound irreducible} \eaq
\end{lemma}

\begin{proof}

First, we note  that for each irreducible diagram $\dsi \in \Si_{[0,t]}(\mathrm{ir})$, we can find a subdiagram $\dsi' \subset \dsi$ such that $\dsi'$ is minimally irreducible in $[0,t]$, i.e., $\dsi' \in \Si_{[0,t]}(\mathrm{mir})$.  Note that the choice of subdiagram $\dsi'$ is not necessarily unique. Conversely, given a minimally irreducible diagram $\dsi' \in \Si_{[0,t]}(\mathrm{mir})$, we can add any diagram $\si'' \in \Si_{[0,t]}$ to $\si'$, thereby creating a new irreducible diagram $\si:= \si' \cup \si'' \in \Si_{[0,t]}(\mathrm{ir})$. 
By these considerations, we easily deduce
\baq 
 \mathop{\int}\limits_{\Pi_T\Si_{[0,t]}(\mathrm{ir}) }  \d [\dsi] \, 1_{\str \dsi \str \geq 2}    K(\si)     
     & \leq &    \Bigg(       \mathop{\int}\limits_{\Pi_T\Si_{[0,t]}(\mathrm{mir}) }  \d [\dsi']   1_{\str \dsi' \str \geq 2}     K(\dsi')      \Bigg)               \Bigg(1+ \mathop{\int}\limits_{\Pi_T\Si_{[0,t]}} \d [\dsi''] K(\dsi'') \Bigg)    \label{eq: reduction to min irr 1} \\
     & + &    \Bigg(       \mathop{\int}\limits_{\Pi_T\Si_{[0,t]}(\mathrm{mir}) }  \d [\dsi']  1_{\str \dsi' \str =1}    K(\dsi')      \Bigg)               \Bigg( \mathop{\int}\limits_{\Pi_T\Si_{[0,t]}} \d [\dsi''] K(\dsi'') \Bigg)    \label{eq: reduction to min irr 2}
\eaq
The $1+\cdot$ in \eqref{eq: reduction to min irr 1} covers the case in which the diagram $\dsi$ was itself minimally irreducible, and hence no diagrams $\dsi''$ are added to $\dsi'$.
In \eqref{eq: reduction to min irr 2}, one always has to add at least one pair to $\dsi'$, since $\str \dsi\str \geq 2$ but $\str \dsi' \str=1$. 
In fact, the  equivalence classes in the inequality could be dropped, i.e., one can omit  the projections $\Pi_T$  and replace $\d[\dsi], \d [\dsi'], \d [\dsi'']$  by $\d\dsi, \d \dsi', \d \dsi''$, respectively.

We recall that if  a diagram $\dsi$ with $\str \dsi \str=1$ is irreducible (or minimally irreducible) in the interval $I$, then its time-coordinates are fixed to be the boundaries of $I$; i.e., there is only one equivalence class of such diagrams:
\beq  \label{eq: 1 minimally irreducible diagram}
 \mathop{\int}\limits_{\Pi_T\Si_{[0,t]}(\mathrm{ir}) }  \d [\dsi]   1_{\str \dsi \str = 1}   K(\si) =  \mathop{\int}\limits_{\Pi_T\Si_{[0,t]}(\mathrm{mir}) }  \d [\dsi]  1_{\str \dsi \str = 1}    K(\si) =  k(t)  
 \eeq

The unconstrained integral over all (equivalence classes of) diagrams, that appears in \eqref{eq: reduction to min irr 1} and \eqref{eq: reduction to min irr 2}, can be performed as follows
\baq
 \mathop{\int}_{\Pi_T \Si_{[0,t]}}      \d [\dsi] K(\si) & = &   \sum_{n \geq 1}   \mathop{\int}\limits_{0< u_1< \ldots < u_n<t}  \d u_1 \ldots \d u_n   \mathop{\int}\limits_{v_i >u_i}     \d v_1 \ldots \d v_n \left( \mathop{\prod}\limits_{ i =1}^n  k(v_i-u_i) \right)   \nonumber \\
&\leq  &    \sum_{n \geq 1}   \mathop{\int}\limits_{0< u_1< \ldots < u_n<t}  \d u_1 \ldots \d u_n    ( \norm k \norm_1   )^n  =   \sum_{n \geq 1} \frac{t^n}{n!}   ( \norm k \norm_1   )^n  =  \e^{t \norm k\norm_1 }-1
\label{eq: unconstrained integration diagrams}
\eaq

Next, we perform the integral over (equivalence classes of) minimally irreducible diagrams. For $\si \in \Si_{[0,t]}(\mathrm{mir})  $ with $\str \dsi \str=n >1$, the relative order of the times $u_i,v_i$ is  fixed as follows:
 \beq
 0= u_1 \leq  u_2 \leq v_1 \leq u_3 \leq v_2  \leq  u_4 \leq \ldots      \leq  v_{n-2}  \leq  u_n \leq v_{n-1}\leq  v_n=t \eeq
 We have hence
\baq
&&     \int_{\bbR^+} \d t \,  \e^{at }    \mathop{\int}\limits_{\Pi_T\Si_{[0,t]}(\mathrm{mir}) }  \d \dsi   K(\si)    1_{\str \dsi\str=n}    =   \int_0^\infty   \d v_1k(v_1-u_1)   \e^{a (v_1-u_1)}  \,       \int_0^{v_{1} } \d u_{2}  \int_{v_{1}}^{\infty } \d v_{2}    \ldots     \nonumber  \\[1mm]     
   &  &     \qquad \ldots \qquad    \int_{v_{n-5}}^{v_{n-4} } \d u_{n-2}     \int_{v_{n-3}}^{\infty } \d v_{n-2}       \ldots     \int_{v_{n-3}}^{v_{n-2} }  \d u_{n-1}  \int_{v_{n-2}}^{\infty } \d v_{n-1}    \e^{a(v_{n-1}-v_{n-2})}   k(v_{n-1}-u_{n-1})       \nonumber    \\[1mm]    
   && \quad \qquad \qquad       \int_{v_{n-2}}^{v_{n-1} } \d u_n   \int_{v_{n-1}}^{\infty } \d v_n     \e^{a(v_n-v_{n-1})}   k(v_n-u_n)    .  \label{eq: calculation minimally irreducible}  
\eaq
First we extend the domain of integration of $u_{n}$  from $[v_{n-2},v_{n-1}]$ to $(-\infty,v_{n-1}]$ and we estimate the integrals over the variables $u_n$ and $v_n$ by 
\beq
\int_{-\infty}^{v_{n-1}} \d u_{n}  \int_{v_{n-1}}^{\infty}  \d v_n  \e^{a (v_n- v_{n-1})}  k(v_n-u_n)  \leq   \norm t e^{at} k \norm_1
\eeq
Next, we perform the integration over $u_{n-1}, v_{n-1}$ in the same way, we continue the procedure until only the variable $v_1$ is left (note that $u_1=0$ is fixed).  The integral over $v_1$ gives $\norm \e^{at} k \norm_1$. 
This yields the bound
\beq  \label{eq: result minimally irreducible} 
 \textrm{LHS of \eqref{eq: calculation minimally irreducible} }  \leq        \norm \e^{at} k \norm_1 \times  \norm t e^{at} k \norm_1^{n-1}
\eeq
We are ready to evaluate the Laplace transform of \eqref{eq: reduction to min irr 1}-\eqref{eq: reduction to min irr 2}. Using \eqref{eq: unconstrained integration diagrams}, we  bound 
\beq
 \Bigg(1+ \mathop{\int}\limits_{\Pi_T\Si_{I}} \d [\dsi] K(\si) \Bigg)  \leq \e^{t \norm k\norm_1}, \qquad     \Bigg( \mathop{\int}\limits_{\Pi_T\Si_{I}} \d [\dsi] K(\si) \Bigg) \leq \e^{t \norm k\norm_1}-1 \leq t \norm k\norm_1 \e^{t \norm k\norm_1}
\eeq
Combining this with  \eqref{eq: 1 minimally irreducible diagram}  and  \eqref{eq: result minimally irreducible} , and summing over $n \geq 2$, we obtain
\beq
   \int \d t \,  \e^{ a t }   \mathop{\int}\limits_{\Pi_T \Si_{[0,t]}(\mathrm{ir}) }  \d [ \dsi]  1_{\str \dsi \str \geq 2} \, K(\dsi)     \leq \norm \e^{ {\tilde a} t } k\norm_1 \frac{\norm t\e^{ {\tilde a} t } k\norm_1 }{1-\norm t\e^{ {\tilde a} t } k\norm_1}  +    \norm  k\norm_1 \norm t\e^{ {\tilde a} t } k\norm_1
\eeq
where the two terms on the RHS correspond to  \eqref{eq: reduction to min irr 1} and \eqref{eq: reduction to min irr 2}, respectively.  This ends the proof of \eqref{eq: bound irreducible}.
The bound in \eqref{eq: bound irreducible classes all}  follows by adding $\norm \e^{at}k\norm_1$, which is the contribution of $\str\dsi\str=1$ (see \eqref{eq: 1 minimally irreducible diagram}), to \eqref{eq: bound irreducible}.   The bound  \eqref{eq: bound minimally irreducible classes all} is proven by summing  \eqref{eq: result minimally irreducible}  over $n \geq 1$.

\end{proof}

\bibliographystyle{plain}
\bibliography{mylibrary08}

\end{document}